\documentclass[11pt]{article}

\usepackage{jheppub}
\usepackage[utf8]{inputenc}
\usepackage{amsmath,amssymb,amsfonts, amsthm, mathtools,bbm}
\usepackage{pmat}
\usepackage[mathscr]{eucal}
\usepackage{bigints}
\usepackage{tabu}
\usepackage{cancel}
\usepackage{makecell}
\usepackage[font={small}]{caption}
\usepackage[font={small}]{subcaption}
\usepackage{bm}
\usepackage{esint}
\usepackage{extarrows}
\usepackage[new]{old-arrows}
\usepackage{tensor}
\usepackage{xspace}
\usepackage[dvipsnames]{xcolor}
\usepackage{hyperref}
\hypersetup{urlcolor=Cerulean, citecolor=red, linkcolor=MidnightBlue}
\usepackage{url}
\usepackage{tikz}
\usetikzlibrary{decorations.markings,intersections, decorations.pathreplacing,angles,quotes,, shapes.geometric, arrows,matrix,positioning, patterns}
\usepackage{tikz-cd}
\usepackage{graphicx}
\usepackage{float}

\newcommand{\mscr}[1]{\mathscr{#1}}
\newcommand{\mfk}[1]{\mathfrak{#1}}
\newcommand{\msf}[1]{\mathsf{#1}}
\newcommand{\mbb}[1]{\mathbb{#1}}
\newcommand{\mbbmss}[1]{\mathbbmss{#1}}
\newcommand{\mbf}[1]{\mathbf{#1}}
\newcommand{\mbs}[1]{\boldsymbol{#1}}
\newcommand{\mtt}[1]{\mathtt{#1}}
\newcommand{\mcal}[1]{\mathcal{#1}}
\newcommand*{\oli}[1]{{\overbracket[.7pt][-1pt]{#1}}}

\newcommand{\newl}{\medskip\noindent}

\newtheorem{thr}{\bf\small\color{NavyBlue} Theorem}[section]
\newtheorem*{thr*}{\bf\small\color{NavyBlue} Theorem}
\newtheorem{lem}{\bf\small\color{RedOrange}{Lemma}}[section]
\newtheorem{example}{\bf\color{Turquoise}Example}[section]
\newtheorem{definition}{\bf\small\color{ForestGreen}{Definition}}[section]
\newtheorem{remark}{\bf\small \color{Blue}Remark}[section]
\newtheorem{proposition}{\bf\small\color{ProcessBlue} Proposition}[section]
\newtheorem*{proposition*}{\bf\small\color{ProcessBlue} Proposition}
\newtheorem{corollary}{\bf\small\color{Maroon} Corollary}[section]

\newcommand\numberthis{\addtocounter{equation}{1}\tag{\theequation}}

\font\seveneurm=eurm7
\font\eighteurm=eurm8
\font\teneurm=eurm10

\newcommand{\g}{\text{\teneurm g}}
\newcommand{\sg}{\text{\eighteurm g}}
\newcommand{\vsg}{\text{\seveneurm g}}
\newcommand{\n}{\text{\teneurm n}}
\newcommand{\sn}{\text{\eighteurm n}}
\newcommand{\vsn}{\text{\seveneurm n}}

\newcommand{\ns}{\text{\teneurm n}_{\text{\normalfont\fontsize{6}{6}\selectfont NS}}}
\newcommand{\sns}{\text{\eighteurm n}_{\normalfont\text{\fontsize{6}{6}\selectfont NS}}}
\newcommand{\ra}{\text{\teneurm n}_{\normalfont\text{\fontsize{6}{6}\selectfont R}}}
\newcommand{\sra}{\text{\eighteurm n}_{\normalfont\text{\fontsize{6}{6}\selectfont R}}}

\newcommand{\nsns}{\text{\teneurm n}_{\text{\normalfont\fontsize{6}{6}\selectfont NSNS}}}
\newcommand{\snsns}{\text{\eighteurm n}_{\normalfont\text{\fontsize{6}{6}\selectfont NSNS}}}
\newcommand{\nsra}{\text{\teneurm n}_{\normalfont\text{\fontsize{6}{6}\selectfont NSR}}}
\newcommand{\snsra}{\text{\eighteurm n}_{\normalfont\text{\fontsize{6}{6}\selectfont NSR}}}
\newcommand{\rans}{\text{\teneurm n}_{\normalfont\text{\fontsize{6}{6}\selectfont RNS}}}
\newcommand{\srans}{\text{\eighteurm n}_{\normalfont\text{\fontsize{6}{6}\selectfont RNS}}}
\newcommand{\rara}{\text{\teneurm n}_{\normalfont\text{\fontsize{6}{6}\selectfont RR}}}
\newcommand{\srara}{\text{\eighteurm n}_{\normalfont\text{\fontsize{6}{6}\selectfont RR}}}

\newcommand{\dmmoduli}{{\oli{\mscr{M}}}_{\sg,\sn}}

\DeclareMathOperator{\Hom}{Hom}
\DeclareMathOperator{\RHom}{RHom}
\DeclareMathOperator{\Isom}{Isom}
\DeclareMathOperator{\Ext}{Ext}

\DeclareMathOperator{\Span}{span}
\DeclareMathOperator{\Ker}{Ker}
\DeclareMathOperator{\Gr}{Gr}

\DeclareMathOperator{\SL}{SL}
\DeclareMathOperator{\Sp}{Sp}
\DeclareMathOperator{\Aut}{Aut}

\DeclareMathOperator{\Spec}{Spec}

\DeclareMathOperator{\Sch}{Sch}
\DeclareMathOperator{\QCoh}{QCoh}
\DeclareMathOperator{\Coh}{Coh}
\DeclareMathOperator{\MC}{MC}
\DeclareMathOperator{\SSp}{SupSp}
\DeclareMathOperator{\CSp}{ComSp}
\DeclareMathOperator{\im}{im}
\DeclareMathOperator{\Tor}{Tor}

\DeclareMathOperator{\Hilb}{\underline{Hilb}}
\DeclareMathOperator{\Quot}{Quot}
\DeclareMathOperator{\Gras}{Grass}
\DeclareMathOperator{\Map}{Map}
\DeclareMathOperator{\Pic}{Pic}
\DeclareMathOperator{\Spf}{Spf}
\DeclareMathOperator{\Proj}{Proj}

\DeclareMathOperator{\Res}{Res}

\DeclareMathOperator{\Supp}{Supp}

\DeclareMathOperator{\coker}{coker}

\DeclareMathOperator{\dol}{[\![}
\DeclareMathOperator{\dor}{]\!]}

\renewcommand*{\H}{{{\rm{H}}}}

\newcommand*{\NS}{\text{NS}}
\newcommand*{\RA}{\text{R}}
\newcommand{\BV}[1]{{\mbs{\msf{BV}}}(#1)}

\DeclareMathOperator{\bSgnm}{\overline{\mathfrak{S}}^{\mathbf m}_{\sg,\sn}}

\DeclareMathOperator{\scurve}{\mscr{M}_{\sg,\sn}}
\DeclareMathOperator{\cscurve}{\overline{\mscr{M}}_{\sg,\sn}}
\DeclareMathOperator{\sspin}{\mscr{S}^{\mbf{m}}_{\sg,\sn}}
\DeclareMathOperator{\csspinm}{\overline{\mscr{S}}^{\mbf{m}}_{\sg,\sn}}
\DeclareMathOperator{\csspinN}{\overline{\mscr{S}}^{N}_{\sg,\sn}}
\DeclareMathOperator{\spspin}{\mscr{PS}^N_{\sg,\sn}}

\DeclareMathOperator{\Sym}{Sym}
\DeclareMathOperator{\dep}{depth}
\DeclareMathOperator{\bos}{bos}
\DeclareMathOperator{\ev}{ev}

\DeclareMathOperator{\sDef}{sDef}
\DeclareMathOperator{\Art}{Art}
\DeclareMathOperator{\sArt}{sArt}

\DeclareMathOperator{\sLoc}{sLoc}
\DeclareMathOperator{\Sets}{Sets}

\allowdisplaybreaks

\title{\centering On the Existence of Heterotic-String and \\ Type-II-Superstring Field Theory Vertices}
\author[a,b]{Seyed Faroogh Moosavian}
\author[a,b]{Yehao Zhou}
\affiliation[a]{Perimeter Institute for Theoretical Physics, Waterloo, Ontario, Canada N2L 2Y5}
\affiliation[b]{University of Waterloo, Waterloo, ON, Canada}
\emailAdd{sfmoosavian@pitp.ca}
\emailAdd{yzhou3@pitp.ca}
\abstract{We consider two problems associated to the space of Riemann surfaces that are closely related to string theory. We first consider the problem of existence of heterotic-string and type-II-superstring field theory vertices in the product of spaces of bordered surfaces parameterizing left- and right-moving sectors of these theories. It turns out that this problem can be solved by proving the existence of a solution to the BV quantum master equation in moduli spaces of bordered spin-Riemann surfaces. We first prove that for arbitrary genus $\g$, $\ns$ Neveu-Schwarz boundary components, and $\ra$  Ramond boundary components such solutions exist. We also prove that these solutions are unique up to homotopies in the category of BV algebras. Furthermore, we prove that there exists a map in this category under which these solutions are mapped to fundamental classes of Deligne-Mumford stacks of associated punctured spin curves. These results generalize the work of Costello on the existence of a solution to the BV quantum master equations in moduli spaces of bordered ordinary Riemann surfaces which, through the work of Sen and Zwiebach, are related to the existence of bosonic-string vertices, and their relation to fundamental classes of Deligne-Mumford stacks of associated punctured ordinary curves. Using the existence of solutions to the BV quantum master equation in moduli spaces of spin curves, we prove that heterotic-string and type-II-superstring field theory vertices, for arbitrary genus $\g$ and arbitrary number of any type of boundary components, exist. Furthermore, we prove the existence of a solution to the BV quantum master equation in spaces of bordered $\mscr{N}=1$ super-Riemann surfaces for arbitrary genus $\g$, $\ns$ Neveu-Schwarz boundary components, and $\ra$  Ramond boundary components. Secondly, we turn to the moduli problem of stable $\mscr{N}=1$ SUSY curves. Using Vaintrob's deformation theory of complex superspaces and Deligne's model for the compactification of the space of $\mscr{N}=1$ SUSY curves, we prove the representability of this moduli problem by a proper smooth Deligne-Mumford superstack. This is a generalization of the work of Deligne and Mumford on the representability of the moduli problem of stable ordinary curves by a Deligne-Mumford stack. Along the way, we also elaborate some possibly-familiar aspects of the formulation of superstring perturbation theory using spin curves. In particular, we explain a subtlety in the gluing of Ramond punctures on spin curves.}

\begin{document}
	\maketitle

\section{Introduction}\label{sec:introduction}

The best-understood formulation of superstring theory is the superstring perturbation theory. Much of the mathematical difficulty of the superstring perturbation theory in the current formulation comes from the intricate mathematics of spin-Riemann or super-Riemann surfaces and stacks\footnote{
To have a uniform and mathematically-precise language throughout the paper, we use the term stack or superstack for spaces of (possibly {\it punctured}) ordinary, spin- or $\mscr{N}=1$ SUSY curves. Stack and superstack in this context are synonyms for orbifold and superorbifold respectively, which are what these spaces are due to the existence of singularities \cite{HarrisMumford198202, Harer1988, Ludwig200707, LeBrunRothstein198803,Hodgkin199509}. These terms can be replaced with the possibly-more-familiar terms moduli or supermoduli space. A concise and accessible discussion of the stacky nature of these spaces can be found in section 3.2.1 of \cite{DonagiWitten2013a}. We also refer to sections 2 and 3 of \cite{CodogniViviani201706}, and also \cite{PerezRuiperezSalas199702}, for relevant definitions of superstacks and coarse superspaces. A thorough treatment of stacks, especially stacks of curves, Deligne-Mumford stacks, and coarse moduli spaces can be found in \cite{Olsson2016}. As far as we aware, it is not known that what the spaces of {\it bordered} ordinary, spin- or super-Riemann surfaces are. As we prove in this paper, the moduli of {\it punctured} $\mscr{N}=1$ SUSY curves is a Deligne-Mumford superstack.} or superstacks thereof. The appearance of these spaces is resulted from two different approaches to superstring perturbation theory in the RNS formalism, 
\begin{enumerate}
	\item The natural setting for the formulation of the RNS superstring theory is the language of supergeometry. The worldsheet theory is written on an $\mscr{N}=1$ super-Riemann surface and one needs to deal with superstacks of these surfaces to compute the scattering amplitudes \cite{ Witten2012a,Witten2012b,Witten2012c}. The concept of super-Riemann surface, as an appropriate generalization of ordinary Riemann surfaces for doing superstring perturbation theory, was introduced by Baranov and Schwarz \cite{BaranovSchwarz198510} and independently by Friedan \cite{Friedan1986a}. 
	
	\item In the second formulation, the so-called {\it picture-changing formalism}, superstring amplitudes are defined as integrals over the moduli stack of spin-Riemann surfaces where the string measure contains the so-called picture-changing operators (PCOs) \cite{FriedanMartinecShenker1986}. Since the worldsheet theory contains fermions, the Riemann surfaces are equipped with a choice of spin structure, i.e. a line bundle $\mscr{E}$ over the surface such that $\mscr{E}^{\otimes 2}\equiv \mscr{E}\otimes\mscr{E}\simeq \omega$, where $\omega$ is the canonical line bundle over the surface\footnote{The canonical line bundle $\omega$ is isomorphic to the holomorphic cotangent bundle over the surface.}. In physical terms, a choice of spin structure corresponds to the phase shift of the fermions after parallel transport along a closed $1$-cycle of the surface. There are two possible boundary conditions by going around a closed $1$-cycle: periodic or antiperiodic boundary conditions\footnote{The boundary conditions we are referring to are in the radial quantization, i.e. after mapping the cylinder, with coordinates $(\sigma,\tau)$, to the complex $z$-plane via $z=\exp(\tau+i\sigma)$.}. The fermions with periodic boundary condition on the plane (or anti-periodic boundary condition on the cylinder) form the so-called Neveu-Schwarz (NS) sector of the Hilbert space, and the fermions with antiperiodic boundary condition on the plane (or periodic boundary condition on the cylinder) form the so-called Ramond (R) sector of the Hilbert space. The antiperiodic boundary condition on R states is one of the main sources of difficulties in superstring theory.	
	
\end{enumerate}
As a result there are two types of punctures in superstring theory, depending on whether the corresponding vertex operator belongs to the NS or R sector. We will call the corresponding punctures NS or R punctures. It has been argued in \cite{VerlindeVerlinde1987a} and has been further elaborated in \cite{Sen2014b, SenWitten2015} that the two formulations are equivalent. In this paper, we use both formalism and it will be hopefully clear which framework we are working within.

\subsection{Statement of the Problems}

\newl One of the surprising feature of the superstring theory is that the modular invariance removes the UV divergences of the theory. However, the theory suffers from IR divergences \cite{Witten2012c, Sen2015d}. The IR divergences come from the points at infinity\footnote{The compactification of the moduli stack of punctured Riemann surfaces without boundary corresponds to the addition of complex co-dimension one divisor to the uncompatified moduli stack. The resulting space, the Deligne-Mumford space \cite{DeligneMumford1969a}, is a compact space without boundary. The divisor points thus should be thought of as points at infinity. There is a subtle difference between this case and the case of open-closed superstring theories, where the superstack has boundaries, i.e. the compactification locus consists of real codimension one divisors. The presence of these boundaries is related to the presence of various anomalies. For more on this see section 9.2 of \cite{Witten2012c}.} of the superstack or equivalently the moduli stack of spin curves. In the picture-changing formalism, on one hand, one has to deal with moduli stack of  punctured spin curves and on the other hand, as was mentioned, the IR divergences of the theory come from the points at infinity of the compactified moduli stack. One thus needs to deal with the compactification of the moduli stack of punctured spin curves. 
However, unlike the case of bosonic-string theory, where the amplitudes are obtained by an integration over the moduli stack of ordinary Riemann surface, there is no natural moduli stack to be integrated over in the superstring theories, in either the supergeometry or the picture-changing formulation. The reason is as follows; 1) for heterotic-string theory the fields of right-moving sector, having $\mscr{N}=1$ supersymmetry, are parameterized by $(z|\theta)$ while the fields of left-moving sector, having no supersymmetry, are parametrized by $\widetilde{z}$. Therefore, there is no notion of complex conjugate variable for $\theta$ and therefore we can not set $\widetilde{z}=\bm\bar{z}$ since the transition functions may depend on $\theta$ on different patches; 2) For type-II theories, the fields of right-moving sector, having $\mscr{N}=1$ supersymmetry, are parameterized by $(z|\theta)$ while the fields of left-moving sector, again having $\mscr{N}=1$, are parametrized by $(\widetilde{z}|\widetilde{\theta})$. Again, there is no natural notion of complex conjugation between  $(z|\theta)$ and $(\widetilde{z}|\widetilde{\theta})$ since $\theta$ and $\widetilde{\theta}$ can have different spin structures. Therefore, in the supergeometry formulation, one needs to integrate over an appropriate integration cycle with appropriate behaviors at infinity inside an appropriate product of left and right stack and/or superstacks, where the appropriate moduli or superstacks are determined by the worldsheet of the theory \cite{Witten2012a,Witten2012b}. In the picture-changing formalism, the construction of a proper integration cycle involves the delicate procedure of vertical integration \cite{ErlerKonopkaSachs201312, Sen2014b,SenWitten2015,ErlerKonopka2017a}. The integrands associated to the amplitudes are invariant under the {\it spin-$\mscr{E}$ mapping-class group}, the subgroup of the mapping-class group of the surface that preserves the spin structure $\mscr{E}$, and thus naturally can be integrated over the moduli stack of spin curves with appropriate signature. One also need to sum over all spin structures for at least two crucial reasons: 1) we need to integrate over $\dmmoduli$, the compactified moduli stack of genus-$\g$ surfaces with $\n$ punctures, and as such, we thus need a mapping-class-group-invariant integrand. Since the action of mapping-class group generically changes  the spin structure from $\mscr{E}$ to say $\mscr{E}'$, the summation over all spin structures guarantees that the resulting expression is mapping-class-group-invariant. 2) the summation over all spin structures is equivalent to imposing GSO projection in the path integral \cite{SeibergWitten1986a}.

\newl The compactified moduli stack of punctured spin curves can be viewed as a covering of the compactified moduli stack of ordinary curves, where the covering map corresponds to the forgetting of spin structures. Therefore, the region near the compactification divisors in the moduli stack of spin curves is mapped into the region near the compactification divisors in the moduli stack of ordinary curves. A good set of coordinate near the compactification divisor can be constructed by the so-called {\it plumbing-fixture construction} \cite{Masur1976a,Wolpert1987a}. In this construction, the local coordinate around two punctures on a single surface or disjoint surfaces are glued together using appropriate gluing relations\footnote{We refer to the gluing of the local coordinates around the punctures as gluing of punctures.}. The parameter(s) of the gluing together with the moduli of the component surfaces form a coordinate chart near the compactification divisor of the moduli stack. To compactify the stack of spin curve, one is thus interested in the following question

\newl \hypertarget{q:first question}{{\bf\small Question 1:} {\it Consider two punctures of the same type on a single surface equipped with a choice of spin structure or on two separate surfaces, each of them equipped with a choice of spin structure. Glue them using the appropriate gluing relation. What is the spin structure of the resulting surface and what does happen to it in the limit that the resulting surface develops a node\footnote{A node or an ordinary double point on an ordinary, spin curve, or super-Riemann surface is a point whose neighborhood is homeomorphic to two disk joined at a single point, the node.}?}}

\newl There is another subtlety in the gluing of punctures in superstring theory.  Since the total number of PCOs\footnote{We exclude the cases where the surface admits conformal-Killing spinors. These cases can be treated separately.} on a genus-$\g$ spin curve with $\n\equiv (\ns,\ra)$ punctures, as we will explain later, is equal to the fermionic dimension of the corresponding superstack, we need to insert $2\g-2+\ns+\frac{1}{2}\ra$ PCOs. From this relation, it turns out that the gluing of two NS punctures reproduces the correct number of PCOs in the resulting surface. However, the total picture number of the surface resulting from the gluing of two R punctures is one less that the one it should be. One thus has to {\it add} a PCO by hand on the surface with a suitable prescription and check that it reproduces the desired results \cite{Sen2014b}. Therefore, one is interested in the following question

\newl \hypertarget{q:second question}{{\bf\small Question 2:} {\it Consider two R punctures on a single surface equipped with a spin structure, or on two separate surfaces each of which is equipped with a choice of spin structure. Glue them using the appropriate gluing relation. What is the reason for the addition of the extra picture-changing operator from the supergeometry point of view?}}

\newl The aim is to show that there is a purely geometric explanation for this missing PCO, otherwise it is obvious that one has to include a PCO in the definition of kinetic term of the Ramond sector. 

\newl Although these questions are trivial for experts of the field, we felt that it might be useful for some readers.  

\newl However, the main objectives of this work are
\begin{enumerate}
	\item the proof of the existence of heterotic-string and type-II-superstring field theory vertices;
	
	\item the proof of the representability of the moduli problem of $\mscr{N}=1$ (possibly punctured) stable SUSY curves. 
\end{enumerate}
Let us briefly explain each of these problems. 

\subsubsection{The Existence of String Vertices}

\newl In the formulation of any string field theory, two choices have to be made: 1) a choice of the worldsheet theory, and 2) a choice of string vertices\footnote{For more on these choices see Section \ref{subsec:introduction to the proof of existence of string vertices}.}. These two choices provide us with a way to describe the interaction vertices of the associated string field theory. String vertices are some regions inside the product of moduli spaces parameterizing left- and right-moving sectors of the relevant worldsheet theory. If we denote the moduli space parameterizing left- and right-moving sectors by $\mscr{M}_L$ and $\mscr{M}_R$, respectively, then we have the following 
\begin{itemize}
	\item For heterotic-string theory $\mscr{M}^{\H}_L=\mscr{M}_{\sg}(\sns+\sra)$\footnote{Here, by $\sns+\sra$, we mean that we ignore all the complications associated to these borders on spin- or super-Riemann surfaces, and we think of them as borders on an ordinary Riemann surface.}, the moduli space of genus-$\g$ ordinary Riemann surfaces with $\ns+\ra$ borders, and $\mscr{M}^{\H}_R=\mscr{SM}_{\sg}(\ns,\ra)$, the moduli space of genus-$\g$ super-Riemann surfaces with $\ns$ NS boundary components and $\ra$ R boundary components. Let us denote a choice heterotic-string vertices corresponding to genus-$\g$ vertices with $\ns$ NS boundary components, and $\ra$ R boundary components by ${\mscr{V}}^{\H}_{\sg}(\ns,\ra)$. Then, the genus-$\g$ heterotic-string field theory interaction vertices $\mcal{I}^{\H}_{\sg}(\ns,\ra)$ is given by 
	\begin{equation}
	\mcal{I}^{\H}_{\sg}(\ns,\ra)=\bigintsss_{{\mscr{V}}^{\H}_{\vsg}(\sns,\sra)}\mbs{\Omega}_{\mbs{T}^{\H}},
	\end{equation}
	where ${\mscr{V}}^{\H}_{\sg}(\ns,\ra)\subset \mscr{M}^{\H}_L\times\mscr{M}^{\H}_R$, and $\mbs{\Omega}_{\mbs{T}^{\H}}$ is some correlation function computed in the worldsheet theory $\mbs{T}^{\H}$ of heterotic-string theory.   
	
	\item For type-II-superstring theory $\mscr{M}^{\text{II}}_L=\mscr{SM}_{\sg}(\ns,\ra)=\mscr{M}^{\text{II}}_R$. Let us denote a choice type-II-superstring vertices corresponding to genus-$\g$ vertices with $\nsns$ NSNS, $\rans$ RNS, $\nsra$ NSR, and $\rara$ RR boundary components by ${\mscr{V}}^{\text{II}}_{\sg}(\nsns,\rans,\nsra,\rara)$. Then, the genus-$\g$ type-II-superstring field theory interaction vertices $\mcal{I}^{\text{II}}_{\sg}(\nsns,\rans,\nsra,\rara)$ is given by 
	\begin{equation}
	\mcal{I}^{\text{II}}_{\sg}(\nsns,\rans,\nsra,\rara)=\bigintsss_{{\mscr{V}}^{\text{II}}_{\vsg}(\snsns,\srans,\snsra,\srara)}\mbs{\Omega}_{\mbs{T}^\text{II}},
	\end{equation}
	where ${\mscr{V}}^{\text{II}}_{\sg}(\nsns,\rans,\nsra,\rara)\subset \mscr{M}^{\text{II}}_L\times\mscr{M}^{\text{II}}_R$, and $\mbs{\Omega}_{\mbs{T}^\text{II}}$ is some correlation function computed in the worldsheet theory $\mbs{T}^{\text{II}}$ of type-II-superstring theory.
\end{itemize}
The gauge-invariance of the string field theory action is guaranteed if the string vertices satisfy the BV quantum master equation (BV QME). However, it is not {\it a priori} guaranteed that such subspaces exists. In this paper, we show that ${\mscr{V}}^{\H}_{\sg}(\ns,\ra)$ and ${\mscr{V}}^{\text{II}}_{\sg}(\nsns,\rans,\nsra,\rara)$ {\it always} exist. 

\subsubsection{The Moduli Problem of Stable $\mscr{N}=1$ SUSY Curves}
One problem of interest in super algebraic geometry is the moduli problem of stable $\mscr{N}=1$ (possibly punctured) SUSY curves, and one can ask the representability of the associated moduli functor. By this in the case of the moudli problem of stable $\mscr{N}=1$ (possibly punctured) SUSY curves, we mean a fibered functor $\mscr{F}$ which maps a superchemes to the category of $\mscr{N}=1$ (possibly punctured) SUSY curves. The representability requires the existence of a superscheme $\mscr{SM}$ such that $\mscr{F}$ is isomorphic to the functor of points of $\mscr{SM}$. If it happened to be the case, then $\mscr{SM}$ is the fine moduli space for this moduli problem. 

\newl There is a catch as it is usual with the moduli problem of stable curves \cite{DeligneMumford1969a}. Since stable $\mscr{N}=1$ SUSY curves have nontrivial auromorphisms, there is no fine moduli spaces for the representability of the moduli functor, i.e. there is no superscheme $\mscr{SM}$ that classifies all such curves. One thus need to find an alternative solution. One such solutions is to extend the category of superschemes to the {\it category of superstacks}. Then, the moduli functor might be representable by a superstack. So the question is the following: {\it is the moduli functor associated to the moduli problem of stable $\mscr{N}=1$ {\normalfont (}possibly punctured{\normalfont )} SUSY curves representable by a stack?} As we will show the moduli functor is more rigid and similar to the case of (possibly punctured) ordinary and spin curves can be represented by a Deligne-Mumford (DM) superstack.

\subsection{Contributions of the Paper}
Although we elaborate on many points, but the main results of this paper can be summarized as follows

\newl \underline{\small\bf 1) The Existence and Uniqueness of Solution to the BV QME in $\mscr{S}(\ns,\ra)$}\\
Let $\mscr{S}(\ns,\ra)$ denote the moduli space of possibly
disconnected spin-Riemann surfaces with $\ns$ NS boundary components and $\ra$ R boundary components. We denote the moduli stack of spin curves of connected genus-$\g$ spin-Riemann surfaces with $\ns$ NS boundary components and $\ra$ boundary components by $\mscr{S}_{\sg}(\ns,\ra)$. We associate a complex to such moduli spaces
\begin{alignat*}{2}
\mcal F(\mscr{S})&\equiv \bigoplus_{\sns,\sra} \mscr{C}_*\left(\frac{\mscr{S} (\ns,\ra)}{S^1\wr S_{\sns,\sra}}\right).
\end{alignat*}
$\mscr{C}_*$ is the functor of normalized singular chains with coefficient in any field containing $\mbb{Q}$, $S^1\wr S_{\sns,\sra}$ is the wreath product $(S^1)^{\sns+\sra}\rtimes S_{\sns,\sra}$. $S_{\sns,\sra}$ is the group that permutes NS boundaries and R boundaries separately among each other. The quotient by the wreath product turns $\mscr{S}(\ns,\ra)$ to the space of bordered spin-Riemann surfaces with unparametrized boundary. We also denote the contribution from connected genus-$\g$ surfaces by $\mcal{F}_{\sg,\sns,\sra}(\mscr{S})$. $\mcal F(\mscr{S})$ carries the structure of a commutative differential-graded
algebra (CDGA)\footnote{A differential-graded algebra $\mcal{A}$ is a graded algebra equipped with a map $d:\mcal{A}\longrightarrow\mcal{A}$ which has degree $-1$ (chain-complex convention) or degree $+1$ (cochain-complex convention) with product $\cdot:\mcal{A}\times\mcal{A}\longrightarrow\mcal{A}$ that satisfies the following two conditions:
	\begin{enumerate}
		\item $d\circ d=0$: this means that $d$ turns $\mcal{A}$ to a chain or cochain complex depending on its degree.
		
		\item {the graded Leibniz rule:} For any two elements $x,y\in\mcal{A}$, we have
		$$d(x\cdot y)=dx\cdot y+(-1)^{\deg(x)}x\cdot dy,$$
		where $\deg(x)$ is the degree of the element $x$ in the graded algebra $\mcal{A}$. 
\end{enumerate}}, where differential $d$ is the boundary map of chain complexes and the product comes from disjoint union of surfaces. Furthermore, we define the following gluing operation: let $\Delta^{\NS}$ ($\Delta^{\RA}$) denotes the operation of gluing of two NS (respectively R) boundaries. We define  
\begin{alignat*}{2}
\Delta &\equiv \Delta^{\text{NS}}+\Delta^{\text{R}}.
\end{alignat*}
This turns $\mcal{F}(\mscr{S})$ into a BV algebra\footnote{Recall that a BV algebra $\mcal{A}$ is defined as an algebra with the following structures
	\begin{enumerate}
		\item $\mcal{A}$ is a commutative differential-graded algebra with differential $d$.
		\item it is equipped with an odd differential $\delta$, i.e. $\delta^2=0$.
		\item $\delta$ satisfies $[\delta,d]=0$, where $[\cdot,\cdot]$ is a graded commutator.
\end{enumerate}}. We can now state one of the results of this paper
\begin{thr*}[The Existence and Uniqueness of Solution to the BV QME in $\mscr{S}_{\sg}(\ns,\ra)$]
	For each triple $(\g,\ns,\ra)$ with $2\g-2+\ns+\ra>0$, there exists an element $\mcal{V}_{\sg}(\ns,\ra)\in \mathcal{F}_{\sg,\sns,\sra}(\mscr{S})$ of homological degree $6\g-6+2\ns+2\ra$, with the following properties
	\begin{enumerate}
		\item $\mcal{V}_{0}(3,0)$ is the fundamental cycle of $\mscr{S}_{0}(3,0)/S_1\wr S_{3,0}$, i.e. 0-chain of coefficient $1/3!$.
		
		\item $\mcal{V}_{0}(1,2)$ is the fundamental cycle of $\mscr{S}_{0}(1,2)/S_1\wr S_{1,2}$, i.e. 0-chain of coefficient $1/2!$.
		
		\item The generating function 
		\begin{equation}\label{eq:the generating function of the solution to the BV QME in spin moduli I}
		\mcal{V}\equiv\sum_{2\sg-2+\sns+\sra>0}\hbar ^{2\sg-2+\sns+\sra}\mcal{V}_{\sg}(\ns,\ra)\in \hbar \mathcal{F}(\mscr{S})[\![\hbar]\!],
		\end{equation}
		satisfies the {\normalfont BV QME}
		\begin{equation*}
		(d+\hbar\Delta)\exp\left(\frac{\mcal{V}}{\hbar}\right)=0.
		\end{equation*}
		
		\item $\mcal{V}$ is unique up to homotopy in the category of BV algebras. 
	\end{enumerate}
\end{thr*}
For the proof of this result see Section \ref{subsec:the proof of existence and uniqueness} and especially Theorem \ref{the:the existence and uniqueness of the solution to the BV QME in spin moduli}. This result is the generalization of Proposition 10.1.1\footnote{The numbers of theorem, proposition, etc of \cite{Costello200509} that we are referring to are the ones that appear in the published version not the Arxiv preprint. We explicitly mention it if we are referring to the ArXiv version.} of \cite{Costello200509}. The physical interpretation of the result of Costello is the existence and uniqueness of bosonic-string field theory vertices. However, we cannot make a similar claim in the case of heterotic-string and type-II-superstring theories. We need to work more to establish the existence of vertices in these theories, as we explain below. 

\newl One might ask {\it what is the importance of this result for string field theory?} As we argue below, the existence of heterotic-string and type-II superstring field theory vertices can be deduced from this result. The reason is as follows
\begin{enumerate}
	\item We can associate BV algebras to complexes associated to spaces parameterizing worldsheets of heterotic-string and type-II superstring field theory. These BV algebras are quasi-isomorphic to BV algebra $\mcal{F}(\mscr{S})$ which means that their corresponding homology groups are isomorphic. 
	
	\item On the other hand, if two differential-graded algebras are quasi-isomorphic, their sets of homotopy classes of solutions of the Maurer-Cartan equation are isomorphic. The BV QME can be considered as the Maurer-Cartan equation of the corresponding BV algebra (see Lemma 5.2.1 of \cite{Costello200509}). Therefore, if two BV algebras are quasi-isomorphic, their sets of homotopy classes of solutions of the BV QME are isomorphic (see Lemma 5.3.1 and Definition 5.4.1 of \cite{Costello200509}). 
	
	\item Putting 1 and 2 together and noting that heterotic-string and type-II superstring field theory vertices must satisfy the BV QME, we conclude that they exist. 
\end{enumerate}

\newl \underline{\small\bf 2) Solutions to the BV QME in $\mscr{S}(\ns,\ra)$ and Fundamental Classes of DM Stacks}\\
Let $\overline{\mscr{S}}_{\sns,\sra}$ denote the moduli stack of stable spin curves with $\ns$ NS punctures and $\ra$ Ramond punctures, and $\overline{\mscr{S}}_{\sg,\sns,\sra}$ denote the contribution coming from connected genus-$\g$ surfaces. We associate a chain complex to these moduli spaces
\begin{equation}
\mcal{F}(\overline{\mscr{S}})\equiv \bigoplus_{\sns,\sra}\mscr{C}_*\left(\frac{\overline{\mscr{S}}_{\sns,\sra}}{S_{\sns,\sra}}\right).
\end{equation}
$S_{\sns,\sra}$ permutes the set of NS and R punctures among each other separately. After taking the homology, and setting $\Delta=0$, $\mcal{F}(\overline{\mscr{S}})$ turns to a BV algebra. We then define the following formal sum
\begin{equation}
[\overline{\mscr{S}}]\equiv\sum_{\substack{\sg,\sns,\sra \\ 2\sg-2+\sns+\sra\ge 0}}\hbar^{2\sg-2+\sns+\sra}[\overline{\mscr{S}}_{\sg,\sns,\sra}/S_{\sns,\sra}]\in\hbar  \H_*(\mcal{F}(\overline{\mscr{S}}))[\![\hbar]\!],
\end{equation}
where $[\overline{\mscr{S}}_{\sg,\sns,\sra}/S_{\sns,\sra}]$ is the fundamental class of $\overline{\mscr{S}}_{\sg,\sns,\sra}/S_{\sns,\sra}$. One can then establish the following Theorem which provides a link between $\mcal{V}$, the solution of the BV quantum master equation in the BV algebra $\mcal{F}(\mscr{S})$ defined in \eqref{eq:the generating function of the solution to the BV QME in spin moduli I}, and $[\overline{\mscr{S}}]$. 

\begin{thr*}[Solutions to the BV QME in $\mscr{S}(\ns,\ra)$ and Fundamental Classes of DM Stacks]
	There is a map $\mcal{F}(\mscr{S})\longrightarrow \H_*(\mcal{F}(\overline{\mscr{S}}))$ in the homotopy category of {\normalfont BV} algebras that sends $\mcal{V}\longrightarrow [\overline{\mscr{S}}]$. This means that in the homotopy category of {\normalfont BV} algebras $\mcal{V}_{\sg,\sns,\sra}$ is homotopy-equivalent to the orbifold fundamental class $[\overline{\mscr{S}}_{\sg,\sns,\sra}/S_{\sns,\sra}]$ of the moduli stack of connected stable genus-$\g$ spin curves with $\ns$ {\normalfont NS} punctures and $\ra$ {\normalfont R} punctures, i.e. it is homotopy-equivalent to the whole moduli stack. 
\end{thr*}
For the proof of this result see Theorem \ref{the:solutions to the BVQME and fundamental classes of DM spaces} in Section \ref{subsec:relation of solutions to the BVQME in spin moduli and fundamental classes of DM spaces}.  

\newl \underline{\small\bf 3) The Existence of Solution to the BV QME in $\mscr{SM}(\ns,\ra)$}\\
Let $\mscr{SM}(\ns,\ra)$ denote the space of bordered $\mscr{N}=1$ super-Riemann surfaces with $\ns$ NS and $\ra$ R boundary components. We can associate a singular superchain\footnote{For the definition of singular superchains on superspaces see Section \ref{subsec:singular homology on superspaces}.} complex to these spaces
\begin{equation}
\mcal{F}(\mscr{SM})\equiv \bigoplus_{\sns,\sra}\mscr{SC}_*\left(\frac{\mscr{SM}(\ns,\ra)}{S^1\wr S_{\sns,\sra}}\right)
\end{equation}
where $\mscr{SC}_*$ is the functor of normalized singular simplicial superchains of top fermionic dimension with coefficients in any field containing $\mbb{Q}$. This is a differential-graded algebra, where the derivative is the boundary operation, as defined in Definition \ref{def:the definition of singular superchains on superspaces}, acting on singular superchains. We can define the gluing operation as follows
\begin{equation}
\Delta_{\mfk{s}}\equiv \Delta^{\NS}_{\mfk{s}}+\Delta^{\RA}_{\mfk{s}},
\end{equation}
where $\Delta_{\mfk{s}}^{\NS}$ denotes the gluing of two NS boundaries and $\Delta^{\RA}_{\mfk{s}}$ is the gluing operation of two R boundaries. This turns $\mcal{F}(\mscr{SM})$ to a BV algebra. We then have the following
\begin{thr*}[The Existence of Solution to BV QME in $\mscr{SM}(\ns,\ra)$]
	There exists a solution to the {\normalfont BV QME} in $\mcal{F}(\mscr{SM})$.
\end{thr*}
For the proof of this Theorem see Section \ref{subsec:proof of the existence of solution to BVQME in the supermoduli of bordered surfaces}. 

\newl We might need to stress one point here. We have constructed a solution to the BV QME in the space of bordered spin-Riemann surfaces. In principle, we should be able to {\it thicken} this solution in fermionic direction to find a solution to the BV QME in the moduli space of bordered $\mscr{N}=1$ super-Riemann surfaces. However, life is not that simple. As we will rigorously prove in Section \ref{subsec:lifting subspace of reduced space}, to construct a supermanifold $M$, we can start from its reduced space $M_{\text{red}}$ and thicken it in the fermionic direction. However, this defines $M$ up to {\it homology}\footnote{This can also be seen from the result of \cite{Batchelor197909} where it is proven that any {\it smooth} thickening can be extended to a smooth supermanifold, and conversely any smooth supermanifold may be obtained in this way.}. This means that any two different thickenings in the fermionic directions can be connected by a homotopy, which is intuitively due to the fact that fermionic directions are infinitesimal. Therefore, by considering the solutions to the BV QME in $\mscr{S}(\ns,\ra)$ {\it as a smooth subspace} and thickening it in all fermionic directions, we get a subspace of $\mscr{SM}(\ns,\ra)$ which {\it does not necessarily satisfies the {\normalfont BV QME}}. Therefore, we need to come-up with another way to prove the existence of such solutions in $\mscr{SM}(\ns,\ra)$, as we do in Section \ref{subsec:proof of the existence of solution to BVQME in the supermoduli of bordered surfaces}.

\newl Perhaps we must stress that the idea of thickening does not work in the realm of complex supermanifolds. It is established that there are always obstructions in extending a holomorphic thickening to a complex supermanifold \cite{EastwoodLeBrun198610}. However, we are only interested in integration over these spaces, and for that reason, the results of integration over different thickenings are the same, and do not depend on the choice of thickening in the fermionic direction \cite{Witten2012a}.

\newl \underline{\small\bf 4) The Existence of Heterotic-String and Type-II-Superstring Vertices}\\
The next result of the paper is the proof of existence of heterotic-string and type-II superstring field theory vertices. As we stated above, the proof of this result depends on the first result, i.e. the existence of solution to the BV QME in $\mscr{S}(\ns,\ra)$.

\newl To state the result for the heterotic-string field theory, we define the following complexes
\begin{alignat}{2}
\mcal{F}((\mscr{M}^\text{H}_L\times\mscr{M}^\text{H}_R)^{\mbf{D}})&\equiv \bigoplus_{\sns,\sra}\mscr{SC}_*\left(\frac{(\mscr{M}^\text{H}_L(\ns+\ra)\times\mscr{M}^\text{H}_R(\ns,\ra))^\mbf{D}}{S_L^{\text{H}}\times S^{\text{H}}_R}\right),
\end{alignat}
where superscript $\mbf{D}$ denotes diagonal and subscript red means that we consider the reduced spaces, and 
\begin{alignat}{2}
S^\text{H}_L\equiv S_1\wr S_{\sns+\sra}, \qquad S^\text{H}_R\equiv S_1\wr S_{\sns,\sra},
\end{alignat}
where $S_1\wr S_{\sns+\sra}=(S_1)^{\sns+\sra}\rtimes S_{\sns+\sra}$ and $S_1\wr S_{\sns,\sra}=(S^1)^{\sns+\sra}\rtimes S_{\sns,\sra}$. We then have the following result
\begin{thr*}[The Existence of Heterotic-String Vertices]
	There exists a solution to the {\normalfont BV QME} in $\mcal{F}((\mscr{M}^\text{\H}_L\times\mscr{M}^\text{\H}_R)^{\mbf{D}})$. The genus-$\g$ part is a genus-$\g$ heterotic-string field theory vertex with $\ns$ {\normalfont NS} and $\ra$ {\normalfont R} boundary components. 
\end{thr*}
For the proof of this Theorem see Section \ref{subsec:the case of heterotic-string vertices}.

\newl To state the result for the case of type-II-superstring field theory, we define the following complexes 
\begin{alignat}{2}
\mcal{F}((\mscr{M}^{\text{II}}_L\times\mscr{M}^{\text{II}}_R)^{\mbf{D}}) &\equiv \bigoplus_{\sns,\sra}\mscr{SC}_*\left(\frac{(\mscr{M}^{\text{II}}_L(\ns,\ra)\times\mscr{M}^{\text{II}}_R(\ns,\ra))^\mbf{D}}{S^{\text{II}}_L\times S^{\text{II}}_R}\right).
\end{alignat} 
where 
\begin{alignat}{2}
S^{\text{II}}_L\equiv S_1\wr S_{\sns,\sra}=S^{\text{II}}_R.
\end{alignat}
We then have the following result
\begin{thr*}[The Existence of Type-II-Superstring Vertices]
	There exists a solution to the {\normalfont BV QME} in $\mcal{F}((\mscr{M}^\text{\normalfont II}_L\times\mscr{M}^\text{\normalfont II}_R)^{\mbf{D}})$. The genus-$\g$ part is a genus-$\g$ type-II-string field theory vertex with $\nsns$ {\normalfont NSNS}, $\rans$ {\normalfont RNS}, $\nsra$ {\normalfont NSR}, and $\rara$ {\normalfont RR} boundary components. 
\end{thr*}
For the proof of this Theorem see Section \ref{subsec:the case of type-II-superstring vertices}.

\newl \underline{\small\bf 5) The Representability of the Moduli Problem of Stable $\mscr{N}=1$ SUSY Curves}\\
There is a model for the compactification of the space of smooth genus-$\g$ possibly punctured $\mscr{N}=1$ SUSY curves introduced by Deligne \cite{Deligne198709}. Let $\overline{\mscr{SM}}_{\sg,\sns,\sra}$ denotes the space of stable genus-$\g$ $\mscr{N}=1$ SUSY curve with $\ns$ NS punctures and $\ra$ Ramond punctures. Then, we have the following result

\begin{thr*}[The Representability of the Moduli Problem of Stable $\mscr{N}=1$ SUSY Curves]
	Assume that a triple of natural numbers $(\g,\ns,\ra)$ satisfies $\ra\in 2\mathbb Z_{\ge 0}$, and $2\g-2+\ns+\ra>0$.
	\begin{enumerate}
		\item[$(1)$] $\overline{\mscr{SM}}_{\sg,\sns,\sra}$ is a smooth proper Deligne-Mumford superstack of dimension $$\left(3\g-3+\ns+\ra\big|2\g-2+\ns+\frac{\ra}{2}\right).$$ Moreover its bosonic truncation $\overline{\mscr{SM}}_{\sg,\sns,\sra;\text{\normalfont bos}}$ is the moduli stack $\overline{\mscr{S}}_{\sg,\sns,\sra}$ of punctured stable spin curves.
		
		\item[$(2)$] $\overline{\mscr{SM}}_{\sg,\sns,\sra}$ has a coarse moduli superspace $\overline{\mathcal{SM}}_{\sg,\sns,\sra}$ which is actually bosonic and is also the coarse moduli space for the bosonic quotient $\overline{\mscr{SM}}_{\sg,\sns,\sra;\text{\normalfont ev}}$. Moreover its reduced subspace $\overline{\mathcal{SM}}_{\sg,\sns,\sra;\text{\normalfont red}}$ is the coarse moduli space $\overline{\mathcal{S}}_{\sg,\sns,\sra}$ of punctured stable spin curves.
	\end{enumerate}
\end{thr*}
For the proof of this Theorem see Section \ref{subsec:moduli of stable SUSY curves}. 

\newl There is an alternative model for the compactification of the space of stable $\mscr{N}=1$ SUSY curve different from the one proposed by Deligne based on studying the spaces of stable log twisted SUSY curves \cite{Wakabayashi1602}. This model will not be pursued here. Recently, the moduli spaces of {\it smooth} SUSY curves with R punctures have also been constructed in the literature \cite{OttVoronov1910,BruzzoHernandezRuiperez1910}. 

\subsection{Notations and Remarks}\label{subsec:notation and remarks}

For the convenience of reader, we collect some notations used in this paper

\begin{itemize}
	\item $\mscr{M}_{\sg,\sn}$: the moduli stack of genus-$\g$ Riemann surfaces with $\n$ punctures.
	
	\item $\overline{\mscr{M}}_{\sg,\sn}$: the moduli stack of stable genus-$\g$ Riemann surfaces with $\n$ punctures.
	
	\item $\partial\overline{\mscr{M}}_{\sg,\sn}$: the boundary (i.e. points at infinity) of $\overline{\mscr{M}}_{\sg,\sn}$. It is isomorphic to $\overline{\mscr{M}}_{\sg,\sn,}\setminus \mscr{M}_{\sg,\sn}$. 
	
	\item $\mscr{S}_{\sg,\sns,\sra}$: the moduli stack of genus-$\g$ spin curves with $\ns$ Neveu-Schwarz punctures and $\ra$ Ramond punctures. The moduli stack of genus-$\g$ spin curves with $\ns$ Neveu-Schwarz punctures and $\ra$ Ramond punctures with an even spin structure is denoted as $\mscr{S}_{\sg,\sns,\sra}^{\text{ev}}$ for the even spin structure, and $\mscr{S}_{\sg,\sns,\sra}^{\text{od}}$ for an odd spin structure. 
	
	\item $\overline{\mscr{S}}_{\sns,\sra}$: the stable-curve compactification of the moduli stack of possible disconnected spin curves with $\ns$ Neveu-Schwarz punctures and $\ra$ Ramond punctures
	
	\item $\overline{\mscr{S}}_{\sg,\sns,\sra}$: the compactified moduli stack of connected genus-$\g$ spin curves with $\ns$ Neveu-Schwarz punctures and $\ra$ Ramond punctures. This moduli stack is the component of $\overline{\mscr{S}}_{\sns,\sra}$ coming from the connected genus-$\g$ surfaces. The moduli stack of genus-$\g$ spin curves with $\ns$ Neveu-Schwarz punctures and $\ra$ Ramond punctures with a specific spin structure is denoted as $\overline{\mscr{S}}_{\sg,\sns,\sra}^{\text{ev}}$ for the even spin structure, and $\overline{\mscr{S}}_{\sg,\sns,\sra}^{\text{od}}$ for the odd spin structure.

	\item $\partial \overline{\mscr{S}}_{\sg,\sns,\sra}$: the boundary (i.e. points at infinity) of the compactified moduli stack of genus-$\g$ spin curves with $\ns$ Neveu-Schwarz punctures and $\ra$ Ramond punctures. It is isomorphic to $\overline{\mscr{S}}_{\sg,\sns,\sra}\setminus \mscr{S}_{\sg,\sns,\sra}$.
	
	\item  \hypertarget{notation for stack of punctured spin curve}{$\mscr{S}^{\mbf{m}}_{\sg,\sn}$:}  The moduli stack of genus-$\g$ spin curve with $\n$ punctures such that the vector $\mbf{m}\equiv (m_1,\cdots,m_{\sn})$ is a vector whose components $m_a$ are the order of vanishing at punctures $\mfk{p}_a$. Let us make this a bit more precise. As we will explain in section \ref{sec:background}, there is a distinguished subbundle $\mcal{D}$ inside the tangent bundle $T\mcal{R}$ of an $\mscr{N}=1$ super-Riemann surface $\mcal{R}$ which is nonintegrable. This means that there exists a short exact sequence
	\begin{equation}
	0\longrightarrow \mcal{D}\longrightarrow T\mcal{R}\longrightarrow T\mcal{R}/\mcal{D} \longrightarrow 0. 
	\end{equation}
	The bundle $T\mcal{R}/\mcal{D}$ is generated by $\partial_z$. When there is no punctures, we will see that the bundle $\mcal{D}^2$ is also generated by $\partial_z$. Once punctures are introduced, the bundle $T\mcal{R}/\mcal{D}$ is still generated by $\partial_z$. However, it is not necessarily true that $\mcal{D}^2$ is generated by $\partial_z$. For example, we will see in section \ref{sec:background} that in the presence of Ramond puncture $\mcal{D}^2$ is actually generated by $P(z)\partial_z$, where $P(z)$ has a simple zero at the location of the Ramond puncture $\mfk{p}_a$, i.e. $P(z)=z-z_a$. We thus have $\mcal{D}^2\simeq T\mcal{R}/\mcal{D}\otimes \mcal{O}(-\mscr{D})$, where
	\begin{equation}
	\mscr{D}\equiv \sum_{a=1}^{\sra}\mscr{D}_a=\sum_{a=1}^{\sra}\mfk{p}_a,
	\end{equation}
	is the divisor defined by the locations of Ramond punctures $\mfk{p}_a$ at $z=z_a$. This means that a section of $\mcal{D}^2$ is a section of $T\mcal{R}/\mcal{D}$ which vanishes along the divisor $\mscr{D}$. In general, $\mcal{D}^2$ is generated by $P(z)\partial_z$, where now $P(z)$ has higher-order zeros at the location of punctures. If near the puncture $\mfk{p}_a$, $P(z)$ has zero of order $m_a$, then $\mcal{D}^2\simeq T\mcal{R}/\mcal{D}\otimes \mcal{O}(-\mscr{D})$ where now
	\begin{equation}
	\mscr{D}=\sum_{a=1}^\sn \mscr{D}_a=\sum_{a=1}^\sn m_a\mfk{p}_a. 
	\end{equation}  
	This means that a section of $\mcal{D}^2$ is a section of $T\mcal{R}/\mcal{D}$ which vanishes along the divisor $\mscr{D}$ with zeros of order $m_a$ at the puncture $\mfk{p}_a$, i.e. $P(z)=\prod_a(z-z_a)^{m_a}$. Once we reduce from an $\mscr{N}=1$ super-Riemann surface to the underlying spin curve, $\mcal{D}$ is essentially a choice of spin structure. Therefore, $\mscr{S}^{\mbf{m}}_{\sg,\sn}$ is the moduli stack of genus-$\g$ spin curve with $\n$ punctures such that the vector $\mbf{m}$ is the vector of the order of vanishing at punctures. For the cases we are interested in, the order of vanishing are as follows: $m_a=0$ for NS punctures and $m_a=1$ for Ramond punctures.
	
	\item  $\mscr{S}(\sns,\sra)$: the moduli space of possibly disconnected spin-Riemann surfaces with $\ns$ Neveu-Schwarz boundary components (see Definition \ref{def:NS boundary components}) and $\ra$ Ramond boundary components (see Definition \ref{def:R boundary components}). 
	
	\item $\mscr{S}_{\sg}(\sns,\sra)$: the moduli space of connected genus-$\g$ spin-Riemann surfaces with $\ns$ Neveu-Schwarz boundary components and $\ra$ Ramond boundary components. 
	
	\item $\widetilde{\mscr{S}}(\sns,\sra)$: the moduli space of stable spin-Riemann surfaces, decorated at each Neveu-Schwarz puncture with a ray in the tangent space, at each Ramond punctures with a ray in the spinor bundle, at each Neveu-Schwarz node with a ray in the tensor product of rays of the tangent spaces at each side, and at each Ramond node with a ray in the spinor bundle. We emphasize that a surface in $\widetilde{\mscr{S}}(\sns,\sra)$ does not have boundary components. $\widetilde{\mscr{S}}(\sns,\sra)$ is however homotopy-equivalent to $\mscr{S}(\sns,\sra)$. As we will explain, there is a crucial fact about the gluing of Ramond punctures in $\widetilde{\mscr{S}}(\sns,\sra)$ which makes it more suitable than $\mscr{S}(\sns,\sra)$ to work with.
	
	\item $\mscr{SM}_{\sg,\sns,\sra}$: the moduli space of genus-$\g$ $\mscr{N}=1$ super-Riemann surfaces with $\ns$ Neveu-Schwarz punctures and $\ra$ Ramond punctures.
	
	\item $\overline{\mscr{SM}}_{\sg,\sns,\sra}$: the compactified superstack of genus-$\g$ $\mscr{N}=1$ super-Riemann surfaces with $\ns$ Neveu-Schwarz punctures and $\ra$ Ramond punctures.
	
	\item $\mscr{SM}_{\sg}(\ns,\ra)$: the moduli space of genus-$\g$ $\mscr{N}=1$ super-Riemann surfaces with $\ns$ Neveu-Schwarz boundary components and $\ra$ Ramond boundary components.
\end{itemize}

\newl {\bf\small Remark 1:} When we say {\it puncture}, we do not mean removing the point from the surface. The surfaces are still closed with marked points on them. There is a delicate difference between marked points and punctures in the presence of spin structure. We only use this terminology to cope with the literatures of superstring theory. We also would like to mention that that since a Ramond puncture is part of the defining data of a superconformal structure, it defines a divisor (we make this statement precise in section \ref{subsec:punctures on super-Riemann surface}). Therefore, Ramond divisor is a better terminology, as has been used in \cite{Witten2012b,Witten2012c}. 

\newl {\bf\small Remark 2:} We use the notation $(z|\theta)$ to denote the local coordinates on a $\mscr{N}=1$ super-Riemann surface. We also use the notation $(p|q)$ to denote the dimension of supermanifolds, supervector spaces, etc with $p$ even directions and $q$ fermionic directions. An exception is when we talk about pseudoforms. In this case, $p$ is related to the form degree of pseudoforms and $-q$ is the so-called picture number of pseudoforms, as we shall explain in Appendix \ref{app:the picture-changing operation}.

\newl {\bf\small Remark 3:} A surface equipped with a line bundle $\mscr{E}$ whose $r$\textsuperscript{th} tensor power is isomorphic to the canonical line bundle, $\mscr{E}^{\otimes r}\simeq \omega$, is called an $r$-spin curve and $\mscr{E}$ is called a generalized spin structure. We are interested in the case that $r=2$ corresponding to the ordinary spin structure, and as such it should be called a $2$-spin curve. To avoid cluttering, we call them spin curves. A discussion of spin curves defined over a base scheme can be found in Appendix \ref{app:compactification of the moduli stack of spin curves}.

\newl {\bf\small Remark 4:} We denote the {\it signature} of a genus-$\g$ spin curve or $\mscr{N}=1$ super-Riemann surface with $\ns$ NS punctures and $\ra$ R punctures as $(\g;\ns,\ra)$.

\newl {\bf\small Remark 5:} For {\it punctured surfaces}, we use ordinary, spin- or SUSY curves (the language of algebraic geometry where we can work with curves over an algebraically-closed field $\mbb{F}$\footnote{One can define curves over fields which are not algebraically-closed. We are not considering such curves in this work.}) and ordinary, spin- or super-Riemann surfaces (the language of analytic geometry where $\mbb{F}=\mbb{C}$) interchangeably throughout the text. We use the terms {\it bordered} ordinary, spin- or super-Riemann surfaces for surfaces with geodesic boundary components. For example, We call a genus-$\g$ Riemann surface with $\ns$ Neveu-Schwarz punctures and $\ra$ Ramond punctures equipped with a spinor bundle $\mscr{E}$, a genus-$\g$ punctured spin curve\footnote{There are generalized $r$-spin curves, the curves equipped with a line bundle $\mscr{E}$ whose $r$\textsuperscript{th} tensor power is isomorphic to $\omega$.}. Sometimes we omit $\mscr{N}=1$ but we always mean $\mscr{N}=1$ (smooth or stable) SUSY curve.

\newl {\bf\small Remark 6:} When it is needed to describe more than one surface, we denote the signature of the $a$\textsuperscript{th} surface by $(\g_a;\ns^{a},\ra^a)$. Here, $\ns^a$ ($\ra^a$) could denote either the $a$\textsuperscript{th} NS (respectively R) puncture or NS (respectively R) boundary. 

\newl The organization of the paper is as follows. In Section \ref{sec:background}, we define spin curves, punctures on them, and the gluing of punctures. In Section \ref{sec:gluing of punctures on spin curves and torsion-free sheaves}, we investigate the relation between gluing of punctures on spin curves and the compactification of the moduli stack of spin curves, and also a subtlety in the gluing of Ramond punctures. We then establish the existence and uniqueness of a solution to the BV quantum master equation in the moduli space of bordered spin-Riemann surfaces, and their relation to fundamental classes of DM spaces in Section \ref{sec:the existence and uniqueness of of solution to the BV QME in the moduli space of bordered spin-Riemann surfaces}. In Section \ref{sec:the existence of solution to the BVQME in supermoduli of bordered surfaces}, we prove the existence of solution to the BV QME in the moduli space of bordered $\mscr{N}=1$ super-Riemann surfaces. Section \ref{sec:superstring vertices} is devoted to the proof of the existence of heterotic-string and type-II-superstring vertices. Finally, we turn to the moduli problem of $\mscr{N}=1$ SUSY curves in Section \ref{sec:on the moduli problem of N=1 SUSY Curves}.  In the Appendix \ref{app:the picture-changing operation}, we explain the picture-changing operation from canonical, path-integral and supergeometry point of view, and their equivalence. The compactification of the moduli stack of spin curves is the subject of Appendix \ref{app:compactification of the moduli stack of spin curves}.

\section{Spin Curves, Punctures and their Gluing}\label{sec:background}
In this section, we define spin curves and various punctures on them. We then explain the gluing of punctures which can be used to construct new surfaces. 

\subsection{Spin Curves over $\mbb{C}$}\label{subsec:spin curves over C}

The gauge-fixed worldsheet action of the RNS-superstring theory has superconformal (gauge) invariance and, as such, the worldsheet theory should only be sensitive to the superconformal structure of the worldsheet. This means that the worldsheet cannot be a general $(1|1)$ complex supermanifold whose transition functions are arbitrary superholomorphic functions. Instead, it should be a $(1|1)$ complex supermanifold with superconformal transition functions, i.e. an $\mscr{N}=1$ super-Riemann surface. 

\newl An $\mscr{N}=1$ super-Riemann surface is a complex supermanifold of dimension $(1|1)$ with the following properties 
\begin{enumerate}
	\item[1)] It can be locally mapped to $\mbb{C}^{1|1}$. Each local chart is parametrized by $(z|\theta)$, where $z$ is a local complex commuting coordinate and $\theta$ is a complex anti-commuting coordinate. Since we are interested in super-Riemann surfaces appearing in the heterotic-string theory, the complex conjugate of $\theta$ does not appear, and as such there is no reality condition on $\theta$ \cite{Cohn1988a}. A similar argument holds for the case of surfaces appearing in the type-II superstring theories. 
	
	\item[2)] The transition function between the overlapping charts $(z|\theta)$ and $(z'|\theta')$ is given by superconformal coordinate transformations, i.e. those coordinate transformations that satisfy the following relations
	\begin{equation}\label{eq:equation of superconformal coordinate transformation}
	D_\theta z'=\theta'D_\theta\theta',
	\end{equation}
	where $D_\theta$ is the spinor derivative
	\begin{equation}\label{eq:spinor derivative}
	D_{\theta}\equiv \frac{\partial}{\partial\theta}+\theta\frac{\partial}{\partial z}. 
	\end{equation}
	A local coordinate $(z|\theta)$ in which $D_{\theta}$ takes the form \eqref{eq:spinor derivative} is called a {\it superconformal coordinate system}. On any super-Riemann surface there is a local superconformal coordinate $(z|\theta)$ such that $D_{\theta}$ has the form \eqref{eq:spinor derivative}. For the proof of this fact, see Lemma $1.2$ of \cite{LeBrunRothstein198803} or Lemma $3.1$ of \cite{DonagiWitten2013a}. $D_\theta$ generates a $(0|1)$-dimensional subbundle $\mcal{D}$ of the tangent bundle such that if $s$ is a nonzero local section of $\mcal{D}$, then $s$ and $s^2\equiv\frac{1}{2}\{s,s\}$ are linearly independent\footnote{As we explain later, this statement fails to hold in the presence of Ramond punctures.}, i.e. the section is {\it completely nonintegrable}. 
\end{enumerate} 
Another useful and equivalent way to define an $\mscr{N}=1$ super-Riemann surface which is closer to the supergravity description is as follows. A $(1|1)$-dimensional supermanifold $\mcal{R}$ is locally parametrized by ${\boldsymbol{z}}\equiv(x_1,x_2;\theta_1,\theta_2)$, where $x_a$ and $\theta_a$ for $a=1,2$ are real. The geometry of the worldsheet is described by a supermetric. Dealing with supermetric is easier if we use the language of frame fields.  As such we introduce the vierbein $\boldsymbol{E}_A(\mbs{z})=E\indices{_A^B}(\mbs{z})\partial_B$ on each chart $(\boldsymbol{z}_\alpha,\boldsymbol{\mscr{U}}_\alpha)$, subject to the torsion constraints required to remove the redundant components \cite{Howe197303}. Here $A,B=x_1,x_2,\theta_1,\theta_2$. There is always a choice of supercoordinate $\mbs{z}\equiv(z,\bm\bar{z};\theta,\bm\bar{\theta})$, where $\bm\bar{z}$ and $\bm\bar{\theta}$ are the complex conjugate of $z$ and $\theta$ respectively, such that the complex superfield $\mbs{E}(\mbs{z})=E\indices{^A}(\mbs{z})\partial_A$ is proportional to the spinor derivative \eqref{eq:spinor derivative}
\begin{equation}\label{eq:the form of zweibein in conformal gauge}
\mbs{E}(\mbs{z})=\rho(\mbs{z})D_\theta,
\end{equation}
for some function $\rho(\mbs{z})$. $\{\mbs{E},\overline{\mbs{E}};\mbs{E}^2,\overline{\mbs{E}}^2\}$ form a basis for the tangent space at each point $\mbs{z}$. \eqref{eq:the form of zweibein in conformal gauge} is invariant precisely under coordinate transformations \eqref{eq:equation of superconformal coordinate transformation}. If we choose to work with superconformal coordinate transformations in the overlap of patches $\boldsymbol{\mscr{U}}_\alpha$ and $\boldsymbol{\mscr{U}}_\beta$, we end-up with an atlas which defines a superconformal structure on $\mcal{R}$. A $(1|1)$-dimensional supermanifold endowed with such a superconformal structure is called an $\mscr{N}=1$ super-Riemann surface. For a mathematically-precise definition of a family of {\it smooth} $\mscr{N}=1$ super-Riemann surfaces, see Definition \ref{def:stable SUSY curves}.

\newl All super-Riemann surfaces we are dealing with in heterotic-string and type-II-superstring theories have trivial topology in the fermionic direction. This restriction is implemented by endowing the surface with the DeWitt topology. Although nontrivial topology in fermionic direction is also possible \cite{CraneRabin198503a,CraneRabin198503b}, it can be argued that they are not contribute to the path integrals in string theory \cite{CraneRabin198812}. 

\newl In the case of ordinary Riemann surface, the operator $\partial_z$ transforms as $\left(\frac{\partial z'}{\partial z}\right)^{-1}\partial_z$. Then, the Dolbeault operator $\partial\equiv dz\otimes \partial_z$ is globally-defined. Similarly, on an $\mscr{N}=1$ super-Riemann surface, the spinor derivative transforms as 
\begin{equation}
D_{\theta'}=D_{\theta'}\theta D_{\theta}.
\end{equation}
We can define an operator
\begin{equation}\label{eq:super-Dolbeault operator}
D\equiv [dz|d\theta]\otimes D_{\theta},
\end{equation}
which can be thought of as the super version of the usual Dolbeault operator. $[dz|d\theta]$ denotes a section of the canonical (or Berezinian) line bundle (similar to the ordinary case that $dz$ is a section of the usual canonical bundle). Under a coordinate transformation $(z|\theta)\longrightarrow(z'|\theta')$, we have
\begin{alignat}{2}
D'&=[dz'|d\theta']\otimes D_{\theta'}=\mtt{Ber}\left(\frac{\partial(z'|\theta')}{\partial(z|\theta)}\right)\cdot \left(D_{\theta'}\theta\right) [dz|d\theta]\otimes D_{\theta}\nonumber
\\
&=\left(D_\theta\theta'\right)\cdot\left(D_{\theta'}\theta\right)[dz|d\theta]\otimes D_{\theta}=[dz|d\theta]\otimes D_{\theta}=D, 
\end{alignat}
where $\mtt{Ber}(\cdots)$ denotes the Berezinian or superdeterminant\footnote{Note that $\cdot$ in the above expression does not mean anything special and we inserted it for clarity of expressions.}. We thus conclude that $D$ is a globally-defined operator and $D_\theta$ is a section of the line bundle dual to the canonical line bundle.

\newl The generic solutions of \eqref{eq:equation of superconformal coordinate transformation} are given by \cite{RoslySchwarzVoronov1988a}
\begin{alignat}{3}\label{eq:solution of superconformal coordinate transformation}
z'&=f(z)+\theta (\partial_zf(z))\varepsilon(z),\nonumber
\\
\theta'&=\left(\partial_zf(z)\right)^{\frac{1}{2}}\left(\theta+\varepsilon(z)+\frac{1}{2}\theta\varepsilon(z)\partial_z\varepsilon(z)\right),
\end{alignat}
where $f(z)$ and $\varepsilon(z)$ are commutating and anticommuting functions and both of which are holomorphic functions of $z$. In the absence of odd moduli for the super-Riemann surface, i.e. $\varepsilon(z)=0$, the transition functions \eqref{eq:solution of superconformal coordinate transformation} reduce to
\begin{alignat}{2}\label{eq:transition function split super-Riemann surface}
z'&=f(z),\nonumber
\\
\theta'&=\left(\partial_zf(z)\right)^{\frac{1}{2}}\theta,
\end{alignat}
on the intersection of two superconformal patches. The first equation is the transition function between two overlapping charts on an ordinary Riemann surface $\mcal{R}$, i.e. a conformal transformation, and the second is the coordinate transformation of the fibers of a fiber bundle on $\mcal{R}$. Together, they are transition functions of the total space of a spinor bundle $\mscr{E}$ over $\mcal{R}$. The choice of square root in the transition function of $\theta$ corresponds to the choice of spin structure. For a mathematically-precise definition of spin curves see Definitions \ref{Defn_SpinCurve} and \ref{Defn_Marked_RelSpinCurve}. 

\newl By going around each 1-cycle of $\mcal{R}$, $\theta$ can be either periodic or anti-periodic. For a genus-$\g$ surface $\mcal{R}$, since $\dim (H_1(\mcal{R},\mbb{Z}))=2\g$, there are $2\g$ possible non-trivial 1-cycles, and therefore $2^{2\sg}$ possible choices of signs for the transition functions of $\theta$, i.e. there are $2^{2\sg}$ spin structures. We thus conclude that {\it a genus-$\g$ super-Riemann surface with vanishing odd moduli is an ordinary genus-$\g$ Riemann surface. Since the transition function for $\theta$ depends on the square root of $f(z)$, we get a genus-$\g$ spin curve}. This discussion shows that $\mscr{S}_{\sg}$, the moduli stack of genus-$\g$ spin curves is the reduced space\footnote{The reduced space of a supermanifold is an ordinary manifold which can be obtained by putting equal to zero all of the odd coordinates and all (possible) odd moduli of the supermanifold. For more on this see \cite{Witten2012a}.} of $\mscr{SM}_{\sg}$, the superstack of genus-$\g$ $\mscr{N}=1$  super-Riemann surfaces. It is also a $2^{2\sg}$-fold covering of $\mscr{M}_{\sg}$, the moduli stack of ordinary genus-$\g$ Riemann surfaces. We thus have the following 
\begin{equation}
\mscr{SM}_{\sg}\xhookleftarrow{\iota_{\mscr{S}}} \mscr{S}_{\sg}\xlongrightarrow{\pi_\mscr{M}} \mscr{M}_{\sg},
\end{equation}
where $\iota_{\mscr{S}}$ is a closed embedding \cite{DonagiWitten2013a}, and $\pi_{\mscr{M}}$ is a covering map of degree $2^{2\sg}$. There is a similar diagram for the punctured spin curves and compactified moduli stacks. In the construction of the compactification divisor of $\overline{\mscr{S}}_{\sg}$ (or $\overline{\mscr{S}}_{\sg,\sns,\sra}$ in the presence of NS and R punctures) using plumbing of surfaces, one needs to deal with {\it gluing of spin bundles} on the original surface(s). The question that arises is the proper way to glue spin bundles to get a well-defined notion of spin structure. It turns out that this is not a trivial task \cite{Cornalba1989a,Jarvis1994a}. We clarify some aspects of this construction in Section \ref{subsec:gluing of punctures}, and give a detailed account of the construction in Appendix \ref{app:compactification of the moduli stack of spin curves}.

\subsection{Punctures on Spin and Super-Riemann Surfaces}\label{subsec:punctures on super-Riemann surface}

In heterotic-string theory, one has to specify the boundary condition for the worldsheet fermions, which could be either periodic or anti-periodic. The inclusion of both types of boundary conditions is necessary for the unitarity of loop amplitudes \cite{GrossHarveyMartinecRohm1985b}. Since these two types of boundary conditions cannot be continuously deformed into each other, the total Hilbert space $\mscr{H}$ consists of two sectors with their own vacua: the Neveu-Schwarz (NS) sector $\mscr{H}_{\text{NS}}$ corresponding to the periodic boundary condition and the Ramond (R) sector $\mscr{H}_{\text{R}}$ corresponding to the antiperiodic boundary condition
\begin{equation}\label{eq:Hilbert space decomposition}
\mscr{H}=\mscr{H}_{\text{NS}}\oplus\mscr{H}_{\text{R}}.
\end{equation}
The external states can be from either of these two sectors. From the state-operator correspondence, there are thus two types of punctures associated to the external states \cite{Alwarez-GaumeNelsonGomezSieraVafa1988}
\begin{enumerate}
	\item[$\bullet$] {\bf\small Neveu-Schwarz or Super Punctures:} This type of puncture is just a marked point on the surface. In a local coordinate $(z|\theta)$, any point $(z|\theta)=(z_0|\theta_0)$ can be the location of an NS puncture. Any NS puncture defines a divisor $\Delta_\alpha$ given by $z=z_0+\alpha\,\theta$ and $\theta=\theta_0+\alpha$, for an anti-commuting variable $\alpha$. This divisor is the orbit through the NS puncture located at $(z_0|\theta_0)$ and is generated by $D_\theta$\footnote{Consider a point $p$ on a super-Riemann surface and an odd tangent vector $v$. Then, there is a subspace of dimension $(0|1)$ that passes through $p$ and has $v$ as the tangent vector.} \cite{Witten2012b}. There can be any number of NS puncture on an $\mscr{N}=1$ super-Riemann surface. Therefore, inserting an NS puncture on a super-Riemann surface increases the complex dimension of its superstack by $(1|1)$;
	
	\item[$\bullet$] {\bf\small Ramond or Spin Punctures:} This type of puncture is actually a singularity in the superconformal structure (and not the surface itself), and so is part of the defining data of a superconformal structure. What do we mean by this is the following. Away from an R puncture, there is a local superconformal coordinate system $(z|\theta)$ such that
	\begin{equation}
	D=[dz|d\theta]\otimes D_{\theta}=[dz|d\theta]\otimes \left(\frac{\partial}{\partial \theta}+\theta\frac{\partial}{\partial z}\right). 
	\end{equation}
	The local expression of the spinor derivative satisfies the following relation
	\begin{equation}
	D_{\theta}^2=\frac{1}{2}\{D_{\theta},D_{\theta}\}=\frac{\partial}{\partial z}.
	\end{equation}
	In particular, $D_{\theta}^2$ is linearly independent of $D_{\theta}$ and non-zero everywhere. However, near an R puncture, there is a local coordinate $(w|\eta)$ such that
	\begin{equation}
	D=[dw|d\eta]\otimes D_{\eta}=[dw|d\eta]\otimes \left(\frac{\partial}{\partial \eta}+\eta(w-w_0)\frac{\partial}{\partial w}\right),
	\end{equation} 
	which satisfies
	\begin{equation}
	D_\eta^2=\frac{1}{2}\{D_{\eta},D_{\eta}\}=(w-w_0)\frac{\partial}{\partial w},
	\end{equation}
	where $w_0$ is the location of insertion of the R puncture. This expression shows that $D_{\theta}^2$ vanishes at $w-w_0=0$, i.e. along the divisor defined by the location of the R puncture\footnote{In general, $D_{\eta}^2=(w-w_0)^n\frac{\partial}{\partial w}$ for $n\ge 1$ near a puncture. Such a surface is called a super-Riemann surface with level-$n$ parabolic structure. The case of $n=1$ corresponds to the superconformal structure in the presence of Ramond puncture.}. A possible change of coordinates is
	\begin{equation}
	(w|\eta)\,\longrightarrow\, (z|\theta)=(w-w_0|(w-w_0)^{1/2}\eta),
	\end{equation}
	where $D$ retain its usual form as \eqref{eq:super-Dolbeault operator}. However, this change of coordinate introduces cuts between pairs of R punctures. Since by going around an R punctures $\theta\longrightarrow -\theta$, there can not be a single R puncture, i.e. R punctures come in pairs\footnote{For another explanation see section $4.2.2$ of \cite{Witten2012b}.}. Each pair of R punctures contributes two even and two odd moduli to the complex dimension of superstack. However, one of the odd moduli of a pair of R punctures can be removed by the conformal-Killing spinor of a disk with two R punctures \cite{Giddings1992a}. Therefore, inserting a pair of R punctures on a super-Riemann surface increases the complex dimension of its superstack by $(2|1)$. 
\end{enumerate}

On a spin curve, the punctures are given by a point on the surface and a choice of $D_{\theta}$. Just as the case of super-Riemann surfaces, the choice of $D_{\theta}$ depends on whether the puncture is of NS or R type. We now turn to the gluing of punctures on a spin- or an $\mscr{N}=1$ super-Riemann surface.

\subsection{Gluing of Punctures on Spin and Super-Riemann Surfaces}\label{subsec:gluing of punctures}
In this section, we describe the gluing of two punctures on a spin or an $\mscr{N}=1$ super-Riemann surface. Intuitively, the gluing corresponds to propagation of states between the glued states, just like a Feynman diagram. The direct sum structure of the decomposition of the total Hilbert space \eqref{eq:Hilbert space decomposition} shows that there can not be propagation between a state from NS sector and a state from R sector, since otherwise there can be mixing of states. We can thus only glue two punctures associated to states from the same sector of the theory, i.e. we can glue either two NS punctures or two R punctures. On the other hand, the gluing can be done either on a single surface or on two separate surfaces. If we consider the underlying spin surfaces of the component super-Riemann surfaces, the surfaces are endowed with spin structures, the gluing relations are different for the two types of punctures. When two punctures are glued, the corresponding spin bundles are glued as well. For a spin curve, we thus need a gluing relation that describes the gluing of surfaces as we as gluing of spinor bundles.  

\newl Consider two punctures on a single or two disconnected $\mscr{N}=1$ super-Riemann surfaces. The neighborhoods of these punctures can be described by local superconformal coordinates $(z|\theta)$ and $(w|\eta)$. The punctures are located at $(z_0|\theta_0)$ and $(w_0|\eta_0)$. In terms of these coordinates, the gluing relations can be described as follows
\begin{enumerate}
	\item[$\bullet$] {\bf\small Gluing of two NS Punctures:} To describe the gluing relation, we define the following coordinates
	\begin{alignat}{2}
	x_1&\equiv z-z_0-\theta\theta_0, \qquad & \qquad \widehat{x}_1&\equiv \theta-\theta_0,\nonumber
	\\
	x_2&\equiv w-w_0-\eta\eta_0, \qquad & \qquad \widehat{x}_2&\equiv \eta-\eta_0.
	\end{alignat}
	Two NS punctures can then be glued via the following superconformal mapping \cite{Cohn1988a}
	\begin{alignat}{1}\label{eq:gluing relation NS punctures I}
	x_1x_2+t^2&=0,	\nonumber
	\\
	x_1\widehat{x}_2-t\widehat{x}_1&=0, \nonumber
	\\
	x_2\widehat{x}_1+t\widehat{x}_2&=0,  \nonumber
	\\
	\widehat{x}_1\widehat{x}_2&=0.
	\end{alignat}
	$t$ is a complex parameter. For non-zero $t$, the two middle equations are equivalent which can be seen by multiplying the second equation by $x_2$ (or the third equation by $x_1$) and using the first equation. However, these equations are independent for $t=0$. One usually is interested in the case that the punctures are located at $(z_0|\theta_0)=(0|0)=(w_0|\eta_0)$ in the respective local coordinates. In this case, the gluing relations can be simply written as
	\begin{alignat}{1} \label{eq:gluing relation NS punctures II}
	zw+t^2&=0,\nonumber
	\\
	z\eta-t\theta&=0,\nonumber
	\\
	w\theta+t\eta&=0, \nonumber
	\\
	\theta\eta&=0.
	\end{alignat}
	To connect with the usual notation in the literature, one defines $q_{\text{NS}}\equiv-t^2$. These relations reduce to the plumbing of punctures on ordinary Riemann surfaces by setting $\theta=\eta=0$ and $q\equiv q_{\text{NS}}$. The one-parameter family of glued surfaces are $\mscr{N}=1$ super-Riemann surfaces parametrized by the even and odd moduli of the component surface(s) and the extra complex parameter $t$. It is easy to check that the resulting surfaces have correct number of even and odd moduli. 
	
	\item[$\bullet$] {\bf\small Gluing of two R Punctures:}
	Two R punctures can be glued using the following relations \cite{Witten2012b}
	\begin{alignat}{1}\label{eq:gluing relation R punctures I}
	(z-z_0)(w-w_0)+t^2&=0, \nonumber
	\\
	(\theta-\theta_0)+\alpha\mp\mfk{i}\,(\eta-\eta_0)&=0,
	\end{alignat}
	where $\alpha$ is an odd gluing parameter. The origin of this parameter is an extra symmetry $\theta \longrightarrow \theta+\alpha$ of the divisor defined by an R puncture. We will explain the idea behind this extra odd symmetry in section \ref{subsubsec:gluing of R punctures}. Putting the punctures at $(z_0|\theta_0)=(0|0)=(w_0|\eta_0)$ in the respective local coordinates, these relations become
	\begin{alignat}{1} \label{eq:gluing relation R punctures II}
	zw+t^2&=0, \nonumber
	\\
	\theta+\alpha \mp \mfk{i}\,\eta&=0.
	\end{alignat}
	One can connect these to the notation in the literature by defining $q_{\text{R}}\equiv-t^2$. Again, these relations reduce to the plumbing of punctures on ordinary Riemann surfaces by setting $\theta=\eta=0$ and $q\equiv q_{\text{R}}$. The two-parameter family of glued surfaces are $\mscr{N}=1$ super-Riemann surfaces parametrized by the even and odd moduli of the component surface(s) and the extra complex parameters $t$ and $\alpha$\footnote{As we mentioned, the odd gluing parameter $\alpha$ enters the gluing relations due to an extra symmetry of the divisors defined by the glued R punctures, as we discuss in section \ref{subsec:subtlety in the gluing of R punctures}. However, there are special cases that $\alpha$ cannot be considered as a modulus. Essentially, these are the cases that the component surfaces have superconformal automorphisms that acts on the divisors associated to the glued punctures by $\theta\longrightarrow\theta+\alpha$ (or $\eta\longrightarrow\eta+\alpha$). These automorphisms can then be used to remove the extra symmetry of the R punctures. In such cases, there is no odd gluing parameter in the gluing relations of R punctures and can be studied separately.}. It is easy to check that the resulting surfaces have correct number of even and odd moduli.
\end{enumerate}

\newl {\bf\small Remark:} The gluing of two R punctures can also be described as follows  \cite{Cohn1988a}. We first define the following coordinates
\begin{equation}
y_1^2\equiv -(z-z_0), \qquad \qquad y_2^2\equiv +(w-w_0).
\end{equation}
Two R punctures can then be glued via the following superconformal mapping
\begin{alignat}{1}\label{eq:gluing relation R punctures III}
\left(y_1\mp\frac{1}{4}\theta\,\widehat{t}\right)\left(y_2-\frac{1}{4}\eta\,\widehat{t}\right)-t&=0,	\nonumber
\\
\mp\theta y_2+\eta y_1+t\,\widehat{t}&=0.
\end{alignat}
$\widehat{t}$ is an odd gluing parameter. In this paper, we only use the gluing relation \eqref{eq:gluing relation R punctures II} for the gluing of two R punctures.

\newl The gluing relations of both NS punctures \eqref{eq:gluing relation NS punctures II} and R punctures \eqref{eq:gluing relation R punctures II} are two-valued. This is related to the imposition of GSO projection by summing over different spin structures. For more explanation see section $6.2.3$ of \cite{Witten2012b}. 

\newl We were interested in the gluing of two R punctures on spin-Riemann surfaces. {\it How do the gluing relations change when we project super-Riemann surfaces to their underlying spin curves?} The gluing relations of NS punctures remains the same. However, the gluing relations of R punctures will change a bit. The reason is that the symmetry of the divisor defined by an R punctures is lost by projecting to the underlying spin curve, i.e. there is no parameter $\alpha$. Therefore, we use the following gluing relations of two R punctures on an spin-Riemann surfaces become
\begin{alignat}{1} \label{eq:gluing relation R punctures on spin curves}
zw+t^2&=0, \nonumber
\\
\theta\pm \mfk{i}\,\eta&=0.
\end{alignat}
As usual, the gluing of even coordinates is the usual plumbing fixture of ordinary Riemann surfaces. The gluing of odd coordinates is interpreted as the gluing of the corresponding spin bundles which induces the spin bundle, i.e. the choice of spin structure, on the resulting surface. For a mathematically-precise description of gluing of punctures on spin curves see Section \ref{app: gluing marked points}.

\section{Nodal Spin Curves and Torsion-Free Sheaves}\label{sec:gluing of punctures on spin curves and torsion-free sheaves}
In this section, we consider the gluing of two punctures of the same kind on a single or two disconnected spin curves using plumbing prescription given in section \ref{subsec:gluing of punctures}. We describe how the gluing provides a proper compactification of the moduli stack of smooth spin curves. This means that starting from the component surface(s)\footnote{We write surface(s) to include the situation that the glued punctures are located on a single surface.} equipped with spin structures, one can construct a family of degenerating surfaces equipped with spin structures such that 1) spin structures are induced from the spin structures of the component surface(s), 2) their moduli can be described in terms of the moduli of component surface(s), and 3) they properly compactify the moduli stack of irreducible spin curves as constructed in \cite{Cornalba1989a, Jarvis1994a,Jarvis1998a,Jarvis2000a}. 

\subsection{Degenerations of Spin Curves}\label{subsec:gluing of punctures on spin curves}
The plumbing fixture construction of a family of degenerating spin curves can be applied to either a single surface or two (and more) surfaces. To be concrete, we consider only the gluing of two surfaces. All considerations are applicable for the other types of gluing. 

\newl Consider two split surfaces $\mcal{R}(\mbs{m},\mathbf{0})$ and $\mcal{R}'(\mbs{m}',\mathbf{0})$, where $\mbs{m}$ and $\mbs{m}'$ denote the even moduli of the surfaces. $\mcal{R}$ has signature $(\g;\ns,\ra)$ and is equipped with a spin structure $\mscr{E}$ and $\mcal{R}'$ has signature $(\g';\ns',\ra')$ and is equipped with a spin structure $\mscr{E}'$. To glue punctures, we consider the local coordinate around the glued punctures on $\mcal{R}$ and $\mcal{R}'$ to be $z$ and $z'$, respectively. Since the surfaces are equipped with spin structure, we also consider the coordinates in the fiber of the respective spinor bundles to be $\theta$ and $\theta'$. These coordinate define the following spinor derivatives acting on the respective spinor bundles
\begin{equation}
D_{\theta}=\frac{\partial}{\partial \theta}+\theta\frac{\partial}{\partial z}, \qquad D_{\theta'}=\frac{\partial}{\partial \theta'}+\theta'\frac{\partial}{\partial z'}. 
\end{equation} 
We would like to understand what is the effect of plumbing fixture on the spinor derivative. These will tells us the spinor bundle on the resulting family of surfaces parametrized by $\mbs{m}$, $\mbs{m}'$, and the gluing parameters.

\subsubsection{Neveu-Schwarz Degenerations}
Consider the NS punctures $\mfk{p}$ and $\mfk{p}'$ on $\mcal{R}$ and $\mcal{R}'$, and the gluing relation \eqref{eq:gluing relation NS punctures II} for the NS puncture. For $t\ne 0$, the independent equations are
\begin{alignat}{1}\label{eq:gluing of NS punctures for non-zero t}
zz'+t^2&=0,\nonumber
\\
z\theta'-t\theta&=0.
\end{alignat}
By putting the first equation into the second one, we get $z'\theta+t\theta'=0$, which is the third equation of \eqref{eq:gluing relation NS punctures II}. Using these relations, we have
\begin{alignat}{2}
\frac{\partial}{\partial \theta}&=\frac{\partial z'}{\partial \theta}\frac{\partial}{\partial z'}+\frac{\partial \theta'}{\partial \theta}\frac{\partial}{\partial \theta'}=-\frac{z'}{t}\frac{\partial}{\partial \theta'}, \nonumber
\\
\theta\frac{\partial}{\partial z}&=-\frac{t\theta'}{z'}\left(\frac{\partial z'}{\partial z}\frac{\partial}{\partial z'}+\frac{\partial \theta'}{\partial z}\frac{\partial}{\partial \theta'}\right)=-\frac{t\theta'}{z'}\left(\frac{z'^2}{t^2}\frac{\partial}{\partial z'}+\frac{z'\theta'}{t^2}\frac{\partial}{\partial \theta'}\right)=-\frac{z'\theta'}{t}\frac{\partial}{\partial z'}.
\end{alignat}
Note that we have used the fact that $\frac{\partial z'}{\partial \theta}=0$. This is because we use one of the last two equations of \eqref{eq:gluing relation NS punctures II}, and as such, $z'$ is independent of $\theta$. Thus, the spinor derivative $D_{\theta}$ transforms as
\begin{equation}\label{eq:NS degeneration spinor derivatives}
D_{\theta}=-\frac{z'}{t}D_{\theta'}.
\end{equation}
It is clear from this form that the spinor derivative blows up as $t\longrightarrow 0$. What does this blow-up mean in terms of the compactification of the moduli stack of spin curves as described in appendix \ref{app:compactification of the moduli stack of spin curves}?
When $t\ne 0$, the spin bundles $\mscr{E}$ and $\mscr{E}'$ are glued together in a natural way to form a spin bundle defined over the whole glued surface. However, when $t=0$, the second and third equations of \eqref{eq:gluing relation NS punctures II} become independent, and \eqref{eq:NS degeneration spinor derivatives} does not make sense. We thus have the following gluing relations at $t=0$
\begin{equation}\label{eq:gluing of NS puncture at t=0}
zz'=0, \qquad z\theta'=0,\qquad z'\theta=0.
\end{equation}
There are two possible cases
\begin{enumerate}
	\item {\it $z=0$ or $z'=0$}: since the equations \eqref{eq:gluing of NS puncture at t=0} are symmetric, let us assume that $z'=0$. The third equation of \eqref{eq:gluing of NS puncture at t=0} is trivially satisfied. Therefore, the odd parameter $\theta$ is left as a free parameter. The corresponding spinor derivative $D_{\theta}$ generates a rank-$1$ bundle with odd fibers, i.e. a spin structure. Therefore, a spinor bundle is defined over the whole glued surface. 
	\item {\it $z=z'=0$}: This case is when, by definition, a node is formed, see figure \ref{fig:an NS node}. In this case, all the equations are trivially satisfied. We thus left with two odd parameters $\theta$ and $\theta'$ whose corresponding spinor derivative $D_{\theta}$ and $D_{\theta'}$ are independent and generate a two-dimensional vector space over the node. Therefore, formation of an NS degeneration is characterized by formation of a $2$-dimensional vector space over the node. In this sense, the NS degenerations do not appear naturally in the compactification of the moduli stack of spin curves.
\end{enumerate}
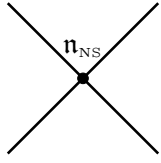
\begin{figure} \centering 
	\begin{tikzpicture}
	\draw[line width=1pt]  (-1,1) -- (1,-1);
	\draw[line width=1pt]  (1,1) -- (-1,-1);
	\draw[fill] (0,0)node[yshift=.4cm]{$\mfk{n}_{\text{\fontsize{3}{3}\selectfont NS}}$} circle (2pt); 
	\end{tikzpicture}
	\caption{In the gluing of two NS punctures, a node $\mfk{n}_{\text{\fontsize{3}{3}\selectfont NS}}$ is formed when $z=z'=0$.}
	\label{fig:an NS node} 
\end{figure}
Therefore the question is that: {\it what is the proper compactification of the moduli stack of spin curves in which {\normalfont NS} degenerations appear naturally?} As has been is explained in Appendix \ref{app:compactification of the moduli stack of spin curves}, the divisor points of a proper compactification of the moduli stack of spin curves consist of triples $(\widehat{\mcal{R}},\mscr{E},\psi)$. $\widehat{\mcal{R}}$ is a semi-stable curve with copies of $\mbb{P}^1$ intersecting the rest of the curve in at most two points\footnote{Such irreducible components of a semi-stable curve are called {\it exceptional curves}.}, and $\mscr{E}$ is a line bundle which has degree one on each $\mbb{P}^1$. The contraction of all $\mbb{P}^1$ described by the blowup map $\rho: \widehat{\mcal{R}}\longrightarrow \mcal{R}^{\bullet}$ turns $\widehat{\mcal{R}}$ to a stable curve $\mcal{R}^{\bullet}$, see figure \ref{fig:blowup geometry of NS node}. The pushforward $\rho_*\widehat{\mscr{E}}$ is a rank-one torsion-free sheaf. $\psi:(\rho_*\widehat{\mscr{E}})^{\otimes 2}\longrightarrow \omega_{\mcal{R}^{\bullet}}$, where $\omega_{\mcal{R}^{\bullet}}$ is the canonical sheaf of $\mcal{R}^\bullet$, is a homomorphism between the tensor product of $\rho_*\widehat{\mscr{E}}$ and the canonical bundle over $\mcal{R}^{\bullet}$ with the following properties, $1)$ it is an isomorphism where $\rho_*\widehat{\mscr{E}}$ is locally-free, and $2)$ its cokernel, i.e. $\omega_{\mcal{R}^{\bullet}}/\mtt{Img}\,\psi$, has length one where $\rho_*\mtt{\mscr{E}}$ is failed to be locally-free, i.e. the singular points of $\mcal{R}^{\bullet}$. 

\newl In the case of an NS node, we thus need to consider the blowup geometry. This means that we turn an NS node to a $\mbb{P}^1$ over which a line bundle $\widehat{\mscr{E}}$ lives. The pushforward of this line bundle by the blowup map $\mtt{b}:\mcal{R}^{\mtt{b}}\longrightarrow \mcal{R}^\bullet$, where $\mcal{R}^{\mtt{b}}$ denotes the blow-up geometry, is a rank-one torsion-free sheaf over the surface with node which, according to \cite{Jarvis1994a, Jarvis1998a,Jarvis2000a}, define a spin structure over the nodal curve.
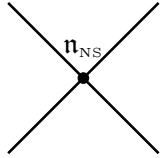
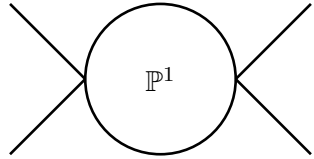
\begin{figure}[H]
	\begin{subfigure}{.45\textwidth}\centering
		\begin{tikzpicture}
		\draw[line width=1pt]  (-1,1) -- (1,-1);
		\draw[line width=1pt]  (1,1) -- (-1,-1);
		\draw[fill] (0,0)node[yshift=.4cm]{$\mfk{n}_{\text{\fontsize{3}{3}\selectfont NS}}$} circle (2pt); 
		\end{tikzpicture}
		\caption{The formation of an NS node.}
	\end{subfigure}
	~ 
	\begin{subfigure}{.45\textwidth}\centering
		\begin{tikzpicture}
		\draw[line width=1pt]  (-1,1) -- (0,0);
		\draw[line width=1pt]  (-1,-1) -- (0,0); 
		\draw[line width=1pt]  (3,1) -- (2,0);
		\draw[line width=1pt]  (3,-1) -- (2,0);
		\draw[line width=1pt] (1,0)node{$\mbb{P}^1$} circle (1cm); 
		\end{tikzpicture}
		\caption{The blowup geometry of an NS node.}
	\end{subfigure}
	\caption{The figure (a) shows the formation of an NS node when $z=z'=0$. $\theta$ and $\theta'$ are two free parameters whose corresponding spinor derivatives $D_{\theta}$ and $D_{\theta'}$ generate a rank-2 vector space over the node. The figure (b) shows the blowup geometry where a $\mbb{P}^1$ intersect the rest of the curve in two points.}
	\label{fig:blowup geometry of NS node}
\end{figure}

\subsubsection{Ramond Degenerations}\label{subsubsec:gluing of R punctures}
Consider R punctures $\mfk{q}$ and $\mfk{q}'$ on $\mcal{R}$ and $\mcal{R}'$, and the gluing relation \eqref{eq:gluing relation R punctures on spin curves} for R punctures
\begin{alignat}{1}
zz'+t^2&=0,\nonumber
\\
\theta\mp\mfk{i}\,\theta'&=0.
\end{alignat}
Let us repeat the same exercise as above. We have
{
	\begin{alignat}{2}
	\frac{\partial}{\partial \theta}&=\mp\mfk{i}\,\frac{\partial}{\partial \theta'},\nonumber
	\\
	\theta z\frac{\partial}{\partial z}&=\pm\mfk{i}\,\theta'\left(-\frac{t^2}{z'}\right)\left(\frac{z'^2}{t^2}\right)\frac{\partial}{\partial z'}=\mp\mfk{i}\,\theta'z'\frac{\partial}{\partial z'}.
	\end{alignat}
	We thus have
	\begin{equation}\label{eq:spinor derivative in the gluing of R puncture}
	D_\theta=\mp\mfk{i}\,D_{\theta'}.
	\end{equation}
	The relation shows that the bundles generated by $D_{\theta}$ and $D_{\theta'}$ 
	are isomorphic. 
	Therefore, irrespective of the values of $z$ and $z'$, there is a rank-$1$ bundle over the glued curve which equips the surface with spin structure. Unlike the NS-type degenerations, there is no need to consider the blowup geometry of the surface with node. This fact shows that R-type degenerations appear in the natural compactification of the moduli stack of spin curves.
	
	\newl We can thus provide the answer to the \hyperlink{q:first question}{\small\bf Question 1} posed in section \ref{sec:introduction}. 
	
	\newl In the gluing of NS punctures, there are two separate situation. In the moduli stack of spin curves and away from the compactification divisor, the spin structure of the component surfaces are glued together using the gluing relations \eqref{eq:gluing of NS punctures for non-zero t} for non-zero $t$ and \eqref{eq:gluing of NS puncture at t=0} for $t=0$. The spin bundle over the gluing tube is generated by either $D_{\theta}$ or $D_{\theta'}$, related by \eqref{eq:NS degeneration spinor derivatives}. However, \eqref{eq:gluing of NS puncture at t=0} also shows that exactly at the node, i.e. $z=z'=0$, there are free parameters $\theta$ and $\theta'$ whose corresponding spinor derivatives generate a rank-$2$ vector space over the node which cannot be naturally thought of as a spin bundle. The resolution is to consider the blowup geometry of the node on which a rank-one bundle is defined. The pushforward of this bundle under the blowup map is a rank-$1$ torsion-free sheaf over the node. This procedure provides a proper compactification of the moduli stack of spin curves \cite{Cornalba1989a}.  
	
	\newl In the gluing of R punctures, there is no ambiguity and the gluing relations gives a rank-one bundle over the glued surface. 
	As we mentioned above, the gluing of R punctures on spin curves rather than the parent $\mscr{N}=1$ super-Riemann surface is accompanied by a missing odd parameter. This has a consequence when one defines the superstring amplitudes in terms of spin curves and the corresponding moduli stack of such surfaces. We turn to this consequence in the next section.
	
	\subsection{PCOs and Gluing of Ramond Punctures on Spin Curves}\label{subsec:subtlety in the gluing of R punctures}
	In this section, we consider the gluing of two R punctures on the spin curves $\mcal{R}$ and $\mcal{R}'$ equipped with spin structures $\mscr{E}$ and $\mscr{E}'$. Equation \eqref{eq:spinor derivative in the gluing of R puncture} shows that
	the parameter $\alpha$ is not visible from the point of view of the spin curves. The question is then {\it what is the consequence of losing the symmetry whose parameter is $\alpha$ by studying the gluing on the underlying spin curve(s) rather than the parent $\mscr{N}=1$ super-Riemann surfaces?} To answer this question, we first need to understand the answer to the following question.
	
	\newl {\it What is the role of parameter $\alpha$ from the point of view of the gluing of super-Riemann surfaces?} As we explained in section \ref{subsec:punctures on super-Riemann surface}, the location of an R puncture defines a divisor on a super-Riemann surface. A generic even superconformal vector field which preserves the superconformal structure in the presence of an R puncture has the following form\footnote{For more details see section 4.2 of \cite{Witten2012b}.} 
	\begin{equation}
	v_{\mfk{e}}=F(z)\left(f(z)\frac{\partial}{\partial z}+\frac{f'(z)}{2}\theta\frac{\partial}{\partial \theta}\right),
	\end{equation}
	where $f(z)$ is some holomorphic function of $z$ and $F(z)$ vanishes along the divisor defined by the R puncture. It is thus clear that this vector field vanishes along that divisor. However, a generic odd superconformal vector field which also preserves the superconformal structure in the presence of an R puncture has the following form
	\begin{equation}
	v_{\mfk{o}}=g(z)\left(\frac{\partial}{\partial \theta}-F(z)\theta\frac{\partial }{\partial z}\right),
	\end{equation}
	where $g(z)$ is some holomorphic function of $z$. This vector field is in general non-vanishing along the divisor defined by the R puncture and is proportional to $\partial_{\theta}$. Such a vector field generates the transformation $\theta\longrightarrow\theta+\alpha$ which is a symmetry of the divisor. This means that when we are gluing two R punctures, the gluing is not unique (the odd coordinate of each of the punctures is defined up to addition by an odd parameter) and the gluing relation is defined up to addition by an odd parameter $\alpha$. Therefore, the divisors in the superstack defined by an R-type degeneration are fiber bundles over $\mscr{SM}_{\sg-1;\sns\sra+2}$ (for nonseparating R-type degenerations where a single genus-$\g$ surface with $\ns$ NS puncture and $\ra$ R punctures degenerates into a surface whose normalization has signature $(\g-1,\ns,\ra+2)$) and $\mscr{SM}_{\sg_1,\ns^{1},\ra^{1}}\underset{\text{\fontsize{6}{6}\selectfont R}}{\cup}\mscr{SM}_{\sg_2,\sns^2,\sra^2}$ (for separating R-type degenerations where a genus $\g=\g_1+\g_2$ surface with $\ns=\ns^1+\ns^2$ NS punctures and $\ra=\ra^1+\ra^2-2$ R punctures degenerates into two surfaces with signatures $(\g_1;\ns^1,\ra^1)$ and $(\g_2;\ns^2,\ra^2)$). The fibers of these fibration are isomorphic to $\mbb{C}^{0|1}$ parameterized by an odd parameter. On the other hand, the integration over $\alpha$, which together with the other even and odd parameters provide a coordinate chart near the divisor of $\mscr{SM}_{\sg,\sns,\sra}$, is necessary to produce the correct propagator in the Ramond sector \cite{Ramond1971a,Witten2012b}. 
	
	\newl If one starts from the supergeometry definition of the superstring amplitudes, and tries to reduce the computation to the underlying moduli stack of spin curves, one has to integrate over the odd moduli of the superstack. As we explained in section \ref{subsubsec:VV interpretation of PCO}, the effect of integration over complex odd moduli of superstack is the introduction of PCOs to the string measure on the moduli stack of spin curves. Since $\alpha$ is one of the odd moduli of the resulting surface, it has to be integrated over. The parameter $\alpha$ enters the gluing relations of two R punctures due to a freedom in the gluing and it is not associated to intrinsic odd moduli of the component surfaces. Therefore, the resulting PCO has to be added on the gluing tube using some prescription that produces the correct results. However, if one uses the picture-changing formalism as a basic definition to define the string amplitudes, the parameter $\alpha$ does not enter the computation since the symmetry of the R puncture does not exists anymore. The effect of losing this symmetry by reducing everything on the underlying spin curve is as follows. If we integrate all of the odd moduli, there will be the $2\g-2+\ns+\frac{1}{2}\ra$ PCOs in the string measure. Let us glue two R punctures on spin curves with signatures $(\g_1;\ns^1,\ra^1)$ and $(\g_2;\ns^2,\ra^2)$. A simple counting shows that the number of PCOs on $\mcal{R}\underset{\text{\fontsize{6}{6}\selectfont R}}{\cup}\mcal{R}'$, the resulting glued surface with signature $(\g_1+\g_2;\ns^1+\ns^2,\ra^1+\ra^2-2)$, is
	\begin{equation}\label{eq:number of PCOs on the glued surface}
	\text{\# of PCOs on $\mcal{R}\underset{\text{\fontsize{6}{6}\selectfont R}}{\cup}\mcal{R}'$}=2(\g_1+\g_2)-2+(\ns^1+\ns^2)+\frac{\ra^1+\ra^2-2}{2}-1.
	\end{equation}
	This is one less than the number it should be. On the other hand, to deal with IR divergences, it is required that the off-shell amplitudes to be defined in such a way that the choice of local coordinates around the punctures and the choice of PCOs be gluing-compatible on the separating-type degenerations \cite{PiusRudraSen2013a,PiusRudraSen2014a}. For PCOs, this essentially means that the choice of PCOs on the glued surface must be induced from the choice of PCOs on the component surfaces. However, \eqref{eq:number of PCOs on the glued surface} shows that the choice of PCOs with picture number $-\frac{1}{2}$ is not gluing-compatible, i.e. the number of PCOs on the resulting surface is not equal to the sum of the number of PCOs on the component surfaces. The resolution is as follows, 1) for computing the two-point function, we choose one of the R states that are being glued to have picture number $-\frac{3}{2}$, and 2) for other situations, we introduce a wighted average of PCOs inserted on a cycle on the tube joining the two glued R punctures. For more details, see sections $6.2$ and $6.3$ of \cite{Sen2014b}. In this way, the number of PCOs on the glued surface match with the one that is induced from the component surfaces.     
	
	\newl We can thus provide the answer to the \hyperlink{q:second question}{\small\bf Question 2} posed in section \ref{sec:introduction}. 
	
	\newl On an $\mscr{N}=1$ super-Riemann surface with R punctures, each of the R punctures defines a divisor with a symmetry $\theta\longrightarrow\theta+\alpha$. This symmetry is lost by projecting onto the underlying spin curve. Therefore, there is a missing odd parameter. On the other hand, the number of PCOs is equal to the fermionic dimension of the corresponding superstack. As such, there is a missing PCO when we glue two R punctures on spin curves. We need to compensate the missing PCO by adding an extra one on the gluing tube as prescribed in \cite{Sen2014b,Sen2015a}.

\section{Existence and Uniqueness of Solution to the BV QME in $\mscr{S}(\ns,\ra)$}\label{sec:the existence and uniqueness of of solution to the BV QME in the moduli space of bordered spin-Riemann surfaces}

In this section, we discuss on of the main results of this paper, i.e. 1) the existence and uniqueness of a solution to the BV QME in $\mscr{S}(\ns,\ra)$, and 2) the relation of these solutions to the fundamental class of DM stack of corresponding spin curves. For our main purposes, we introduce an alternative model for the moduli space of bordered spin-Riemann surfaces. This model is a natural generalization of the model for the moduli space of bordered ordinary Riemann surfaces introduced in \cite{KimuraStasheffVoronov9307}. After setting the stage, we prove the existence and uniqueness of a solution to the BV quantum master equation, and find that it can be mapped to the fundamental classes of DM spaces. 

\subsection{The Strategy of Proof}

Since the procedure for proving the results of this section is long, we provide a summary of the strategy of proof in this subsection. We first describe the problems in the context of bordered ordinary Riemann surfaces, and the results obtained in \cite{Costello200509}. We then briefly explain the generalization of these results to the case of bordered spin-Riemann surfaces considered in this paper. 

\newl Let us first define the Neveu-Schwarz and Ramond boundary components. For this purpose, we define a more general notion of a superconformal cycle on an $\mscr{N}=1$ super-Riemann surface \cite{Belopolsky1997b}.

\begin{definition}[Neveu-Schwarz Cycles]\label{def:NS cycles}
	A cycle\footnote{By {\it cycle} we mean a closed $C^{\infty}$ sub-supermanifold of real codimension $(1|1)$.} on an $\mscr{N}=1$ super-Riemann surface is called a Neveu-Schwarz cycle if it is isomorphic to $S^{1|1}$. If we denote the coordinates along $S^{1|1}$ by $(x|\theta)$, then  they are identified as $(x+2\pi|-\theta)\sim(x|\theta)$ on an {\normalfont NS} cycle.
\end{definition}
This definition shows that on an spin curve, a cycle with antiperiodic boundary condition is an NS cycle. Similarly, we can define NS boundary components
\begin{definition}[Neveu-Schwarz Boundary Component]\label{def:NS boundary components}
	An NS boundary component of an $\mscr{N}=1$ super-Riemann surface is a cs sub supermanifolds isomorphic to $S^{1|1}\otimes_{\mbb{R}}\mbb{C}$ of dimension\footnote{By dimension $(m|n)$ of a cs supermanifold, we mean even real dimension $m$ and odd complex dimension $n$.} $(1|1)$. 
\end{definition}

\begin{definition}[Ramond Cycles]\label{def:R cycles}
	A cycle on an $\mscr{N}=1$ super-Riemann surface is called a Ramond cycle if it is isomorphic to $S^{1}\times\mbb{R}^{0|1}$. If we denote the coordinates along $S^{1}\times\mbb{R}^{0|1}$ by $(x|\theta)$, then they are identified as $(x+2\pi|\theta)\sim (x|\theta)$ on an {\normalfont R} cycle.
\end{definition}
This definition shows that on an spin curve, a cycle with periodic boundary condition is an R cycle. Similarly, we can define R boundary components
\begin{definition}[Ramond Boundary Component]\label{def:R boundary components}
	An R boundary component of an $\mscr{N}=1$ super-Riemann surface is a cs sub supermanifolds isomorphic to $S^{1}\times\mbb{R}^{0|1}\otimes_{\mbb{R}}\mbb{C}$ of dimension $(1|1)$. 
\end{definition}

\newl One of the main result of \cite{Costello200509} is the proof of existence and uniqueness, in the sense we explain below, of the solution of the BV quantum master equation for the BV algebra associated to the moduli space of bordered ordinary Riemann surfaces. From the work of Sen and Zwiebach \cite{SenZwiebach199311,SenZwiebach199408}, this is the problem of existence and uniqueness of the closed bosonic-string vertices. Let $\mscr{C}_*$ be the functor of normalized singular simplicial chains with coefficients in any field $\mbb{F}$ containing $\mbb{Q}$, and $\mscr{M}(\n)$ be the moduli space of Riemann surfaces with $\n$ boundaries. Associated to these spaces, we can define the following complex
\begin{equation}
\mcal{F}(\mscr{M})\equiv \bigoplus_{\sn}\mscr{C}_*\left(\frac{\mscr{M}(\n)}{S^1\wr S_\sn}\right),
\end{equation}
where $S^1$ is a circle action, $S_{\sn}$ denotes the permutation group, and $S^1\wr S_\sn\equiv (S^1)^{\sn}\ltimes S_\sn$ is the wreath product group. The action of wreath product on $\mscr{M}(\n)$ turns it into the moduli space of Riemann surfaces with unlabeled boundaries. $\mcal{F}(\mscr{M})$ has the structure of a BV algebra by considering $d$ to be the boundary operator acting on chains and $\delta$ to be $\Delta$, the operation of sewing two boundaries on the same surface \cite{SenZwiebach199408,Costello200509}. It is clear that there is an $S^1$-worth possible ways of gluing. This is called the Sen-Zwiebach BV algebra \cite{SenZwiebach199408}. We denote the part of $\mcal{F}(\mscr{M})$ coming from genus-$\g$ surfaces with $\n$ boundaries as $\mcal{F}_{\sg,\sn}(\mscr{M})$. To prove the existence and uniqueness of a solution to the BV QME, Costello considered a model for $\mscr{M}(\n)$ introduced by Kimura, Stasheff, and Voronov \cite{KimuraStasheffVoronov9307}. We denote these spaces as $\widetilde{\mscr{M}}(\n)$\footnote{The model in \cite{KimuraStasheffVoronov9307} is really for the space of surfaces all whose connected components have negative Euler number.}. This is the moduli space of Riemann surfaces in $\overline{\mscr{M}}_{\sn}$, the stable-curve compactification of the moduli stack of Riemann surfaces with $\n$ punctures, decorated at each puncture with a ray in the tangent space and decorated at each node by a ray in the tensor product of tangent spaces in the two sides. $\widetilde{\mscr{M}}(\n)$ is homotopy-equivalent to $\mscr{M}(\n)$. The reason is clear: 1) each puncture is equipped with a ray, and the phase of this ray parametrizes a circle which can be thought of as a boundary, and 2) each node is equipped with a ray which can be used to open the node. Similar to $\mscr{M}(\n)$, we can associate a complex to $\widetilde{\mscr{M}}(\n)$ as follows
\begin{equation}
\mcal{F}(\widetilde{\mscr{M}})\equiv \bigoplus_{\sn}\mscr{C}_*\left(\frac{\widetilde{\mscr{M}}(\n)}{S^1\wr S_\sn}\right).
\end{equation}
The quotient by the wreath product turns $\widetilde{\mscr{M}}(\n)/S^1\wr S_\sn$ to the space of Riemann surfaces possibly with nodes, unordered marked points (without any ray in the tangent space), unparametrized boundary, but with a ray in the tensor product of tangent spaces at each side for each node. If $d$ is the boundary operation, and $\Delta$ is the operation of taking two marked points and gluing them (we give a more precise definition in Section \ref{subsec:an alternative model for space of bordered spin-Riemann surfaces}), $\mscr{F}(\widetilde{\mscr{M}})$ turns into a BV algebra. It is clear that the gluing can be done in $S^1$ possible ways. For brevity, let us define
\begin{equation}
\mbbmss{X}(\n)\equiv \frac{\widetilde{\mscr{M}}(\n)}{S^1\wr S_\sn},
\end{equation}
and denote its fundamental chain by $[\mbbmss{X}(\n)]$. After passing to homology, these classes satisfy \cite{Costello200509}
\begin{equation}\label{eq:the BV quantum master action for X(n)}
d[\mbbmss{X}(\n)]+\Delta [\mbbmss{X}(\n+2)]=0.
\end{equation}
If $\mbbmss{X}_\sg(\n)$ denotes the subspace of connected genus-$\g$ surfaces, and $[\mbbmss{X}_\sg(\n)]$, after passing to homology, its fundamental class, we can define the following formal sum
\begin{equation}
[\mbbmss{X}]\equiv \sum_{\substack{\sg,\sn=0 \\ 2\sg-2+\sn>0}}^{\infty}\hbar^{2\sg-2+\sn} [\mbbmss{X}_\sg(\n)].
\end{equation}
Using this and \eqref{eq:the BV quantum master action for X(n)}, we have
\begin{equation}
(d+\hbar\Delta)\exp\left(\frac{[\mbbmss{X}]}{\hbar}\right)=0, 
\end{equation}
i.e. $[\mbbmss{X}]$ satisfies the BV quantum master equation. A more transparent way of writing this equation is 
\begin{equation}
d[\mbbmss{X}_\sg(\n)]+\Delta[\mbbmss{X}_{\sg-1}(\n+2)]+\frac{1}{2}\sum_{\substack{\sg_1+\sg_2=\sg \\ \sn_1+\sn_2=\sn-2}}\{[\mbbmss{X}_{\sg_1}(\n_1)],[\mbbmss{X}_{\sg_2}(\n_2)]\}=0. 
\end{equation}
Using this result, we have (Proposition 10.1.1\footnote{The numbers of theorem, proposition, etc of \cite{Costello200509} that we are referring to are the ones that appear in the published version not the Arxiv preprint.} of \cite{Costello200509}) 
\begin{proposition*}[The Existence and Uniqueness of Closed Bosonic-String Vertices]
	For any $\g$ and $\n$ such that $2\g-2+\n>0$, there exist elements $\mcal{V}_{\sg}(\n)\in\mcal{F}_{\sg,\sn}(\mscr{M})$, i.e. which can be thought of as genus-$\g$ bosonic-string vertices with $\n$ punctures, of homological degree $6\g-6+2\n$ such that
	\begin{enumerate}
		\item $\mcal{V}_{0}(3)$ is the fundamental cycle of $\mscr{M}_0(3)/S_3$, i.e. a $0$-chain of degree $\frac{1}{3!}$.
		
		\item Consider the following formal power series
		\begin{equation}\label{eq:the generating function of the bosonic-string vertices}
		\mcal{V}\equiv \sum_{\substack{\sg,\sn=0 \\ 2\sg-2+\sn>0}}^{\infty}\hbar^{2\sg-2+\sn}\mcal{V}_{\sg}(\n)\in\hbar\mcal{F}(\mscr{M})[\![\hbar]\!].
		\end{equation}
		$\mcal{V}$ satisfies the BV quantum master equation 
		\begin{equation*}
		(d+\hbar\Delta)\exp\left(\frac{\mcal{V}}{\hbar}\right)=0.
		\end{equation*}
		
		\item $\mcal{V}$ is unique up to homotopy in the category of BV algebras. 
	\end{enumerate}
\end{proposition*}
At this point, one might ask why do we need to use $\widetilde{\mscr{M}}(\n)$ and its associated complex rather than $\mscr{M}(\n)$ and its associated complex? The reason is that the above proposition is proven for $\mcal{F}(\widetilde{\mscr{M}})$. However, the BV algebras $\mcal{F}(\mscr{M})$ and  $\mcal{F}(\widetilde{\mscr{M}})$ are quasi-isomorphic (see Lemma 10.4.2 of \cite{Costello200509}) which means that their corresponding homology groups are isomorphic \cite{Costello200509}. On the other hand, if two differential-graded algebras are quasi-isomorphic, their sets of homotopy classes of solutions of the Maurer-Cartan equation are isomorphic \cite{Costello200509}. The BV quantum master equation can be considered as the Maurer-Cartan equation of the corresponding BV algebra (see Lemma 5.2.1 of \cite{Costello200509}). Therefore, if two BV algebras are quasi-isomorphic, their sets of homotopy classes of solutions of the BV quantum master equation are isomorphic (see Lemma 5.3.1 and Definition 5.4.1 of \cite{Costello200509}). We can then say that once the existence and uniqueness of the solution of the BV quantum master equation in the BV algebra $\mcal{F}(\widetilde{\mscr{M}})$ is proven, we can use the quasi-isomorphism of $\mcal{F}(\widetilde{\mscr{M}})$ and $\mcal{F}(\mscr{M})$, and the above considerations, to conclude the existence and uniqueness of the solution of the BV quantum master equation, which we interpret as closed bosonic-string vertices, in the BV algebra  $\mcal{F}(\mscr{M})$.

\newl Let us now describe the analogous result for the case of bordered spin-Riemann surfaces which has been proven in Section \ref{subsec:the proof of existence and uniqueness}. Analogous to $\widetilde{\mscr M}(\n)$ which is a model for $\mscr{M}(\n)$, we introduce the space $\widetilde{\mscr{S}}(\ns,\ra)$ which is a nice model for $\mscr{S}(\ns,\ra)$. It is the moduli space of stable spin curves, decorated at each NS puncture with a ray in the tangent space, at each R puncture
with a ray in the spinor bundle, at each NS node with a ray in the tensor product of
rays of the tangent spaces at each side, and at each R node with a ray in the spinor line bundle. $\widetilde{\mscr{S}}(\ns,\ra)$ is again homotopy-equivalent to $\mscr{S}(\ns,\ra)$. The reason is similar to the one explained in the case of $\mcal{F}(\mscr{M})$ and $\mcal{F}(\widetilde{\mscr{M}})$. Associate to these spaces, and completely analogous to $\mcal{F}(\mscr{M})$ and $\mcal{F}(\widetilde{\mscr{M}})$, we define two complexes
\begin{alignat*}{2}
\mcal F(\mscr{S})&\equiv \bigoplus_{\sns,\sra} \mscr{C}_*\left(\frac{\mscr{S} (\ns,\ra)}{S^1\wr S_{\sns,\sra}}\right),
\\
\mcal F(\widetilde{\mscr{S}})&\equiv \bigoplus _{\sns,\sra} \mscr{C}_*\left(\frac{\widetilde{\mscr{S}}(\ns,\ra)}{S^1\wr S_{\sns,\sra}}\right). \numberthis
\end{alignat*}
$S^1\wr S_{\sns,\sra}$ is the wreath product $(S^1)^{\sns+\sra}\rtimes S_{\sns,\sra}$. $S_{\sns,\sra}$ is the group that permutes NS boundaries and R boundaries separately among each other. The quotient by the wreath product turns $\widetilde{\mscr{S}}(\ns,\ra)$ to the space of bordered spin-Riemann surfaces possibly with nodes, unordered marked points (without any ray in either the tangent space or the spinor bundle), unordered boundary, but with a ray in the tensor product of tangent spaces at each side for each NS node and a ray in the spinor bundle at each R node. Let $\mathcal F_{\sg,\sns,\sra}(\mscr{S})$ be the part coming from connected surfaces of genus $\g$ with $\ns$ NS boundaries and $\ra$ R boundaries. $\mcal F(\mscr{S})$ and $\mcal F(\widetilde{\mscr{S}})$ naturally carry structures of commutative differential graded
algebra, where differential is the boundary map of chain complexes and the product comes from disjoint union of surfaces. Since $\mscr{S} (\ns,\ra)$ is homotopy-equivalent to $\widetilde{\mscr{S}}(\ns,\ra)$, $\mcal F(\mscr{S})$ and $\mcal F(\widetilde{\mscr{S}})$ are quasi-isomorphic. Furthermore, we define the following gluing operation\footnote{We have used the same notation of $\Delta$ for gluing operation on the spin curves as well as ordinary Riemann surface. In this paper, we exclusively work with spin curves.} acting on surfaces in $\widetilde{\mscr{S}}(\ns,\ra)$
\begin{alignat*}{2}
\Delta &\equiv \Delta^{\text{NS}}+\Delta^{\text{R}+}+\Delta^{\text{R}-},
\\
\{\cdot,\cdot\} &\equiv \{\cdot,\cdot\}_{\text{NS}}+\{\cdot,\cdot\}_{\text{R}}.
\end{alignat*}
$\Delta^{\text{NS}}$ is the operation of gluing two NS punctures on a surface, and $\Delta^{\text{R}\pm}$ is the operation of gluing of two R punctures on a surface. We note that surfaces in $\widetilde{\mscr{S}}(\ns,\ra)$ are equipped with a ray in the spinor bundle at each R puncture. Therefore, unlike $\mscr{S}(\ns,\ra)$, there are two canonical choices for the gluing of R punctures on a single surface in $\widetilde{\mscr{S}}(\ns,\ra)$. $\Delta^{\text{R}+}$ and $\Delta^{\text{R}-}$ denote these two possible choices, i.e. two rays associated to the glued punctures are either of the same phase $\Delta^{\text{R}+}$ or of the opposite phase $\Delta^{\text{R}-}$. The operation $d+\hbar\Delta$, where $d$ is as usual the boundary map on chain complexes, turns $\mcal F(\widetilde{\mscr{S}})$ into a BV algebra. For brevity, let us define 
\begin{equation}
\mbbmss{X}(\ns,\ra)\equiv \frac{\widetilde{\mscr{S}}(\ns,\ra)}{S^1\wr S_{\sns,\sra}},
\end{equation}
and denote its fundamental chain by $[\mbbmss{X}(\ns,\ra)]$. After passing to homology, these classes satisfy (see the proof of Lemma \ref{lem:Lemma 4.1} in Section \ref{subsec:the proof of existence and uniqueness})
\begin{equation}\label{eq:the BV quantum master equation for spin curves with rays in each node}
d[\mbbmss{X}(\ns,\ra)]+\Delta^{\text{NS}}[\mbbmss{X}(\ns+2,\ra)]+\Delta^{\text{R}+}[\mbbmss{X}(\ns,\ra+2)]]+\Delta^{\text{R}-}[\mbbmss{X}(\ns,\ra+2)]=0.
\end{equation}
If $\mbbmss{X}_\sg(\ns,\ra)$ denotes the subspace of connected genus-$\g$ surfaces, and $[\mbbmss{X}_\sg(\ns,\ra)]$ its fundamental class, we can define the following formal sum\footnote{We have used the same notation for the formal sum of $[\mbbmss{X}]$ for formal power series of $[\mbbmss{X}_\vsg(\sn)]$ as well. In this paper however $[{\mbbmss{X}}]$ exclusively mean the formal sum \eqref{eq:the formal power series for spin curves with rays in each node}.}
\begin{equation}\label{eq:the formal power series for spin curves with rays in each node}
[{\mbbmss{X}}]\equiv \sum_{\substack{\sg,\sn=0 \\ 2\sg-2+\sns+\sra>0}}^{\infty}\hbar^{2\sg-2+\sns+\sra} [\mbbmss{X}_\sg(\ns,\ra)].
\end{equation}
Using this and \eqref{eq:the BV quantum master equation for spin curves with rays in each node}, we have (see Lemma \ref{lem:Lemma 4.1} in section \ref{subsec:the proof of existence and uniqueness})
\begin{equation}
(d+\hbar\Delta)\exp\left(\frac{[{\mbbmss{X}}]}{\hbar}\right)=0,
\end{equation}
i.e. $[{\mbbmss{X}}]$ satisfies the BV quantum master equation. A more transparent way to write this equation is 
\begin{alignat*}{2}
\hphantom{+}d[\mbbmss{X}_{\sg}(\ns,\ra)]&+\Delta^{\text{NS}}[\mbbmss{X}_{\sg-1}(\ns+2,\ra)]+\Delta^{\text{R}+}[\mbbmss{X}_{\sg-1}(\ns,\ra+2)]]+\Delta^{\text{R}-}[\mbbmss{X_{\sg-1}}(\ns,\ra+2)]
\\
&+\frac{1}{2}\sum_{\substack{\sg_1+\sg_2=\sg \\ \sns^1+\sns^2=\sns+2 \\ \sra^1+\sra^2=\sra}}\left\{[\mbbmss{X}_{\sg_1}(\ns^1,\ra^1)],[\mbbmss{X}_{\sg_2}(\ns^2,\ra^2)]\right\}_{\text{NS}}\\
&+\frac{1}{2}\sum_{\substack{\sg_1+\sg_2=\sg \\ \sns^1+\sns^2=\sns \\ \sra^1+\sra^2=\sra+2}}\left\{[\mbbmss{X}_{\sg_1}(\ns^1,\ra^1)],[\mbbmss{X}_{\sg_2}(\ns^2,\ra^2)]\right\}_{\text{R}}=0. \numberthis
\end{alignat*}
Using this result, we can show one of the main results of this paper (Theorem \ref{the:the existence and uniqueness of the solution to the BV QME in spin moduli} in Section \ref{subsec:the proof of existence and uniqueness})
\begin{thr*}[The Existence and Uniqueness of Closed Superstring Vertices]
	For each triple $(\g,\ns,\ra)$ with $2\g-2+\ns+\ra>0$, there exists an element ${\mcal{V}}_{\sg}(\ns,\ra)\in \mathcal{F}_{\sg,\sns,\sra}(\mscr{S})$, of homological degree $6\g-6+2\ns+2\ra$, with the following properties
	\begin{enumerate}
		\item ${\mcal{V}}_{0}(3,0)$ is the fundamental cycle of $\mscr{S}_{0}(3,0)/S_{3,0}$, i.e. 0-chain of coefficient $1/3!$.
		
		\item ${\mcal{V}}_{0}(1,2)$ is the fundamental cycle of $\mscr{S}_{0}(1,2)/S_{1,2}$, i.e. 0-chain of coefficient $1/2!$.
		
		\item The generating function 
		\begin{equation}\label{eq:the generating function of the solution to the BV QME in the moduli space of bordered spin-Riemann surfaces II}
		{\mcal{V}}\equiv\sum_{2\sg-2+\sns+\sra>0}\hbar ^{2\sg-2+\sns+\sra}{\mcal{V}}_{\sg}(\ns,\ra)\in \hbar \mathcal{F}(\mscr{S})[\![\hbar]\!],
		\end{equation}
		satisfies the BV quantum master equation
		\begin{equation*}
		(d+\hbar\Delta)\exp\left(\frac{{\mcal{V}}}{\hbar}\right)=0.
		\end{equation*}
		
		\item ${\mcal{V}}$ is unique up to homotopy in the category of BV algebras. 
	\end{enumerate}
\end{thr*}
Once again the quasi-isomorphism of $\mcal{F}(\mscr{S})$ and $\mcal{F}(\widetilde{\mscr{S}})$ make it possible to conclude the results for $\mcal{F}(\mscr{S})$. 

\newl Another main result of \cite{Costello200509}, which again we generalize to the  case of bordered spin-Riemann surfaces in Section \ref{subsec:relation of solutions to the BVQME in spin moduli and fundamental classes of DM spaces}, is the relation between the solution of the BV quantum master equation in the BV algebra  $\mcal{F}_{\sg,\sn}(\mscr{M})$ and the fundamental class of the DM stack of stable ordinary Riemann surfaces. Let $\overline{\mscr{M}}_{\sn}$ denote the moduli stack of stable Riemann surfaces with $\n$ punctures, and $\overline{\mscr{M}}_{\sg,\sn}$ denote the moduli stack of stable connected genus-$\g$ Riemann surfaces with $\n$ punctures. One can define a commutative differential-graded algebra as follows
\begin{equation}
\mcal{F}(\overline{\mscr{M}})=\bigoplus_{\sn}\mscr{C}_*\left(\frac{\overline{\mscr{M}}_{\sn}}{S_{\sn}}\right), 
\end{equation}  
and turn it into a BV algebra by setting $\Delta=0$. We then define the following formal sum
\begin{equation}
[\overline{\mscr{M}}]\equiv \sum_{\substack{\sg,\sn \\ 2\sg-2+\sn>0}}\hbar^{2\sg-2+\sn}[\overline{\mscr{M}}_{\sg,\sn}/S_{\sn}]\in\hbar \mcal{F}(\overline{\mscr{M}})[\![\hbar]\!],
\end{equation}
where $[\overline{\mscr{M}}_{\sg,\sn}/S_{\sn}]$ is the fundamental chain of $\overline{\mscr{M}}_{\sg,\sn}/S_\sn$. The following theorem provides the link between $\mcal{V}$, the solution of the BV quantum master equation in the BV algebra $\mcal{F}(\mscr{M})$ defined in \eqref{eq:the generating function of the bosonic-string vertices}, and $[\overline{\mscr{M}}]$ (Theorem 10.4.1 of \cite{Costello200509})
\begin{thr*}[Bosonic-String Vertices and Fundamental Classes of DM Stacks]
	There is a map in the homotopy category of {\normalfont BV} algebras that sends $\mcal{V}\longrightarrow [\overline{\mscr{M}}]$. This means that the bosonic-string vertices $\mcal{V}_{\sg}(\n)$ is homotopy-equivalent in the homotopy category of {\normalfont BV} algebras to an orbifold fundamental chain $[\overline{\mscr{M}}_{\sg,\sn}/S_\sn]$ of the compactified moduli stack of connected genus-$\g$ Riemann surfaces with $\n$ punctures, i.e. it is homotopy-equivalent to the whole compactified moduli stack. 
\end{thr*}
Let us present the analogous result for the case of bordered spin-Riemann surfaces. Let $\overline{\mscr{S}}_{\sns,\sra}$ denote the moduli stack of stable spin curves with $\ns$ NS punctures and $\ra$ Ramond punctures, and $\overline{\mscr{S}}_{\sg,\sns,\sra}$ denote the moduli stack of stable connected spin curves with genus $\g$, $\ns$ NS punctures and $\ra$ Ramond punctures. One again can define the following chain complex
\begin{equation}
\mcal{F}(\overline{\mscr{S}})\equiv \bigoplus_{\sns,\sra}\mscr{C}_*\left(\frac{\overline{\mscr{S}}_{\sns,\sra}}{S_{\sns,\sra}}\right).
\end{equation}
$S_{\sns,\sra}$ permutes the set of NS and R punctures among each other separately. After taking the homology, and setting $\Delta=0$, we have a BV algebra. We then define the following formal sum
\begin{equation}
[\overline{\mscr{S}}]\equiv\sum_{\substack{\sg,\sns,\sra \\ 2\sg-2+\sns+\sra\ge 0}}\hbar^{2\sg-2+\sns+\sra}[\overline{\mscr{S}}_{\sg,\sns,\sra}/S_{\sns,\sra}]\in\hbar  \H_*(\mcal{F}(\overline{\mscr{S}}))[\![\hbar]\!],
\end{equation}
where $[\overline{\mscr{S}}_{\sg,\sns,\sra}/S_{\sns,\sra}]$ is the fundamental class of $\overline{\mscr{S}}_{\sg,\sns,\sra}/S_{\sns,\sra}$. One can then establish the following theorem which provides a link between ${\mcal{V}}$, the solution of the BV quantum master equation in the BV algebra $\mcal{F}(\mscr{S})$ defined in \eqref{eq:the generating function of the solution to the BV QME in the moduli space of bordered spin-Riemann surfaces II}, and $[\overline{\mscr{S}}]$ (Proposition 4.5.1 in Section \ref{subsec:relation of solutions to the BVQME in spin moduli and fundamental classes of DM spaces})

\begin{thr*}[Superstring Vertices and Fundamental Classes of DM Stacks]
	There is a map in the homotopy category of {\normalfont BV} algebras that sends ${\mcal{V}}\longrightarrow [\overline{\mscr{S}}]$. This means that ${\mcal{V}}_{\sg,\sns,\sra}$ is homotopy-equivalent in the homotopy category of {\normalfont BV} algebras to the orbifold fundamental class $[\overline{\mscr{S}}_{\sg,\sns,\sra}/S_{\sns,\sra}]$ of the compactified moduli stack of connected genus-$\g$ spin curves with $\ns$ {\normalfont NS} punctures and $\ra$ {\normalfont R} punctures, i.e. it is homotopy-equivalent to the whole compactified moduli stack. 
\end{thr*}
This ends our brief summary of the results in this section. 

\subsection{Gluing of Boundary Components}\label{subsec:gluing punctures} 
In this section, we explain the gluing of boundary components on spin-Riemann surfaces. We take two bordered spin-Riemann surfaces $\mscr C_1$ and $\mscr C_2$. Suppose that $\ns^1$ and $\ns^2$, the number of NS boundary components on $\mscr C_1$ and $\mscr C_2$ respectively, are nonzero. We can then pick one NS boundary component from each curve, say that the $i_1$th NS boundary component of $\mscr C_1$ and $i_2$th NS boundary component of $\mscr C_2$, and glue them to get a new spin curve $\mscr C_1\cup \mscr C_2$, and the spin structure $\mscr{E}$ on $\mscr C_1\cup \mscr C_2$ is the direct sum of pushforward of spin structures $\mscr{E}_1$ and $\mscr{E}_2$ via the gluing map $\mscr C_1\sqcup \mscr C_2\to \mscr C_1\cup \mscr C_2$. This can be easily generalized to families of bordered spin-Riemann surfaces hence giving rise to an operation on corresponding stacks
\begin{equation}
\mathfrak m^{\text{NS}} _{i_1,i_2}:{\mscr{S}}_{\sg_1}(\ns^1,\ra^1)\times {\mscr{S}}_{\sg_2}(\ns^2,\ra^2)\to \partial{\mscr{S}}_{\sg_1+\sg_2}(\ns^1+\ns^2-2,\ra^1+\ra^2),
\end{equation}
where $\ns^i$ and $\ra^i$ for $i=1,2$ are the number of NS and R boundary components on $\mscr{C}_i$. This can be defined for a bordered spin-Riemann surfaces with at least two NS boundary components, say the $i$th and the $j$th ones, as well
\begin{equation}
\mathfrak g^{\text{NS}}_{i,j}:{\mscr{S}}_{\sg}(\ns,\ra)\to \partial{\mscr{S}}_{\sg+1}(\ns-2,\ra).
\end{equation}
Note that these maps depends on the choice of $i_1$, $i_2$, $i$, and $j$, in an equivariant way. For example, if $\sigma_{i_1,i_1'}\in S_{n_1}$ switches $i_1$ and $i_1'$ (it acts on ${\mscr{S}}_{\sg_1}(\ns^1,\ra^1)$ naturally), then we have $\mathfrak m^{\text{NS}} _{i_1',i_2}=\mathfrak m^{\text{NS}} _{i_1,i_2}\circ \sigma_{i_1,i_1'}$.

We can also define the gluing operation for a pair of R boundary components belong to two bordered spin-Riemann surfaces: the gluing of spin structures is simply choosing a square root of gluing data for the twisted canonical sheaf (since the map $b:\mscr E^{\otimes 2}\to \omega$ is non-degenerate at R-nodes, as we explained in section \ref{subsec:gluing of punctures on spin curves}). Note that the choice of the sign of square root does not matter because changing the sign corresponds to a $\mathbb Z/2$-automorphism of the spinor bundle on one of the curve (fix the spinor bundle on the other one), but this $\mathbb Z/2$-automorphism has been quotient out in the definition of the space of bordered spin-Riemann surfaces, so the resulting spin structure is canonically isomorphic to each other. The gluing of borders give rises to an operation on the corresponding moduli spaces
\begin{equation}
\mathfrak m^{\text{R}}_{j_1,j_2}:{\mscr{S}}_{\sg_1}(\ns^1,\ra^1)\times {\mscr{S}}_{\sg_2}(\ns^2,\ra^2)\to \partial{\mscr{S}}_{\sg_1+\sg_2}(\ns^1+\ns^2,\ra^1+\ra^2-2).
\end{equation}
So far there was no difference with the classical gluing property of the moduli space of bordered ordinary Riemann surfaces ${\mscr{M}}_{\sg}(\n)$. However, it turns out that there is no self-gluing operation of R boundary components\footnote{This means the gluing of two R boundary components on the same surface.}: the reason is that the $\mathbb Z/2$ ambiguity can not be removed by doing an automorphism (there is only one curve, an automorphism affects two boundary components simultaneously), i.e. there is no canonical choice of gluing data. Nevertheless, we will see in Section \ref{subsec:an alternative model for space of bordered spin-Riemann surfaces} that this ambiguity can be cured by introducing a new model for $\mscr{S}(\ns,\ra)$.

\subsection{Connected Components of $\partial\overline{\mscr{S}}_{\sg,\sns,\sra}$}
A crucial step for the proof of the existence of bosonic-string vertices used in \cite{Costello200509} is to establish a bound on the homological dimension of $\mscr{M}_{\sg,\sn}/S_{\sn}$ for $(\g,\n)\ne(0,3)$ (see the proof of Proposition 10.1.1 in \cite{Costello200509})
\begin{equation}
\H_i(\mscr{M}_{\sg,\sn}/S_{\sn},\mbb{C})=0, \qquad i\ge 6\g-7+2\n.
\end{equation}
We will prove a similar result for stacks of spin curves. We find that a similar result holds in general other than a few exceptional cases. It turns out that to establish the result, we need two know the number of connected components of $\partial \overline{\mscr{S}}_{\sg,\sns,\sra}$. Therefore, we describe the components of $\partial \overline{\mscr{S}}_{\sg,\sns,\sra}$ in this section. The main result is

\begin{thr}[Connected Components of $\partial\overline{\mscr{S}}_{\sg,\sns,\sra}$]
	For each connected component of $\overline{\mscr{S}}_{\sg,\sns,\sra}$, its boundary is connected, except for $\overline{\mscr{S}}^{\mbf m}_{0,4}$\footnote{Notice that here we have four punctures whose puncturing pattern is encoded in the vector $\mbf{m}=(m_1,m_2,m_3,m_4)$, as we explained in Section \ref{subsec:notation and remarks}. They can be either four NS punctures, two NS punctures and two R punctures, or four R punctures.} and $\overline{\mscr{S}}^{\text{ev}}_{1,1,0}$, in which cases the number of boundary components are 3 and 2, respectively. \label{the:connected components of stacks of spin curves}
\end{thr}

\begin{proof}
	We divide the proof into different cases. 
	
	\newl \textbf{The Case $\g=0$:} Since $\overline{\mscr{S}}_{0,\sns,\sra}\cong \overline{\mscr{M}}_{0,\sns+\sra}$, and the result for the latter is classical \cite{DeligneMumford1969a}.
	
	\newl For other cases, notice that the classical result of connectedness of $\partial\overline{\mscr{M}}_{\sg,\sn}$, together with the finiteness and flatness of the morphism $f:\overline{\mscr{S}}_{\sg,\sns,\sra}\longrightarrow\overline{\mscr{M}}_{\sg,\sns+\sra}$, imply that every connected component of $\partial \overline{\mscr{S}}_{\sg,\sns,\sra}$ maps surjectively to $\partial\overline{\mscr{M}}_{\sg,\sns+\sra}$. We make use of this fact in the following way: if we can show that for some connected substack $\mathcal S\subset \partial\overline{\mscr{M}}_{\sg,\sns+\sra}$, the preimage $f^{-1}(\mathcal S)$ lies in the same connected component of $\partial \overline{\mscr{S}}_{\sg,\sns,\sra}$, then it follows that $\partial \overline{\mscr{S}}_{\sg,\sns,\sra}$ is connected.

	\newl \textbf{The Case $\ra\neq 0$:} This separates into 2 subcases:
	\begin{enumerate}
		\item[•] $\g\ge 2$: Consider the locus $\mathcal S$ of $\partial\overline{\mscr{M}}_{\sg,\sns+\sra}$ consisting of gluing an elliptic curve with a genus $\g-1$ curve at a node, with one punctures on the elliptic curve and all other punctures on the other component. We claim that every connected component of the preimage $f^{-1}(\mathcal S)$ can be connected to the component representing the configuration shown in Figure \ref{fig:a genus-(g-1) and a genus-1 degeneration configuration}, i.e. one R puncture on the elliptic curve and all other punctures are on the other component. This can be seen from the degeneration pattern of Figure \ref{fig:the degeneration pattern for genus-1 and genus g-1 joined at a node}. 
		\begin{figure}[]\centering
			\begin{tikzpicture}[scale=.6]
			\filldraw[line width=1pt,rotate=-60] (-4,0) -- (2,0) node[below]{$\g-1$}; 
			
			\filldraw[red, line width=1pt,rotate=-60] (-.5,0) circle (3pt);
			\filldraw[red, line width=1pt,rotate=-60] (.5,0) circle (3pt);
			\filldraw[red, line width=1pt,rotate=-60] (1.5,0) circle (3pt);
			
			\filldraw[line width=1pt, rotate=-60,Cerulean] (-1.5,0) circle (3pt);
			\filldraw[line width=1pt, rotate=-60,Cerulean] (-3,0) circle (3pt);

			\begin{scope}[xscale=-1,xshift=2.5cm]
			\filldraw[line width=1pt,rotate=-60] (-4,0) -- (2,0) node[below]{$1$}; 
			
			\filldraw[line width=1pt, rotate=-60,Cerulean] (-1,0) circle (3pt); 
			\end{scope}
			\end{tikzpicture}
			\caption{A surface with a node whose connected components are an elliptic curve with a single R puncture, and a genus-$(\g-1)$ curve with $\ra$ R punctures and $\ns$ NS punctures. The red dots denote NS punctures and blue dots denote R punctures.}\label{fig:a genus-(g-1) and a genus-1 degeneration configuration}
		\end{figure}
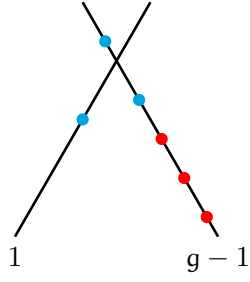
		\begin{figure}[]\centering
			\hspace*{-0.1\linewidth}
			\begin{tikzpicture}[scale=.35]
			
			
			\draw[line width=1pt,rotate=-60] (-5,0) -- (3,0) node[below]{$\g-1$}; 
			
			\filldraw[red, line width=1pt,rotate=-60] (-.5,0) circle (3pt);
			\filldraw[red, line width=1pt,rotate=-60] (.5,0) circle (3pt);
			\filldraw[red, line width=1pt,rotate=-60] (1.5,0) circle (3pt);
			
			\filldraw[line width=1pt, rotate=-60,Cerulean] (-1.5,0) circle (3pt);

			\begin{scope}[xscale=-1,xshift=2.5cm]
			\filldraw[line width=1pt,rotate=-60] (-5,0) -- (3,0) node[below]{$1$}; 
			
			\filldraw[line width=1pt, rotate=-60,Cerulean] (0,0) circle (3pt);
			\filldraw[line width=1pt, rotate=-60,red] (+1,0) circle (3pt);  
			\end{scope}
			
			
			\draw[->,line width=1pt] (4,4) -- node [rotate=23, above] {\text{degeneration}} (8.5,6); 
			
			\draw[->,line width=1pt] (4,-2) -- node [rotate=-23, below] {\text{degeneration}} (8.5,-4); 
			
			
			\draw[line width=1pt,,rotate=-60] (8,10) -- (15,10) node[below]{$0$};

			\filldraw[red, line width=1pt,rotate=-60] (11.5,10) circle (3pt);

			\begin{scope}[xscale=-1,xshift=-27cm]
			\draw[line width=1pt,rotate=-60] (8,10) -- (15,10) node[below]{$1$}; 
			
			\filldraw[line width=1pt, rotate=-60,Cerulean] (13,10) circle (3pt);
			
			\end{scope}
			
			\begin{scope}[xscale=-1,xshift=-30cm]
			\draw[line width=1pt,rotate=-60] (8,10) -- (15,10) node[below]{$\g-1$}; 
			
			\filldraw[red, line width=1pt,rotate=-60] (8.5,10) circle (3pt);
			\filldraw[red, line width=1pt,rotate=-60] (9.5,10) circle (3pt);
			\filldraw[red, line width=1pt,rotate=-60] (10.5,10) circle (3pt);
			\filldraw[Cerulean, line width=1pt,rotate=-60] (11.5,10) circle (3pt);
			\end{scope}
			
			
			\begin{scope}[yshift=12cm]

			\draw[line width=1pt,,rotate=-60] (8,10) -- (15,10) node[below]{$0$};

			\filldraw[Cerulean, line width=1pt,rotate=-60] (11.5,10) circle (3pt);

			\begin{scope}[xscale=-1,xshift=-27cm]
			\draw[line width=1pt,rotate=-60] (8,10) -- (15,10) node[below]{$1$}; 
			
			\filldraw[line width=1pt, rotate=-60,red] (13,10) circle (3pt);
			
			\end{scope}
			
			\begin{scope}[xscale=-1,xshift=-30cm]
			\draw[line width=1pt,rotate=-60] (8,10) -- (15,10) node[below]{$\g-1$}; 
			
			\filldraw[red, line width=1pt,rotate=-60] (8.5,10) circle (3pt);
			\filldraw[red, line width=1pt,rotate=-60] (9.5,10) circle (3pt);
			\filldraw[red, line width=1pt,rotate=-60] (10.5,10) circle (3pt);
			\filldraw[Cerulean, line width=1pt,rotate=-60] (11.5,10) circle (3pt);
			
			\end{scope}
			\end{scope}

			
			\draw[<-,line width=1pt] (20.5,6) -- node [ above] {\text{degeneration}} (25.5,6); 
			
			\draw[<-,line width=1pt] (20.5,-5) -- node [ above] {\text{degeneration}} (25.5,-5); 
			
			
			\begin{scope}[xshift=32cm,yshift=5.8cm]
			\draw[line width=1pt,rotate=-60] (-5,0) -- (3,0) node[below]{$\g-1$}; 
			
			\filldraw[red, line width=1pt,rotate=-60] (.5,0) circle (3pt);
			\filldraw[red, line width=1pt,rotate=-60] (1.5,0) circle (3pt);
			\filldraw[red, line width=1pt,rotate=-60] (2.5,0) circle (3pt);

			\filldraw[line width=1pt, rotate=-60,Cerulean] (-1.5,0) circle (3pt);
			\filldraw[line width=1pt, rotate=-60,Cerulean] (-.5,0) circle (3pt);
			
			\begin{scope}[xscale=-1,xshift=2.5cm]
			\filldraw[line width=1pt,rotate=-60] (-5,0) -- (3,0) node[below]{$1$}; 
			
			\filldraw[line width=1pt, rotate=-60,red] (0,0) circle (3pt);
			\end{scope}
			\end{scope}
			
			
			\begin{scope}[xshift=32cm,yshift=-5.8cm]
			\draw[line width=1pt,rotate=-60] (-5,0) -- (3,0) node[below]{$\g-1$}; 
			
			\filldraw[red, line width=1pt,rotate=-60] (.5,0) circle (3pt);
			\filldraw[red, line width=1pt,rotate=-60] (1.5,0) circle (3pt);
			\filldraw[red, line width=1pt,rotate=-60] (2.5,0) circle (3pt);
			
			\filldraw[line width=1pt, rotate=-60,Cerulean] (-1.5,0) circle (3pt);

			\begin{scope}[xscale=-1,xshift=2.5cm]
			\filldraw[line width=1pt,rotate=-60] (-5,0) -- (3,0) node[below]{$1$}; 
			
			\filldraw[line width=1pt, rotate=-60,Cerulean] (0,0) circle (3pt);
			\end{scope}
			\end{scope}
			
			\end{tikzpicture}
			\caption{The degeneration pattern of a surface with a node consists of two connected components, a genus-$(\g-1)$ curve with a number of NS and R punctures, and a genus-$1$ surface with a single NS or R puncture. The red dots denote the NS punctures and blue dots denote the R punctures.}\label{fig:the degeneration pattern for genus-1 and genus g-1 joined at a node}
		\end{figure}
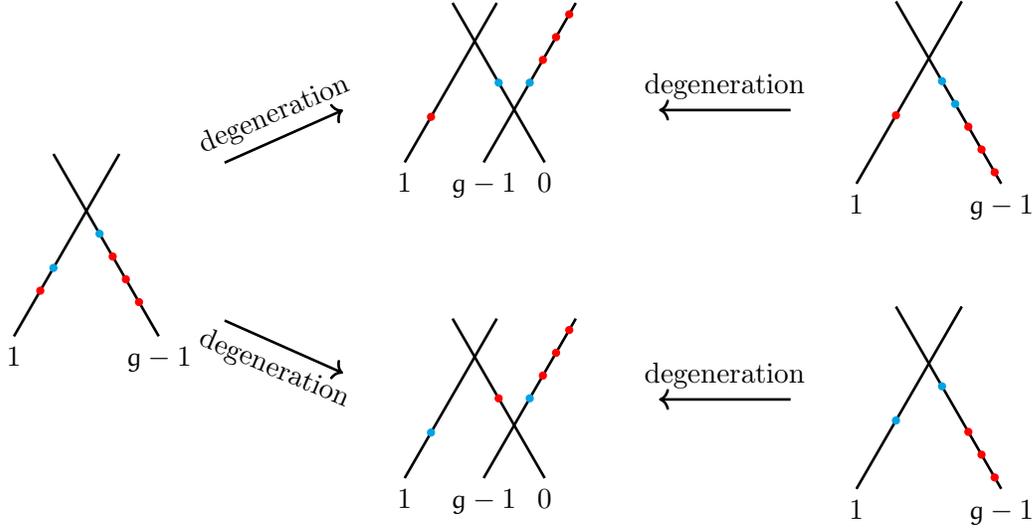

		\item[•] $\g=1$ and $\ra\ge 2$: Consider the locus $\mathcal S$ of $\partial\overline{\mscr{M}}_{1,\sns+\sra}$ consisting of gluing $\ns+\ra$ projective lines at $\ns+\ra$ node, with a punctures at each component (there is a unique such curve up to isomorphism). We claim that every connected component of the preimage $f^{-1}(\mathcal S)$ can be connected to the component representing consecutively $\ra$ components with R punctures on them. This can be seen from the degeneration pattern of figure \ref{fig:the degeneration pattern of a surface with connected components joined as a node}.
		\begin{figure}\centering
			\begin{tikzpicture}[scale=.35]
			
			
			\draw[line width= 1pt,rotate=45 ] (-4,0) -- (4,0);
			\filldraw[Cerulean, line width=1pt,rotate=45] (0,0) circle (3pt);
			\draw[line width=1pt,rotate=20] (3,4) -- (3,0);
			\draw[line width=1pt,rotate=20,dashed] (3,4) -- (3,7);

			\begin{scope}[xscale=-1,xshift=3cm]
			\draw[line width= 1pt,rotate=45 ] (-4,0) -- (4,0);
			\filldraw[red, line width=1pt,rotate=45] (0,0) circle (3pt);
			\draw[line width=1pt,rotate=20] (3,4) -- (3,0);
			\draw[line width=1pt,rotate=20,dashed] (3,4) -- (3,7);
			\end{scope}
			
			
			\draw[<-,line width= 1pt] (5,1) -- node[above]{degeneration} (10,1);
			\draw[<-,line width= 1pt] (27,1) -- node[above]{degeneration} (22,1)  ;

			
			\draw[line width= 1pt] (12,0) -- (20,0);
			\filldraw[red, line width=1pt] (15,0) circle (3pt);
			\filldraw[Cerulean, line width=1pt] (17,0) circle (3pt);
			\draw[line width=1pt,rotate=20] (12.5,0) -- (12.5,-6);
			\draw[line width=1pt,rotate=20,dashed] (12.5,0) -- (12.5,3);
			
			\begin{scope}[xscale=-1,xshift=-32cm]
			\draw[line width=1pt,rotate=20] (12.5,0) -- (12.5,-6);
			\draw[line width=1pt,rotate=20,dashed] (12.5,0) -- (12.5,3);
			\end{scope}

			
			\begin{scope}[xshift=36cm]
			\draw[line width= 1pt,rotate=45 ] (-4,0) -- (4,0);
			\filldraw[red, line width=1pt,rotate=45] (0,0) circle (3pt);
			\draw[line width=1pt,rotate=20] (3,4) -- (3,0);
			\draw[line width=1pt,rotate=20,dashed] (3,4) -- (3,7);

			\begin{scope}[xscale=-1,xshift=3cm]
			\draw[line width= 1pt,rotate=45 ] (-4,0) -- (4,0);
			\filldraw[Cerulean, line width=1pt,rotate=45] (0,0) circle (3pt);
			\draw[line width=1pt,rotate=20] (3,4) -- (3,0);
			\draw[line width=1pt,rotate=20,dashed] (3,4) -- (3,7);
			\end{scope}
			\end{scope}
			\end{tikzpicture}
			\caption{The degeneration pattern of a surface consists of multiple connected components joined at several nodes. The red dot denotes the NS puncture, and blue dots denote the R punctures.} \label{fig:the degeneration pattern of a surface with connected components joined as a node}
		\end{figure}
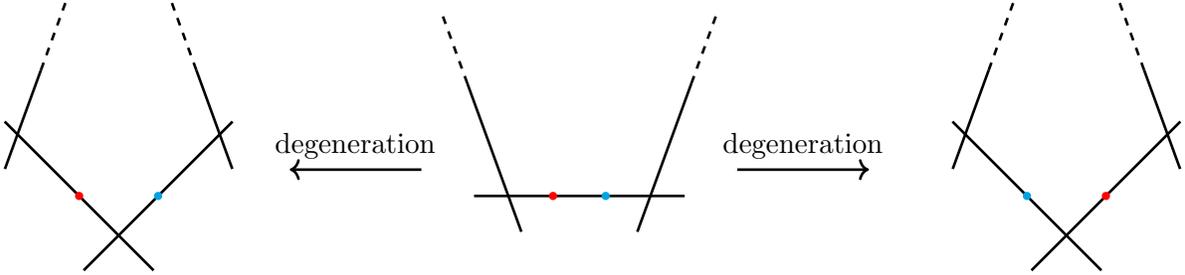
	\end{enumerate}
	
	\textbf{The Case $\ra=0$:} This separates into 6 sub-cases:
	\begin{enumerate}
		\item $\g\ge 3$:  Consider the locus $\mathcal S$ of $\partial\overline{\mscr{M}}_{\sg,\sns}$ consisting of gluing an elliptic curve with a genus-$(\g-1)$ curve at a node, with all punctures on the elliptic curve. For the even spin structure component, the preimage $f^{-1}(\mathcal S)$ consists of two connected components arising from gluing spin structures of the same parity. They are connected via degenerations shown in Figure \ref{fig:the degeneration patterns of a spin curve with a node and even spin structure}. 
		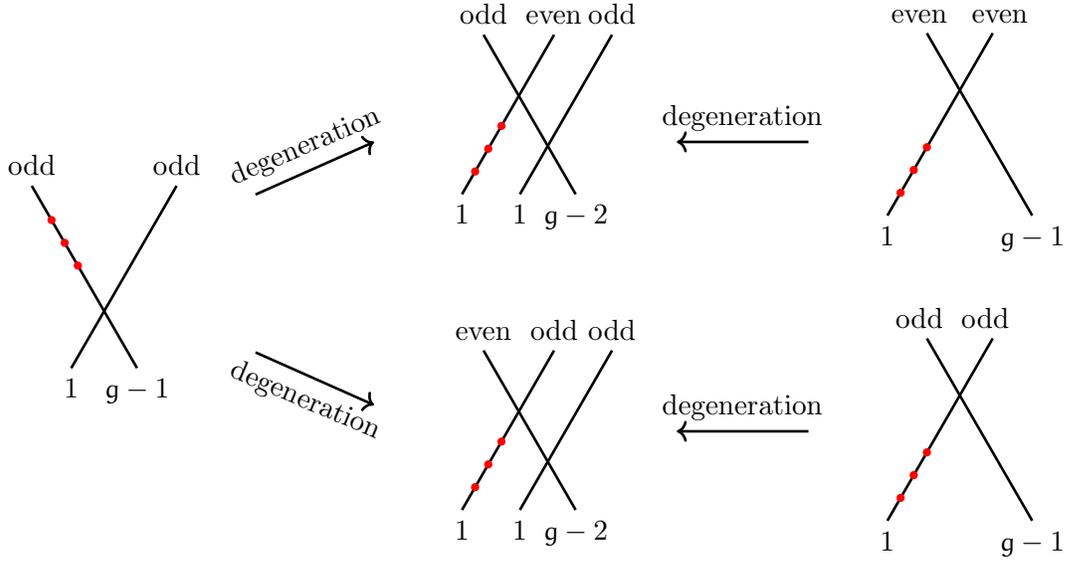
\begin{figure}[]\centering
			\hspace*{-0.1\linewidth}
			\begin{tikzpicture}[scale=.35]
			
			\begin{scope}[xshift=-2cm]
			\draw[line width=1pt,rotate=-60] (-5,0)node[above]{odd} -- (3,0) node[below]{$\g-1$}; 
			
			\filldraw[red, line width=1pt,rotate=-60] (-3.5,0) circle (3pt);
			\filldraw[red, line width=1pt,rotate=-60] (-2.5,0) circle (3pt);
			\filldraw[red, line width=1pt,rotate=-60] (-1.5,0) circle (3pt);

			\begin{scope}[xscale=-1,xshift=-.5cm]
			\filldraw[line width=1pt,rotate=-60] (-5,0)node[above]{odd} -- (3,0) node[below]{$1$};  
			\end{scope}
			\end{scope}
			
			
			\draw[->,line width=1pt] (4,4) -- node [rotate=23, above] {\text{degeneration}} (8.5,6); 
			
			\draw[->,line width=1pt] (4,-2) -- node [rotate=-23, below] {\text{degeneration}} (8.5,-4); 
			
			
			\draw[line width=1pt,,rotate=-60] (8,10)node[above]{even} -- (15,10) node[below]{$\g-2$}; 
			
			\begin{scope}[xscale=-1,xshift=-28cm]
			\draw[line width=1pt,rotate=-60] (8,10)node[above]{odd} -- (15,10) node[below]{$1$}; 
			\filldraw[red, line width=1pt,rotate=-60] (14,10) circle (3pt);
			\filldraw[red, line width=1pt,rotate=-60] (13,10) circle (3pt);
			\filldraw[red, line width=1pt,rotate=-60] (12,10) circle (3pt);
			\end{scope}
			
			\begin{scope}[xscale=-1,xshift=-30.2cm]
			\draw[line width=1pt,rotate=-60] (8,10)node[above]{odd} -- (15,10) node[below]{$1$}; 
			\end{scope}
			
			
			\begin{scope}[yshift=12cm]
			
			\draw[line width=1pt,,rotate=-60] (8,10)node[above]{odd} -- (15,10) node[below]{$\g-2$}; 
			
			\begin{scope}[xscale=-1,xshift=-28cm]
			\draw[line width=1pt,rotate=-60] (8,10)node[above]{even} -- (15,10) node[below]{$1$}; 
			\filldraw[red, line width=1pt,rotate=-60] (14,10) circle (3pt);
			\filldraw[red, line width=1pt,rotate=-60] (13,10) circle (3pt);
			\filldraw[red, line width=1pt,rotate=-60] (12,10) circle (3pt);
			\end{scope}
			
			\begin{scope}[xscale=-1,xshift=-30.2cm]
			\draw[line width=1pt,rotate=-60] (8,10)node[above]{odd} -- (15,10) node[below]{$1$}; 
			\end{scope}
			
			\end{scope}

			
			\draw[<-,line width=1pt] (20,6) -- node [ above] {\text{degeneration}} (25,6); 
			
			\draw[<-,line width=1pt] (20,-5) -- node [ above] {\text{degeneration}} (25,-5); 
			
			
			\begin{scope}[xshift=32cm,yshift=5.8cm]
			\draw[line width=1pt,rotate=-60] (-5,0)node[above,xshift=-.1cm]{even} -- (3,0) node[below]{$\g-1$};

			\begin{scope}[xscale=-1,xshift=2.5cm]
			\filldraw[line width=1pt,rotate=-60] (-5,0)node[above,xshift=.1cm]{even} -- (3,0) node[below]{$1$}; 
			
			\filldraw[red, line width=1pt,rotate=-60] (2,0) circle (3pt);
			\filldraw[red, line width=1pt,rotate=-60] (1,0) circle (3pt);
			\filldraw[red, line width=1pt,rotate=-60] (0,0) circle (3pt);
			\end{scope}
			\end{scope}
			
			
			\begin{scope}[xshift=32cm,yshift=-5.8cm]
			\draw[line width=1pt,rotate=-60] (-5,0)node[above,xshift=-.1cm]{odd} -- (3,0) node[below]{$\g-1$};

			\begin{scope}[xscale=-1,xshift=2.5cm]
			\filldraw[line width=1pt,rotate=-60] (-5,0)node[above,xshift=-.1cm]{odd} -- (3,0) node[below]{$1$}; 
			\filldraw[red, line width=1pt,rotate=-60] (2,0) circle (3pt);
			\filldraw[red, line width=1pt,rotate=-60] (1,0) circle (3pt);
			\filldraw[red, line width=1pt,rotate=-60] (0,0) circle (3pt);
			\end{scope}
			\end{scope}
			
			\end{tikzpicture}
			\caption{The degeneration patterns of a spin curve with a node and even spin structure. One of the connected components is a genus-$(\g-1)$ curve with $\ns$ NS punctures, and the other component is an elliptic curve. The red dots denote NS punctures. In the figure, $\ns=3$.}\label{fig:the degeneration patterns of a spin curve with a node and even spin structure}
		\end{figure}
		Similarly for odd spin structure component, the preimage $f^{-1}(\mathcal S)$ consists of two connected components arising from gluing spin structures of opposite parity. They are connected via degenerations shown in Figure \ref{fig:The degeneration patterns of a spin curve with a node and odd spin structure} 
		\begin{figure}\centering
			\hspace*{-0.1\linewidth}
			\begin{tikzpicture}[scale=.35]
			
			\begin{scope}[xshift=-2cm]
			\draw[line width=1pt,rotate=-60] (-5,0)node[above]{even} -- (3,0) node[below]{$\g-1$}; 
			
			\filldraw[red, line width=1pt,rotate=-60] (-3.5,0) circle (3pt);
			\filldraw[red, line width=1pt,rotate=-60] (-2.5,0) circle (3pt);
			\filldraw[red, line width=1pt,rotate=-60] (-1.5,0) circle (3pt);

			\begin{scope}[xscale=-1,xshift=-.5cm]
			\filldraw[line width=1pt,rotate=-60] (-5,0)node[above]{odd} -- (3,0) node[below]{$1$};  
			\end{scope}
			\end{scope}
			
			
			\draw[->,line width=1pt] (4,4) -- node [rotate=23, above] {\text{degeneration}} (8.5,6); 
			
			\draw[->,line width=1pt] (4,-2) -- node [rotate=-23, below] {\text{degeneration}} (8.5,-4); 
			
			
			\draw[line width=1pt,,rotate=-60] (8,10)node[above]{odd} -- (15,10) node[below]{$\g-2$}; 
			
			\begin{scope}[xscale=-1,xshift=-28cm]
			\draw[line width=1pt,rotate=-60] (8,10)node[above]{odd} -- (15,10) node[below]{$1$}; 
			\filldraw[red, line width=1pt,rotate=-60] (14,10) circle (3pt);
			\filldraw[red, line width=1pt,rotate=-60] (13,10) circle (3pt);
			\filldraw[red, line width=1pt,rotate=-60] (12,10) circle (3pt);
			\end{scope}
			
			\begin{scope}[xscale=-1,xshift=-30.2cm]
			\draw[line width=1pt,rotate=-60] (8,10)node[above]{odd} -- (15,10) node[below]{$1$}; 
			\end{scope}
			
			
			\begin{scope}[yshift=12cm]
			
			\draw[line width=1pt,,rotate=-60] (8,10)node[above]{even} -- (15,10) node[below]{$\g-2$}; 
			
			\begin{scope}[xscale=-1,xshift=-28cm]
			\draw[line width=1pt,rotate=-60] (8,10)node[above]{even} -- (15,10) node[below]{$1$}; 
			\filldraw[red, line width=1pt,rotate=-60] (14,10) circle (3pt);
			\filldraw[red, line width=1pt,rotate=-60] (13,10) circle (3pt);
			\filldraw[red, line width=1pt,rotate=-60] (12,10) circle (3pt);
			\end{scope}
			
			\begin{scope}[xscale=-1,xshift=-30.2cm]
			\draw[line width=1pt,rotate=-60] (8,10)node[above]{odd} -- (15,10) node[below]{$1$}; 
			\end{scope}
			
			\end{scope}

			
			\draw[<-,line width=1pt] (20,6) -- node [ above] {\text{degeneration}} (25,6); 
			
			\draw[<-,line width=1pt] (20,-5) -- node [ above] {\text{degeneration}} (25,-5); 
			
			
			\begin{scope}[xshift=32cm,yshift=5.8cm]
			\draw[line width=1pt,rotate=-60] (-5,0)node[above,xshift=-.1cm]{odd} -- (3,0) node[below]{$\g-1$};

			\begin{scope}[xscale=-1,xshift=2.5cm]
			\filldraw[line width=1pt,rotate=-60] (-5,0)node[above,xshift=.1cm]{even} -- (3,0) node[below]{$1$}; 
			
			\filldraw[red, line width=1pt,rotate=-60] (2,0) circle (3pt);
			\filldraw[red, line width=1pt,rotate=-60] (1,0) circle (3pt);
			\filldraw[red, line width=1pt,rotate=-60] (0,0) circle (3pt);
			\end{scope}
			\end{scope}
			
			
			\begin{scope}[xshift=32cm,yshift=-5.8cm]
			\draw[line width=1pt,rotate=-60] (-5,0)node[above,xshift=-.1cm]{even} -- (3,0) node[below]{$\g-1$};

			\begin{scope}[xscale=-1,xshift=2.5cm]
			\filldraw[line width=1pt,rotate=-60] (-5,0)node[above,xshift=-.1cm]{odd} -- (3,0) node[below]{$1$}; 
			\filldraw[red, line width=1pt,rotate=-60] (2,0) circle (3pt);
			\filldraw[red, line width=1pt,rotate=-60] (1,0) circle (3pt);
			\filldraw[red, line width=1pt,rotate=-60] (0,0) circle (3pt);
			\end{scope}
			\end{scope}
			
			\end{tikzpicture}
			\caption{The degeneration patterns of a spin curve with a node and odd spin structure. One of the connected components is a genus-$(\g-1)$ curves with $\ns$ NS punctures, and the other component is an elliptic curve. The red dots denote NS punctures. In the figure, $\ns=3$.}\label{fig:The degeneration patterns of a spin curve with a node and odd spin structure}
		\end{figure}
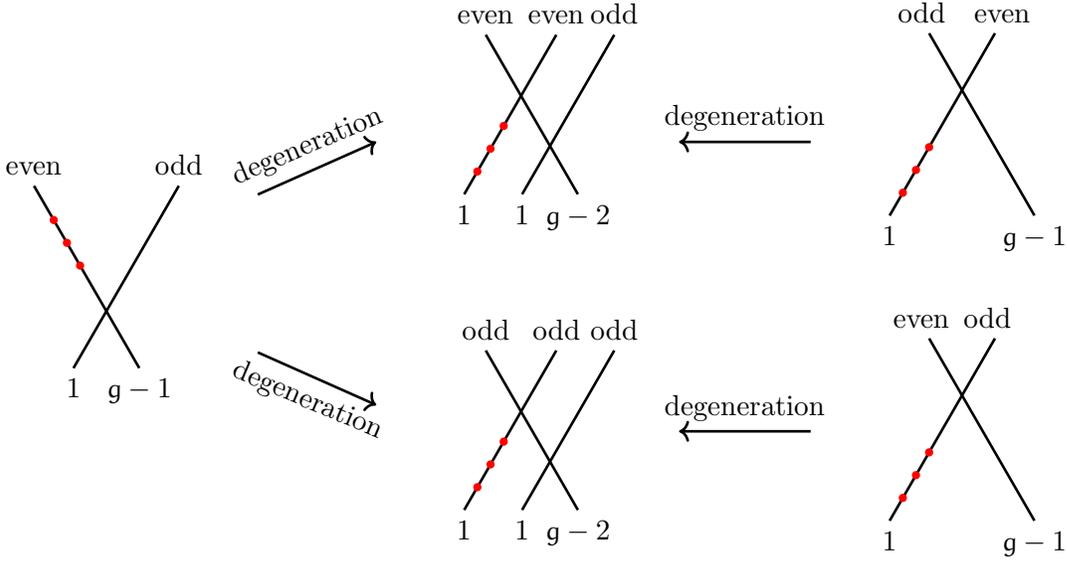
		\item $\g\ge 1$ and $\ns\ge 2$: Consider the locus $\mathcal S$ of $\partial\overline{\mscr{M}}_{\sg,\sns}$ consisting of gluing a projective line with a genus-$\g$ curve at a node, with all punctures on the projective line. It is easy to see that the preimage $f^{-1}(\mathcal S)$ restricted to each parity component is connected, since the parity of spin structure on the genus-$\g$ curve is the same as the one of the nodal curve. 
		
		\item $\g=1$ and $\ns=1$ with odd spin structure: The boundary of $\partial\overline{\mscr{M}}_{1,1}$ is the unique genus-1 nodal curve with a puncture. The preimage $\partial\overline{\mscr{M}}_{1,1}$ is connected since the only odd spin structure comes from gluing the spin structure of $\mathbb P^1$ with one NS puncture and two R punctures (note that the spin structure $\mscr E$ is isomorphic to the structure sheaf) in a way that the global section of the spinor bundle $\mscr E$ (which is one dimensional) is preserved.
		
		\item $\g=2$ and $\ns=1$ with odd spin structure: Consider the locus $\mathcal S$ of $\partial\overline{\mscr{M}}_{2,1}$ consisting of gluing two elliptic curves at a node. Let us label them by $E_1$ and $E_2$. We choose the only punctures to be on $E_2$. The preimage $f^{-1}(\mathcal S)$ consists of two connected components arising from gluing spin structures of opposite parity. They are connected via degeneration to the same boundary:
		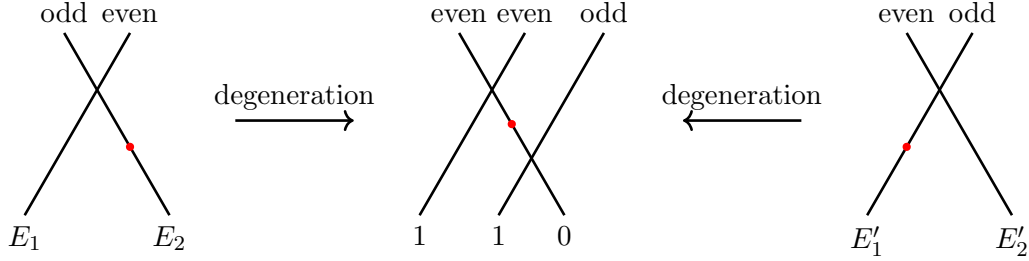
\begin{figure}[H]\centering
			\hspace*{-0.1\linewidth}
			\begin{tikzpicture}[scale=.35]
			
			
			\draw[line width=1pt,rotate=-60] (-5,0)node[above]{odd} -- (3,0) node[below]{$E_2$}; 
			
			\filldraw[red, line width=1pt,rotate=-60] (0,0) circle (3pt);

			\begin{scope}[xscale=-1,xshift=2.5cm]
			\filldraw[line width=1pt,rotate=-60] (-5,0)node[above]{even} -- (3,0) node[below]{$E_1$};
			\end{scope}
			
			
			\draw[->,line width=1pt] (4,1.) -- node [above] {\text{degeneration}} (8.5,1.); 
			
			
			\begin{scope}[xshift=15cm]
			\draw[line width=1pt,rotate=-60] (-5,0)node[above]{even} -- (3,0) node[below]{$0$}; 
			
			\filldraw[red, line width=1pt,rotate=-60] (-1,0) circle (3pt);

			\begin{scope}[xscale=-1,xshift=2.5cm]
			\filldraw[line width=1pt,rotate=-60] (-5,0)node[above]{even} -- (3,0) node[below]{$1$};
			\end{scope}
			
			\begin{scope}[xscale=-1,xshift=-.5cm]
			\filldraw[line width=1pt,rotate=-60] (-5,0)node[above]{odd} -- (3,0) node[below]{$1$};
			\end{scope}
			\end{scope}
			
			
			\draw[<-,line width=1pt] (21,1.) -- node [above] {\text{degeneration}} (25.5,1.); 
			
			
			\begin{scope}[xshift=32cm]
			\draw[line width=1pt,rotate=-60] (-5,0)node[above]{even} -- (3,0) node[below]{$E'_2$};

			\begin{scope}[xscale=-1,xshift=2.5cm]
			\filldraw[line width=1pt,rotate=-60] (-5,0)node[above]{odd} -- (3,0) node[below]{$E'_1$};
			\filldraw[red, line width=1pt,rotate=-60] (0,0) circle (3pt);
			\end{scope}
			\end{scope}
			
			\end{tikzpicture}
			\caption{The degeneration pattern of a genus-$2$ spin curve with a node and odd spin structure. The component surfaces are elliptic curves and there is an NS puncture, the red dot, on one of them.}
		\end{figure}
		\item $\g=2$ and $\ns=0$ with odd spin structure: Consider the locus $\mathcal S$ of $\partial\overline{\mscr{M}}_{2,0}$ consisting of gluing two \textit{isomorphic} elliptic curves $E_1$ and $E_2$ at a node. The preimage $f^{-1}(\mathcal S)$ is connected since spin structures on two elliptic curves must have opposite parity, i.e. either even on $E_1$ and odd on $E_2$ or the other way around, and they are isomorphic via pullback along the isomorphism of switching two components.
		
		\item $\g=2$ and $\ns=0$ or $1$ with even spin structure: This is the hardest case, we first consider the locus $\mathcal S_1$ of $\partial\overline{\mscr{M}}_{2,\sns}$ consisting of gluing two elliptic curves at a node. The preimage $f^{-1}(\mathcal S_1)$ has two connected components since spin structures on two elliptic curves must have the same parity, i.e. either even-even or odd-odd, denoted by $f^{-1}(\mathcal S_1)_{\mfk{e}}$ and $f^{-1}(\mathcal S_1)_\mfk{o}$. For a general point $p$ in $\mathcal S_1$, its preimage consists of $10$ points:  nine of which lies in $f^{-1}(\mathcal S_1)_{\mfk{e}}$ and one of which lies in $f^{-1}(\mathcal S_1)_{\mfk{o}}$. To show that $\partial \overline{\mscr{S}}_{2,\sns,0}$ is connected, it suffices to show that the unique point lying in $f^{-1}(\mathcal S_1)_{\mfk{o}}$ can be connected one of points lying in $f^{-1}(\mathcal S_1)_{\mfk{e}}$.
		
		We temporarily switch to another locus $\mathcal S_2$ of $\partial\overline{\mscr{M}}_{2,\sns}$ consisting of irreducible curves with one node. Its preimage $f^{-1}(\mathcal S_1)$ has two connected components: the one with an NS node and another one with an R node\footnote{As we explained in section \ref{subsec:gluing punctures}, there is no self contraction $\overline{\mscr{S}}_{1,\sns,2}\to\partial \overline{\mscr{S}}_{2,\sns,0}$. In fact, the connectedness of the second component follows from the construction in \ref{subsec:an alternative model for space of bordered spin-Riemann surfaces}.}. Let us focus on the component with an R node. We conclude the proof of this case by the following degeneration procedure.
		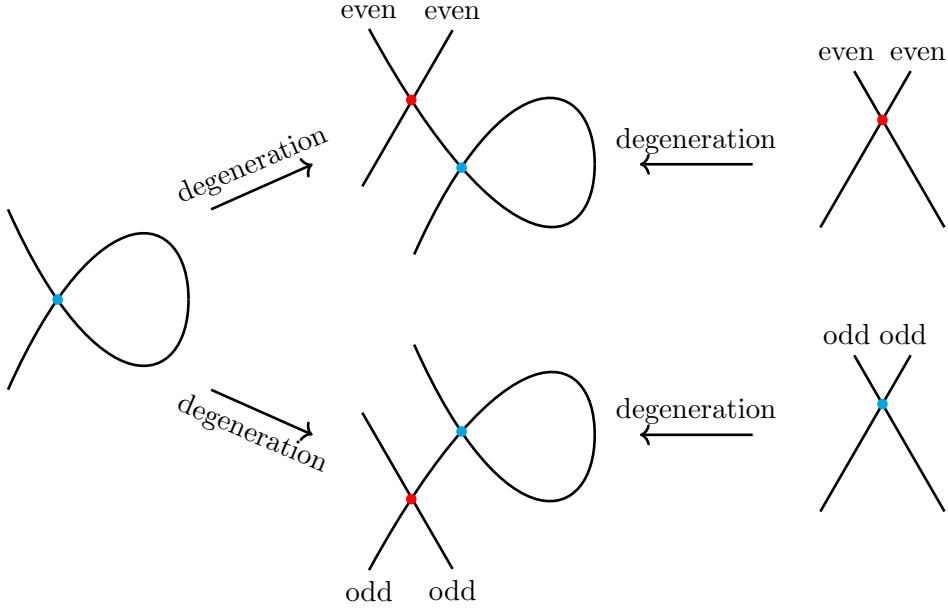
\begin{figure}[h!]\centering
			\hspace*{-0.1\linewidth}
			\begin{tikzpicture}[scale=.3]
			
			\draw[line width=1pt] (0,4) .. controls (4,-5) and (8,-4) .. (8,0);
			\begin{scope}[yscale=-1]
			\draw[line width=1pt] (0,4) .. controls (4,-5) and (8,-4) .. (8,0);
			\end{scope}

			\filldraw[Cerulean, line width=1pt] (2.2,0) circle (5pt);
			
			
			\draw[->,line width=1pt] (9,4) -- node [rotate=23, above] {\text{degeneration}} (13.5,6); 
			
			\draw[->,line width=1pt] (9,-4) -- node [rotate=-23, below] {\text{degeneration}} (13.5,-6);
			
			
			\begin{scope}[xshift=+18cm,yshift=+6cm]
			\draw[line width=1pt] (-2,6)node[above]{even} .. controls (4,-5) and (8,-4) .. (8,0);
			\begin{scope}[yscale=-1]
			\draw[line width=1pt] (0,4) .. controls (4,-5) and (8,-4) .. (8,0);
			\end{scope}
			
			\draw[line width=1pt,rotate=60,xshift=+3cm,yshift=+1.5cm] (-5,0)  -- (3,0)node[above]{even}; 
			
			\filldraw[red, line width=1pt] (-.13,2.85) circle (5pt);
			
			\filldraw[Cerulean, line width=1pt] (2.1,-.15) circle (5pt);
			\end{scope}
			
			
			\begin{scope}[xshift=+18cm,yshift=-6cm,yscale=-1]
			\draw[line width=1pt] (-2,6)node[below]{odd} .. controls (4,-5) and (8,-4) .. (8,0);
			\begin{scope}[yscale=-1]
			\draw[line width=1pt] (0,4) .. controls (4,-5) and (8,-4) .. (8,0);
			\end{scope}
			
			\draw[line width=1pt,rotate=60,xshift=+3cm,yshift=+1.5cm] (-5,0)  -- (3,0)node[below]{odd}; 
			
			\filldraw[red, line width=1pt] (-.13,2.85) circle (5pt);
			
			\filldraw[Cerulean, line width=1pt] (2.1,-.15) circle (5pt);
			\end{scope}

			
			\draw[<-,line width=1pt] (28,6) -- node [ above] {\text{degeneration}} (33,6); 
			
			\draw[<-,line width=1pt] (28,-6) -- node [ above] {\text{degeneration}} (33,-6); 
			
			
			\begin{scope}[xshift=40cm,yshift=5.8cm]
			\draw[line width=1pt,rotate=-60] (-5,0)node[above,xshift=-.1cm]{even} -- (3,0);

			\begin{scope}[xscale=-1,xshift=2.5cm]
			\filldraw[line width=1pt,rotate=-60] (-5,0)node[above,xshift=.1cm]{even} -- (3,0); 
			
			\filldraw[red, line width=1pt,rotate=-60] (-2.5,0) circle (5pt);
			\end{scope}
			\end{scope}
			
			
			\begin{scope}[xshift=40cm,yshift=-6.8cm]
			\draw[line width=1pt,rotate=-60] (-5,0)node[above,xshift=-.1cm]{odd} -- (3,0);

			\begin{scope}[xscale=-1,xshift=2.5cm]
			\filldraw[line width=1pt,rotate=-60] (-5,0)node[above,xshift=-.1cm]{odd} -- (3,0); 
			\filldraw[Cerulean, line width=1pt,rotate=-60] (-2.5,0) circle (5pt);
			\end{scope}
			\end{scope}
			\end{tikzpicture}
			\caption{The degeneration pattern of genus-$2$ curves with NS or R nodes. The red dots denote NS nodes and blue dots denote R nodes.}
		\end{figure}
	\end{enumerate}
	
	It remains to deal with the exceptional cases: $\overline{\mscr{S}}^{\text{ev}}_{1,1,0}$. It consists of two connected components, corresponding to two even spin structures associated to the unique genus-1 nodal curve with a puncture. One of these spin structures comes from the gluing of a pair of NS punctures on a $\mathbb P^1$ with three NS punctures, and the other one comes from the gluing of the spin structure of $\mathbb P^1$ with one NS puncture and two R punctures in a way that the global section of $\mscr E$ is not preserved.
\end{proof}
We can state the following corollary to the Theorem \ref{the:connected components of stacks of spin curves}
\begin{corollary}\label{cor:top minus one homology}
	Let permutation groups $S_{\sns}$ and $S_{\sra}$ act on {\normalfont NS} and {\normalfont R} punctures naturally, and define $S_{\sns,\sra}\equiv S_{\sns}\times S_{\sra}$. Then for $(\g,\ns+\ra)\neq (3,0)$ we have
	\begin{align*}
	\H_{6\sg-6+2(\sns+\sra)}({\mscr{S}}_{\sg,\sns,\sra} /S_{\sns,\sra},\mathbb C)=0,
	\end{align*}
	and
	\begin{align*}
	\H_{6\sg-7+2(\sns+\sra)}({\mscr{S}}_{\sg,\sns,\sra} /S_{\sns,\sra},\mathbb C)=\begin{cases}
	\mathbb C, & (\g,\ns,\ra)=(0,2,2)\,\,\text{\normalfont or }(1,1,0)^{\text{ev}},\\
	0, & \qquad\qquad  \text{\normalfont otherwise}.
	\end{cases}
	\end{align*}
\end{corollary}

\begin{proof}
	The top homology is trivial since ${\mscr{S}}_{\sg,\sns,\sra}$ is noncompact. For the other homology groups, we first compute the group without quotient, using the Poincar\'e duality
	\begin{equation}
	\H_{6\sg-7+2(\sns+\sra)}({\mscr{S}}_{\sg,\sns,\sra},\mathbb C)\cong \H^1 (\overline{\mscr{S}}_{\sg,\sns,\sra},\partial \overline{\mscr{S}}_{\sg,\sns,\sra},\mathbb C),
	\end{equation}
	Observe that $\pi_1(\overline{\mscr{S}}_{\sg,\sns,\sra})$ is trivial since the spin mapping-class group is generated by Dehn twists, and the boundary of $\overline{\mscr{S}}_{\sg,\sns,\sra}$ trivializes those Dehn twists. It follows from the long exact sequence of cohomologies for the pair $(\overline{\mscr{S}}_{\sg,\sns,\sra},\partial \overline{\mscr{S}}_{\sg,\sns,\sra})$ that
	\begin{equation}
	\H^1 (\overline{\mscr{S}}_{\sg,\sns,\sra},\partial \overline{\mscr{S}}_{\sg,\sns,\sra},\mathbb C)\cong \H^0 (\partial \overline{\mscr{S}}_{\sg,\sns,\sra},\mathbb C)/\H^0 (\overline{\mscr{S}}_{\sg,\sns,\sra},\mathbb C).
	\end{equation}
	This identification shows that $\H_{6\sg-7+2(\sns+\sra)}({\mscr{S}}_{\sg,\sns,\sra},\mathbb C)$ is canonically identified with $\mathbb C^{\# \text{Boundary Components}-1}$. This proves the case when $(\g,\ns,\ra)$ is not one of $(0,4,0)$, $(0,2,2)$, $(0,0,4)$, or $(1,1,0)^{\text{ev}}$, since there is only one boundary component. For $(\g,\ns,\ra)=(1,1,0)^{\text{ev}}$, it is easy to see that the permutation group $S_{\sns,\sra}$ preserves those boundary components, hence taking quotient does not affect the homology, i.e. they are still one dimensional. The action of $S_4$ on the boundary of $\overline{\mscr{M}}_{0,4}$ (which consists of 3 points) is transitive, hence after taking quotient the homology becomes trivial. It remains to examine the action of $S_2\times S_2$ on the the boundary of $\overline{\mscr{M}}_{0,4}$. The order of $S_2\times S_2$ is 4, its image in $S_3$ is either trivial or $S_2$. Since we know that the action must be nontrivial: there is a $S_2$ component corresponding to switching two NS punctures. Hence the image of $S_2\times S_2$ in $S_3$ is $S_2$, after taking the quotient, and the homology is one dimensional.
\end{proof}

\begin{remark}\label{rem:inteserction pairing}
	It follows from the Corollary {\normalfont \ref{cor:top minus one homology}} that the stack $\partial\overline{\mscr{S}}_{\sg,\sns,\sra} /S_{\sns,\sra}$ has at most two connected components. The intersection pairing between $\H_{6\sg-7+2(\sns+\sra)}({\mscr{S}}_{\sg,\sns,\sra}/S_{\sns,\sra},\mathbb C)$ and $\H_1(\overline{\mscr{S}}_{\sg,\sns,\sra} /S_{\sns,\sra}, \partial \overline{\mscr{S}}_{\sg,\sns,\sra} /S_{\sns,\sra},\mathbb C)$ {\normalfont (}when it is not zero{\normalfont)} is given by connecting two boundary components with a generic smooth arc and counting the intersection number with a $(6\g-7+2(\ns+\ra))$-cycle.
\end{remark}

\subsection{An Alternative Model for $\mscr{S}(\ns,\ra)$}\label{subsec:an alternative model for space of bordered spin-Riemann surfaces}
As we discussed in Section \ref{subsec:gluing of punctures}, there is an ambiguity of gluing R boundary components on a single spin-Riemann surfaces due to a $\mbb{Z}/2$-automorphism acting on the spinor bundle over the surface. This motivates the following question: {\it can we define a classifying space for spin curves such that there is no such ambiguity?} The answer turns out to be positive. We can indeed define such a space by natural generalization of the space considered in \cite{KimuraStasheffVoronov9307} as a model for the space of bordered ordinary Riemann surfaces. In this section, we elaborate the generalization of this construction for the space of bordered spin-Riemann surfaces. 

\newl $\mscr{S}_{\sg}(\ns,\ra)$ has a nice model $\widetilde{\mscr{S}}_{\sg}(\ns,\ra)$, which is the moduli stack of stable Riemann surfaces, i.e. the surfaces belong to $\overline{\mscr{S}}_{\sg,\sns,\sra}$, decorated at each NS puncture with a ray in the tangent space, at each R puncture with a ray in the spinor line bundle $\mscr E$, at each NS node with a ray in the tensor product of rays of the tangent spaces at each side, and at each R node with a ray in the spinor line bundle $\mscr E$. This is an orbifold with corners whose interior consists of smooth spin curves, so it is homotopy-equivalent to $\mscr S_{\sg}(\ns,\ra)$. According to Section \ref{subsec:gluing of punctures}, we have the natural gluing operation which give rises to an operation on the corresponding moduli stacks
\begin{alignat}{2}
\mathfrak m^{\text{NS}} _{i_1,i_2}:\widetilde{\mscr{S}}_{\sg_1}(\ns^1,\ra^1)\times \widetilde{\mscr{S}}_{\sg_2}(\ns^2,\ra^2) &\to \widetilde{\mscr{S}}_{\sg_1+\sg_2}(\ns^1+\ns^2-2,\ra^1+\ra^2),\nonumber
\\
\mathfrak m^{\text{R}}_{i_1,i_2}:\widetilde{\mscr{S}}_{\sg_1}(\ns^1,\ra^1)\times \widetilde{\mscr{S}}_{\sg_2}(\ns^2,\ra^2) &\to \widetilde{\mscr{S}}_{\sg_1+\sg_2}(\ns^1+\ns^2,\ra^1+\ra^2-2),
\end{alignat}
as well as self-gluing for NS punctures
\begin{equation}
\mathfrak g^{\text{NS}} _{i,j}:\widetilde{\mscr{S}}_{\sg}(\ns,\ra)\to \widetilde{\mscr{S}}_{\sg+1}(\ns-2,\ra).
\end{equation}
Note that the gluing operation takes two rays at two NS points to their tensor product. 

\newl On the other hand, the ambiguity of gluing R boundaries for $\mscr{S}(\ns,\ra)$ disappears due to the fact that there are two canonical choices of gluing R punctures for surfaces in $\widetilde{\mscr{S}}(\ns,\ra)$: two rays are either of the same phase or of the opposite phase. We denote the corresponding self-gluing operation $\mathfrak g^{\text{R}+} _{i,j}$ and $\mathfrak g^{\text{R}-} _{i,j}$:
\begin{equation}
\mathfrak g^{\text{R}\pm} _{i,j}:\widetilde{\mscr{S}}_{\sg}(\ns,\ra)\to \widetilde{\mscr{S}}_{\sg+1}(\ns,\ra-2).
\end{equation}
Note that all of these gluing operations are equivariant under the permutation group action. Therefore, they can be defined on the quotient space without referring to an explicit choice of punctures, i.e. we have 
\begin{alignat}{2}\hspace*{-0.2\linewidth}
\mathfrak m^{\text{NS}} :\widetilde{\mscr{S}}_{\sg_1}(\ns^1,\ra^1)/S_{\sns^1,\sra^1}\times \widetilde{\mscr{S}}_{\sg_2}(\ns^2,\ra^2)/S_{\sns^2,\sra^2} &\to \widetilde{\mscr{S}}_{\sg_1+\sg_2}(\ns^1+\ns^2-2,\ra^1+\sra^2)/S_{\sns^1+\sns^2-2,\sra^1+\sra^2},\nonumber
\\
\mathfrak m^{\text{R}} :\widetilde{\mscr{S}}_{\sg_1}(\ns^1,\ra^1)/S_{\sns^1,\sra^1}\times \widetilde{\mscr{S}}_{\sg_2}(\ns^2,\ra^2)/S_{\sns^2,\sra^2} &\to \widetilde{\mscr{S}}_{\sg_1+\sg_2}(\ns^1+\ns^2,\ra^1+\ra^2-2)/S_{\sns^1+\sns^2,\sra^1+\sra^2-2},\nonumber
\\
\mathfrak g^{\text{NS}} :\widetilde{\mscr{S}}_{\sg}(\ns,\ra)/S_{\sns,\sra}&\to \widetilde{\mscr{S}}_{\sg+1}(\ns-2,\ra)/S_{\sns-2,\sra},\nonumber
\\
\mathfrak g^{\text{R}\pm} :\widetilde{\mscr{S}}_{\sg}(\ns,\ra)/S_{\sns,\sra}&\to \widetilde{\mscr{S}}_{\sg+1}(\ns,\ra-2)/S_{\sns,\sra-2}.
\end{alignat}

\subsection{The Proof of Existence and Uniqueness}\label{subsec:the proof of existence and uniqueness}

We are finally in a position to present one of the main results of this work, i.e. the existence and uniqueness, in the sense we explain, of a solution to the BV QME. Let $\mscr{C}_*$ be the functor of normalized singular simplicial chains with coefficients in any field $\mbb{F}$ containing $\mbb{Q}$. We define two complexes
\begin{alignat}{2}\label{eq:the chain complexes associated to S and tilde S}
	\mathcal F(\mscr{S})&\equiv \bigoplus_{\sns,\sra} \mscr{C}_*(\mscr{S} (\ns,\ra)/S^1\wr S_{\sns,\sra}), \nonumber
	\\
	\mathcal F(\widetilde{\mscr{S}})&\equiv \bigoplus_{\sns,\sra} \mscr{C}_*(\widetilde{\mscr{S}} (\ns,\ra)/S^1\wr S_{\sns,\sra}),
\end{alignat}
where $S^1\wr S_{\sns,\sra}$ is the wreath product $(S^1)^{\sns+\sra}\rtimes S_{\sns,\sra}$.Let $\mathcal F_{\sg,\sns,\sra}(\mscr{S})$ (respectively $\mathcal F_{\sg,\sns,\sra}(\widetilde{\mscr{S}})$) be the part coming from connected surfaces of with $\g$ handles, $\ns$ NS boundary components (respectively NS punctures) and $\ra$ R boundary components (respectively R punctures). $\mathcal F(\mscr{S})$ and $\mathcal F(\widetilde{\mscr{S}})$ naturally carry structures of commutative differential graded algebra, where the differential is the boundary map of chain complexes and the product comes from disjoint union of surfaces. Since $\mscr{S}(\ns,\ra)$ is homotopy-equivalent to $\widetilde{\mscr{S}}(\ns,\ra)$\footnote{Remember that 1) the phase of ray at at a puncture parametrizes a circle, which can be thought of as a boundary component, and 2) the ray at a node can be used to open it.}, $\mathcal F(\mscr{S})$ and $\mathcal F(\widetilde{\mscr{S}})$ are quasi-isomorphic. Let us define the following spaces
\begin{equation}
\mbbmss{X}(\ns,\ra)\equiv \widetilde{\mscr{S}} (\ns,\ra)/S^1\wr S_{\ns,\ra}.
\end{equation}
$\mbbmss{X}(\ns,\ra)$ is the moduli stack of spin curves with unordered punctures, together with at each NS node, a ray in the tensor product of the tangent spaces at each side, and at each R node, a ray in the spinor bundle, and unparametrized boundaries since there is no ray on punctures anymore. The boundary of $\mbbmss{X}_\sg(\ns,\ra)$, i.e. the locus of nodal curves, has real codimension one inside $\mbbmss{X}_\sg(\ns,\ra)$.

\newl Consider the moduli stack of a curve $\mscr{C}\in \mbbmss{X}_{\sg-1}(\ns+2,\ra)$ together with a chosen pair of NS punctures and a choice of gluing at these two punctures. The choice does not matter because the permutation symmetry has been quotient out. There is an $S^1$ possible ways of gluing rays.  Note that it can be identified with an irreducible component of $\partial \mbbmss{X}_\sg(\ns,\ra)$. Let us denote this component by $\mbbmss{Y}^{\text{NS}}_\sg(\ns,\ra)$. Similarly, we can define the space $\mbbmss{Y}^{\text{R}+}_\sg(\ns,\ra)$ and $\mbbmss{Y}^{\text{R}-}_\sg(\ns,\ra)$, they are differed by a choice of gluing of spin structures at punctures. There is a sequence of maps:
\begin{equation}
\mbbmss{X}_{\sg-1}(\ns+2,\ra)\longleftarrow \mbbmss{Y}^{\text{NS}}_\sg(\ns,\ra)\longrightarrow \mbbmss{X}_{\sg}(\ns,\ra),
\end{equation}
then define the operator $\Delta^{\text{NS}}: \mscr{C}_*(\mbbmss{X}_{\sg-1}(\ns+2,\ra))\to \mscr{C}_*(\mbbmss{X}_{\sg}(\ns,\ra))$ to be the minus of pulling back to $\mbbmss{Y}^{\text{NS}}_\sg(\ns,\ra)$ followed by pushing forward to $\mbbmss{X}_{\sg}(\ns,\ra)$. Similarly for $\Delta^{\mbb{R}\pm}:\mscr{C}_*(\mbbmss{X}_{\sg-1}(\ns,\ra+2))\to \mscr{C}_*(\mbbmss{X}_{\sg}(\ns,\ra))$. A similar procedure applies to a pair of punctures coming from different connected curves, which gives rise to brackets $\{\cdot,\cdot\}_{\text{NS}}$ and $\{\cdot,\cdot\}_{\text{R}}$. Note that the pulling back procedure increases the degree of chain by one, since $\mbbmss{Y}^{\bullet}_\sg\to \mbbmss{X}_{\sg-1}$\footnote{$\bullet$ denotes either NS or R.} is an $S^1$ fibration.
\begin{definition}
	Let $\mathcal F_{k}(\widetilde{\mscr{S}})$ denote the subspace spanned by chains on the space
	of surfaces with at least $k$ modes of either type. We define the following operators
	\begin{alignat}{2}
	\Delta&\equiv \Delta^{\text{\normalfont NS}}+\Delta^{\text{\normalfont R}+}+\Delta^{\text{\normalfont R}-}:\mathcal F_{k}(\widetilde{\mscr{S}})\to \mathcal F_{k+1}(\widetilde{\mscr{S}}),\nonumber
	\\
	\{\cdot,\cdot\}&\equiv \{\cdot,\cdot\}_{\text{\normalfont NS}}+\{\cdot,\cdot\}_{\text{\normalfont R}}:\mathcal F_{k}(\widetilde{\mscr{S}})\otimes \mathcal F_{l}(\widetilde{\mscr{S}})\to \mathcal F_{k+l+1}(\widetilde{\mscr{S}}).
	\end{alignat}
	The operator $\widehat{d}\equiv d+\hbar\Delta$ makes $\mathcal F(\widetilde{\mscr{S}})$ into a BV algebra.
\end{definition}

One can then prove the following lemma which helps us in proving the solution to the BV QME in $\widetilde{S}(\ns,\ra)/S_{S^1\wr S_{\sns,\sra}}$. 

\begin{lem}\label{lem:Lemma 4.1}
	Define the subspace $\partial _i \mbbmss{X}_\sg(\ns,\ra)$ to be the locus of spin curves with at least $i$ nodes, and consider the following decreasing filtration on $\mathcal F(\widetilde{\mscr{S}})$
	\begin{equation}\label{eq:filteration of the BV algebra}
	F^i\mathcal F(\widetilde{\mscr{S}})\equiv \Span\left\{\alpha\in \mathcal F_{*}(\widetilde{\mscr{S}})|\,\text{$\alpha$ is supported in }\partial_i\mbbmss{X}_\sg(\ns,\ra)\right\}.
	\end{equation}
	Let $[\mbbmss{X}_\sg(\ns,\ra)]$ be the fundamental class of $(\mbbmss{X}_\sg(\ns,\ra),\partial \mbbmss{X}_\sg(\ns,\ra))$, then it satisfies the BV QME{\normalfont :}
	\begin{alignat*}{2}
	d[\mbbmss{X}_\sg(\ns,\ra)]&+\Delta^{\text{\normalfont NS}} [\mbbmss{X}_{\sg-1}(\ns+2,\ra)]+\Delta^{\text{\normalfont R}+} [\mbbmss{X}_{\sg-1}(\ns,\ra+2)]+\Delta^{\text{\normalfont R}-} [\mbbmss{X}_{\sg-1}(\ns,\ra+2)]
	\\
	&+\frac{1}{2}\sum _{\substack{\sg_1+\sg_2=\sg\\ \sns^1+\sns^2=\sns+2 \\ \sra^1+\sra^2=\sra}}\left\{[\mbbmss{X}_{\sg_1}(\ns^1,\ra^1)],[\mbbmss{X}_{\sg_2}(\ns^2,\ra^2)]\right\}_{\text{\normalfont NS}}
	\\
	&+\frac{1}{2}\sum _{\substack{\sg_1+\sg_2=\sg\\\sns^1+\sns^2=\sns\\ \sra^1+\sra^2=\sra+2}}\left\{[\mbbmss{X}_{\sg_1}(\ns^1,\ra^1)],[\mbbmss{X}_{\sg_2}(\ns^2,\ra^2)]\right\}_{\text{\normalfont R}}=0,
	\end{alignat*}
	in the {\normalfont BV} algebra $\left(\H_*(\mathcal F(\widetilde{\mscr{S}})/F^2\mathcal F(\widetilde{\mscr{S}})),d+\hbar\Delta\right)$. 
\end{lem}

\begin{proof}
	By definition of $\Delta$ and $\{\cdot,\cdot\}$, they are of degree one with respect to the filtration, i.e. $\Delta(F^k)\subset F^{k+1}$ and $\{F^k,F^l\}\subset F^{k+l+1}$. It follows that $\Gr _F^i\mathcal F(\widetilde{\mscr{S}})$, where $\Gr _F^i\mathcal F(\widetilde{\mscr{S}})$ is the $i$th graded component of the graded algebra associated to the BV algebra $\mathcal F(\widetilde{\mscr{S}})$ and the grading is induced by the filtration $F$\footnote{Let $(F^M_i)_{i\in\mbb{Z}}$ be a decreasing filtration of a module $M$ over a commutative ring $R$. We can define
			\begin{equation*}
			M_i\equiv F_i^M/\sum_{j>i}F^M_j, \qquad i,j\in\mbb{Z}.
			\end{equation*}
			Then, we can define
			\begin{equation*}
			\text{Gr}_F M\equiv \bigoplus_{i\in\mbb{Z}}M_i.
			\end{equation*}
			$\text{Gr}_F M$ is the associated graded R-module of $M$. Assuming that $M$ has an algebra structure $*:M\times M\longrightarrow M$ such that $F_i^M\times F_j^M$ is sent to $F_{i+j}^M$, we can define an algebra structure on $\text{Gr}M$ as follows. Consider the canonical projections $\pi_i:F^M_i\longrightarrow M_i$. For two elements $a_i\in M_i$ and $a_j\in M_j$, we can define the multiplication operation $\boldsymbol{\cdot}:\text{Gr}_F M\times \text{Gr}_F M\longrightarrow \text{Gr}_F M$ on $\text{Gr}_F M$
			\begin{equation*}
			a_i\boldsymbol{\cdot} a_j\equiv \pi_{i+j}(a'_i*a'_j),
			\end{equation*}
			where $a'_i=\pi_i^{-1}(a_i)$ and $a'_j=\pi_i^{-1}(a_j)$ are lifts of $a_i$ and $a_j$. The filtered algebra $\text{Gr}_F M$ endowed with this multiplication is called the associated graded $R$-algebra of $M$.}, is the relative chain complex $(\mscr{C}_*(\partial _i \mbbmss{X})/\mscr{C}_*(\partial _{i+1} \mbbmss{X}),d)$, hence 
	\begin{equation}
	\H_*(\Gr _F^i\mathcal F(\widetilde{\mscr{S}}))\cong \bigoplus_{\sns,\sra}\H_*(\partial _i \mbbmss{X}(\ns,\ra),\partial_{i+1} \mbbmss{X}(\ns,\ra)).
	\end{equation}
	$\Delta^{\text{\normalfont NS}}$ acts on the homology class: $\H_k(\Gr _F^i\mathcal F(\widetilde{\mscr{S}}))\to \H_{k+1}(\Gr _F^{i+1}\mathcal F(\widetilde{\mscr{S}}))$, i.e. it is a map 
	\begin{equation}
	\H_k(\partial _i \mbbmss{X}(\ns+2,\ra),\partial _{i+1} \mbbmss{X}(\ns+2,\ra))\to \H_{k+1}(\partial _{i+1} \mbbmss{X}(\ns,\ra),\partial _{i+2} \mbbmss{X}(\ns,\ra)).
	\end{equation}
	Similar statement holds for $\Delta^{\text{\normalfont R}\pm}$. 
	
	\newl We look closer to the case $i=0$: there is a fundamental class $[\mbbmss{X}(\ns+2,\ra)]\in \H_{*}(\mbbmss{X}(\ns+2,\ra),\partial  \mbbmss{X}(\ns+2,\ra))$, its image
	\begin{equation}
	\Delta ^{\text{\normalfont NS}} [\mbbmss{X}(\ns+2,\ra)]\in\H_{*+1}(\partial \mbbmss{X}(\ns,\ra),\partial _{2} \mbbmss{X}(\ns+2,\ra)),
	\end{equation}
	is by construction $-[\mbbmss{Y}^{\text{\normalfont NS}}(\ns+2,\ra)]$, i.e. minus the fundamental class of the corresponding irreducible component of $\partial \mbbmss{X}(\ns,\ra)$. Adding the contribution of $\Delta^{\text{\normalfont R}\pm}$, we see that
	\begin{equation}
	\Delta ^{\text{\normalfont NS}} [\mbbmss{X}(\ns+2,\ra)]+\Delta ^{\text{\normalfont R}+} [\mbbmss{X}(\ns,\ra+2)]+\Delta^{\text{\normalfont R}-} [\mbbmss{X}(\ns,\ra+2)]=-[\partial \mbbmss{X}(\ns,\ra)].
	\end{equation}
	On the other hand, the $\mathbb Q$-orientation of $\partial \mbbmss{X}(\ns,\ra)$ is induced from the $\mathbb Q$-orientation of $\mbbmss{X}(\ns,\ra)$, i.e. the fundamental class of the boundary $\partial \mbbmss{X}(\ns,\ra)$ is induced from the boundary of fundamental class of $\mbbmss{X}(\ns,\ra)$. Note that the map is exactly the boundary map in the following exact sequence of homologies:
	\begin{align*}
	0=\H_{6\sg-6+2\sns+2\sra}(\mbbmss{X}(\ns,\ra),\partial _2 \mbbmss{X}(\ns,\ra))\longrightarrow  \H_{6\sg-6+2\sns+2\sra}(\mbbmss{X}(\ns,\ra),\partial \mbbmss{X}(\ns,\ra))\\
	\overset{d}{\longrightarrow } \H_{6\sg-7+2\sns+2\sra}(\partial \mbbmss{X}(\ns,\ra),\partial _2 \mbbmss{X}(\ns,\ra))\longrightarrow  \H_{6\sg-7+2\sns+2\sra}(\mbbmss{X}(\ns,\ra),\partial _2 \mbbmss{X}(\ns,\ra)).
	\end{align*}
	It then follows that
	\begin{equation}
	d[\mbbmss{X}(\ns,\ra)]=[\partial \mbbmss{X}(\ns,\ra)],
	\end{equation}
	in the BV algebra $\left(\H_*(\mathcal F(\widetilde{\mscr{S}})/F^2\mathcal F(\widetilde{\mscr{S}})),d+\Delta\right)$, whence
	\begin{equation}\label{eq:the BV master equation for the classes [X(m,n)]}
	d([\mbbmss{X}(\ns,\ra)])+\Delta ^{\text{\normalfont NS}} ([\mbbmss{X}(\ns+2,\ra)])+\Delta ^{\text{\normalfont R}+} ([\mbbmss{X}(\ns,\ra+2)])+\Delta ^{\text{\normalfont R}-} ([\mbbmss{X}(\ns,\ra+2)])=0.
	\end{equation}
	We can consider the following formal sum 
	\begin{equation}
	\sum_{\sns,\sra}[\mbbmss{X}(\ns,\ra)]=\exp \left(\sum_{\sg,\sns,\sra}\hbar^{2\sg-2+\sns+\sra}[\mbbmss{X}_{\sg}(\ns,\ra)]\right)\equiv\exp\left([\mbbmss{X}]\right).
	\end{equation}
	Using this sum and \eqref{eq:the BV master equation for the classes [X(m,n)]}, we have
	\[
	(d+\hbar\Delta)\exp\left(\frac{[\mbbmss{X}]}{\hbar}\right)=0,
	\]
	i.e. the class $[\mbbmss{X}]$ satisfies the BV QME.
\end{proof}
Using this result, we have the following Theorem
\begin{thr}[The Existence and Uniqueness of Solution to the BV QME in $\mcal{F}(\mscr{S})$]\label{the:the existence and uniqueness of the solution to the BV QME in spin moduli}
	For each triple $(\g,\ns,\ra)$ with $2\g-2+\ns+\ra>0$, there exists an element ${\mcal{V}}_{\sg}(\ns,\ra)\in \mathcal{F}_{\sg,\sns,\sra}(\widetilde{\mscr{S}})$, of homological degree $6\g-6+2\ns+2\ra$, with the following properties.
	\begin{enumerate}
		\item [$\text{\normalfont 1.}$] ${\mcal{V}}_{0}(3,0)$ is the fundamental cycle of $\mbbmss{X}_{0}(3,0)$, i.e. 0-chain of coefficient $1/6$;
		
		\item [$\text{\normalfont 2.}$] ${\mcal{V}}_{0}(1,2)$ is the fundamental cycle of $\mbbmss{X}_{0}(2,1)$, i.e. 0-chain of coefficient $1/2$;
		
		\item[$\text{\normalfont 3.}$] The generating function
		\begin{equation}
		{\mcal{V}}\equiv \sum_{\substack{\sg,\sns,\sra \\ 2\sg-2+\sns+\sra>0}}\hbar ^{2\sg-2+\sns+\sra} {\mcal{V}}_{\sg}(\ns,\ra)\in \hbar \mathcal{F}(\widetilde{\mscr{S}})[\![\hbar]\!],
		\end{equation}
		satisfies the BV quantum master equation
		\begin{equation}
		(d+\hbar\Delta)\exp\left(\frac{{\mcal{V}}}{\hbar}\right)=0.
		\end{equation}
		
		\item[$\text{\normalfont 4.}$] 	Such solution ${\mcal{V}}$ is unique up to homotopy through such elements.
	\end{enumerate}
\end{thr}

\begin{proof}
	Define a dg Lie algebra $\mathfrak{g}\equiv \oplus _i \mathfrak{g}_i$ where $\mathfrak{g}_i$ is the set of
	\begin{equation}\label{eq:the generating function of superstring vertices II}
	{\mcal{V}}\equiv \sum_{\substack{\sg,\sns,\sra \\ 2\sg-2+\sns+\sra>0}}\hbar ^{2\sg-2+\sns+\sra} {\mcal{V}}_{\sg}(\ns,\ra)\in \hbar \mathcal{F}(\widetilde{\mscr{S}})[\![\hbar]\!],
	\end{equation}
	such that ${\mcal{V}}_{\sg}(\ns,\ra)\in \mathcal{F}_{\sg,\sns,\sra}(\widetilde{\mscr{S}})$ and $\deg {\mcal{V}}_{\sg}(\ns,\ra)=6\g-5+2\ns+2\ra+i$. Consider the filtration \eqref{eq:filteration of the BV algebra}, $\mathfrak{g}=F^0\mathfrak{g}\supset F^1\mathfrak{g}\supset\cdots$. We define $\mathfrak{g}'$ to be $\mathfrak{g}/F^2\mathfrak{g}$. It follows that $\forall k\ge 0$ and $\forall i\ge 0$, $\H_{i}(F^k\mathfrak{g}/F^{k+1}\mathfrak{g})=0$, for there is no element of non-negative degree. Obviously, the quotient map $\mathfrak{g} \to \mathfrak{g}'$ induces isomorphisms
	\begin{equation}
	\H_{i}(F^k\mathfrak{g}/F^{k+1}\mathfrak{g})\cong \H_{i}(F^k\mathfrak{g}'/F^{k+1}\mathfrak{g}'),
	\end{equation}
	when $k=0$ or $k=1$. Moreover, when $k\ge 2$, there is no chain of degree $\ge -2$, so $\H_{i}(F^k\mathfrak{g}/F^{k+1}\mathfrak{g})=0$ when $i\ge -2$. To sum up, the quotient map $\mathfrak{g} \to \mathfrak{g}'$ induces isomorphisms
	\begin{equation}
	\H_{i}(\Gr _F\mathfrak{g})\cong \H_{i}(\Gr _F\mathfrak{g}')\qquad i=0,-1,-2,
	\end{equation}
	where $\Gr_F\mfk{g}$ is the graded algebra associated to the Lie algebra $\mfk{g}$ and the grading is induced by the filtration $F$. Let $\MC(\mfk{g})$ be the set of elements of $\mfk{g}$ satisfying the Maurer-Cartan equation, and $\pi_0(\MC(\mfk{g}))$ be  the set of homotopy-equivalence classes of solutions of the Maurer-Cartan equation in $\mfk{g}$ (see Definition 5.1.2 and Lemma 5.2.1 of \cite{Costello200509}). It follows from the Lemma 5.3.1 of \cite{Costello200509} that the map 
	\begin{equation}
	\pi_0(\MC(\mathfrak{g}))\to \pi_0(\MC(\mathfrak{g}')),
	\end{equation}
	is an isomorphism. Since the set of homotopy-equivalence classes of solutions of the Maurer-Cartan equation in the Lie algebra $\mfk{g}$, defined using \eqref{eq:the generating function of superstring vertices II}, can be identified with the set of homotopy-equivalence classes of solutions ${\mcal{V}}$ of the BV QME in $\mcal{F}(\widetilde{\mscr{S}})$, the existence of ${\mcal{V}}$ follows immediately from the Lemma \ref{lem:Lemma 4.1}, as the quotient map sends ${\mcal{V}}_{0}(3,0)$ to $[\mbbmss{X}_0(3,0)]$, and we have seen that the classes $[\mbbmss{X}_{\sg}(\ns,\ra)]$ satisfy the BV QME. It is thus suffices to show that $[\mbbmss{X}_\sg(\ns,\ra)]$ is the unique solution to the BV QME in the BV algebra $\H_*(\mathcal F(\widetilde{\mscr{S}})/F^2\mathcal F(\widetilde{\mscr{S}}))$ such that the degree $(\g,\ns,\ra)=(0,3,0)$ part is the fundamental class. In fact, suppose that there is another system of solutions
	\begin{equation*}
	{\mcal{V}}'_{\sg}(\ns,\ra)\in \H_{6\sg-6+2\sns+2\sra}(\mbbmss{X}_\sg(\ns,\ra),\partial \mbbmss{X}_\sg(\ns,\ra)),
	\end{equation*}
	then it must satisfy 
	\begin{equation*}
	d{\mcal{V}}'^{\text{ev}}_{1}(1,0)+\Delta^{\text{\normalfont NS}}[\mbbmss{X}_{0}(3,0)]+\Delta^{\text{\normalfont R}-}{\mcal{V}}'_{0}(1,2)=0.
	\end{equation*}
	Note that there is no $\Delta^{\text{\normalfont R}+}$ term\footnote{Since the operator $\Delta ^{\text{R}+}$ preserves the global section of the spin structure $\mscr E$ on $\mathbb P^1$ with one NS puncture an two R punctures, hence producing an odd spin structure.}. It then follows from the exact sequence of homologies
	\begin{align*}
	0=\H_{2}(\mbbmss{X}_1(1,0)^{\text{ev}})\longrightarrow  \H_{2}(  X_1(1,0)^{\text{ev}},\partial  \mbbmss{X}_1(1,0)^{\text{ev}})\overset{d}{\longrightarrow } \H_{1}(\partial \mbbmss{X}_1(1,0)^{\text{ev}})\longrightarrow  \H_{1}( \mbbmss{X}_1(1,0)^{\text{ev}}),
	\end{align*}
	that the image of $\Delta^{\text{\normalfont NS}}[\mbbmss{X}_{0}(3,0)]+\Delta^{\text{\normalfont R}-}{\mcal{V}}'_{0}(1,2)$ in $\H_{1}(\mbbmss{X}_{1,1,0}^{\text{ev}})$ is trivial. On the other hand, we already know that $[\mbbmss{X}_\sg(\ns,\ra)]$ is a solution, so the same argument implies that the image of $\Delta^{\text{\normalfont NS}} [\mbbmss{X}_{0}(3,0)]+\Delta^{\text{\normalfont R}-} [\mbbmss{X}_0(1,2)])$ in $\H_{1}(\mbbmss{X}_1(1,0)^{\text{ev}})$ is trivial. Comparing these equations we find that 
	\begin{equation*}
	\Delta^{\text{\normalfont R}-}({\mcal{V}}'_{0}(1,2)-[\mbbmss{X}_0(1,2)])=0,
	\end{equation*}
	in $\H_{1}(\mbbmss{X}_1(1,0)^{\text{ev}})$. However, we know that $\mbbmss{X}_1(1,0)^{\text{ev}}$, after reducing to the coarse moduli, is just a sphere with a real blow-up at two points (i.e. the boundary of $\widetilde{\mscr{S}}^{\text{ev}}_1(1,0)$), and $\Delta ^{\text{\normalfont R}-}$ sends a zero cycle in $\widetilde{\mscr{S}}_0(1,2)$ to the circle fundamental class then embed it isomorphically to one of the boundary circle of $\mbbmss{X}_1(1,0)^{\text{ev}}$\footnote{$\mbbmss{X}_1(1,0)^{\text{ev}}$ is a copy of $\mbb{P}^1$ which two boundaries at north and south poles of $\mbb{P}^1$. The real blow-up of these boundaries are circles. We are referring to one of these circles as the boundary circle.}. In particular, $\Delta^{\text{\normalfont R}-}$ is injective. As a consequence, we see that
	\begin{equation*}
	{\mcal{V}}'_{0}(1,2)-[\mbbmss{X}_0(1,2)]=0.
	\end{equation*}
	Now the uniqueness of solution follows from induction and the fact that
	\begin{equation*}
	d:\H_{6\sg-6+2\sns+2\sra}(\mbbmss{X}(\ns,\ra),\partial  \mbbmss{X}(\ns,\ra))\longrightarrow\H_{6\sg-7+2\sns+2\sra}(\partial\mbbmss{X}_\sg(\ns,\ra),\partial _2 \mbbmss{X}_\sg(\ns,\ra)),
	\end{equation*}
	is injective. This means that lower-degree terms completely determine higher-degree terms via the BV QME.
\end{proof}

\begin{proof}[Another Proof]
	Alternatively, one can prove this theorem using the filtration in Proposition 10.1.1 of \cite{Costello200509}, but the argument there need to be modified since the vanishing of top minus one degree homologies is not always true for the moduli of spin curves. Nevertheless we can take the filtration $\mathfrak g=F^1\mathfrak g\supset F^2\mathfrak g\supset F^3\mathfrak g\cdots$ by letting $F^k\mathfrak g$ generated by the set of $\mcal V$ such that $\mcal V_{\sg}(\ns,\ra)$ is zero for $2\g-2+\ns+\ra<k$, when $k\ge 3$, and $F^2\mathfrak g$ generated by the set of $\mcal V$ such that $\mcal V_{\sg}(\ns,\ra)$ is zero for $2\g-2+\ns+\ra<k$ except for $(\g,\ns,\ra)=(0,2,2)$. The point is that 
	\begin{equation*}
	H_{i}(F^k\mathfrak{g}/F^{k+1}\mathfrak{g})=0, \qquad \text{for $k\ge 2$, \quad and $i=0,-1,-2$},
	\end{equation*}
	which follows from Corollary \ref{cor:top minus one homology}. It follows from the Lemma 5.3.1 of \cite{Costello200509} that the map 
	\begin{equation*}
	\pi_0(\MC(\mathfrak{g}))\to \pi_0(\MC(\mathfrak{g}')),
	\end{equation*}
	is isomorphism, where $\mathfrak{g}'=\mathfrak{g}/F^2\mathfrak{g}$. It remains to solve the QME in $\mathfrak{g}'$, and this is a consequence of Lemma \ref{lem:Lemma 4.1} since $\partial _i \mbbmss{X}_\sg(\ns,\ra)$ is empty for 
	\begin{equation*}
	(\g,\ns,\ra)=(0,3,0),(0,1,2),(0,2,2),(1,1,0).
	\end{equation*}
\end{proof}
This theorem shows that there exists certain subspaces ${\mcal{V}}_{\sg}(\ns,\ra)$ inside $\widetilde{\mscr{S}}_{\sg}(\ns,\ra)/S^1\wr S_{\sns,\sra}$ that satisfy the BV QME. However, we are interested in finding a solution to the BV QME in $\mcal{F}(\mscr{S})$, and hence find a certain subspace inside $\mscr{S}_{\sg}(\ns,\ra)/S^1\wr S_{\sns,\sra}$. To show that there indeed exists such a subspace, we proceed as follows. As we have explained, $\mathcal F(\mscr{S})$ and $\mathcal F(\widetilde{\mscr{S}})$ naturally carry structures of commutative differential-graded algebra, where the differential is the boundary map of chain complexes and the product comes from disjoint union of surfaces. We can also define a BV algebra structure on $\mathcal F(\mscr{S})$ in a similar way to the definition of $\Delta$ operator of $\mathcal F(\widetilde{\mscr{S}})$. Then the Lemma 10.4.2 of {\normalfont \cite{Costello200509}} can be applied to this case, and we have a homotopy-equivalence of BV algebras:
\begin{equation}\label{eq:the homotopy equivalence of chain complexes associated to the S and tilde S}
\mathcal F(\mscr{S})\cong \mathcal F(\widetilde{\mscr{S}}).
\end{equation}
Therefore, all statements in Theorem \ref{the:the existence and uniqueness of the solution to the BV QME in spin moduli} remain true for $\mathcal F(\mscr{S})$, i.e. there exists an element (unique up to homotopy)\footnote{By abuse of notation, we use the same notation $\widehat{\mcal{V}}_{\vsg}(\sns,\sra)$ for subspace of $\mscr{S}_{\vsg}(\sns,\sra)/S^1\wr S_{\sns,\sra}$.}
\begin{equation}\label{eq:the superstring vertices inside the complex associated to the moduli of spin curves}
{\mcal{V}}\equiv \sum_{\substack{\sg,\sns,\sra \\ 2\sg-2+\sns+\sra>0}}\hbar ^{2\sg-2+\sns+\sra}{\mcal{V}}_{\sg}(\ns,\ra)\in \hbar \mathcal{F}(\mscr S)[\![\hbar]\!],
\end{equation}
and ${\mcal{V}}_{0}(3,0)=[\mscr{S}_{0}(3,0)/S_{3,0}]$ and ${\mcal{V}}_{0}(1,2)=[\mscr{S}_0(1,2)/S_{1,2}]$, satisfying the BV QME
\begin{equation}
(d+\hbar\Delta)\exp\left(\frac{{\mcal{V}}}{\hbar}\right)=0.
\end{equation}
We have thus established that there exist certain subspace ${\mcal{V}}_{\sg}(\ns,\ra)$ in $\mscr{S}_{\sg}(\ns,\ra)/S^1\wr S_{\sns,\sra}$ which satisfy the BV QME.  

\subsection{Relation of Solutions to Fundamental Classes of $\overline{\mscr{S}}_{\sg,\sns,\sra}$}\label{subsec:relation of solutions to the BVQME in spin moduli and fundamental classes of DM spaces}

In this section, we prove that there is a map in the homotopy category of BV algebras that sends the solution ${\mcal{V}}$ of the BV quantum master equation in the BV algebra $\mcal{F}(\widetilde{\mscr{S}})$ to the fundamental class of the Deligne-Mumford stacks of associated {\it punctured} spin curves. Consider the following complex
\begin{equation}
\mathcal F(\overline{\mscr{S}})\equiv \bigoplus_{\sns,\sra} \mscr{C}_*(\overline{\mscr{S}}_{\sns,\sra}/S_{\sns,\sra}).
\end{equation}
We pass to homology of this complex by taking the boundary operation on chains. We further turn the homology groups to a BV algebra by setting $\Delta=0$ \cite{Costello200509}. Projecting down to the part coming from connected genus-$\g$ surfaces, i.e. to $\overline{\mscr{S}}_{\sg,\sns,\sra}/S_{\sns,\sra}$, we have a morphism of linear spaces
\begin{equation*}
\pi :\mathcal F(\widetilde{\mscr{S}})\to \mathcal \H_{*}(\mcal F(\overline{\mscr{S}})).
\end{equation*}
In fact, we have

\begin{thr}[Relation of Solutions to the BV QME and Fundamental Classes of DM Stacks]\label{the:solutions to the BVQME and fundamental classes of DM spaces}
	$\pi$ is a homomorphism of {\normalfont BV} algebras, and $\pi(\widehat{\mcal{V}}_{\sg}(\ns,\ra))$ is the fundamental class $[\overline{\mscr{S}}_{\sg,\sns,\sra}/S_{\sns,\sra}]$.
\end{thr}

\begin{proof}
	The first statement follows from similar argument for Lemma 10.4.3 in \cite{Costello200509}. In the proof of Theorem \ref{the:the existence and uniqueness of the solution to the BV QME in spin moduli}, we have seen that the quotient map
	\begin{equation*}
	\mathcal F(\widetilde{\mscr{S}})\longrightarrow\H_*(\mathcal F(\widetilde{\mscr{S}})/F^2\mathcal F(\widetilde{\mscr{S}})),
	\end{equation*}
	takes ${\mcal{V}}_{\sg}(\ns,\ra)$ to the fundamental class of $\mbbmss{X}_\sg(\ns,\ra)$. Composing this map with a further quotient 
	\begin{equation*}
	\H_*(\mathcal F(\widetilde{\mscr{S}})/F^2\mathcal F(\widetilde{\mscr{S}}))\longrightarrow\H_*(\mathcal F(\widetilde{\mscr{S}})/F\mathcal F(\widetilde{\mscr{S}})),
	\end{equation*}
	and then projecting down to $\overline{\mscr{S}}_{\sg,\sns,\sra}/S_{\sns,\sra}$ sends ${\mcal{V}}_{\sg}(\ns,\ra)$ to the image of $[\mbbmss{X}_\sg(\ns,\ra)]$ in $\H_*(\mathcal F(\widetilde{\mscr{S}})/F\mathcal F(\widetilde{\mscr{S}}))$, which is $[\overline{\mscr{S}}_{\sg,\sns,\sra}/S_{\sns,\sra}]$. 
\end{proof}		
\newl On the other hand, as we explained, there is a homotopy-equivalence \eqref{eq:the homotopy equivalence of chain complexes associated to the S and tilde S} between the BV algebras $\mcal{F}(\mscr{S})$ and $\mcal{F}(\widetilde{\mscr{S}})$. Therefore, this fact together with Theorem \ref{the:solutions to the BVQME and fundamental classes of DM spaces} implies that there is a homomorphism in the homotopy category of BV algebras
\begin{equation*}
\pi:\mathcal F(\mscr{S})\to \H_{*}\left(\mathcal F(\overline{\mscr{S}})\right),
\end{equation*}
which takes ${\mcal{V}}_{\sg}(\ns,\ra)$, as certain subspaces of $\mscr{S}_{\sg}(\ns,\ra)/S^1\wr S_{\ns,\ra}$ which appear in \eqref{eq:the superstring vertices inside the complex associated to the moduli of spin curves}, to the fundamental class $[\overline{\mscr{S}}_{\sg,\sns,\sra}/S_{\sns,\sra}]$. This establishes the result we were looking for.

\section{Existence of Solutions to the BV QME in $\mscr{SM}(\ns,\ra)$}\label{sec:the existence of solution to the BVQME in supermoduli of bordered surfaces}
The main aim of this section is to give a proof for the existence of the solution to the BV QME in $\mscr{SM}$ (or more precisely $\mscr{SM}(\ns,\ra)/S^1\wr S_{\sns,\sra}$). The basic idea that we use is the one we have used in previous sections, i.e. we give $\mscr{SM}(\ns,\ra)$  the structure of a BV algebra, and then show that there is a quasi-isomorphisms between this BV algebra and the BV algebra associated to the underlying moduli space of bordered spin-Riemann surfaces. 

\newl One might ask {\it why we do not try the seemingly simpler method of lifting the solution to the BV quantum master equation in the space of bordered spin-Riemann surfaces to the corresponding space of bordered $\mscr{N}=1$ super-Riemann surfaces?} As we show in Section \ref{subsec:lifting subspace of reduced space}, the solution to the BV QME in $\mscr{S}(\ns,\ra)/S^1\wr S_{\sns,\sra}$ can always be lifted to $\mscr{SM}(\ns,\ra)/S^1\wr S_{\sns,\sra}$. However, this lifting is unique up to homotopy. As such, the lifting of a solution in $\mscr{S}(\ns,\ra)/S^1\wr S_{\sns,\sra}$ to $\mscr{SM}(\ns,\ra)/S^1\wr S_{\sns,\sra}$ does not necessary give us a solution to BV QME in $\mscr{SM}(\ns,\ra)/S^1\wr S_{\sns,\sra}$. One thus need to prove the existence of solution to BV QME in $\mscr{SM}(\ns,\ra)/S^1\wr S_{\sns,\sra}$ by a different method. One of these methods is using the quasi-isomorphism of BV algebras and the isomorphism of the set of homotopy-equivalent solutions to the BV QME of such BV algebras. There are other methods for such proof. We mention one of them in the end of Section \ref{subsec:lifting subspace of reduced space}. The background for this section is provided in Section \ref{subsec:definitions, elementary results, and super-analog of classical theorems}.

\subsection{Obstructions of Splitness and Projectedness}

It is interesting to ask when a superscheme $(X,\mathcal O_X)$ is projected or split. To simplify the life, we restrict to the case of smooth superschemes over $\mathbb C$. Since $\mathbb A^{p|q}$ is automatically split, the definition of smoothness together with Lemma \ref{Lemma_Superspace_EtaleLiftingSplitProjected} implies that smooth superschemes over $\mathbb C$ splits locally. So the question becomes whether local projections (resp. splittings) glue to a global projection (resp. splitting). This is similar to the situation of principal $G$-bundles, where bundles are locally trivializable but not necessarily globally trivial, and the obstruction is exactly captured by the cohomology class $\H ^1(X,G)$. Here the gauge group $G$ is replaced by the group of automorphisms of local models, i.e. $\Aut(\wedge ^{\bullet}\mscr V)$ for a free $\mathcal O_{X_{\text{\normalfont bos}}}$-sheaf $\mscr V$.

\newl To analyse this problem in detail, we adopt the usual set-up for lifting problems: let $X_n$ be the superscheme $(X,\mathcal{O}_X/\mathfrak{n}_X^{n})$, and assume that $X_n$ splits\footnote{$X_1$ is just $X_{\text{\normalfont bos}}$ and $X_2$ always splits.}, i.e. $\mathcal O_{X_n}=\oplus_{i=0}^{n-1}\wedge^i\mscr V$. Let $G_n$ be the subgroup of $\Aut(\wedge^{\bullet}\mathcal V)$ whose image in $\Aut(\oplus_{i=0}^{n-1}\wedge^i\mscr V)$ is the identity, and let $H_n\equiv G_n/G_{n+1}$. Then we fix a cover $\{U_i\}$ of $X$ such that $\mathcal O_{X_{n+1}}$ splits on $U_i$. On overlap $U_{ij}\equiv U_i\cap U_j$, splittings on $U_i$ and $U_j$ are related by an isomophism of $\oplus_{i=0}^{n+1}\wedge^i\mscr V$ which maps to identity in $\Aut(\oplus_{i=0}^{n-1}\wedge^i\mscr V)$, i.e. an element $h_{ij}$ in $H_n$; moreover on $U_{ijk}\equiv U_i\cap U_j\cap U_k$, $h_{ij}h_{jk}=h_{ik}$, hence $\{h_{ij}\}$ defines a class in $\H^1(X,H_n)$, and is trivial if and only if there is a 
collection $\{h_i\in \Gamma(U_i,H_n)\}$, such that $h_{ij}=h_ih_j^{-1}$, whence is equivalent to a splitting of $X_{n+1}$ over $X_n$.

\newl On the other hand, it is easy to see that
\begin{equation}
H_n\cong T^{(-)^n}_X\otimes _{\mathcal O_{X_{\text{\normalfont bos}}}}\bigwedge^n\mscr V,
\end{equation}
where $T^+_X\equiv T_{X_{\text{\normalfont bos}}}$ and $T^-_X\equiv \mscr V^{\vee}$. To sum up, we have 

\begin{proposition}
	Assume $X_n$ splits, then there is an element
	\begin{equation}
	\text{ob}(X_n)\in \H^1\left(X,T^{(-)^n}_X\otimes \bigwedge^n\mscr V\right),
	\end{equation}
	whose vanishing is necessary and sufficient for the splitting of $X_{n+1}$. If it vanishes, then the space of splittings of $X_{n+1}$ is a torsor\footnote{For the definition of torsor see \cite{torsorstackproject}.} under
	\begin{equation}
	\H^0\left(X,T^{(-)^n}_X\otimes \bigwedge^n\mscr V\right).
	\end{equation}
\end{proposition}

\begin{corollary}
	If $\H^1\left(X,T^{(-)^n}_X\otimes \wedge^n\mathcal V\right)$ vanishes for all $n\ge 2$, then $X$ splits, and the space of splittings is a torsor under
	\begin{equation}
	\H^0\left(X,\bigoplus_{n\ge 1}T^{(-)^n}_X\otimes \bigwedge^n\mscr V\right).
	\end{equation}
\end{corollary}

For projectedness, the analysis is analogous. Thanks to Lemma \ref{Lemma_Superspace_Projected}, we can forget about the odd grading part and analyze the projection problem for embedding $X_{\text{\normalfont bos}}\hookrightarrow X_{\text{\normalfont ev}}$ of ordinary schemes.
Let $X_n^0$ be the scheme $(X,\mathcal{O}_X^0/(\mathcal{O}_X^1)^{2n})$, and assume that $X_n^0$ splits, i.e. $\mathcal O_{X_n^0}=\oplus_{i=0}^{n-1}\wedge^{2i}\mathcal V$. An argument similar to the splitness case that we will skip shows the following:
\begin{proposition}
	Assume $X_n^0$ splits, then there is an element
	\begin{equation}
	\text{ob}(X_n^0)\in \H^1\left(X,T_{X_{\text{\normalfont bos}}}\otimes \bigwedge^{2n}\mscr V\right),
	\end{equation}
	whose vanishing is necessary and sufficient for the projectedness of $X_{n+1}^0$. If it vanishes, then the space of projection of $X_{n+1}^0$ onto $X_{n}^0$ is a torsor under
	\begin{equation}
	\H^0\left(X,T_{X_{\text{\normalfont bos}}}\otimes \bigwedge^{2n}\mscr V\right).
	\end{equation}
\end{proposition}

\begin{corollary}
	If $\H^1\left(X,T_{X_{\text{\normalfont bos}}}\otimes \wedge^{2n}\mathcal V\right)$ vanishes for all $n\ge 1$, then $X$ is projected, and the space of projections is a torsor under
	\begin{equation}
	\H^0\left(X,\bigoplus_{n\ge 1}T_{X_{\text{\normalfont bos}}}\otimes \bigwedge^{2n}\mscr V\right).
	\end{equation}
\end{corollary}

\begin{remark}\label{rem:the definition of supermanifold}
	A supermanifold is defined to be a superspace $X$ such that $X_{\text{\normalfont bos}}$ is a $C^{\infty}$-manifold and locally on $X$, $\mathcal O_X$ is isomorphic to $\wedge^{\bullet} V$ for a vector bundle $V$ on $X_{\text{\normalfont bos}}$\footnote{What is called supermanifold in \cite{DonagiWitten2013a} is a special class of superspace defined in this work, and $C^{\infty}$-supermanifold in the sense of \cite{DonagiWitten2013a} is simply called supermanifold here.}. The above arguments also apply when $X$ is a supermanifold. In fact, $\H^1\left(X,T^{(-)^n}_X\otimes \wedge^n  V\right)$ vanishes automatically since sections of vector bundles form a fine sheaf. Hence supermanifolds are always split.
\end{remark}

In general it is hard to compute the the obstruction class $\text{ob}(X_n)$ or $\text{ob}(X_n^0)$ except for small dimensions and small $n$. In \cite{DonagiWitten2013a}, an interesting example is constructed as follows: assume that $\g>1$, then there is a morphism $f:\mscr{SM}_{\sg,1,0}\to \mscr{SM}_{\sg,0,0}$ which is defined by the projection from the universal super-Riemann surfaces (SRS) to the base. Choose a SRS $\mscr C$ whose spin structure is denoted by $\mscr E$, and an element $\eta\in \H^1 (C,\mscr E^{\vee})$, then $\eta$ is an odd tangent vector at the point $[\mscr C]\in \mscr{SM}_{\sg,0,0}(\mathbb C)$, which can be identified with a superaffine space $\mathbb A^{0|1}$, denoted by $\mathbb C_{\eta}$. Consider $\mscr C_{\eta}\equiv f^{-1}\mathbb C_{\eta}$. It is a smooth superscheme of dimension $(1|2)$ and is a first order odd deformation of $\mscr C$. Now the obstruction class $\text{ob}(C_{\eta,2})$ lives in the cohomology
\begin{equation}
\H^1\left(\mscr C,T_{\mscr C}\otimes \wedge^2(\mscr E\oplus \mathcal O_{\mscr C})\right)\cong \H^1 (\mscr C,\mathcal E^{\vee}).
\end{equation}
An explicit computation in Lemma 3.5 of \cite{DonagiWitten2013a} shows that $\text{ob}(\mscr C_{\eta,2})$ is identified with $\eta$ via this isomorphism, in particular, it is nontrivial if $\eta\neq 0$.

\newl In \cite{DonagiWitten2013a}, Donagi and Witten used this example and a comparison between obstruction classes of an ambient superscheme and its smooth closed super subscheme (\textit{loc. cit.} Corollary 2.10) to conclude that for $\g\ge 2$, the even spin component $\mscr{SM}_{\sg,1,0}^{\text{\normalfont ev}}$ is non-projected. Moreover, for $\g\ge 5$, they construct a class of non-split super subschemes of supermoduli $\mscr{SM}_{\sg,0,0}$ to conclude that the supermoduli in these cases are non-projected. This example has an interesting consequence which is slightly stronger than non-projectedness.
\begin{proposition}
	Assume $g\ge 2$, then there is no projection $\pi: \mscr{SM}_{\sg,1,0}^{\text{\normalfont ev}}\to \mscr{M}_{\sg,1}$ such that the following diagram commutes
	\begin{center}
		\begin{tikzcd}
		\mscr{S}_{\sg,1,0}^{\text{\normalfont ev}}\arrow[r,"i", hook] \arrow[dr, "f" swap] &\mscr{SM}_{\sg,1,0}^{\text{\normalfont ev}}\arrow[d,"\pi"]\\
		&\mscr{M}_{\sg,1}.
		\end{tikzcd}
	\end{center}
\end{proposition}

\begin{proof}
	Assume that there is such projection $\pi: \mscr{SM}_{\sg,1,0}^{\text{\normalfont ev}}\to \mscr{M}_{\sg,1}$ making the diagram commutative. Choose a spin curve $\mscr C$ of genus $\g$ with even spin structure $\mscr E$ such that $\H^1(\mscr C,\mscr  E)=0$. Note that a generic choice of $[\mscr  C]$ in $\mscr {S}_{\sg,0,0}^{\text{\normalfont ev}}(\mathbb C)$\footnote{$\mscr {S}_{\sg,0,0}^{\text{\normalfont ev}}(\mathbb C)$ denotes complex points of the stack $\mscr {S}_{\sg,0,0}^{\text{\normalfont ev}}$, i.e. $\mscr C$ is a spin curve over complex numbers.} will satisfy $\H^1(\mscr C,\mscr  E)=0$ since the locus in $\mscr{S}_{\sg,0,0}^{\text{\normalfont ev}}$ where $\H^1(\mscr  C,\mscr E)=0$ is open by semicontinuity of cohomology, and nonempty by explicit construction on a hyperelliptic curve\footnote{The divisor $\mathfrak{b}_S^{(-1)}$ in the section $4$ of \cite{Mumford1971a}, for example.}. Now consider the commutative diagram
	\begin{center}
		\begin{tikzcd}[row sep=scriptsize, column sep=scriptsize]
		& \mscr {S}_{\sg,1,0} ^{\text{\normalfont ev}}\arrow[rr, hook] \arrow[dd]  \arrow[dr, "f" swap] & & \mscr{SM}_{\sg,1,0}^{\text{\normalfont ev}} \arrow[dl, "\pi"] \arrow[dd] \\
		& & \mscr{M}_{\sg,1} \\
		& \mscr{S}_{\sg,0,0} ^{\text{\normalfont ev}}\arrow[rr, hook] \arrow[dr, "f'" swap] & & \mscr{SM}_{\sg,0,0}^{\text{\normalfont ev}}  \\
		& & \mscr{M}_{\sg,0} \arrow[from=uu, crossing over].
		\end{tikzcd}
	\end{center}
	Note that $f$ and $f'$ are \' etale coverings. Take an \'etale neighborhood $U$ of $[\mscr  C]\in \mscr {M}_{\sg,0}(\mathbb C)$\footnote{$\mscr {M}_{\sg,0}(\mathbb C)$ denotes the complex points of the stack $\mscr {M}_{\sg,0}$, i.e. $\mscr C$ is a genus $\sg$ curve over complex numbers.} such that $f'^{-1}(U)$ is isomorphic to a disjoint union of $U$. Take the component of $f'^{-1}(U)$ which contains the chosen spin structure, and call it $U_1$. Denote the base change of $U$ (resp. $U_1$) to $\mscr{M}_{\sg,1}$ (resp. $\mscr{S}_{\sg,1,0} ^{\text{\normalfont ev}}$) by $V$ (resp. $V_1$). Then, there is a commutative diagram
	\begin{center}
		\begin{tikzcd}
		V_1\arrow[r,"i", hook] \arrow[dr, "f|_{V_1}" swap] &\pi^{-1}(V)\arrow[d,"\pi"]\\
		&V.
		\end{tikzcd}
	\end{center}
	Since $f|_{V_1}$ is an isomorphism, this diagram says that $\pi^{-1}(V)$ is projected. Note that given a non-zero $\eta\in \H^1 (\mscr  C,\mscr  E^{\vee})$, the superscheme $\mscr  C_{\eta}$ is still a super subscheme of $\pi^{-1}(V)$, hence the same argument in Proposition 4.1 of \cite{DonagiWitten2013a} shows that $\pi^{-1}(V)$ is not projected, which is a contradiction. This concludes the proof.
\end{proof}

\subsubsection{Lifting Space of the Bosonic Truncation of a Supermanifold}\label{subsec:lifting subspace of reduced space}

As we have seen so far, we could find certain subspaces ${\mcal{V}}_{\sg}(\ns,\ra)$ of the moduli stack of bordered spin-Riemann surfaces $\mscr{S}_{\sg}(\ns,\ra)$ (or more precisely $\mscr{S}_{\sg}(\ns,\ra)/S^1\wr S_{\sns,\sra}$) with the following properties
\begin{enumerate}
	\item ${\mcal{V}}_{\sg}(\ns,\ra)$ satisfies the BV QME;
	
	\item Under a map in the homotopy category of BV algebras between chain complexes defined in \eqref{eq:the chain complexes associated to S and tilde S}, ${\mcal{V}}_{\sg}(\ns,\ra)$ is mapped to the fundamental class of the stable-curve compactification of the moduli stack of punctured spin curves;
\end{enumerate}
The question that we are interested to address is whether we can lift such solutions to $\mscr{S}_{\sg}(\ns,\ra)/S^1\wr S_{\sns,\sra}$ and then interpret the lifting as a solution to the BV QME in $\mscr{S}_{\sg}(\ns,\ra)/S^1\wr S_{\sns,\sra}$? We thus need to show that the lifting is unique. We investigate this question in the following. 

\newl In a precise language, we would like to understand the following question: given a smooth superscheme $(X,\mathcal O_X)$, and a smooth closed subscheme $Y_{\text{bos}}\subset X_{\text{bos}}$\footnote{From the definition \ref{Defn_Superspace}, $X_\text{bos}$ is the locally-ringed topological space $(X,\mathcal O_X/\mathfrak{n}_X)$ of codimension $n$, where $\mfk{n}_X$ is the ideal generated by the odd part $\mathcal O_X^1$ of the sheaf $\mcal{O}_X$.}, what is the deformation-obstruction of lifting $Y_{\text{bos}}$ to a closed super subscheme $Y\subset X$ of codimension $(n|0)$? To simplify our discussion, we assume that $X$ is a scheme.

\newl Let us start with assuming that $X_{\text{bos}}$ is affine and $X$ is split, i.e. there is an isomorphism 
\begin{equation}
\mathcal O_X\cong \mathcal O_{X_{\text{bos}}}\oplus \bigoplus _{i\ge 1}^m{\bigwedge}^i\mscr V,
\end{equation}
for a locally-free sheaf $\mscr V$ of rank $m$, and $Y_{\text{bos}}$ is generated by $\{f_1,f_2,\cdots,f_n\}\subset \Gamma(X,\mathcal O_{X_{\text{bos}}})$, a regular sequence of length $n$. Then there is a canonical lifting of $Y_{\text{bos}}$, namely by pulling back $Y_{\text{bos}}$ along the projection $X\to X_{\text{bos}}$. This is equivalent to setting the ideal defining $Y$ to be generated by the same set of generators $\{f_1,f_2,\cdots,f_n\}$. Suppose that there is another lifting $Y'$  defined by $n$ even functions $\{f_1',f_2',\cdots,f_n'\}\subset \Gamma(X,\oplus _{i\ge 0}{\wedge}^{2i}\mscr V)$ such that $f_i'-f_i\in \Gamma(X,\oplus _{i\ge 1}{\wedge}^{2i}\mscr V)$. Moreover, two sets $\{f_1',f_2',\cdots,f_n'\}$ and $\{f_1'',f_2'',\cdots,f_n''\}$ generates the same ideal if and only if $f_i''-f_i'\in \{f_1,f_2,\cdots,f_n\}\cdot \Gamma(X,\oplus _{i\ge 1}{\wedge}^{2i}\mscr V)$. Hence the liftings have one to one correspondence to the Hom set 
\begin{align}
\Hom _{\mathcal O_{X_{\text{bos}}}}\left(\mathcal I_{Y_{\text{bos}}},\bigoplus _{i\ge 1}^m{\bigwedge}^{2i}\mscr V\big|_{Y_{\text{bos}}}\right)
\end{align}
where $\mathcal I_{Y_{\text{bos}}}$ is the ideal defining $Y_{\text{bos}}$. 

\newl If $X_{\text{bos}}$ is not necessarily affine, then we need to lift $Y_{\text{bos}}$ order by order, and we possibly encounter obstructions. Let $X_k$ be the commutative space $(X,\mathcal O_X^{+}/(\mathcal O_X^-)^{2k+2})$. From our previous discussion, we know that to lift $Y_{\text{bos}}$ to a subspace of $X$, it suffices to lift it to a subspace of $X_{\text{ev}}$\footnote{The locally-ringed topological space $(X,\mathcal O_X^0)$ is the \textbf{bosonic quotient} of $X$, denoted by $X_{\text{\normalfont ev}}$. For more on this, see Section \ref{subsec:definitions, elementary results, and super-analog of classical theorems}.}, by lifting it to $X_k$ recursively. Suppose that there is a lifting $Y_k\subset X_k$, we want to know the deformation-obstruction space of lifting it to $X_{k+1}$, this  is completely classical, in fact we have:
\begin{proposition}
	There is an element 
	\begin{equation}
	\text{ob}(Y_k)\in \H ^1\left(Y_{\text{bos}},\mathcal N_{Y_{\text{bos}}/X_{\text{bos}}}\otimes{\bigwedge}^{2k+2}\mscr V\big|_{Y_{\text{bos}}}\right),
	\end{equation}
	whose vanishing is necessary and sufficient for the existence of lifting $Y_{k+1}$ of $Y_k$. $\mathcal N_{Y_{\text{bos}}/X_{\text{bos}}}$ is the normal bundle of $Y_{\text{bos}}$ in $X_{\text{bos}}$. If it vanishes, then the space of lifting $Y_{k+1}$ is a torsor under
	\begin{equation}
	\Hom _{\mathcal O_{X_{\text{bos}}}}\left(\mathcal I_{Y_{\text{bos}}},{\bigwedge}^{2k+2}\mscr V|_{Y_{\text{bos}}}\right)=\H ^0\left(Y_{\text{bos}},\mathcal N_{Y_{\text{bos}}/X_{\text{bos}}}\otimes{\bigwedge}^{2k+2}\mscr V\big|_{Y_{\text{bos}}}\right).
	\end{equation}
\end{proposition}

\begin{proof}
	This is essentially Theorem 6.4.5 of \cite{FantechiIllusieGottsche2005} with modifications. We will basically follow \textit{loc. cit.}
	
	\newl Step (a): construction of $\text{ob}(Y_k)$. We have following diagram of coherent sheaves on $X_{k+1}$ with exact row and columns
	\begin{center}
		\begin{tikzcd}
		&0\arrow[d] & &0 \arrow[d] \\
		
		&\mathcal I_{Y_{\text{bos}}}\otimes{\bigwedge}^{2k+2}\mscr V\arrow[d] & &\mathcal I_{Y_{k}}\arrow[d] \\
		
		0 \arrow[r]& {\bigwedge}^{2k+2}\mscr V\arrow[d] \arrow[r] &\mathcal O_{X_{k+1}} \arrow[r] &\mathcal O_{X_{k}} \arrow[r] \arrow[d] &0,\\
		
		&{\bigwedge}^{2k+2}\mscr V|_{Y_{\text{bos}}} \arrow[d]& &\mathcal O_{Y_k}\arrow[d]\\
		
		&0& &0
		\end{tikzcd}
	\end{center}
	Consider the induced maps $\alpha:\mathcal I_{Y_{\text{bos}}}\otimes{\bigwedge}^{2k+2}\mathcal V\to \mathcal O_{X_{k+1}}$ and $\beta :\mathcal O_{X_{k+1}}\to \mathcal O_{Y_k}$, and we also have $\im(\alpha)\subset \ker(\beta)$. Define $\widetilde {\mathcal O}_{X_{k+1}}:=\ker(\beta)/\im(\alpha)$ and we have short exact sequence of coherent sheaves on $X_{k+1}$
	\begin{center}
		\begin{tikzcd}
		0 \arrow[r]& {\bigwedge}^{2k+2}\mscr V|_{Y_{\text{bos}}} \arrow[r] &\widetilde {\mathcal O}_{X_{k+1}} \arrow[r] &\mathcal I_{Y_{k}} \arrow[r] &0,
		\end{tikzcd}
	\end{center}
	Similar to \textit{loc. cit.} we need to check that this is in fact a sequence of $\mathcal O_{X_k}$ sheaves. Since this is a local property and we can assume that $X$ splits and $\mscr V$ is free. In this case $X_k$ is isomorphic to $X_{\text{bos}}\times \Spec A_k$ for an Artinian ring $A_k$ and this goes back to the situation of \textit{loc. cit.} hence concludes the proof.
	
	\newl Step (b). It is enough to prove that extensions are in a natural bijection with splittings of the above short exact sequence. Assume we are given a splitting $\xi: \mathcal I_{Y_{k}}\to \widetilde {\mathcal O}_{X_{k+1}}$, and let $\mathcal I_{Y_{k+1}}$ be the preimage of $\xi( \mathcal I_{Y_{k}})$ in $\ker(\beta)\subset\mathcal O_{X_{k+1}}$. We can then complete the H-shaped diagram to a commutative diagram with exact rows and columns
	\begin{center}
		\begin{tikzcd}
		&0\arrow[d] & 0\arrow[d] &0 \arrow[d] \\
		
		0 \arrow[r]&\mathcal I_{Y_{\text{bos}}}\otimes{\bigwedge}^{2k+2}\mscr V\arrow[d] \arrow[r]& \mathcal I_{Y_{k+1}} \arrow[d] \arrow[r] &\mathcal I_{Y_{k}}\arrow[d] \arrow[r] &0 \\
		
		0 \arrow[r]& {\bigwedge}^{2k+2}\mscr V\arrow[d] \arrow[r] &\mathcal O_{X_{k+1}} \arrow[r] \arrow[d] &\mathcal O_{X_{k}} \arrow[r] \arrow[d] &0,\\
		
		0 \arrow[r]&{\bigwedge}^{2k+2}\mscr V|_{Y_{\text{bos}}} \arrow[d] \arrow[r]& \mathcal O_{Y_{k+1}} \arrow[d]\arrow[r]&\mathcal O_{Y_k}\arrow[d] \arrow[r] &0\\
		
		&0 &0 &0 
		\end{tikzcd}
	\end{center}
	and $\mathcal O_{Y_{k+1}}$ is the required lifting. To show that it's generated by a regular sequence of length $n$, we can again assume that $X$ splits and $\mscr V$ is free. In this case $Y_{k+1}$ is flat over $A_{k+1}$ (see \textit{loc. cit.}) hence the statement follows from \textit{loc. cit.} Lemma 9.3.4\footnote{There $n$ is set to $1$ but the proof presented there works for general $n$.}. Conversely, given a lifting $\mathcal O_{Y_{k+1}}$ we can build such commutative diagram as above by exactness and the subsheaf $\mathcal I_{Y_{k+1}}/\im (\alpha)\subset \ker(\beta)/\im(\alpha)$ maps isomorphically to $\mathcal I_{Y_{k}}$ hence gives a splitting of $\widetilde {\mathcal O}_{X_{k+1}} \to \mathcal I_{Y_{k}}$.
	
	\newl Step (c). We have only shown that $\text{ob}(Y_k)\in \Ext ^1 _{\mathcal O_{X_{\text{bos}}}}\left(\mathcal I_{Y_{\text{bos}}},{\bigwedge}^{2k+2}\mscr V\big|_{Y_{\text{bos}}}\right)$. To proceed further we notice that there is a spectral sequence \begin{equation*}
	\Ext ^q_{\mathcal O_{Y_{\text{bos}}}}\left(\Tor ^{\mathcal O_{X_{\text{bos}}}}_p\left(\mathcal I_{Y_{\text{bos}}},\mathcal O_{Y_{\text{bos}}}\right),{\bigwedge}^{2k+2}\mscr V\big|_{Y_{\text{bos}}}\right)\Longrightarrow _p \Ext ^{p+q} _{\mathcal O_{X_{\text{bos}}}}\left(\mathcal I_{Y_{\text{bos}}},{\bigwedge}^{2k+2}\mscr V\big|_{Y_{\text{bos}}}\right),
	\end{equation*}
	which gives rise to a short exact sequence 
	\begin{alignat*}{2}
	&0\longrightarrow \H ^1\left(Y_{\text{bos}},\mathcal N_{Y_{\text{bos}}/X_{\text{bos}}}\otimes{\bigwedge}^{2k+2}\mscr V\big|_{Y_{\text{bos}}}\right) \longrightarrow \Ext ^1 _{\mathcal O_{X_{\text{bos}}}}\left(\mathcal I_{Y_{\text{bos}}},{\bigwedge}^{2k+2}\mscr V\big|_{Y_{\text{bos}}}\right)
	\\
	\longrightarrow &\Hom _{\mathcal O_{Y_{\text{bos}}}}\left(\Tor ^{\mathcal O_{X_{\text{bos}}}}_1\left(\mathcal I_{Y_{\text{bos}}},\mathcal O_{Y_{\text{bos}}}\right),{\bigwedge}^{2k+2}\mscr V\big|_{Y_{\text{bos}}}\right)=\Gamma\left(Y_{\text{bos}},{\bigwedge}^{2}\mathcal N_{Y_{\text{bos}}/X_{\text{bos}}}\otimes{\bigwedge}^{2k+2}\mscr V\big|_{Y_{\text{bos}}}\right).
	\end{alignat*}
	Take an open covering $\{U_i\}\to X_{\text{bos}}$ such that $U_i$ are affine and split. From our previous discussions, we see that the obstruction vanishes on $U_i$, in particular the image of $\text{ob}(Y_k)$ in $\Gamma(Y_{\text{bos}},{\wedge}^{2}\mathcal N_{Y_{\text{bos}}/X_{\text{bos}}}\otimes{\wedge}^{2k+2}\mscr V|_{Y_{\text{bos}}})$, denoted by $\overline{\text{ob}(Y_k)}$, vanishes on $U_i$. Since $U_i$ is a covering, this implies that $\overline{\text{ob}(Y_k)}$ is trivial, hence $\text{ob}(Y_k)$ lives in the subspace
	\begin{equation}
	\H ^1\left(Y_{\text{bos}},\mathcal N_{Y_{\text{bos}}/X_{\text{bos}}}\otimes{\bigwedge}^{2k+2}\mscr V\big|_{Y_{\text{bos}}}\right).
	\end{equation}
	This concludes the proof.
\end{proof}

\begin{corollary}\label{cor:the lifting is a torsor}
	If $\H ^1\left(Y_{\text{bos}},\mathcal N_{Y_{\text{bos}}/X_{\text{bos}}}\otimes{\wedge}^{2k+2}\mscr V|_{Y_{\text{bos}}}\right) $ vanishes for all $k$, then the set of lifting of $Y_{\text{bos}}$ is a torsor under 
	\begin{equation}
	\Hom _{\mathcal O_{X_{\text{bos}}}}\left(\mathcal I_{Y_{\text{bos}}},\bigoplus _{i\ge 1}^m{\bigwedge}^{2i}\mscr V|_{Y_{\text{bos}}}\right).
	\end{equation}
\end{corollary}
The case we are interested is when $X$ is a supermanifold, in the sense of Remark \ref{rem:the definition of supermanifold}. Above arguments also apply when $X$ is a supermanifold, the case that we are interested in. In that case, the coholomoly group $\H ^1\left(Y_{\text{bos}},\mathcal N_{Y_{\text{bos}}/X_{\text{bos}}}\otimes{\bigwedge}^{2k+2}\mscr V|_{Y_{\text{bos}}}\right) $ vanishes automatically since sections of vector bundles form a fine sheaf. So the space of liftings has the structure of an affine space, hence contractible. As a consequence, any two liftings can be connected by a homotopy.

\newl We thus end-up with the conclusion we were looking: We have shown that 
\begin{enumerate}
	\item We can always a solution to the BV QME in $\mscr{S}_{\sg}(\ns,\ra)/S^1\wr S_{\sns,\sra}$ to $\mscr{SM}_{\sg}(\ns,\ra)/S^1\wr S_{\sns,\sra}$. The lifting is unique up to homology;
	
	\item Precisely because the lifted solution is defined up to homology, we cannot consider the lifted solution as a solution to the BV QME in $\mscr{SM}_{\sg}(\ns,\ra)/S^1\wr S_{\sns,\sra}$.
\end{enumerate}

\subsection{Singular Homology on Superspaces}\label{subsec:singular homology on superspaces}
In this section, we discuss singular homology on superspaces. We mainly follow \cite{VoronovZorich198703} and \cite{Voronov1991}. From now on we focus on  cs-manifolds\footnote{Here 'cs' means complex-super but be aware that this is by definition a real manifold with a complex structure sheaf. Dimension $(p|q)$ means real even dimension $p$ and complex odd dimension $q$.} (possibly with boundaries and corners). Consider the $n$-dimensional simplex $\Delta^n$ and embed it into a cs-manifold $\Delta^{n|m}$ defined by the trivial $m$-dimensional bundle on $\Delta^n$ as its first order part. 
\begin{definition}[Singular Superchains on Superspaces]\label{def:the definition of singular superchains on superspaces}
	Fix a cs-manifold $M$ of dimension $(p|q)$ such that $\mathcal V\equiv \mathcal O_M[1]$ is a complex vector bundle of dimension $q$. A singular supersimplex $\sigma$ of dimension $(n|m)$ on M is a morphism of superspaces $\sigma:\Delta^{n|m}\to M$ such that the associated bundle map $\underline{\mbb C}^m\to\sigma^*\mathcal V$ is \textbf{injective}. Define the singular superchain complex valued in a ring $R$, denoted by $C_{n|m}(M;R)$, to be the free $R$ module with basis singular supersimplices on $M$ of dimension $(n|m)$. And also define the boundary map $\partial_{n|m}:C_{n|m}(M;R)\to C_{n-1|m}(M;R)$ by $$\partial_{n|m}\sigma:=\sum_{i=0}^n(-1)^i\partial_i\sigma$$where $\partial_i\sigma$ is the singular supersimplex induced by restriction of $\sigma$ to sub-superpace $\Delta^{n-1|m}$ of $\Delta^{n|m}$, which is topologically the i'th boundary of $\Delta ^n$ with the same orientation as in the definition of singular chain complex.\\
	
	The singular homology of $M$ of fermionic degree $m$ with coefficient ring $R$, denoted by $\H_{\bullet |m}(M;R)$ is defined to be the homology of the complex $\left(C_{\bullet |m}(M;R),\partial_{\bullet|m}\right)$.
\end{definition}
We then have the following result
\begin{proposition}\label{Homology_SupMfd}
	Fix a cs-manifold $M$ of dimension $(p|q)$ such that $\mathcal V\equiv \mathcal O_M[1]$ is a complex vector bundle of dimension $q$. We have isomorphism:
	$$\H_{n|m}(M;R)\cong \H_n\left(\text{\normalfont Grass}_m(\mathcal V);R\right),$$
	in particular, $\H_{n|0}(M;R)\cong\H_{n|q}(M;R)\cong \H_{n}(M_{\text{\normalfont bos}};R)$. Here $\text{\normalfont Grass}_m(\mathcal V)$ means the Grassmannian bundle of $m$ dimensional subspaces of $\mathcal V$.
\end{proposition}

\begin{proof}
	The proof goes similarly to the proof of Theorem 3 and 4 of \cite{VoronovZorich198703}. We apply Corollary 1 of \cite{VoronovZorich198703} and see that the chain complex $\left(C_{\bullet |m}(M;R),\partial_{\bullet|m}\right)$ is quasi-isomorphic to the chain complex $\left(C_{n,m}',d_n\right)$ of free $R$ modules generated by $n$-simplices on $M_{\text{bos}}$, together with a choice of $m$-dimensional subbundle of pull-back of $\mathcal V$ to the simplex, and $d_n$ is alternating sum of restriction to boundaries. It is easy to see that this is exactly $\left(C_{\bullet}(\text{\normalfont Grass}_m(\mathcal V),R),\partial_{\bullet}\right)$.
\end{proof}

\begin{remark}
	The definition of homology does not guarantee that for any morphism between cs-manifolds $f:M\to N$, there is a natural morphism of homologies. However, if the associated bundle map $\mathcal V_M\to f^*\mathcal V_N$ is \textbf{injective}, then $f$ indeed induces a natural morphism of homologies $f_{*}:\H_{n|m}(M;R)\to \H_{n|m}(N;R)$. In fact, in this case the isomorphism in {\normalfont Proposition \ref{Homology_SupMfd}} is natural, i.e. there is a commutative diagram
	\begin{center}
		\begin{tikzcd}
		\H_{n|m}(M;R)\arrow[r,"\cong"] \arrow[d, "f_*", swap] &\H_n\left(\text{\normalfont Grass}_m(\mathcal V_M);R\right)\arrow[d,"f_*"]\\
		\H_{n|m}(N;R)\arrow[r,"\cong"]  &\H_n\left(\text{\normalfont Grass}_m(\mathcal V_N);R\right).
		\end{tikzcd}
	\end{center}
\end{remark}

\newl Using Proposition \ref{Homology_SupMfd}, we can prove the existence of solution to the BV QME in the moduli space of bordered $\mscr{N}=1$ super-Riemann surfaces. For that purpose, we can define a further generalization of the Kimura et al spaces for the moduli space of bordered $\mscr{N}=1$ super-Riemann surface. As asserted in Proposition \ref{Homology_SupMfd}, the homology group with top fermionic dimensions for these spaces are isomorphic to their reduced space, which is the generalization of the Kimura et al spaces for bordered spin-Riemann surface we introduced in Section \ref{subsec:an alternative model for space of bordered spin-Riemann surfaces}, the proof for the existence of solution to the BV QME in the moduli space of bordered $\mscr{N}=1$ super-Riemann surfaces goes through without change. We will not attempt to reproduce that proof here. Instead, we present a shorter and more elegant proof for the existence of such solutions.   

\subsection{Proof of the Existence of Solution}\label{subsec:proof of the existence of solution to BVQME in the supermoduli of bordered surfaces}

In this section we give the proof for the existence of solution to the BV QME in $\mscr{SM}(\ns,\ra)/S^1\wr S_{\sns,\sra}$. Consider the following complex
\begin{equation}
\mcal{F}(\mscr{SM})\equiv \bigoplus_{\sns,\sra}\mscr{SC}_*\left(\frac{\mscr{SM}(\ns,\ra)}{S^1\wr S_{\sns,\sra}}\right),
\end{equation}
where $\mscr{SC}_*$ is the functor of normalized singular simplicial superchains of top fermionic dimension with coefficients in any field containing $\mbb{Q}$. This complex has the structure of a dg algebra. As usual the product comes from the disjoint union of surfaces, which gives $\mcal{F}(\mscr{SM})$ an algebra structure, and the differential comes from the boundary operation on singular superchains, as defined in Definition \ref{def:the definition of singular superchains on superspaces}. We would like to give $\mcal{F}(\mscr{SM})$ the structure of a BV algebra. For that purpose, we define the following gluing operations of boundaries: let $\widetilde\Delta^{\NS}_{\mfk{s}}$ (respectively $\widetilde\Delta^{\RA}_{\mfk{s}}$) be the operation of gluing two NS (respectively, R) boundaries with a full $S^1$ twist. There is one additional complication in defining $\widetilde\Delta^{\RA}_{\mfk{s}}$, namely we let it increase the odd dimension of the superchain by 1. The question is that why we can do that for R boundaries while we do not have this freedom for NS boundaries? From Definition \ref{def:R boundary components} of R boundaries, they are isomorphic to $S^{1}\times \mbb{R}^{0|1}\otimes_{\mbb{R}}\mbb{C}$, and as such there is a translation symmetry along the only complex odd direction of $\mbb{R}^{0|1}\otimes_{\mbb{R}}\mbb{C}$. This also makes sure that while we glue R boundaries, the resulting surface does have correct number of global fermionic parameters. From Definition \ref{def:NS boundary components} of NS boundaries there is no such freedom.

\newl We would like to consider the operation of gluing as an operation acting on singular superchains. For this purpose, consider a singular superchain associated to $\mbbmss{SM}(\ns,\ra)\equiv\mscr{SM}(\ns,\ra)/S^1\wr S_{\sns,\sra}$, which we denote by $C_{m|d_f}(\mbbmss{SM}(\ns,\ra);\mbb{F})$. Note that $d_f=2\g-2+\ns+\ra/2$ is the top fermionic dimension of the part of $\mbbmss{SM}(\ns,\ra)$ coming from connected genus-$\g$ surfaces. This singular superchain is generated by the $(m|d_f)$-singular supersimplices $\sigma^{m|d_f}_{\mfk{s}}:\Delta^{m|d_f}\longrightarrow \mbbmss{SM}(\ns,\ra)$, i.e. a morphism of superspaces $\sigma^{m|d_f}_{\mfk{s}}:\Delta^{m|d_f}\longrightarrow\mbbmss{SM}(\ns,\ra)$ with properties given in Definition \ref{def:the definition of singular superchains on superspaces}. We would like to define the operations $\widetilde\Delta^{\NS}_{\mfk{s}}$ and $\widetilde\Delta^{\RA}_{\mfk{s}}$ on the singular superchains. We thus consider the operations corresponding to $\widetilde\Delta^{\NS}_{\mfk{s}}$ and $\widetilde\Delta^{\NS}_{\mfk{s}}$ for all values of $\ns$ and $\ra$ which are acting on singular superchains and define\footnote{These operations are not canonically defined. See Footnote 6 on page 5 of \cite{Costello200509}.}
\begin{alignat}{2}
\Delta^{\NS}_{\mfk{s}}&:\mcal{F}^{\mfk{b}|\mfk{f}}(\mscr{SM})\longrightarrow \mcal{F}^{\mfk{b}+1|\mfk{f}}(\mscr{SM}),\nonumber
\\
\Delta^{\RA}_{\mfk{s}}&:\mcal{F}^{\mfk{b}|\mfk{f}}(\mscr{SM})\longrightarrow \mcal{F}^{\mfk{b}+1|\mfk{f}+1}(\mscr{SM}), 
\end{alignat}
where superscript $\mfk{b}|\mfk{f}$ denotes the $\mbb{Z}_2$-graded homological degree of the respective superchains. It is clear that this operation is compatible with gluing, i.e. $[\Delta_{\mfk{s}},d_{\mfk{s}}]=0$, where $\Delta_{\mfk{s}}\equiv \Delta_{\mfk{s}}^{\NS}+\Delta_{\mfk{s}}^{\RA}$, and $d_{\mfk{s}}\equiv \partial$ is the boundary map acting on singular superchains, as defined in Definition \ref{def:the definition of singular superchains on superspaces}. We denote this BV algebra as $\BV{\mscr{SM}}\equiv (\mcal{F}(\mscr{SM});d_{\mfk{s}},\Delta_{\mfk{s}})$.

\newl Our next task is to define a map between the BV algebras $\BV{\mscr{SM}}$ and $\BV{\mscr{S}}\equiv (\mcal{F}(\mscr{S});d,\Delta)$, the BV associated to the moduli space of bordered spin-Riemann surfaces, and show that it is a quasi-isomorphism. We define the map $\varphi:\BV{\mscr{SM}}\longrightarrow\BV{\mscr{S}}$ by sending any singular superchain in $\mcal{F}(\mscr{SM})$ to its underlying singular chain, $d_{\mfk{s}}$ to $d$, and $\Delta_{\mfk{s}}$ to $\Delta$. Proposition \ref{Homology_SupMfd} asserts that the $\mbb{Z}_2$-graded homology groups of the singular supercomplex associated to $\mbbmss{SM}(\ns,\ra)$ with top fermionic degree are isomorphic to the singular homology groups associated to $\mbbmss{S}(\ns,\ra)\equiv \mscr{S}(\ns,\ra)/S^1\wr S_{\sns,\sra}$. This shows that the associated singular chain complexes are quasi-isomorphic. Since this holds for all values of $\ns$ and $\ra$, the BV algebras $\BV{\mscr{SM}}$ and $\BV{\mscr{S}}$ are quasi-isomorphic. If we define 
\begin{equation}
\mbbmss{S}=\bigcup_{\sns,\sra}\mbbmss{S}(\ns,\ra), \qquad  \mbbmss{SM}=\bigcup_{\sns,\sra}\mbbmss{SM}(\ns,\ra), 
\end{equation}
there is a commutative diagram which is shown in Figure \ref{dig:the commutative diagram associated to the quasi-isomorphic BV algebras}.
\begin{figure}\centering 
	\begin{tikzcd}[row sep=huge]
	\mbbmss{SM} \arrow[r,"\mscr{SC}_*"] \arrow[d,swap, "\pi"] &
	\BV{\mscr{SM}}\arrow[d,"\varphi"] 
	\\
	\mbbmss{S}\arrow[r,"\mscr{C}_*"] & \BV{\mscr{S}}. 
	\end{tikzcd}
	\caption{The commutative diagram associated to the map $\varphi$ between the BV algebras $\BV{\mscr{SM}}$ and $\BV{\mscr{S}}$. $\mscr{SC}_*$ associates singular superchains to $\mscr{SM}$, $\mscr{C}_*$ associates singular chains to $\mscr{S}$, and $\pi$ is the map to the underlying bosonic space that sends $\mbbmss{SM}(\ns,\ra)$ to $\mbbmss{S}(\ns,\ra)$. The commutative diagram basically tells us that projecting $\mscr{SM}$ to $\mscr{S}$, and then taking singular chains is the same as associating singular superchains with top fermionic degree to $\mscr{SM}$, and then projecting onto underlying singular chains.}
	\label{dig:the commutative diagram associated to the quasi-isomorphic BV algebras}
\end{figure}
We thus have the following result
\begin{proposition}\label{prop:the qusi-isomorphism of BV algebras of spin- and super-Riemann surfaces}
	The BV algebras $\BV{\mscr{SM}}$ and $\BV{\mscr{S}}$ associated to moduli spaces of bordered spin-Riemann surfaces and moduli spaces of bordered $\mscr{N}=1$ super-Riemannn surfaces are quasi-isomorphic.
\end{proposition}
We already mentioned that the set of homotopy-equivalent solutions of the BV QME for quasi-isomorphic BV algebras are isomorphic. Combining this result together with Proposition \ref{prop:the qusi-isomorphism of BV algebras of spin- and super-Riemann surfaces}, we conclude that
\begin{corollary}[The Existence of Solution to the BV QME in $\BV{\mscr{SM}}$]
	There exists a solution to the BV QME in $\BV{\mscr{SM}}$. 
\end{corollary}
With this result at hand, we conclude this section. 

\section{Existence of Heterotic-String and Type-II-Superstring Vertices}\label{sec:superstring vertices}

In this section, we will prove the existence of heterotic-string and type-II-superstring field theory vertices, i.e. we prove the existence of subspaces in the product of stacks $\mscr{M}_L$ and $\mscr{M}_R$ that respectively parameterize left-moving and right-moving modes of the theory. The worldsheet $\mcal{R}$ of heterotic-string or type-II-superstring theories can be described by an embedding of the following form\footnote{For more on this viewpoint of worldsheets in string theory, see section 5.5 of \cite{Witten2012a}.}
\begin{equation}
\mcal{R}\longhookrightarrow\mcal{R}_L\times\mcal{R}_R,
\end{equation}
where $\mcal{R}_L$ is either an ordinary Riemann surface, in the case of heterotic-string theory, or an $\mscr{N}=1$ super-Riemann surface, in the case of type-II-superstring theories, and $\mcal{R}_R$ is an $\mscr{N}=1$ super-Riemann surface. In the following, we denote by $\pi$ the map  that forgets the spin structure, by $X_{\text{red}}$ the reduced space of the supermanifold $X$, and by $X_{\overline{\text{red}}}\equiv X_{\pi\circ\text{red}}$ the reduced space of the supermanifold together with the forgetting of spin structure of either $X$ or the geometric objects that $X$ parametrizes. For an ordinary surface $\mcal{R}$, we denote a surface which has the complex-conjugate complex structure by $\overline{\mcal{R}}$. 

\subsection{Introduction}\label{subsec:introduction to the proof of existence of string vertices}
\newl In the current understanding of superstring field theory, one must make two choices to be able to formulate such a theory
\begin{enumerate}
	\item {\bf\small a choice of worldsheet superconformal field theory:} this choice is essential to be able to define interaction vertices of superstring field theory. This choice begs one of the most important question in the formulation of superstring field theory, and string theory in general, i.e. the question of {\it background independence}. Since the results of this paper do not have any implication for this choice, the reader can directly go to \hyperlink{choice:string vertex}{\bf\small a choice of string vertex}. We however briefly describe the problem. There are two versions of this question.
	\begin{itemize}
		\item {\bf\small Background-Independence (Weak Version):} In this version, the question is the following: {\it consider two arbitrary $2d$ superconformal field theories which live in the same connected component of the space of $2d$ superconformal field theories with appropriate central charges and supersymmetry in the space of all such superconformal field theories. If we use these two theories to formulate two superstring field theories, to which extent the physical quantities computed using these two theories are sensitive to the choice of $2d$ superconformal field theory?} Since this is a kind of gauge choice, for the consistency of the whole theory, the physical quantities must be independent of this choice. For the case of closed bosonic-string field theory is has been shown that two string field theories formulated using two infinitesimally-different conformal field theories related to each other by some exactly marginal perturbation are related to each other by a field redefinition \cite{SenZwiebach199307, SenZwiebach199311}. In the case of closed superstring field theories formulated in the picture-changing formalism, a similar result has been established recently \cite{Sen201711}. There exist stronger results for open bosonic-string field theory, i.e. the field redefinition for even two finitely-separated backgrounds belonging to the same connected component can be constructed \cite{ErlerMaccaferri201406,ErlerMaccaferri201909}. Some other references on the issue of background-independence in string theory are \cite{Sen199011a,Sen199011b,Sen199201,Witten199208,Witten199210,Witten199306,Shatashvili199303,Shatashvili199311}

		\item {\bf\small Background-Independence (Strong Version):} In general, superconformal field theories are classical backgrounds for the formulation of superstring theory. However, the full quantum theory should be formulated using an arbitrary background. By this we mean more than just starting from a string field theory constructed on a conformal background and then constructing string field theory around a new background by giving VEV to original string fields. Such a theory is a gauge-invariant theory but it explicitly depends on the original conformal background and as such cannot be the solution to the problem of formulating theory on a generic background. String field theory suggest that such a background can be constructed by starting with a conformal background, and then consider the space of all, not necessarily marginal, deformations of this conformal background. This suggest that the space of quantum string background is the space of all $2d$ supersymmetric field theories possibly with some yet unknown features \cite{Zwiebach199606}. If this is the case, then the strong version of the concept of background-independence can be phrased as follows: {\it consider two arbitrary $2d$ supersymmetric field theories with some specific properties. If we use these two theories to formulate two superstring field theories, to which extent the physical quantities computed using these two theories are sensitive to the choice of $2d$ supersymmetric field theory?} As far as we know, this question and more generally the question of string theory on nonconformal backgrounds have not been touched at all. The only attempt in this direction that we are aware of is the work of Zwiebach on the closed bosonic-string field theory on nonconformal backgrounds \cite{Zwiebach199606}.
	\end{itemize}
	
	\newl We also note that there is more to the concept of background-independence in string theory. This motivated the definition of the concept of {\it string background} rather than a conformal background. An example of a deformation of string background without deforming the underlying conformal field theory is provided by the zero-momentum ghost-dilaton state, the component of the dilaton state constructed by the action of operators from the ghost sector \cite{BergmanZwiebach199411,RahmanZwiebach199507,BelopolskyZwiebach199511}. It can be shown that the deformation of a CFT by a zero-momentum ghost-dilaton state is a trivial deformation of the CFT, i.e. it is the old CFT represented in a new basis of states. Despite this fact, such deformation provides a new string background in the sense defined in section 7.2 of \cite{BergmanZwiebach199411}. This shows that the space of deformation of string background is larger than the space of all $2d$ conformal field theories with appropriate properties. Therefore, the independence of the physical quantities from such deformation must also be part of the proof of the weak version of the background-independence. For the case of zero-momentum ghost-dilaton, and also the matter dilaton and the true dilaton, constructed by a combination of matter and ghost dilatons\footnote{For more on the definition of these states see section 5 of \cite{BelopolskyZwiebach199511}.}, in bosonic-string field theory, it is shown that such shifts change the dimensionless string coupling constant, and thus the string background, but the  resulting bosonic-string field theories are related by a field redefinition \cite{BelopolskyZwiebach199511}. Also, as far as we know, the other possible deformations of the quantum string background, similar to the ghost-dilaton deformation of the classical string background, has not been explored in the literature.
	
	\item \hypertarget{choice:string vertex}{{\bf\small a choice of string vertex:}} string vertices are regions inside the moduli space of surfaces on which the relevant theories are defined. They must satisfy the BV QME. This guarantees the gauge-invariance of the action. Then, the remaining region of the moduli space is constructed by gluing of surfaces with lower genus and/or number of punctures, which can be interpreted as Feynman diagrams in string field theory. This region together with the region corresponding to the string vertices must provide a single cover of the moduli space. This essential requirement, related to the modular invariance, is demanded by the unitarity of the theory  \cite{Mandelstam197311,Mandelstam198601,GiddingsMartinec198612,Witten2012c}. Regarding this choice, there is a similar question. This question can be phrased as follows. {\it Let us choose two string vertices. If we formulate two string field theory using these two choices, to which extent the physical quantities computed in these two theories are sensitive to the choice of string vertex?} Again, the consistency of the theory demands that the physical results must be independent of this choice. It has been shown in \cite{HataZwiebach199301} and it has been further elaborated in \cite{SenZwiebach199311,SenZwiebach199408} that two closed bosonic-string field theories constructed using two
	infinitesimally-different vertices give the same gauge-fixed action for the bosonic-string field theory
	upon using two different gauge-fixing conditions. Similar results for finitely-different closed bosonic-string vertices is considered in \cite{CostelloZwiebach201909}. As far as we aware, similar results for heterotic-string and type-II-superstring vertices are not considered in the literature.  
\end{enumerate} 

\newl To construct string vertices, we can either work with {\it punctured} surfaces and consider local (super)conformal coordinates, or we can work with {\it bordered} surfaces. The gluing of local (super)conformal coordinates around the punctures correspond to gluing of boundaries. In this paper, we work with {\it bordered} surfaces to characterize string vertices. We consider the two cases of heterotic-string and type-II superstring theory separately in Sections \ref{subsec:the case of heterotic-string vertices} and \ref{subsec:the case of type-II-superstring vertices}, respectively. 

\subsection{The Case of Heterotic-String Field Theory Vertices}\label{subsec:the case of heterotic-string vertices}
Heterotic-string theories are described by a worldsheet theory with $\mscr{N}=(0,1)$ supersymmetry. The notation means that the right-moving sector has one right-moving supersymmetry, and the left-moving sector is similar to the left-moving sector of the bosonic-string theory. $\mcal{R}$ is a smooth submanifold inside $\mcal{R}^\text{H}_L\times\mcal{R}^\text{H}_R$, whose dimension is $2|1$, and whose reduced space $\mcal{R}_{\text{red}}$ is taken to be such that the bosonic moduli of $\mcal{R}^\text{H}_{R;{\overline{\text{red}}}}$ is the complex conjugate of the moduli of $\mcal{R}^\text{H}_L$ (note that $\mcal{R}_L$ does not have fermionic moduli), and there is no condition on the spin structure. This means that $\mcal{R}_{\text{red}}$ is diagonal in $\mcal{R}^\text{H}_{L}\times\mcal{R}^\text{H}_{R;\text{red}}$. Worldsheets that appear in the definition of $\g$-loop heterotic-string vertices with $\ns$ NS boundary components and $\ra$ R boundary components are parametrized by a cycle $\Gamma^{\text{H}}$ inside $\mscr{M}^\text{H}_{L}\times\mscr{M}^\text{H}_R$. $\mscr{M}^\text{H}_L$ is the moduli space of genus-$\g$ bordered ordinary Riemann surfaces with $\ns+\ra$ boundary components and $\mscr{M}^\text{H}_L$ is the moduli space of genus-$\g$ $\mscr{N}=1$ bordered super-Riemann surfaces with $\ns$ NS boundary components and $\ra$ R boundary components. From what we explained above, the choice of $\Gamma^{\text{H}}$ is such that $\Gamma^{\text{H}}_{\overline{\text{red}}}$ is diagonal in  $(\mscr{M}^\text{H}_L\times\mscr{M}^\text{H}_R)_{\overline{\text{red}}}=\mscr{M}^\text{H}_{L}\times\mscr{M}^\text{H}_{R;\overline{\text{red}}}$. $\Gamma^{\text{H}}$ can be constructed as follows. One first consider a diagonal space $\Gamma^{\text{H}}_{\overline{\text{red}}}$ in $\mscr{M}^\text{H}_{L}\times\mscr{M}^\text{H}_{R;\overline{\text{red}}}$ such that the complex structure of $\mcal{R}^\text{H}_{L}$ parametrized by $\mscr{M}^\text{H}_{L}$ and the complex structure of $\mcal{R}^\text{H}_{R;\overline{\text{red}}}$ parameterized by $\mscr{M}^\text{H}_{R;\overline{\text{red}}}$ are complex conjugate of each other. One then thicken $\Gamma^{\text{H}}_{\overline{\text{red}}}$ in the fermionic directions to get $\Gamma^\text{H}$. Since the thickening is arbitrary, the choice of $\Gamma^{\text{H}}$ is unique up to homolog, as we have explained in Section \ref{subsec:lifting subspace of reduced space}. As it has been explained in section 5.6 of \cite{Witten2012a}, the result of the integration over $\Gamma^{\text{H}}$ is independent of $\Gamma^{\text{H}}$, and it is only sensitive to $\Gamma^{\text{H}}_{\text{red}}$\footnote{To be more precise, suppose that the ideal sheaf of $\mscr{M}^\text{H}_{L}\times\mscr{M}^\text{H}_R$ generated by fermionic directions is $\mathcal I$, then deformation theory implies that the space of all possible lifting of $\Gamma_n\subset \mscr{M}^{\H}_{L}\times\mscr{M}^{\H}_R/\mathcal I^n$ to $\Gamma_{n+1}\subset \mscr{M}_{L}^{\H}\times\mscr{M}^{\H}_R/\mathcal I^{n+1}$ is an affine space modeled on some linear space, hence contractible. This means that for any two liftings $\Gamma$ and $\Gamma'$, there is a homotopy connecting them.}. However, the definition of genus-$\g$ heterotic-string vertices with some number of boundary components requires more elaboration since $\Gamma^{\text{H}}$ must furthermore satisfies the BV QME. If we were dealing with $\mscr{M}^\text{H}_L\times\mscr{M}^\text{H}_R$, i.e. we were letting left and right moduli vary independently, then we could construct a solution of the BV quantum master equation by considering the diagonal product of the solutions of the BV QME in $\mscr{M}_L$ and $\mscr{M}_R$. On the other hand, it is not always clear that if we start from a diagonal solution of the BV QME in $\mscr{M}_L\times\mscr{M}_{R;\overline{\text{red}}}$ and thicken that solution in the fermionic directions, we land on a solution of the BV QME in the product $\mscr{M}_L\times\mscr{M}_R$ since the thickening is defined up to homology. We thus need to consider the problem of the existence of the solution of the BV quantum master equation in $\mscr{M}^\text{H}_L\times\mscr{M}^\text{H}_R$ such that once we reduce that solution to $(\mscr{M}^\text{H}_L\times\mscr{M}^\text{H}_R)_{\overline{\text{red}}}$, the reduced solution parametrizes irreducible surfaces $\mcal{R}_{\overline{\text{red}}}$ which are embedded diagonally inside $(\mcal{R}^\text{H}_L\times\mcal{R}^\text{H}_R)_{\overline{\text{red}}}$, i.e. those surfaces in $(\mcal{R}^\text{H}_L\times\mcal{R}^\text{H}_R)_{\overline{\text{red}}}$ whose left-moving complex structure is complex conjugate of the right-moving complex structure. Thus we call ${\mcal{V}}^{\text{H}}_{\sg}(\ns,\ra)\equiv \Gamma^{\text{H}}$ the genus-$\g$ heterotic-string field theory with $\ns$ NS boundary components and $\ra$ R boundary components if it satisfies the following requirements 
\begin{enumerate}
	\item The complex dimension of ${\mcal{V}}^{\text{H}}_{\sg}(\ns,\ra)$ is $(3\g-3+\ns+\ra|2\g-2+\ns+\ra/2)$;
	
	\item ${\mcal{V}}^{\text{H}}_{\sg}(\ns,\ra)$, as a subspace of $\mscr{M}_L\times\mscr{M}_R$, satisfies the BV quantum master equation;
	
	\item $({\mcal{V}}^{\text{H}}_{\sg}(\ns,\ra))_{\overline{\text{red}}}$ is embedded diagonally in $(\mscr{M}_L\times\mscr{M}_R)_{\overline{\text{red}}}$, where by diagonal embedding, we mean all points in $(\mscr{M}_L\times\mscr{M}_R)_{\overline{\text{red}}}$ of the form $(\mcal{R},\overline{\mcal{R}})$. 
\end{enumerate}

To construct such a solution, we proceed as follows\footnote{Note that for heterotic-string theories $\mscr{M}^\text{H}_L=\mscr{M}(\sns+\sra)$, $\mscr{M}^\text{H}_R=\mscr{SM}(\sns,\sra)$, and $\mscr{M}^\text{H}_{R,\text{red}}=\mscr{S}(\sns,\sra)$.}. We introduce the maps $f^\text{H}_L:\mscr{S}(\ns,\ra) \longrightarrow \mscr{M}^\text{H}_L$ defined by forgetting the spin structure, i.e. 
\begin{equation}
f^\text{H}_{L}(\mcal{R}^\mscr{E})\equiv \mcal{R}, \qquad \forall \mcal{R}^\mscr{E}\in \mscr{S}(\ns,\ra),
\end{equation}
where $\mcal{R}$ is the underlying bosonic surface, and $f^\text{H}_R\equiv J:\mscr{S}(\ns,\ra) \longrightarrow \mscr{M}^\text{H}_{R;\text{red}}$,
where $J$ is the operation of anti-holomorphic involution which sends a spin-Riemann surface $\mcal{R}^{\mscr{E}}$ with a specific complex structure to a spin-Riemann surface with complex-conjugate complex structure $\oli{\mcal{R}}^{\bar{\mscr{E}}}$. Note that the complex-conjugate spin structure $\oli{\mscr{E}}$ is defined such that $\oli{\mscr{E}}^{\otimes 2}\cong \oli{\omega}$, where $\oli{\omega}$ can be identified with the antiholomorphic cotangent bundle of $\mcal{R}^{\mscr{E}}$. Furthermore, these maps can be restricted to $\mscr{S}_{\sg}(\ns,\ra)$. 

\newl Using these maps, we can define
\begin{equation}\label{eq:the map f between spin moduli and product stack}
f^\text{H}\equiv f^\text{H}_L\times f^\text{H}_R:\mscr{S}(\ns,\ra)\longrightarrow\mscr{M}^\text{H}_L\times\mscr{M}^\text{H}_{R;\text{red}}.
\end{equation}
We then have
\begin{alignat}{2}
(\pi\circ f^\text{H}) (\mcal{R}^{\mscr{E}})=\pi (f^\text{H}(\mcal{R}^\mscr{E}))=\pi(\mcal{R},\oli{\mcal{R}}^{\bar{\mscr{E}}})=(\mcal{R},\oli{\mcal{R}}).
\end{alignat}
This shows that the image of a spin Riemann surface $\mcal{R}^{\mscr{E}}$ under $f$ after reduction and forgetting of the spin structure is embedded diagonally in $(\mscr{M}^\text{H}_L\times\mscr{M}^\text{H}_R)_{\overline{\text{red}}}$ such that the left and right complex structures are complex conjugate of each other. To proceed further, we define the following space
\begin{equation}
(\mscr{M}^\text{H}_L\times\mscr{M}^\text{H}_R)_{\text{red}}^{\mbf{D}}\equiv f^{\text{H}}(\mscr{S}(\ns,\ra))= \left\{(\mcal{R},\oli{\mcal{R}}^{\bar{\mscr{E}}})\in(\mscr{M}^\text{H}_L\times\mscr{M}^\text{H}_R)_{\text{red}}\right\},
\end{equation}
i.e. all pairs $(\mcal{R},\oli{\mcal{R}}^{\bar{\mscr{E}}})$ of a surface $\mcal{R}$ with a specific complex structures, and the surface $\oli{\mcal{R}}^{\bar{\mscr{E}}}$ whose complex structure is complex-conjugate of $\mcal{R}$ together with a choice of spin structure compatible with the complex-conjugate complex structure. The superscript $\mbf{D}$ stands for diagonal. The dimension of this space, for genus-$\g$ surfaces with $\ns$ NS boundary components, and $\ra$ R boundary components, is $6\g-6+2\ns+2\ra$, the same is as the space of spin-Riemann surfaces with the same signature. We also consider the thickening of this space in all fermionic dimensions, and denote the resulting space by $(\mscr{M}^\text{H}_L\times\mscr{M}^\text{H}_R)^{\mbf{D}}$. Note that
\begin{alignat*}{2}
\mscr{M}^\text{H}_L&=\mscr{M}(\ns+\ra)=\bigcup_{\sg=0}^\infty\mscr{M}_{\sg}(\ns+\ra),\qquad  \mscr{M}^\text{H}_{R;\text{red}}&=\mscr{S}(\ns,\ra)=\bigcup_{\sg=0}^\infty\mscr{S}_{\sg}(\ns,\ra),
\end{alignat*}
and the thickening is done for all values of $(\g;\ns,\ra)$ in $2\g-2+\ns+\frac{\ra}{2}$ fermionic directions. We can thus define the following chain complex 
\begin{equation}
\mcal{F}((\mscr{M}^\text{H}_L\times\mscr{M}^\text{H}_R)_{\text{red}}^{\mbf{D}})\equiv \bigoplus_{\sns,\sra}\mscr{C}_*\left(\frac{(\mscr{M}^\text{H}_L(\ns+\ra)\times\mscr{M}^\text{H}_R(\ns,\ra))_{\text{red}}^\mbf{D}}{S^\text{H}_L\times S^\text{H}_R}\right),
\end{equation}
where 
\begin{alignat}{2}
S^\text{H}_L\equiv S_1\wr S_{\sns+\sra}, \qquad S^\text{H}_R\equiv S_1^{\sns+\sra}\wr S_{\sns,\sra}, 
\end{alignat}
$\mscr{M}^\text{H}_L(\ns+\ra)\equiv \mscr{M}(\ns+\ra)$, and $\mscr{M}^\text{H}_R(\ns,\ra)\equiv \mscr{SM}(\ns,\ra)$. $S^\text{H}_L$, respectively $S^\text{H}_R$, acts on $\mscr{M}^\text{H}_R$, respectively $\mscr{M}^\text{H}_L$, trivially. We also associate a superchain complex to $(\mscr{M}^\text{H}_L\times\mscr{M}^\text{H}_R)^{\mbf{D}}$, defined by thickening $(\mscr{M}^\text{H}_L\times\mscr{M}^\text{H}_R)^{\mbf{D}}_{\text{red}}$ in all fermionic directions, 
\begin{equation}
\mcal{F}((\mscr{M}^\text{H}_L\times\mscr{M}^\text{H}_R)^{\mbf{D}}) \equiv \bigoplus_{\sns,\sra}\mscr{SC}_*\left(\frac{(\mscr{M}^\text{H}_L(\ns+\ra)\times\mscr{M}^\text{H}_R(\ns,\ra))^\mbf{D}}{S_L^{\text{H}}\times S^{\text{H}}_R}\right),
\end{equation}
where $\mscr{SC}_*$ is the functor of normalized singular simplicial superchains with coefficients in any field containing $\mbb{Q}$ which associates singular simplicial superchains with top fermionic dimension.

\newl The crucial property of the map $f^\text{H}$ defined in \eqref{eq:the map f between spin moduli and product stack} is that it is compatible with the gluing operation, i.e. the operation of gluing and mapping with $f^\text{H}$ commute and there is a commutative diagram. We can thus define the following induced map
\begin{equation}
F^\text{H}_{f^\text{H}}:\mcal{F}(\mscr{S}) \longrightarrow\mcal{F}((\mscr{M}^\text{H}_L\times\mscr{M}^\text{H}_R)_{\text{red}}^{\mbf{D}}).
\end{equation}
If $\mcal{F}(\mscr{M}^\text{H}_L\times\mscr{M}^\text{H}_R)$ is the complex associated to the product $\mscr{M}^\text{H}_L\times\mscr{M}^\text{H}_R$, it has naturally the structure of a differential-graded algebra where the differential $d^\text{H}\equiv d^\text{H}_L\oplus d^\text{H}_R$ is the left $d^\text{H}_L$ and right $d^\text{H}_R$  boundary operations on left and right singular chains, and the product is the disjoint union of left and right surfaces. $\mcal{F}(\mscr{M}^\text{H}_L\times\mscr{M}^\text{H}_R)$ can be turned into a BV algebra by the usual twist-swing operation $\Delta^\text{H}_{LR}\equiv \Delta^\text{H}_L \oplus  \Delta^\text{H}_R$, where $\Delta^\text{H}_L$($\Delta^\text{H}_R$) is the gluing operation defined on the complex $\mcal{F}(\mscr{M}^\text{H}_L)$ ($\mcal{F}(\mscr{M}^\text{H}_R)$). $\mcal{F}((\mscr{M}^\text{H}_L\times\mscr{M}^\text{H}_R)_{\text{red}}^{\mbf{D}})$ inherits this BV-algebra structure. By construction,  $F^{\text{H}}_{f^{\text{H}}}$ is a homorphism of BV algebras, and thus maps solutions of the BV QME onto each other. Therefore, having a solution in $\mcal{F}(\mscr{S})$ whose existence we have proven, we have a solution in $\mcal{F}((\mscr{M}^\text{H}_L\times\mscr{M}^\text{H}_R)_{\text{red}}^{\mbf{D}})$. We have thus constructed a subspace of $(\mscr{M}^\text{H}_L\times\mscr{M}^\text{H}_R)_{\text{red}}^{\mbf{D}}$ that satisfies the BV QME. To construct the heterotic-string vertices, we consider the following maps of BV algebras
\begin{equation}
\widehat{F}^\text{H}_{f^\text{H}}:\mcal{F}((\mscr{M}^\text{H}_L\times\mscr{M}^\text{H}_R)_{\text{red}}^{\mbf{D}})\longrightarrow \mcal{F}((\mscr{M}^\text{H}_L\times\mscr{M}^\text{H}_R)^{\mbf{D}}).
\end{equation}
Since the BV algebra associated to singular superchains with top fermionic dimensions is quasi-isomorphic to the BV algebra of associated bosonic chains,, $\widehat{F}^\text{H}_{f^\text{H}}$ is a quasi-isomorphism of BV algebras. Again using the fact that quasi-isomorphic BV algebras have the same set of homotopy-classes of solutions of the BV QME, we conclude that there is a subspace of $(\mscr{M}^\text{H}_L(\ns+\ra)\times\mscr{M}^\text{H}_R(\ns,\ra))^{\mbf{D}}$, whose genus-$\g$ contribution we denote by ${\mcal{V}}^{\text{H}}_{\sg}(\ns,\ra)$, which satisfies the BV QME. ${\mcal{V}}^{\text{H}}_{\sg}(\ns,\ra)$ is a genus-$\g$ heterotic-string field theory vertex with $\ns$ boundary components and $\ra$ boundary components since
\begin{enumerate}
	\item its complex dimension is by construction $(3\g-3+\ns+\ra|2\g-2+\ns+\ra/2)$;
	
	\item ${\mcal{V}}^{\text{H}}_{\sg}(\ns,\ra)$, as a subspace of $\mscr{M}^\text{H}_L\times\mscr{M}^\text{H}_R$, satisfies the BV QME;
	
	\item by construction, $({\mcal{V}}^{\text{H}}_{\sg}(\ns,\ra))_{\overline{\text{red}}}$ is embedded diagonally in $(\mscr{M}_L\times\mscr{M}_R)_{\overline{\text{red}}}$.
\end{enumerate}
We thus have the following
\begin{thr}[The Existence of Heterotic-String Vertices]
	There exists a solution to the {\normalfont BV QME} in $\mcal{F}((\mscr{M}^\text{\H}_L\times\mscr{M}^\text{\H}_R)^{\mbf{D}})$. The genus-$\g$ part is a genus-$\g$ heterotic-string field theory vertex with $\ns$ {\normalfont NS} and $\ra$ {\normalfont R} boundary components. 
\end{thr}

\subsection{The Case of Type-II-Superstring Field Theory Vertices}\label{subsec:the case of type-II-superstring vertices}	
We now turn to the case of type-II-superstring theories. The worldsheet theory has $\mscr{N}=(1,1)$ supersymmetry. The notation means that the right-moving sector has one right-moving supersymmetry, and the left-moving sector is has one left-moving supersymmetry. The worldsheet $\mcal{R}$ is a smooth submanifold inside $\mcal{R}^{\text{II}}_L\times\mcal{R}^{\text{II}}_R$, whose dimension is $2|2$, and whose reduced space $\mcal{R}_{\text{red}}$ is taken to be such that the bosonic moduli of $\mcal{R}^{\text{II}}_{R;{\text{red}}}$ is the complex conjugate of the bosonic moduli of $\mcal{R}^{\text{II}}_{L;{\overline{\text{red}}}}$, and we do not assume any relation between fermionic moduli. This means that $\mcal{R}_{\overline{\text{red}}}$ is diagonal in $\mcal{R}^{\text{II}}_{L;\overline{\text{red}}}\times\mcal{R}^{\text{II}}_{R;\overline{\text{red}}}$. Worldsheets that appear in the definition of $\g$-loop type-II-superstring vertices with $\nsns$ NSNS boundary components, $\nsra$ NSR boundary components, $\rans$ boundary components, and $\rara$ RR boundary components are parametrized by a cycle $\Gamma^{\text{II}}$ inside $\mscr{M}^{\text{II}}_{L}\times\mscr{M}^{\text{II}}_R$. $\mscr{M}^{\text{II}}_L$  and $\mscr{M}^{\text{II}}_R$ are the moduli space of genus-$\g$ $\mscr{N}=1$ bordered super-Riemann surfaces with $\ns$ NS boundary components and $\ra$ R boundary components. Again from what we explained so far, the choice of $\Gamma^{\text{II}}$ is such that $\Gamma^{\text{II}}_{\overline{\text{red}}}$ is diagonal in  $(\mscr{M}^{\text{II}}_L\times\mscr{M}^{\text{II}}_R)_{\overline{\text{red}}}=\mscr{M}^{\text{II}}_{L;\overline{\text{red}}}\times\mscr{M}^{\text{II}}_{R;\overline{\text{red}}}$. $\Gamma^{\text{II}}$ can be constructed as follows. One first consider a diagonal space $\Gamma^{\text{II}}_{\overline{\text{red}}}$ in $\mscr{M}^{\text{II}}_{L;\overline{\text{red}}}\times\mscr{M}^{\text{II}}_{R;\overline{\text{red}}}$ such that the complex structure of $\mcal{R}^{\text{II}}_{L}$ parametrized by $\mscr{M}^{\text{II}}_{L;\overline{\text{red}}}$ and the complex structure of $\mcal{R}^{\text{II}}_{R;\overline{\text{red}}}$ parameterized by $\mscr{M}^{\text{II}}_{R;\overline{\text{red}}}$ are complex conjugate of each other. One then thicken $\Gamma^{\text{II}}_{\overline{\text{red}}}$ in the fermionic directions to get $\Gamma^{\text{II}}$. Thickening is arbitrary, the choice of $\Gamma^{\text{II}}$ is unique up to homology, the result of the integration over $\Gamma^{\text{II}}$ is independent of $\Gamma^{\text{II}}$, and it is only sensitive to $\Gamma^{\text{II}}_{\text{red}}$. Again, the definition of genus-$\g$ type-II-superstring vertices with some number of boundary components requires more elaboration since $\Gamma^{\text{II}}$ must furthermore satisfies the BV QME. If we were dealing with $\mscr{M}^{\text{II}}_L\times\mscr{M}^{\text{II}}_R$, i.e. we were letting left and right moduli vary independently, then we could construct a solution of the BV QME by considering the diagonal product of the solutions of the BV QME in $\mscr{M}^{\text{II}}_L$ and $\mscr{M}^{\text{II}}_R$. Again, it is not always clear that if we start from a diagonal solution of the BV QME in $\mscr{M}^{\text{II}}_{L;\overline{\text{red}}}\times\mscr{M}^{\text{II}}_{R;\overline{\text{red}}}$ and thicken that solution in the fermionic directions, we land on a solution of the BV quantum master equation in the product $\mscr{M}^{\text{II}}_L\times\mscr{M}^{\text{II}}_R$ since the thickening is defined up to homology. We thus need to consider the problem of existence of the solution of the BV quantum master equation in $\mscr{M}^{\text{II}}_L\times\mscr{M}^{\text{II}}_R$ such that once we reduce that solution to $(\mscr{M}^{\text{II}}_L\times\mscr{M}^{\text{II}}_R)_{\overline{\text{red}}}$, the reduced solution parametrizes irreducible surfaces $\mcal{R}_{\overline{\text{red}}}$ which are embedded diagonally inside $(\mcal{R}^{\text{II}}_L\times\mcal{R}^{\text{II}}_R)_{\overline{\text{red}}}$, i.e. those surfaces in $(\mcal{R}^{\text{II}}_L\times\mcal{R}^{\text{II}}_R)_{\overline{\text{red}}}$ whose left-moving complex structure is complex conjugate of the right-moving complex structure. Thus we call ${\mcal{V}}^{\text{II}}_{\sg}(\nsns,\nsra,\rans,\rara)\equiv \Gamma^{\text{II}}$ the genus-$\g$ type-II-superstring field theory vertex with $\nsns$ NSNS, $\nsra$ RNS, $\rans$ NSR, and $\rara$ RR  boundary components if it satisfies the following requirements 
\begin{enumerate}
	\item The dimension of ${\mcal{V}}^{\text{II}}_{\sg}(\nsns,\nsra,\rans,\rara)$ is
	\begin{alignat}{2}\label{eq:the dimension of type-II-superstring vertices}
	\dim^{\text{e}}_{\mbb{C}}({\mcal{V}}^{\text{II}}_{\sg}(\nsns,\nsra,\rans,\rara))&=3\g-3+\nsns+\nsra+\rans+\rara,\nonumber
	\\
	\dim^{\text{o}}_{\mbb{C}}({\mcal{V}}^{\text{II}}_{\sg}(\nsns,\nsra,\rans,\rara))&=4\g-4+2\nsns+\frac{3}{2}(\nsra+\rans)+\rara,
	\end{alignat}
	where $\dim^{\text{e}}_{\mbb{C}}$ and $\dim^{\text{o}}_{\mbb{C}}$ are bosonic and fermionic dimensions, respectively.
	
	\item ${\mcal{V}}^{\text{II}}_{\sg}(\nsns,\nsra,\rans,\rara)$, as a subspace of $\mscr{M}^{\text{II}}_L\times\mscr{M}^{\text{II}}_R$, satisfies the BV QME;
	
	\item $({\mcal{V}}^{\text{II}}_{\sg}(\nsns,\nsra,\rans,\rara))_{\overline{\text{red}}}$ is embedded diagonally in $(\mscr{M}^{\text{II}}_L\times\mscr{M}^{\text{II}}_R)_{\overline{\text{red}}}$, where by diagonal embedding, we mean all points in $(\mscr{M}^{\text{II}}_L\times\mscr{M}^{\text{II}}_R)_{\overline{\text{red}}}$ of the form $(\mcal{R},\overline{\mcal{R}})$. 
\end{enumerate}

To construct such a solution, we proceed as follows\footnote{Note that for type-II-superstring theories $\mscr{M}^{\text{II}}_L=\mscr{SM}(\sns,\sra)$, $\mscr{M}^{\text{II}}_R=\mscr{SM}(\sns,\sra)$, $\mscr{M}^{\text{II}}_{L;\text{red}}=\mscr{S}(\sns,\sra)$, and $\mscr{M}^{\text{II}}_{R,\text{red}}=\mscr{S}(\sns,\sra)$.}. We consider the fiber product with the twisted diagonal of $\mscr{M}(\ns+\ra)$:

\begin{center}
	\begin{tikzcd} 
	\mscr{S}'(\ns,\ra)\arrow[r,"f^{\text{II}}"] \arrow[d, "\pi" left] & \mscr{S}(\ns,\ra)\times\mscr{S}(\ns,\ra)\arrow[d, "\pi" right]\\
	\mscr{M}(\ns+\ra)\arrow[r,"\mbf{D}^t"] & \mscr{M}(\ns+\ra)\times {\mscr{M}(\ns+\ra)},
	\end{tikzcd}
\end{center}

where $\mbf{D}^t$ is the twisted diagonal defined by sending a surface $\mcal R$ to $(\mcal R,\oli{\mcal R})$. Note that the space $\mscr{S}'(\ns,\ra)$ represents Riemann surfaces with $\ns+\ra$ borders and two choices of spin structures, one with respect to a complex structure and the other one with respect to its complex conjugate. Therefore, $\mscr{S}'(\ns,\ra)$ can be obtained by distributing $\ns$ and $\ra$ such that $\ns+\ra$ is fixed. To each of such distribution of boundary components, one can assign $2^{2\sg}$ choices of spin structures (and also its complex conjugate). Therefore, $\mscr{S}'(\ns,\ra)$ is a disjoint union of $\mscr{S}(\ns,\ra)$. 

\newl Using this construction, we have
\begin{equation}\label{eq:the map f between product reduced superstacks}
f^{\text{II}}:\mscr{S}'(\ns,\ra)\longrightarrow\mscr{M}^{\text{II}}_{L;\text{red}}\times\mscr{M}^{\text{II}}_{R;\text{red}}.
\end{equation}
We then have
\begin{alignat}{2}
(\pi\circ f^{\text{II}}) (\mcal{R}^{\mscr{E}})=\pi (f^{\text{II}}(\mcal{R}^\mscr{E}))=\pi(\mcal{R}^{\mscr{E}_L},\oli{\mcal{R}}^{\mscr{E}_R})=(\mcal{R},\oli{\mcal{R}}),
\end{alignat}
where $\mscr{E}_L$ and $\mscr{E}_R$ are the spin structures of surfaces left- and right-moving sectors. {\it A priori, there is no relation between these spin structures}. This shows that the image of a spin-Riemann surface $\mcal{R}^{\mscr{E}}$ under $f^{\text{II}}$ after forgetting of the spin structure is embedded diagonally in $(\mscr{M}^{\text{II}}_L\times\mscr{M}^{\text{II}}_R)_{\overline{\text{red}}}$ such that the left and right complex structures are complex conjugate of each other. To proceed further, we define the following space
\begin{equation}
(\mscr{M}^{\text{II}}_L\times\mscr{M}^{\text{II}}_R)_{\text{red}}^{\mbf{D}}\equiv f^{\text{II}}(\mscr{S}'(\ns,\ra))= \left\{(\mcal{R}^{\mscr{E}_L},\oli{\mcal{R}}^{\mscr{E}_R})\in(\mscr{M}^{\text{II}}_L\times\mscr{M}^{\text{II}}_R)_{\text{red}}\right\},
\end{equation}
i.e. all pairs $(\mcal{R}^{\mscr{E}_L},\oli{\mcal{R}}^{\mscr{E}_R})$ of a surface $\mcal{R}^{\mscr{E}_L}$ with a specific complex structures equipped with an arbitrary spin structure, and the surface $\oli{\mcal{R}}^{\mscr{E}_R}$ whose complex structure is complex-conjugate of $\mcal{R}$ together with a choice of spin structure compatible with the complex-conjugate complex structure. As we mentioned above, there is no relation between spin structures $\mscr{E}_L$ and $\mscr{E}_R$. As usual, the superscript $\mbf{D}$ stands for diagonal. The dimension of this space, for genus-$\g$ surfaces with $\nsns$ NSNS, $\nsra$ NSR, $\rans$ RNS, and $\rara$ RR boundary components is $6\g-6+2\nsns+2\nsra+2\rans+2\rara$, the same is as the space of spin-Riemann surfaces with the same signature. We also consider the thickening of this space in all fermionic dimensions, and denote the resulting space by $(\mscr{M}^{\text{II}}_L\times\mscr{M}^{\text{II}}_R)^{\mbf{D}}$. Note that
\begin{alignat*}{2}
\mscr{M}^{\text{II}}_{L;\text{red}}&=\mscr{S}(\ns,\ra)=\bigcup_{\sg=0}^\infty\mscr{S}_{\sg}(\ns,\ra)= \mscr{M}^{\text{II}}_{R;\text{red}},
\end{alignat*}
and the thickening is done for all values of $(\g;\ns,\ra)$ in $2\g-2+\ns+\frac{\ra}{2}$ fermionic directions. We can thus define the following chain complex 
\begin{equation}
\mcal{F}((\mscr{M}^{\text{II}}_L\times\mscr{M}^{\text{II}}_R)_{\text{red}}^{\mbf{D}})\equiv \bigoplus_{\sns,\sra}\mscr{C}_*\left(\frac{(\mscr{M}^{\text{II}}_L(\ns,\ra)\times\mscr{M}^{\text{II}}_R(\ns,\ra))_{\text{red}}^\mbf{D}}{S^{\text{II}}_L\times S^{\text{II}}_R}\right),
\end{equation}
where 
\begin{alignat}{2}
S^{\text{II}}_L\equiv S_1\wr S_{\sns,\sra}=S^{\text{II}}_R, 
\end{alignat}
$\mscr{M}^{\text{II}}_L(\ns,\ra)\equiv \mscr{SM}(\ns,\ra)$, and $\mscr{M}^{\text{II}}_R(\ns,\ra)\equiv \mscr{SM}(\ns,\ra)$. $S^{\text{II}}_L$, respectively $S^{\text{II}}_R$, acts on $\mscr{M}^{\text{II}}_R$, respectively $\mscr{M}^{\text{II}}_L$, trivially. We also define the superchain complex associated to $(\mscr{M}^{\text{II}}_L\times\mscr{M}^{\text{II}}_R)^{\mbf{D}}$
\begin{equation}
\mcal{F}^{\text{II}}((\mscr{M}^{\text{II}}_L\times\mscr{M}^{\text{II}}_R)^{\mbf{D}}) \equiv \bigoplus_{\sns,\sra}\mscr{SC}_*\left(\frac{(\mscr{M}^{\text{II}}_L(\ns,\ra)\times\mscr{M}^{\text{II}}_R(\ns,\ra))^\mbf{D}}{S^{\text{II}}_L\times S^{\text{II}}_R}\right).
\end{equation}

\newl The crucial property of the map $f^{\text{II}}$ defined in \eqref{eq:the map f between product reduced superstacks} is that it is compatible with the gluing operation, i.e. the operation of gluing and mapping with $f^{\text{II}}$ commute and there is a commutative diagram. We can thus define the following induced map
\begin{equation}
F^{\text{II}}_{f^{\text{II}}}:\mcal{F}(\mscr{S}') \longrightarrow\mcal{F}((\mscr{M}^{\text{II}}_L\times\mscr{M}^{\text{II}}_R)_{\text{red}}^{\mbf{D}}).
\end{equation}
If $\mcal{F}(\mscr{M}^{\text{II}}_L\times\mscr{M}^{\text{II}}_R)$ is the complex associated to the stack $\mscr{M}^{\text{II}}_L\times\mscr{M}^{\text{II}}_R$, it has naturally the structure of a differential-graded algebra where the differential $d^{\text{II}}\equiv d^{\text{II}}_L\oplus d^{\text{II}}_R$ is the left $d^{\text{II}}_L$ and right $d^{\text{II}}_R$  boundary operations on left and right singular chains, and the product is the disjoint union of left and right surfaces. $\mcal{F}(\mscr{M}^{\text{II}}_L\times\mscr{M}^{\text{II}}_R)$ can be turned into a BV algebra by the usual twist-swing operation $\Delta^{\text{II}}_{LR}\equiv \Delta^{\text{II}}_L \oplus  \Delta^{\text{II}}_R$, where $\Delta^{\text{II}}_L$($\Delta^{\text{II}}_R$) is the gluing operation defined on the complex $\mcal{F}(\mscr{M}^{\text{II}}_L)$ ($\mcal{F}(\mscr{M}^{\text{II}}_R)$). $\mcal{F}((\mscr{M}^{\text{II}}_L\times\mscr{M}^{\text{II}}_R)_{\text{red}}^{\mbf{D}})$ inherits this BV-algebra structure. By construction, $F^{\text{II}}_{f^{\text{II}}}$ is a homomorphism of BV algebras, and thus maps solutions of the BV QME onto each other. Therefore, having a solution in $\mcal{F}(\mscr{S}')$ whose existence we have proven\footnote{Note that $\mscr{S}'(\sns,\sra)$ is a disjoint union of $\mscr{S}(\sns,\sra)$, and as such, the existence of the solution of the BV QME in $\mcal{F}(\mscr{S}')$ follows at once from the existence of solution in $\mcal{F}(\mscr{S})$.}, we have a solution in $\mcal{F}((\mscr{M}^{\text{II}}_L\times\mscr{M}^{\text{II}}_R)_{\text{red}}^{\mbf{D}})$. We have thus constructed a subspace of $(\mscr{M}^{\text{II}}_L\times\mscr{M}^{\text{II}}_R)_{\text{red}}^{\mbf{D}}$ that satisfies the BV QME. To construct the type-II-superstring vertices, we consider the following maps of BV algebras
\begin{equation}
\widehat{F}^{\text{II}}_{f^{\text{II}}}:\mcal{F}((\mscr{M}^{\text{II}}_L\times\mscr{M}^{\text{II}}_R)_{\text{red}}^{\mbf{D}})\longrightarrow \mcal{F}((\mscr{M}^{\text{II}}_L\times\mscr{M}^{\text{II}}_R)^{\mbf{D}}).
\end{equation}
Since the BV algebra associated to singular superchains with top fermionic dimensions is quasi-isomorphic to the BV algebra of associated bosonic chains, $\widehat{F}^{\text{II}}_{f^{\text{II}}}$ is a quasi-isomorphism of BV algebras. Again using the fact that quasi-isomorphic BV algebras have the same set of homotopy-classes of solutions of the BV QME, we conclude that there is a subspace of $(\mscr{M}^{\text{II}}_L(\ns,\ra)\times\mscr{M}^{\text{II}}_R(\ns,\ra))^{\mbf{D}}$, whose genus-$\g$ contribution we denote by ${\mcal{V}}^{\text{II}}_{\sg}(\nsns,\nsra,\rans,\rara)$, which satisfies the BV QME. ${\mcal{V}}^{\text{II}}_{\sg}(\nsns,\nsra,\rans,\rara)$ is the genus-$\g$ type-II-superstring field theory vertex with $\nsns$ NSNS, $\nsra$ NSR, $\rans$ RNS, and $\rara$ RR boundary components since
\begin{enumerate}
	\item its complex dimension is by construction the one given in \eqref{eq:the dimension of type-II-superstring vertices};
	
	\item ${\mcal{V}}^{\text{II}}_{\sg}(\nsns,\nsra,\rans,\rara)$, as a subspace of $\mscr{M}^{\text{II}}_L\times\mscr{M}^{\text{II}}_R$, satisfies the BV quantum master equation;
	
	\item by construction $({\mcal{V}}^{\text{II}}_{\sg}(\nsns,\nsra,\rans,\rara))_{\overline{\text{red}}}$ is embedded diagonally in $(\mscr{M}^{\text{II}}_L\times\mscr{M}^{\text{II}}_R)_{\overline{\text{red}}}$.
\end{enumerate}
We thus have proven the following 
\begin{thr}[The Existence of Type-II-Superstring Vertices]
	There exists a solution to the {\normalfont BV QME} in $\mcal{F}((\mscr{M}^\text{\normalfont II}_L\times\mscr{M}^\text{\normalfont II}_R)^{\mbf{D}})$. The genus-$\g$ part is a genus-$\g$ type-II-string field theory vertex with $\nsns$ {\normalfont NSNS}, $\rans$ {\normalfont RNS}, $\nsra$ {\normalfont NSR}, and $\rara$ {\normalfont RR} boundary components. 
\end{thr}

\section{On the Moduli Problem of  Stable $\mscr{N}=1$ SUSY Curves}\label{sec:on the moduli problem of N=1 SUSY Curves}
This section is devoted to the representability of the moduli problem of stable $\mscr{N}=1$ SUSY curves. We review some basics of superspace. We will basically follow \cite{DonagiWitten2013a,CodogniViviani201706} for various definitions and splitness of superspaces, and follow \cite{VoronovZorich198703,Voronov1991} for homology on superspaces. Unless specified, all rings discussed in this section are of charactersitic $0$, i.e. contains $\mathbb Q$. For early development on moduli problem and deformation theory see \cite{Artin1969,Artin197001,Artin197403,Schlessinger196802}. Deformation theory in the context of supergeometry is developed in \cite{Vaintrob199001,FlennerSundararaman199201}.

\subsection{Definitions, Elementary Results and Super Analogs of Classical Theorems}\label{subsec:definitions, elementary results, and super-analog of classical theorems}
\begin{definition}[Superrings, Ideals and Modules]\label{Defn_Superring}
	A superring is a $\mbb Z_2$-graded ring $R=R^+\oplus R^-$ such that for homogeneous elements $x,y\in R$, $x\cdot y=(-1)^{|x|\cdot|y|}y\cdot x$. An ideal of $R$ is a $\mbb Z_2$-graded ideal of $R$ (considered as an ordinary ring). A module of $R$ is a $\mbb Z_2$-graded $R^+$-module $M=M^+\oplus M^-$, with an $R^+$-homomorphism $a:R^-\otimes_{R^+}M^{\pm}\to M^{\mp}$ such that following diagram 
	\begin{center}
		\begin{tikzcd}
		R^-\otimes_{R^+}R^-\otimes_{R^+}M^{\pm}\arrow[r,"\text{Id}\otimes a"] \arrow[dr,"\mu\otimes \text{Id}",swap] &R^-\otimes_{R^+}M^{\mp}\arrow[d,"a"]\\
		&M^{\pm}
		\end{tikzcd}
	\end{center}
	commutes, where $\mu:R^-\otimes _{R^+}R^-\to R^+$ is the multiplication map.
\end{definition}
Unless specified, we only consider ideals and modules which are defined above, instead of ideals or modules in the ordinary sense.

\begin{definition}[{\bf Superspace}]\label{Defn_Superspace}
	A superspace is a locally-superringed topological space $(X,\mathcal O_X)$, i.e. there is a $\mathbb Z_2$-grading
	\begin{equation}
	\mathcal O_X=\mathcal O_X^+\oplus \mathcal O_X^-.
	\end{equation}
	We also define the ideal generated by $\mathcal O_X^-$ to be $\mathfrak{n}_X$, and define the locally-ringed topological space $(X,\mathcal O_X/\mathfrak{n}_X)$ to be the \textbf{bosonic truncation} of $X$, denoted by $X_{\text{\normalfont bos}}$. The locally-ringed topological space $(X,\mathcal O_X^+)$ will be called the \textbf{bosonic quotient} of $X$, denoted by $X_{\text{\normalfont ev}}$.\\
	
	A morphism between superspaces $f:(X,\mathcal O_X)\to (Y,\mathcal O_Y)$ is a morphism between locally-ringed topological space $f:(X,\mathcal O_X)\to (Y,\mathcal O_Y)$ such that the induced homomorphism $ f^{-1}\mathcal O_Y\to \mathcal O_X$ preserves the $\mathbb Z_2$-grading. The category of superspaces will be denoted by $\SSp$, and the category of locally-commutative-ringed topological space which we call commutative space will be denoted by $\CSp$.
\end{definition}

\begin{remark}\label{Superspace_Z_2Action}
	Given a superspace $(X,\mathcal O_X)$, there is a natural $\mathbb Z_2$ action on it, namely, let $\mathbb Z_2$ acts on $\mathcal O_X^i$ by $(-1)^i$, hence the $\mathbb Z_2$-invariant is exactly $\mathcal O_X^+$. This is where the name bosonic quotient comes from. Moreover this defines a functor:
	\begin{equation*}
	C:\SSp\to \CSp.
	\end{equation*}
\end{remark}

\begin{remark}
	$X_{\text{\normalfont bos}}\equiv (X,\mathcal O_{X_{\text{\normalfont bos}}})$ has a natural superspace structure, namely let the odd part to be trivial. It has a natural closed embedding into $X$ in the category of \textit{superspaces}. In fact, for any locally-commutative-ringed topological space $(X,\mathcal O_X)$, there is a canonical superspace $(X,\mathcal O_X\oplus 0)$ associated to it, and this is a functor:
	\begin{equation*}
	S: \CSp\to \SSp.
	\end{equation*}
\end{remark}

\begin{lem}\label{Superspace_Adjointness}
	$C$ is left adjoint to $S$, and $S$ is left adjoint to the functor $B$ which takes a superspace to its bosonic truncation.
\end{lem}

\begin{proof}
	Suppose that $X\in \SSp$ and $Y\in \CSp$, and given a morphism $f:X\to S(Y)$, since the induced homomorphism $f^{-1}\mathcal O_Y\to \mathcal O_X$ preserves $\mathbb Z_2$-grading and the odd part of $f^{-1}\mathcal O_Y$ is trivial, this map factors through the $\mathbb Z_2$-invariant of $\mathcal O_X$, hence there is morphism $S(X)\to Y$ given by the same map on topological spaces as $f$ but send $f^{-1}\mathcal O_Y$ to $\mathcal O_X^{+}$. The converse is obvious. Hence we see that $C$ is left adjoint to $S$.\\
	
	Suppose that there is a morphism $g:S(Y)\to X$, since the induced homomorphism $ g^{-1}\mathcal O_X\to \mathcal O_Y$ preserves $\mathbb Z_2$ grading and the odd part of $\mathcal O_Y$ is trivial, hence $\mathfrak{n}_Y$ is sent to zero and $g$ factors through $X_{\text{\normalfont bos}}$. The converse is composing with canonical morphism $X_{\text{\normalfont bos}}\hookrightarrow X$. Hence we see that $S$ is left adjoint to $B$.
\end{proof}

\begin{definition}[{\bf Split and Projected Superspace}]\label{Defn_Superspace_Split and Projected}
	A superspace is called \textbf{split} if $\mathcal O_X$ is $\mathbb Z_{\ge 0}$ graded, i.e. $$\mathcal O_X=\bigoplus_{i\ge 0}\mathcal O_X[i],$$such that $\mathcal O_X$ is generated by $\mathcal O_X[0]$ and $\mathcal O_X[1]$, and the $\mathbb Z_2$-grading is induced from the $\mathbb Z_{\ge 0}$ grading. Note that this implies that $\mathfrak{n}_X=\oplus_{i>0}\mathcal O_X[i]$ and $\mathcal O_{X_{\text{\normalfont bos}}}\cong \mathcal O_X[0]$.\\
	
	A superspace is called \textbf{projected} if there is a morphism $\pi:X\to X_{\text{\normalfont bos}}$ of superspaces such that the following diagram commutes{\normalfont :}
	\begin{center}
		\begin{tikzcd}
		X_{\text{\normalfont bos}}\arrow[r,"i", hook] \arrow[dr,"Id" swap] &X\arrow[d,"\pi"]\\
		&X_{\text{\normalfont bos}}.
		\end{tikzcd}
	\end{center}
	It is easy to see that a split superspace is projected.
\end{definition}
The following is obvious:
\begin{lem}\label{Lemma_Superspace_Projected}
	A superspace $X$ is projected if and only if the induced closed immersion $X_{\text{\normalfont bos}}\hookrightarrow X_{\text{\normalfont ev}}$ has a retract $X_{\text{\normalfont ev}}\to X_{\text{\normalfont bos}}$.
\end{lem}

The following lemma gives equivalent definitions of Noetherian superrings.
\begin{lem}
	Let $R=R^+\oplus R^-$ be a superring, then the following are equivalent:
	\begin{enumerate}
		\item $R$ is Noetherian as an ordinary ring, i.e. its ideals (including non-homogeneous ideals) satisfies Ascending Chain Condition $(ACC)$.
		\item $R$ satifies $ACC$ on ideals.
		\item $\Gr_{\mathfrak{n}}R$ satifies $ACC$ on ideals.
		\item $\mathfrak{n}$ is a finitely-generated ideal and $R_{\bos}=R/\mathfrak{n}$ is Noetherian.
		\item $R^-$ is finitely-generated as $R^+$-module and $R_{\bos}=R/\mathfrak{n}$ is Noetherian.
		\item $R^-$ is finitely-generated as $R^+$-module and $R^+$ is Noetherian.
	\end{enumerate}
\end{lem}
\begin{proof}
	$1\Rightarrow 2$ and $6\Rightarrow 5\Rightarrow 4$ are obvious. $6\Rightarrow 1$: $R$ is a finitely-generated $R^+$-module so every left ideal in $R$ is a finitely-generated $R^+$-module, same for right ideals. $2\Rightarrow 6$: if $I_i\subset I_{i+1}\subset \cdots$ is an ascending chain of ideals of $R^+$, then $I_iR=I_i\oplus I_iR^-$ is an ascending chain of homogeneous ideals of $R$ hence stable, so $\{I_i\}$ is stable and $R^+$ is Noetherian; on the other hand, $\mathfrak{n}=R^-\oplus (R^-)^2$ is a finitely-generated $R$-ideal so its odd part $R^-$ is finitely-generated as an $R^+$-module. $4 \Rightarrow 3$: $4$ implies that $\Gr_{\mathfrak{n}}R$ satisfies $6$ hence $\Gr_{\mathfrak{n}}R$ is Noetherian whence $3$ follows. $3\Rightarrow 6$: $3$ implies that $\Gr_{\mathfrak{n}}R$ satisfies 2 so by what we have shown $\Gr_{(R^-)^2}R^+$ is Noetherian and $\Gr_{(R^-)^2}R^-$ is a finitely-generated $\Gr_{(R^-)^2}R^+$-module, so $(R^-)^{2N-1}/(R^-)^{2N+1}=0$ for some $N\ge 1$, thus $(R^-)^{2N-1}=0$ by Nakayama Lemma (since elements in $(R^-)^2$ are nilpotent), whence $R^+$ is Noetherian \footnote{For a commutative ring $S$ with an ideal $I$ such that $I^2=0$, if $\Gr_I S=S/I\oplus I$ is Noetherian then so is $S$, since for every ideal $J$ of $S$, $J/J\cap I\oplus J\cap I$ is finitely-generated $\Gr _I S$-module, so $J/J\cap I$ and $J\cap I$ are finitely-generated $S$ modules, thus $J$ is a finitely-generated ideal.} and $R^-$ is a finitely-generated $R^+$-module \footnote{Similarly, for a commutative ring $S$ with an ideal $I$ such that $I^2=0$, if $S$-module $M$ satisfies that $\Gr_I M=M/IM\oplus IM$ is a finitely-generated $\Gr_I S$-module, then $M$ is a finitely-generated $S$-module.}.
\end{proof}

Analogous to classical case, there is a notion of depth of modules in terms of $\Ext$ groups:

\begin{definition}[{\bf Depth of Module}]
Let $A$ be a local Noetherian superring with residue field $k$, $M$ is a finitely-generated $A$ module, define $\dep _A(M)$ to be smallest $n\in \mbb N$ such that $\Ext^n_A(k,M)$ is nontrivial.
\end{definition}

In fact this equals to the depth of $M$ regarded as an $A^+$-module
\begin{lem}\label{Equivalent_Defn_of_Depth}
$\dep_{A^+}(M)=\dep _A(M)$. Consequently, the depth of $M$ equals to the length of maximal regular sequence of $M$.
\end{lem}
\begin{proof}
First of all, for any $A$-module $N$ with finite length, $\dep _A(M)$ is the smallest integer $n$ such that $\Ext^n_A(N,M)$ is nontrivial. This can be seen from induction on the length of $N$ and the long exact sequence of $\Ext$ groups. Note that $\Tor ^{A^+}_p\left(k,A\right)$ has finite length for all $p$. Then it follows immediately from the spectral sequence
\begin{equation*}
	\Ext ^q_{A}\left(\Tor ^{A^+}_p\left(k,A\right),M\right)\Longrightarrow _p \Ext ^{p+q} _{A^+}\left(k,M\right),
\end{equation*}
that $\dep _{A}(M)=\dep _{A^+}(M)$.
\end{proof}

\begin{definition}[{\bf Cohen-Macaulay Module}]
Let $A$ be a local Noetherian superring, $M$ is a finitely-generated $A$ module. $M$ is called Cohen-Macaulay if $\dep _A(M)$ equals to the dimension of $\Supp(M)$\footnote{$\Supp(M)$ denote the support of the module $M$.}.
\end{definition}

Obviously we have the inequality:
\begin{align*}
\dep_A(M)&=\min \{\dep_{A^+}(M^+),\dep_{A^+}(M^-)\}\\
&\le\max \{\dim \Supp (M^+),\dim \Supp (M^-)\}=\dim \Supp (M),
\end{align*}
which implies the following
\begin{lem}\label{Equivalent_Defn_of_CM}
$M$ is Cohen-Macaulay if and only if both $M^+$ and $M^-$ are Cohen-Macaulay $A^+$ modules and either they have the same dimension of supports or one of them is zero. In particular, ring $A$ itself is Cohen-Macaulay if and only if $A^+$ is a Cohen-Macaulay ring and either $A^-$ is a Cohen-Macaulay module for $A^+$ with $\dim\Supp (A^-)=\dim \Spec A$ or $A^-$ is trivial.
\end{lem}

\begin{definition}[{\bf Superscheme}]
	A superscheme is a superspace $(X,\mathcal O_X)$ such that $(X,\mathcal O_X^+)$ is a scheme and $\mathcal O_X^-$ is a quasi-coherent $\mathcal O_X^+$-module. This implies that $\mathfrak{n}_X/\mathfrak{n}_X^2$, which is isomorphic to $\mathcal O_X^-/(\mathcal O_X^-)^3$, is a quasi-coherent $\mathcal O_{X_{\text{\normalfont bos}}}$-sheaf. A superscheme is called Noetherian if $(X,\mathcal O_X^+)$ is Noetherian and $\mathcal O_X^-$ is a coherent $\mathcal O_X^+$ sheaf. This implies that $(\mathcal O_X^-)^N$ and $\mathfrak{n}_X^{N}$ vanish for large $N$. We say that a superscheme $(X,\mcal O_X)$ is affine if its bosonic quotient $(X,\mcal O_X^+)$ is affine \footnote{By Chevalley's Criterion, this is equivalent to its bosonic truncation $(X,\mcal O_X/\mathfrak{n}_X)$ being affine.}.
\end{definition}

\begin{example}
	The basic example is the super affine space $\mathbb A^{p|q}$, which has the same topological space with $\mathbb A^{p}$, but the structure sheaf
	\begin{equation*}
	\mathcal O_{\mathbb A^{p|q}}=\bigoplus_{i=0}^q\bigwedge^i\mathcal O_{\mathbb A^{p}}^q,
	\end{equation*}
	endowed with the natural $\mathbb Z_2$-grading.
\end{example}

\begin{definition}[{\bf Quasicoherent Sheaves}]
	An $\mathcal O_X$-module $\mathcal F$ on a superscheme $X$ is said to be quasicohenrent if there is an open covering of $X$ such that for each open $U$ in the covering family, $\mathcal F$ is the the cokernel of an \textbf{even} homomorphism$$\mathcal O_U^{K|L}\to \mathcal O_U^{I|J},$$where $I,J,K,L$ are index sets, and $\mathcal O_U^{I|J}:=\mathcal O_U^I\oplus \prod\mathcal O_U^J$. Moreover $\mathcal F$ is called finitely-generated (respectively finite presented) if $I$ and $J$ are finite (respectively all index sets are finite) for all $U$. If $X$ is Noetherian, then we say that $\mathcal F$ is coherent if it is finitely-generated.
\end{definition}
Similar to the classical case, we have the following lemma on the extension of coherent sheaves:
\begin{lem}
	Let $S$ be a Noetherian superscheme, $U\subset X$ is open and $E\in \Coh(\mathcal O_U)$, then $\exists F\in \Coh(\mathcal O_S)$ such that $E=F|_U$.
\end{lem}
\begin{proof}
	First consider the case where $S = \Spec A$ is affine, and let $j:U \to S$ denote the inclusion. The quasi-coherent sheaf $j_*(E)$ corresponds to the $A$-module $M = \Gamma(U,E)$. Given any $u\in U$, there exist finitely many elements $e_1,\cdots,e_n \in M$ which generate the fiber $E(u)$ regarded as a super-vector space over the residue field $k(u)$. By Nakayama Lemma these elements will generate the stalks of $E$ in an open neighbourhood of $u$ in $U$. Therefore by the Noetherian hypothesis, there exist finitely many elements $f_1,\cdots,f_N \in M$ which generate the stalk of $E$ at each point of $U$. If $N\subset M$ is the submodule generated by these elements, then $F=N^{\sim}$ is a coherent prolongation of $E$ to $S$, proving the result in the affine case. In the general case, we can use the Noetherian hypothesis to assume that there is a maximal $V$ on which $E$ can be extended, if $V\neq S$ then take an open affine $V'$ not contained in $V$ and extend $F_V$ to $V\cup V'$, contradicting the maximality of $V$.
\end{proof}

\begin{definition}
	Given a morphism $f:X\to Y$ between Noetherian superschemes, we say that
	\begin{enumerate}
		\item $f$ has property $\mathbf P$ if $f_{\text{\normalfont bos}}:X_{\text{\normalfont bos}}\to Y_{\text{\normalfont bos}}$ has property $\mathbf P$, where $\mathbf P$ can be finite type, separated, proper, finite.
		\item $f$ is flat if $\mathcal O_X$ is a flat $f^{-1}\mathcal O_Y$ module.
		\item $f$ is unramified if $f$ is of finite type and for any geometric point $y$ of $Y$, the fiber $X_y$ is bosonic and \' etale.
		\item $f$ is \' etale if $f$ is flat and unramified.
		\item $f$ is smooth of relative dimension $(p|q)$ if locally on $X$ there is an \' etale morphism to $Y\times \mathbb A^{p|q}$ whose composition with projection $Y\times \mathbb A^{p|q}\to Y$ is $f$.
	\end{enumerate}
\end{definition}

\begin{lem}[Local Criteria of Flatness]\label{Local Criteria of Flatness}
	Let $A$ be a superring, $I$ an ideal of $A$ and $M$ an $A$-module. Assume that either
	\begin{enumerate}
		\item[$(a)$] $I$ is nilpotent, or
		\item[$(b)$] $A$ is local Noetherian, $B$ is a local Noethrian $A$-superalgebra, and $M$ is a finitely-generated $B$-module.
	\end{enumerate}
	Then the following are equivalent
	\begin{enumerate}
		\item[$(1)$] $M$ is $A$-flat.
		\item[$(2)$] $\Tor_1^A(N,M)=0$ for all $A/I$-module $N$.
		\item[$(3)$] $M/IM$ is $A/I$-flat and $I\otimes M\overset{\sim}{\rightarrow} IM$ by natural map.
		\item[$(3')$] $M/IM$ is $A/I$-flat and $\Tor_1^A(A/I,M)=0$.
		\item[$(4)$] The natuarl map $$\Gr_IA\otimes_{A/I}M/IM\to \Gr_IM,$$ is isomorphism.
		\item[$(5)$] $M/I^nM$ is $A/I^n$-flat, for all $n\ge 0$.
	\end{enumerate}
\end{lem}
The proof is the same as in the classical case, see 20.C of \cite{Matsumura1989}, except that the classical Artin-Rees Lemma used in the proof of \textit{loc.cit.} should be enhanced to its super analog proven below.
\begin{lem}[Artin-Rees]\label{Artin-Rees}
	Let $I$ be an ideal of Noetherian superring $A$, $M$ a finitely-generated $A$-module and $N$ a submodule of $M$. Then there exists a $k\ge 1$ such that
	$$I^nM\cap N=I^{n-k}(I^kM\cap N),$$for all $n\ge k$.
\end{lem}
\begin{proof}
	Apparently $I^nM\cap N\supset I^{n-k}(I^kM\cap N)$, so it suffices to prove the reverse inclusion. Define
	\begin{equation*}
	\widetilde{M}^+\equiv I^+M^++I^-M^-, \qquad \widetilde{M}^-\equiv I^+M^-+I^-M^+,
	\end{equation*}
	then $I^nM$ decomposes into even part $(I^+)^{n-1}\widetilde{M}^+$ and odd part $(I^+)^{n-1}\widetilde{M}^-$, so we have $(I^nM\cap N)^+=(I^+)^{n-1}\widetilde{M}^+\cap N^+$. Using the classical Artin-Rees Lemma, there exists a $k$ such that $$(I^+)^{n-1}\widetilde{M}^+\cap N^+=(I^+)^{n-k}((I^+)^{k-1}\widetilde{M}^+\cap N^+)=(I^+)^{n-k}((I^kM)^+\cap N^+),$$for all $n\ge k$. The last term is contained in $I^{n-k}(I^kM\cap N)$, whence the reverse inclusion holds.
\end{proof}

For the rest of this section, we will only consider Noetherian superschemes unless otherwise specified.

\begin{lem}\label{Lemma_Flat_Bosonic}
	Suppose that $f:X\to Y$ is flat and $f^{-1}\mathfrak n_Y$ generates $\mathfrak n_X$, then the commutative diagram 
	\begin{center}
		\begin{tikzcd}
		X\arrow[r, "f"] \arrow[d] & Y\arrow[d]\\
		X_{\text{\normalfont ev}} \arrow[r, " f_{\text{\normalfont ev}}"] &Y_{\text{\normalfont ev}} .
		\end{tikzcd}
	\end{center}
	is Cartesian, and $f_{\text{\normalfont ev}}$ is flat.
\end{lem}

\begin{proof}
	Since $f^{-1}\mathfrak n_Y$ generates $\mathfrak n_X$, the local flatness implies that the natural map is isomorphism:
	$$\mathcal O_{X_{\text{bos}}}\otimes _{f^{-1}\mathcal O_{Y_{\text{bos}}}}f^{-1}\Gr_{\mathfrak{n}_Y}\mathcal O_{Y}\cong \Gr_{\mathfrak{n}_X}\mathcal O_{X}.$$
	Hence the homomorphism between bosonic graded alegebras
	\begin{equation*}
	\mathcal O_{X_{\text{bos}}}\otimes _{f^{-1}\mathcal O_{Y_{\text{bos}}}}\bigoplus _{i\ge 0}f^{-1}\Gr_{\mathfrak{n}_Y}^{2i}\mathcal O_{Y}\to \bigoplus _{i\ge 0}\Gr_{\mathfrak{n}_X}^{2i}\mathcal O_{X},
	\end{equation*}
	is an isomorphism. So $f_{\text{ev}}$ is flat by local flatness again. It is straightforward from the isomorphism between graded algebras that the induced homomorphism $f_{\text{\normalfont ev}}^*\mathcal O_Y\to \mathcal O_X$ is an isomorphism of sheaves, i.e. $X\cong X_{\text{\normalfont ev}}\times _{Y_{\text{\normalfont ev}}} Y$.
\end{proof}

\begin{lem}\label{Lemma_Equiv_Defn_of_Etale}
	$f:X\to Y$ is \' etale if and only if $f_{\text{\normalfont ev}}:X_{\text{\normalfont ev}}\to Y_{\text{\normalfont ev}}$ is \' etale and the commutative diagram 
	\begin{center}
		\begin{tikzcd}
		X\arrow[r, "f"] \arrow[d] & Y\arrow[d]\\
		X_{\text{\normalfont ev}} \arrow[r, " f_{\text{\normalfont ev}}"] &Y_{\text{\normalfont ev}} .
		\end{tikzcd}
	\end{center}
	is Cartesian.
\end{lem}

\begin{proof}
	If $f$ is \' etale, then $X\times_{Y}Y_{\text{bos}}$ is bosonic, since its nilpotent ideal $\mathfrak{n}$ generated by fermionic elements is zero after restriction to fibers, thus it is identically zero by Nakayama Lemma. So we see that $f^{-1}\mathfrak n_Y$ generates $\mathfrak n_X$. By Lemma \ref{Lemma_Flat_Bosonic}, $f_{\text{ev}}$ is flat and the diagram is Cartesian. It follows that $f_{\text{ev}}$ and $f$ have the same fibers hence $f_{\text{ev}}$ is \' etale. The converse is obvious.
\end{proof}

\begin{lem}\label{Lemma_Superspace_EtaleLiftingSplitProjected}
	Suppose that $f:X\to Y$ is \' etale. If $Y$ is split {\normalfont (}resp. projected{\normalfont)}, then so is $X$.
\end{lem}

\begin{proof}
	We use the charaterization of \' etale morphism in Lemma \ref{Lemma_Equiv_Defn_of_Etale}. First of all, $f$ being \' etale implies that $\mathcal O_X^-$ is generated by $f^{-1}\mathcal O_Y^-$, hence there is an isomorphism
	\begin{equation*}
	X_{\text{\normalfont bos}}\cong X_{\text{\normalfont ev}}\times _{Y_{\text{\normalfont ev}}}Y_{\text{\normalfont bos}}.
	\end{equation*}
	Assume that $Y$ is projected, then there is a retract $\pi:Y_{\text{\normalfont ev}}\to Y_{\text{\normalfont bos}}$. To show that $X$ is projected, it suffices to show that there is a retract $X_{\text{\normalfont ev}}\to X_{\text{\normalfont bos}}$. This follows from $f_{\text{\normalfont ev}}$ being \' etale hence there is a unique dotted arrow making following diagram commutes
	\begin{center}
		\begin{tikzcd}
		X_{\text{\normalfont bos}}\arrow[r, equal] \arrow[d, "f_{\text{\normalfont bos}}" swap] &X_{\text{\normalfont bos}}\arrow[d, hook]\\
		Y_{\text{\normalfont bos}} &X_{\text{\normalfont ev}} \arrow[l, "\pi\circ f_{\text{\normalfont ev}}"]\arrow[ul, dotted, "\exists !" description].
		\end{tikzcd}
	\end{center}
	By uniqueness, the composition of embedding $X_{\text{\normalfont bos}}\hookrightarrow X_{\text{\normalfont ev}}$ with the dotted arrow is identity, hence $X$ is projected.\\
	
	Assume that $Y$ is split, then $Y$ is projected hence $X$ is also projected and the previous argument shows that $\mathcal O_{X_{\text{\normalfont ev}}}\cong f_{\text{\normalfont bos}}^*\mathcal O_{Y_{\text{\normalfont ev}}}$ as a sheaf of $O_{X_{\text{\normalfont bos}}}$-algebras. This in turn implies that $\mathcal O_{X}\cong f_{\text{\normalfont bos}}^*\mathcal O_{Y}$ whence the splitness of $X$ follows.
\end{proof}

\begin{lem}\label{Lemma_Nilpotent_Lift_of_Etale_Morphism}
	Suppose that $f:X\to Y$ is \' etale. If $Y\hookrightarrow Y'$ is a nilpotent thickenning, then there exists a unique lifting of \' etale morphism $f':X'\to Y'$ such that its restriction to $Y$ is $f$.
\end{lem}

\begin{proof}
	Consider bosonic quotient of $X,Y,Y'$ and fit them into
	\begin{center}
		\begin{tikzcd}
		X_{\text{\normalfont ev}}\arrow[d, "f_{\text{\normalfont ev}}"] & \\
		Y_{\text{\normalfont ev}} \arrow[r, hook] &Y'_{\text{\normalfont ev}} .
		\end{tikzcd}
	\end{center}
	Now the classical nilpotent lifting result for \' etale morphisms shows that there exists a unique $f'_{\text{\normalfont ev}}:X'_{\text{\normalfont ev}}\to Y'_{\text{\normalfont ev}}$ which makes the diagram Cartesian
	\begin{center}
		\begin{tikzcd}
		X_{\text{\normalfont ev}}\arrow[d, "f_{\text{\normalfont ev}}"] \arrow[r, hook] &X'_{\text{\normalfont ev}} \arrow[d, "f'_{\text{\normalfont ev}}"]\\
		Y_{\text{\normalfont ev}} \arrow[r, hook] &Y'_{\text{\normalfont ev}} .
		\end{tikzcd}
	\end{center}
	Then $X'=X'_{\text{\normalfont ev}}\times_{ Y'_{\text{\normalfont ev}}}Y'$ is a lifting and the uniqueness is obvious.
\end{proof}

\begin{remark}
	There is a notion of {\normalfont (}Artin, Deligne-Mumford{\normalfont)} superstack which is defined in a similar flavor of usual stack {\normalfont \cite{CodogniViviani201706}}, but the exact definition is not important here and we will skip it. Note that all statements remain true if we relax the the conditions to allow DM superstacks {\normalfont \cite{CodogniViviani201706}}.
\end{remark}

\subsubsection{Superschemes over Inductive Limit of Superrings}

Here we include some results for superschemes/sheaves/morphisms under the inverse limit of base superscheme. The setup is similar to the classical algebaric geometry: Let $(I,\le)$ be a directed set, with a distinguished element $0\in I$, suppose that $\{A_{i}\}_{i\in I}$ is an directed system of superrings, with inductive limit $A$.

\begin{lem}\label{Lemma_Limit_of_Morphisms}
	If $X_0, Y_0$ are finite presented surperschemes over $\Spec A_0$, let $X_i,Y_i$ be the base change of $X_0,Y_0$ to $\Spec A_i$, and $X,Y$ be the base change of $X_0,Y_0$ to $\Spec A$, then the natural map 
	\begin{equation*}
	\varphi:\lim_{\substack{\longrightarrow\\i}}\Map_{\Spec A_i}(X_i,Y_i)\to \Map_{\Spec A}(X,Y),
	\end{equation*}
	is an equivalence of sets.
\end{lem}

\begin{proof}
	{\small\bf $\varphi$ is injective:} suppose that two morphisms $f_0,g_0:X_0\to Y_0$ has the same limit $f=g:X\to Y$, consider the kernel of two morphisms $\j_0:\Ker (f_0,g_0)\to X_0$, which is a finite presented immersion, becomes an isomorphism after passing to the limit, hence $j_i:\Ker (f_i,g_i)\to X_i$ is a homeomorphism for some index $i\in I$. Since the ideal defining $\j_i$ is finitely-generated, so locally we can take finitely many generators and they vanishes after taking limit, thus vanishing for some index $i'\in I$, consequently $j_{i'}:\Ker (f_{i'},g_{i'})\to X_{i'}$ is an isomorphism and $f_{i'}=g_{i'}$.
	
	\newl {\bf\small $\varphi$ is surjective:} suppose that $f:X\to Y$ is a morphism, we cover $Y_0$ by finite many affine open sub superschemes $\{U_{\lambda 0}\to Y_0\}$, and let $U_{\lambda }$ be the base change of $U_{\lambda 0}$ to $\Spec A$, then cover $X$ by finite many open affine sub superschemes $\{V_{\lambda \mu}\to X\}$ such that $f(V_{\lambda \mu})\subset U_{\lambda }$. According to Theorem 10.57 of \cite{GortzWedhorn2010}, we can assume that there is an index $i\in I$  and an open covering $\{V_{\lambda \mu i}\to X_i\}$ whose base change to $X$ is $\{V_{\lambda \mu}\to X\}$. Refine $V_{\lambda \mu i}$ such that each open sub superscheme is affine. Let us still use the same notation to denote this refined cover, then $\mcal O_X(V_{\lambda \mu})$ is a finite presented $\mcal O_Y(U_{\lambda })$ superalgebra, then there exists an index $i'$ such that $\mcal O_X(V_{\lambda \mu i'})$ is a finite presented $\mcal O_Y(U_{\lambda i'})$ superalgebra, i.e. we found $f_{\lambda \mu i'}:V_{\lambda \mu i'}\to Y_{i'}$ whose base change to $Y$ is $f$ restricted on $V_{\lambda \mu}$. Finally, increase the index $i'$ such that $f_{\lambda \mu i'}$ agrees on the overlap (because of the injectivity of $\varphi$). Hence $f_{\lambda \mu i'}$ glues to a morphism $f_{i'}:X_{i'}\to Y_{i'}$
\end{proof}

\begin{lem}\label{Lemma_Limit_of_Superschemes}
	If $X$ is a finite presented surperscheme over $\Spec A$, then there exists a finite presented superscheme $X_i$ over $\Spec A_i$ such that its base change to $\Spec A$ is isomorphic to $X$.
\end{lem}

\begin{proof}
	Cover $X$ by open affine sub superschemes $\{U_{\lambda}\to X\}$, then each $\mcal O_X(U_{\lambda})$ is a finite presented $A$ superalgebra, thus $\exists i\in I$ such that there exists finite presented $A_i$ superalgebras $B_{\lambda i}$ whose base change to $A$ is $\mcal O_X(U_{\lambda})$. Let $U_{\lambda}$ be $\Spec B_{\lambda i}$. Now increase the index $i$ such that there exists open sub superschemes $V_{\lambda \mu i}\subset U_{\lambda}$ whose base change to $U_{\lambda}$ is $U_{\mu}\cap U_{\lambda}$, together with isomorphism $\psi _{\lambda \mu i}:V_{\lambda \mu i}\cong V_{\mu \lambda i}$ satisfying cocycle conditions whose base change to $\Spec A$ is the gluing morphism for the cover $U_{\lambda}$. The existence of such open sub superschemes follows from Theorem 10.57 of \cite{GortzWedhorn2010} and the existence of such isomorphisms follows from Lemma \ref{Lemma_Limit_of_Morphisms}. Hence $V_{\lambda \mu i}$ glues to a superscheme $X_i$ over $\Spec A_i$, whose base change to $\Spec A$ is isomorphic to $X$.
\end{proof}

In summary, the natural functor between categories of finite presented schemes$$\lim_{\substack{\longrightarrow\\i}}\Sch_{\text{\normalfont fp}}({S_i})\to \Sch_{\text{\normalfont fp}}({S}),$$ is an equivalence of categories. We also have similar result for finite presented sheaves, whose proof is similar to the bosonic case, which will be omitted:
\begin{lem}\label{Lemma_Limit_of_Map_Between_Sheaves}
	Let $(I,\le)$ be a directed set, suppose that $\{S_{i}\}_{i\in I}$ is an directed system of qcqs\footnote{This is stands for quasi-compact quasi-separated.} superschemes with affine transition morphisms, with inverse limit $S$, then the natural functor between categories of finite presented sheaves $$\lim_{\substack{\longrightarrow\\i}}\QCoh_{\text{\normalfont fp}}(\mcal O_{S_i})\to \QCoh_{\text{\normalfont fp}}(\mcal O_{S}),$$ is an equivalence of categories.
\end{lem}

For later application, we also need a result about flatness:
\begin{proposition}\label{Limit_Flatness}
	If $X$ is a finite presented surperscheme over $\Spec A$, $\mcal F\in \QCoh_{\text{\normalfont fp}}(\mcal O_X)$, $X$ is the inverse limit of a system of finite presented superscheme $X_i$ over $\Spec A_i$ and $\mcal F$ is the inverse limit of $\mcal F_i\in\QCoh_{\text{\normalfont fp}}(\mcal O_{X_i})$, then $\mcal F$ is flat over $\Spec A$ if and only if $\exists i\in I$ such that $\mcal F_i$ is flat over $\Spec A_i$.
\end{proposition}

\begin{proof}
	This is the super analog of Th\' eor\` eme 11.2.6 of \cite{GrothendieckEGAIV1967}. Let's explain why the proof of \textit{loc.cit.} Th\' eor\` eme 11.2.6 works for superschemes as well. In the step I, the key ingredient is to find a Noetherian sub superrings of $A_0$, denoted by $A_0'$ such that there exists $(X_0',\mcal F_0')$ on $\Spec A_0'$ whose base change to $\Spec A_0$ is $(X_0,\mcal F_0)$, this can be done because we can finite many even elements $\{x_i\}_{i=1}^m$ in $A_0^+$ and finite many odd elements $\{\theta_j\}_{j=1}^n$ in $A_0^-$, then we define a homomorphism $$\mathbb Z[X_i|\Theta_j]\to A_0,$$by sending $X_i$ to $x_i$ and $\Theta_j$ to $\theta_j$, its image in $A_0$ is a Noetherian superring. Increase the size of this set and we can apply Lemma \ref{Lemma_Limit_of_Superschemes} and \ref{Lemma_Limit_of_Map_Between_Sheaves}. The step II is a pure topological statement so it works for superschemes as well. The key ingredient of step III is showing that $\Tor ^{A_0}_{1}(A_0/\mathfrak{m}_0,M_0)$ is finitely-generated $B_0$-module in the setup that $A_0$ is a Noetherian superring, $B_0$ is a finite type $A_0$ superalgebra and $M_0$ is a finitely-generated $B_0$-module. This is because we can take a resolution of $A_0/\mathfrak{m}_0$ by finite free $A_0$ modules, with even connecting homomorphisms, by the Noetherian property of $A_0$. Now the same argument in step III shows that the canonical homomorphism
	\begin{equation*}
	\Tor ^{A_0}_{1}(A_0/\mathfrak{m}_0,M_0)\to \Tor ^{A_i}_{1}(A_i/\mathfrak{m}_0A_i,M_i),
	\end{equation*}
	is zero for some $i\in I$, then the flatness of $M_i$ over $B_i$ follows from the super analog of \textit{loc.cit.} Lemme 11.2.4 (stated in the Lemma \ref{Lemma_Canonical_Map_of_Tor_Is_Surjective}) and the local criterion for flatness.
\end{proof}

\begin{lem}\label{Lemma_Canonical_Map_of_Tor_Is_Surjective}
	Let $A$ be a superring, $C$ is an $A$ superalgebra, $M$ is an $A$ module, $N$ is a $C$ module and $A'$ is an $A$ algebra. Let $C'=C\otimes_A A'$, $M'=M\otimes_A A'$ and $N'=N\otimes_A A'$, then the canonical homomorphism
	\begin{equation*}
	\Tor ^{A}_{1}(M,N)\otimes_A A'\to \Tor ^{A'}_{1}(M', N'),
	\end{equation*}
	is surjective.
\end{lem}
The proof is formally the same as the proof of Lemme 11.2.4 of \cite{GrothendieckEGAIV1967}, except that the resolutions $L_1\to L_0\to M\to 0$ and $L'_1\to L'_0\to M'\to 0$ in the \textit{loc.cit.} are required to have even connecting homomorphisms.

\subsubsection{Affineness Criteria}
Here are super analog of Serre's and Chevalley's Criteria for affineness.

\begin{proposition}[Serre's Criterion]\label{Serre's Criterion}
	Let $X$ be a quasi-compact superscheme. Then the following statements are equivalent:
	\begin{enumerate}
		\item[(1)] $X$ is affine.
		\item[(2)] There exists a family of even elements $f_i\in A\equiv \Gamma(X, \mcal O_X)$ such that the open sub superscheme $X_{f_i}:=\{x\in X | f_i(x) \neq 0\}$ of $X$ is affine for all $i$ and the ideal generated by the $f_i$ in $A$ is the unit ideal.
		\item[(3)] $\H^1(X,\mcal F)=0$ for all $\mcal F\in \QCoh(\mcal O_X)$.
		\item[(4)] $\H^1(X,\mcal I)=0$ for all quasi-coherent ideal $\mcal I$.
	\end{enumerate}
\end{proposition}

\begin{proof}
	Since statement $(2)$ only involves even elements $f_i$ and $f_i$ generate unit ideal in $A$ if and only if they generate unit ideal in $A^+$, the equivalence between $(1)$ and $(2)$ follows from classical Serre's criterion.
	
	\newl $(1)\Longrightarrow (3):$ Write $\mcal F=\mcal F^+\oplus \mcal F^-$ be its even and odd parts, considered as $\mcal O_X^+$ module, then $(X,\mcal O_X^+)$ being affine implies that $\H^1(X,\mcal F^{\pm})=0$. 
	
	\newl $(3)\Longrightarrow (4)$ is trivial, so let us prove $(4)\Longrightarrow (1):$ For any quasi-coherent ideal $\mcal J$ of $\mcal O_X^+$, consider $\mcal I=\mcal J\oplus \mcal J\mcal O_X^-\subset\mcal O_X$, which is a quasi-coherent $\mcal O_X$-ideal, so $\H^1(X,\mcal I)=\H^1(X,\mcal J)\oplus \H^1(X,\mcal J\mcal O_X^-)=0$, hence $\H^1(X,\mcal J)=0$ and by classical Serre's criterion, $X_{\text{ev}}$ is affine.
\end{proof}

\begin{proposition}[Chevalley's Criterion]\label{Chevalley's Criterion}
	Suppose that $f:X\to Y$ is a finite morphism between superschemes, then $X$ is affine if and only if $Y$ is affine.
\end{proposition}

\begin{proof}
	Since $f_{\text{bos}}:X_{\text{bos}}\to Y_{\text{bos}}$ is finite, then the result follows from classical Chevalley's criterion.
\end{proof}

\subsubsection{Chow's Lemma for Superschemes}
\begin{thr}
	Let $S$ be a Noetherian superscheme, $f:X\to S$ is of finite type and separated. Then there exists a dense open $U\subset X$ and a projective morphism $p:X'\to X$, such that $p|_{p^{-1}(U)}$ is an isomorphism, and $X'$ is quasi-projective over $S$.
\end{thr}

\begin{proof}
	First we consider the bosonic quotient of $X$ and $S$, denoted by $X_{\ev}$ and $S_{\ev}$, and fit them into a commutative diagram:
	\begin{center}
		\begin{tikzcd}
		X \arrow[drr,"f"]\arrow[dr,"\pi",swap, near end]\arrow[ddr]& &\\
		&X_{\ev}\times _{S_{\ev}}S\arrow[r,"f'",swap] \arrow[d] &S\arrow[d]\\
		&X_{\ev}\arrow[r,"f_{\ev}"] &S_{\ev}
		\end{tikzcd}
	\end{center}
	Both $f_{\ev}$ and $f'$ are of finite type, since $X_{\bos}\to S_{\ev}$ is of finite type and nilpotent ideal defining $X_{\bos}\subset X_{\ev}$ is finitely-generated. $\pi$ is finite, since $\pi_{\bos}:X_{\bos}\to \left(X_{\ev}\times _{S_{\ev}}S\right)_{\bos}=X_{\bos}$ is identity, and nilpotent ideal defining $X_{\bos}\subset X$ is finitely-generated.
	
	\newl By applying the classical Chow's Lemma to $f_{\ev}$, we get a dense open $U_{\ev}\subset X_{\ev}$ and a projective morphism $q:Y\to X_{\ev}$, such that $q|_{q^{-1}(U_{\ev})}$ is an isomorphism, and $Y$ is a quasi-projective $S_{\ev}$-scheme. So the natural morphism $Y\times _{S_{\ev}}S\to X_{\ev}\times _{S_{\ev}}S$ is projective, and $Y\times _{S_{\ev}}S$ is quasi-projective over $S$. Since $\pi$ is finite, pullback of $Y\times _{S_{\ev}}S$ along $\pi$ is also quasi-projective over $S$, and call it $X'$. The projection $p: X'\to X$ is an isomorphism on $p^{-1}(U)$ so it satisfies all conditions.
\end{proof}

\subsubsection{Pushforward and Flatness}
The following is the super analog of Grothendieck's Complex, whose proof is also analogous to the classical case.
\begin{thr}\label{Grothendieck_Complex}
	Let $S=\Spec R$ be a Noetherian affine superscheme, $f:X\to S$ is proper, $\mathcal F\in \Coh(\mathcal O_X)$ and is flat over $S$, then $Rf_*\mathcal F$ is quasi-isomorphic to a perfect complex with cohomological amplitudes in $[0,d]$, where $d$ is the maximal dimension of $\Supp (\mathcal F_s)$, and its formation commutes with arbitrary base change.
\end{thr}

In the same way as the classical situation, Grothendieck's Complex implies the base-change theorem for flat sheaves:
\begin{thr}[{Theorem 5.10 of \cite{FantechiIllusieGottsche2005}}]\label{Proper_Flat_Pushforward}
	Let $\pi : X \to S$ be a proper morphism of Noetherian superschemes, 
	and let $\mathcal F$ be a coherent $\mathcal O_X$-module which is flat over $S$. Then the following statements hold: 
	\begin{enumerate}
		\item For any integer $i$ the functions $s \mapsto \dim_{\kappa(s)} \H^i(X_s,\mathcal F_s)_{+}$ and $s \mapsto \dim_{\kappa(s)} \H^i(X_s,\mathcal F_s)_{-}$ are upper semi-continuous on $S$.
		
		\item The functions $s \mapsto \dim_{\kappa(s)} \chi(X_s,\mathcal F_s)_{+}$ and $s \mapsto \dim_{\kappa(s)} \chi(X_s,\mathcal F_s)_{-}$ are  locally constant on $S$. 
		
		\item Suppose that $S$ is a reduced scheme, and if for some integer $i$, there is a pair of integers $d_{\pm} \ge 0$ such that for all $s\in S$ we have $\dim_{\kappa(s)} \H^i(X_s,\mathcal F_s)_{\pm}=d_{\pm}$, then $R^i\pi_*\mathcal F$ is locally free of rank $(d_+|d_{-})$ and $(R^i\pi_*\mathcal F)_s\to \H^i(X_s,\mathcal F_s)$ is an isomorphism for all $s\in S$.
		\item If for some integer $i$ and point $s\in S$ the map $(R^i\pi_*\mathcal F)_s\to \H^i(X_s,\mathcal F_s)$ is surjective, then there exists an open subscheme $U\subset S$ such that both $\tau_{\le i}R\pi_*\mathcal F|_U$ and $\tau_{\ge i+1}R\pi_*\mathcal F|_U$ are quasi-isomorphic to perfect complex on $U$.
	\end{enumerate}
\end{thr}

\begin{corollary}[Hom Superscheme]\label{Hom_Superscheme}
	Let $\pi:X\to S$ be quasi-projective morphism between Noetherian superschemes, and let $G$ and $F$ be coherent sheaves on $X$. Suppose that $F$ is flat over $S$ and $\Supp (F)$ is proper over $S$. Then the functor $\mathfrak{Hom}_S(G,F)_+: T\mapsto \Hom_{X_T}(G_T,F_T)_+$ is represented by a linear superscheme $\mathbb{V}(\mathcal Q)$ for a coherent sheaf $\mathcal Q$ on $S$.
\end{corollary}

\begin{proof}
	Since $\pi$ is quasi-projective, there exists finite locally free sheaves $E_1,E_2$ on $X$ and a map $\phi:E_2\to E_1$ with $\coker \phi=G$. It follows that
	\begin{equation*}
	\Hom_{X_T}(G_T,F_T)\cong \H^0R\pi_{T*}\underline{R\Hom}_{X_T}([E_{2T}\to E_{1T}], F_T).
	\end{equation*}
	Theorem \ref{Grothendieck_Complex} implies that $R\pi_{*}\underline{R\Hom}_{X}([E_{2}\to E_{1}], F)$ is a perfect complex of non-negative cohomological degrees, whose formation commutes with arbitrary base change. Denote this perfect complex by $\mathcal K^{\bullet}$. $(\mathcal K^{\bullet})^{\vee}$ is of non-positive cohomological degrees, whence
	\begin{equation*}
	\H^0(\mathcal K_T^{\bullet})=\Hom_{\mathcal O_T}(\H^0((\mathcal K^{\bullet})_T^{\vee}),\mathcal O_T).
	\end{equation*}
	So $\mathfrak{Hom}_S(G,F)$ is represented by $\mathbb{V}(\mathcal Q)$ where $\mathcal Q=\H^0((\mathcal K^{\bullet})^{\vee})$.
\end{proof}

\begin{lem}\label{Flatness_Criterion}
	Let $S=\Spec R$ be a Noetherian affine superscheme, $X=\mathbb P^{n|m}_S$, $\mathcal F\in \Coh(\mathcal O_X)$, then $\mathcal F$ is flat over $S$ if and only if $\exists N$ such that $\forall i>N$, $\pi_*\mathcal F(i)$ is flat on $S$.
\end{lem}

\begin{proof}
	This is the super analog of Lemma 5.5 of \cite{FantechiIllusieGottsche2005}. The key point is that $$\mathcal F=\Gamma_{*}(\mathbb P^{n|m}_S, \mathcal F)^{\sim}.$$The reason is that there is a canonical homomorphism $\Gamma_{*}(\mathbb P^{n|m}_S, \mathcal F)^{\sim}\to\mathcal F$ and we know that there is a projection $\mathbb P^{n|m}\to \mathbb P^{n}$, so $\Gamma_{*}(\mathbb P^{n|m}_S, \mathcal F)$ is identified with $\Gamma_{*}(\mathbb P^{n}_S, \mathcal F)$ under this projection, where the latter is considered as a $\mathcal O_{\mathbb P^n}$ module. Now the isomorphism follows from its classical counterpart. The rest of proof is the same as \textit{loc. cit.}
\end{proof}

We also have the super analog of the Formal Function Theorem:
\begin{thr}[Super Version of Formal Function Theorem]
	Let $\pi : X \to S$ be a proper morphism of Noetherian superschemes, 
	and let $\mathcal F$ be a coherent $\mathcal O_X$-module, let $s\in S$. $X_s^{(n)}$ is the $n$'th formal neighborhood of $X_s$ and $(R^i\pi_*\mathcal F)_s^{\wedge}$ be the completion of $R^i\pi_*\mathcal F$ at $s$. Then there is an isomorphism
	$$(R^i\pi_*\mathcal F)_s^{\wedge}\cong \lim_{\substack{\longleftarrow\\n}}\H^i(X_s^{(n)},\mathcal F_s^{(n)}).$$
\end{thr}

\begin{proof}
	Localize $S$ and we can assume that $S=\Spec R$ for a local superring $R=R^+\oplus R^-$, whose closed point is $s$. Let the maximal ideal of $R^+$ be $\mathfrak{m}_0$, and let the ideal of $R$ generated by $R^-$ be $\mathfrak{n}$. Then the maximal idela of $R$ defining $s$ is $\mathfrak{m}_0R+\mathfrak{n}$. Note that $(\mathfrak{m}_0R+\mathfrak{n})^N\subset\mathfrak{m}_0^NR+\mathfrak{n}^N $, so $(\mathfrak{m}_0R+\mathfrak{n})^N=\mathfrak{m}_0^NR$ for large $N$. We first replace $\pi$ by $\pi':X\to S_{\ev}=\Spec R^+$, and the classical Formal Function Theorem states that
	$$(R^i\pi'_*\mathcal F)_s^{\wedge}\cong \lim_{\substack{\longleftarrow\\n}}\H^i(X_s^{(n)},\mathcal F_s^{(n)}).$$
	Note that $R^i\pi'_*\mathcal F$ is the same as $R^i\pi_*\mathcal F$ considered as a $R^+$-module. Now the theorem follows from the fact that $\mathfrak{m}_0R$ and $\mathfrak{m}_0R+\mathfrak{n}$ generate the same topology on $R$ hence two limit systems are cofinal to each other.
\end{proof}

\subsubsection{Super Analog of Grothendieck's Comparison and Existence Theorems}
Let $A$ be an adic Noetherian superring, $I$ an homogeneous ideal of definition of $A$, $Y=\Spec A$, $Y_n=\Spec A/I^n$, and$$\widehat{Y}=\lim_{\substack{\longleftarrow\\n}}Y_n=\Spf A.$$

\begin{thr}[Grothendieck's Comparison Theorem]\label{Grothendieck's_Comparison_Theorem}
	Let $f:X\to Y$ be a finite type morphism between superschemes, and let $\widehat{X}$ be the $I$-adic completion of $X$. Suppose that $\mathcal F\in \Coh(\mathcal O_X)$ has proper support over $Y$, then the base change homomorphisms
	\begin{alignat}{2}
	&(R^qf_*\mathcal F)^{\wedge}\to R^q\widehat{f}_*\widehat{\mathcal F}.\nonumber\\
	&R^q\widehat{f}_*\widehat{\mathcal F}\to \lim_{\substack{\longleftarrow\\n}}R^q(f_{n})_* \mathcal F_n,
	\end{alignat}
	are isomorphisms.
\end{thr}

\begin{proof}
	Write $A=A^+\oplus A^-$ be the decomposition of $A$ into even and odd parts and also $I=I^+\oplus I^-$, where $I^+$ is even and $I^-$ is odd. Since $I^-$ is odd thus is nilpotent, the topology on $\widehat X$ and $\widehat Y$ can also be generated by $(I^+)^n$. Now we regard every thing to be superscheme over the bosonic adic Noetherian ring $A^+$ with $I^+$-topology. Write $\mathcal F=\mathcal F_+\oplus \mathcal F_-$ be decomposition into even and odd parts. By classical Grothendieck's Comparison Theorem, there are natural isomorphisms of $A^+$ modules
	\begin{alignat*}{2}
	&(R^qf_*\mathcal F_{\pm})^{\wedge}\to R^q\widehat{f}_*\widehat{\mathcal F}_{\pm},\\
	&R^q\widehat{f}_*\widehat{\mathcal F}_{\pm}\to \lim_{\substack{\longleftarrow\\n}}R^q(f_{n})_* (\mathcal F_{\pm})_n.
	\end{alignat*}
	The super analog comes from the direct sum of even and odd parts.
\end{proof}

\begin{corollary}\label{Completion_Fully_Faithful}
	In the same setup as {\normalfont Theorem \ref{Grothendieck's_Comparison_Theorem}}, let $\mathcal F,\mathcal G\in \Coh(\mathcal O_X)$ whose supports have an intersection which is 
	proper over $Y$. Then the natural map below is an isomorphism 
	\begin{align}
	\Ext^r(\mathcal F,\mathcal G)^{\wedge}\overset{\sim}{\longrightarrow}\Ext^r(\widehat{\mathcal F},\widehat{\mathcal G}).
	\end{align}
\end{corollary}
\noindent The proof is formally the same as Corollary 8.2.9 in \cite{FantechiIllusieGottsche2005}.

\begin{thr}[Grothendieck's Existence Theorem]\label{Grothendieck's_Existence_Theorem}
	Let $f:X\to Y$ be a separated finite type morphism between superschemes, and let $\widehat{X}$ be the $I$-adic completion of $X$. Then the functor $\mathcal F\mapsto \widehat{\mathcal F}$ where $$\widehat{\mathcal F}:=\lim_{\substack{\longleftarrow\\n}}\mathcal F/I^n\mathcal F,$$from the category of coherent sheaves on $X$ whose support is proper over Y to the category of coherent sheaves on $\widehat{X}$ whose support is proper over $Y$ is an equivalence. 
\end{thr}

\begin{proof}
	The fully faithfulness follows from Corollary \ref{Completion_Fully_Faithful}. Let us prove essential surjectivity. Again we regard everything over the bosonic adic Noetherian ring $A^+$ with $I^+$-topology. Suppose that $F\in \Coh(\mathcal O_{\widehat{X}})$ with proper support over $\widehat Y$, we want to show that there exists $\mathcal F\in \Coh(\mathcal O_{X})$ with proper support over $Y$ such that $F=\widehat{\mathcal F}$. Write $F=F_+\oplus F_-$ be decomposition into even and odd parts. By the classical Grothendieck's Existence Theorem, there exists coherent $\mathcal O_X^+$ sheaves $\mathcal F_+$ and $\mathcal F_-$ whose supports are proper over $Y$, such that $F_{\pm}=\widehat{\mathcal F}_{\pm}$. Now an $\mathcal O_X$-module structure on $\mathcal F_+\oplus\mathcal F_-$ is a pair of homomorphisms$$a:\mathcal O_X^-\otimes_{\mathcal O_X^+}\mathcal F_{\pm}\to \mathcal F_{\mp},$$such that the diagram
	\begin{center}
		\begin{tikzcd}
		\mathcal O_X^-\otimes_{\mathcal O_X^+}\mathcal O_X^-\otimes_{\mathcal O_X^+}\mathcal F_{\pm}\arrow[r,"\text{Id}\otimes a"] \arrow[dr,"\mu\otimes \text{Id}",swap] &\mathcal O_X^-\otimes_{\mathcal O_X^+}\mathcal F_{\mp}\arrow[d,"a"]\\
		&\mathcal F_{\pm}
		\end{tikzcd}
		,
	\end{center}
	commutes. Now every sheaf in this diagram has proper support, hence the existence of the homomorphisms $a$ and the comutativity of this diagram follows from the classical Grothendieck's Existence Theorem.
\end{proof}

\subsubsection{Super Analog of Castelnuovo-Mumford Regularity}
A coherent sheaf $\mathcal F$ on $\mathbb P^{n|s}$ is called $m$-regular if $$\H^i(\mathbb P^{n|s}, \mathcal F(m-i))=0,\qquad \forall i\ge 1.$$
An important trick we will use is that $\mathbb{P}^{n|s}$ is split, in fact $\mathbb{P}^{n|s}=\wedge ^{\bullet}\mathcal O_{\mathbb P^n}(-1)^{\oplus s}$, so coherent sheaf $\mathcal F$ on $\mathbb{P}^{n|s}$ can be regarded as a pair of coherent sheaves $(\mathcal F_+|\mathcal F_{-})$ on $\mathbb{P}^n$ with a homomorphism $a:\mathcal O_{\mathbb P^n}(-1)^{ s}\otimes _{\mathcal O_{\mathbb P^n}}\mathcal F_{\pm}\to \mathcal F_{\mp}$ such that the sum of two homomorphisms
\begin{equation}
\psi=a\circ(\text{Id}\otimes a)\text{ and }\phi=\psi\circ (\sigma\otimes \text{Id}):\mathcal F_{\pm}(-2)^{{s}\choose {2}}\to \mathcal F_{\pm},
\end{equation}
is zero, where 
\begin{equation}
\sigma\in \Aut(\mathcal O_{\mathbb P^n}(-1)^{s}\otimes _{\mathcal O_{\mathbb P^n}}\mathcal O_{\mathbb P^n}(-1)^{s}),
\end{equation}
is the swap of two components. Under this identification, there is an isomorphism 
\begin{equation*}
\H^i(\mathbb P^{n|s}, \mathcal F(r))\cong\H^i(\mathbb P^{n}, \mathcal F_+(r))\oplus \prod \H^i(\mathbb P^{n}, \mathcal F_{-}(r)).
\end{equation*}
Hence $\mathcal F$ being $m$-regular is equivalent to that $\mathcal F_{+}$ and $\mathcal F_{-}$ are both $m$-regular (in the classical sense). With the help of this trick, we can formulate the Castelnuovo-Mumford $m$-regularity theory on super projective space purely in terms of $m$-regularity theory on classical projective space, namely we have

\begin{lem}[Lemma 5.1 of \cite{FantechiIllusieGottsche2005}]
	If $\mathcal F$ is a $m$-regular sheaf on $\mathbb P^{n|s}$, then the following statements hold:
	\begin{enumerate}
		\item The canonical map $\H^0(\mathbb P^{n|s},\mathcal O_{\mathbb P^n}(1))\otimes \H^0(\mathbb P^{n|s},\mathcal F(r))\to \H^0(\mathbb P^{n|s},\mathcal F(r+1))$ is surjective whenever $r \ge m$. 
		\item We have $\H^i(\mathbb P^{n|s},\mathcal F(r)) = 0$ whenever $i > 1$ and $r > m — i$. In other words, if $\mathcal F$ is m-regular, then it is $m'$-regular for all $m'>m$. 
		\item The sheaf $\mathcal F(r)$ is generated by its global sections, and all its higher cohomologies vanish, whenever $r > m$. 
	\end{enumerate}
\end{lem}

\begin{thr}[Theorem 5.3 of \cite{FantechiIllusieGottsche2005}]\label{m-Regularity}
	For any non-negative integers $p$, $n$ and $s$, there exists a  polynomial $F_{p,n,s}$ in $n+1$ variables with integral coefficients, which has the following property:
	Let $k$ be any field, and let $\mathcal F$ be any coherent sheaf on $\mathbb P^{n|s}$, which is isomorphic to a subsheaf of $\mathcal O_{\mathbb P^{n|s}}^{\oplus p}$. Let the Hilbert polynomial of $\mathcal F$  be written in terms of binomial coefficients as
	\begin{equation*}
	\chi(\mathcal F(r))_+=\sum_{i=0}^n a_i {{n}\choose{i}}, \qquad \chi(\mathcal F(r))_-=\sum_{i=0}^n b_i {{n}\choose{i}},
	\end{equation*}
	where $a_0,\cdots,a_n\in \mathbb Z$, and $\chi(\mathcal F(r))_{+}$ $($reps. $\chi(\mathcal F(r))_{-}$ $)$ are even $($resp. odd$)$ Euler characteristics. Then $\mathcal F$ is $m$-regular, where $m=\max\{F_{p,n,s}(a_0,\cdots,a_n), F_{p,n,s}(b_0,\cdots,b_n)\}$.
\end{thr}

\subsubsection{Flattening Stratification}
Let's first recall the definition of flattening stratification. Suppose that $f:X\to S$ is a morphism of superschemes and $\mcal F\in \QCoh(\mcal O_X)$, a flattenning stratification of $\mcal F$ over $S$ (if exists) is a set of locally closed sub superschemes $\{S_i\}_{i\in I}$ of $S$, such that 
\begin{enumerate}
	\item As a set, $|S|=\coprod_{i\in I}|S_i|$.
	\item For any morphism of superschemes $g:T\to S$, $\mcal F_T$ is flat over $T$ if and only if there exists $i\in I$ and $g$ factors through $S_i\hookrightarrow S$.
\end{enumerate}
An elementary case is when $f$ is identity map so $\mcal F$ is a quasi-coherent sheaf on $S$. To study the existence of flattenning stratification, we start from the affine case. Let $R=R^+\oplus R^-$ be a superring (not necessarily Noetherian) and $M=\coker (f:R^{p|q}\to R^{m|n})$ be a finite presented module. We write $f$ in the matrix form
\begin{align*}
f=
\begin{bmatrix}
A       & B  \\
C       & D
\end{bmatrix}.
\end{align*}

\begin{proposition}\label{FlatStrat_Local}
	There exists flattening stratification for $M$ on $\Spec R$. 
\end{proposition}

\begin{proof}
	It is well-known that $M_{\bos}$ on $R_{\bos}=R/\mathfrak{n}$ has  flattening stratification 
	\begin{equation*}
	\coprod_{\substack {0\le k\le m\\0\le l\le n}}U'_{k|l}\to \Spec R_{\bos},
	\end{equation*}
	such that $M_{S_{\bos}}$ is flat of rank $(k|l)$ if an only if $\Spec S_{\bos}\to \Spec R_{\bos}$ factors through stratum $U'_{k|l}$. So if there is a flattening stratification for $M$ on $R$, then it must has the same underlying topological space as the $U'_{k|l}$'s. Take any point $p\in U'_{k|l}$, defined by prime ideal $\mathfrak{p}$. Passing to residue field $k(p)$, $f$ is similar to a matrix of the form
	\begin{equation*}
	\begin{pmat}[{.|.}]
	\text{Id}_{(m-k)\times (m-k)} & \mathbf{0} & \mathbf{0} & \mathbf{0} \cr
	\mathbf{0} & \mathbf{0} & \mathbf{0} & \mathbf{0} \cr\-
	\mathbf{0} & \mathbf{0} & \mathbf{0} & \mathbf{0} \cr
	\mathbf{0} & \mathbf{0} & \mathbf{0} & \text{Id}_{(n-l)\times (n-l)} \cr
	\end{pmat}.
	\end{equation*}
	Hence on $R_{\mathfrak{p}}$ (the localization of $R$ at $p$), $f$ is similar to a matrix of the form
	\begin{equation*}
	\begin{pmat}[{.|.}]
	\text{Id}_{(m-k)\times (m-k)} & \mathbf{0} & \mathbf{0} & \mathbf{0} \cr
	\mathbf{0} & A' & B' & \mathbf{0} \cr\-
	\mathbf{0} & C' & D' & \mathbf{0} \cr
	\mathbf{0} & \mathbf{0} & \mathbf{0} & \text{Id}_{(n-l)\times (n-l)} \cr
	\end{pmat},
	\end{equation*}
	where $A',B',C',D'$ are matrices with input from prime ideal $\mathfrak{p}$, by Gauss elimination algorithm. The same is true if $R_{\mathfrak{p}}$ is replaced by $R_a$ for some $a\in \mathfrak{p}_0$. Then $M_a$ is the cokernel of 
	\begin{align*}
	f'=
	\begin{bmatrix}
	A'       & B'  \\
	C'       & D'
	\end{bmatrix}:R_a^{k+p-m|l+q-n}\to R_a^{k|l}.
	\end{align*}
	Hence the pullback of $M_a$ to a $R_a$-superscheme $S$ is locally free of rank $(k|l)$ if and only if pullback of $A',B',C',D'$ are identically zero, i.e. $S\to \Spec R_a$ factors through the closed super subscheme defined by the ideal generated by entries of $A',B',C',D'$. Let $p$ runs through $U'_{k|l}$ and we get a open super subcheme $V\subset \Spec R$ and a closed super subscheme $U_{k|l}\subset V$ by gluing. This is the flattening stratum for $(k|l)$ piece.
\end{proof}
\begin{remark}
	Since $U_{k|l}$ has the same topological space as $U'_{k|l}$, by the lower semicontinuity of the rank of matrix in the bosonic case, we see that the closure of $U_{k|l}$ is contained in the union of all $U_{k'|l'}$ where $k'\ge k,l'\ge l$.
\end{remark}
By a gluing argument, Proposition \ref{FlatStrat_Local} generalizes to
\begin{corollary}\label{FlatStrat_Module}
	There exists flattening stratification for finite presented sheaves on superschemes.
\end{corollary}
A more general situation is when $f:X\to S$ is projective, and $\mcal F$ is finite presented, in this case the Hilbert polynomials are constant for any flat strata (if exists) by Theorem \ref{Proper_Flat_Pushforward}. Analogous to the classical case, we can add one more condition to the flattening stratification for projective morphisms:
\begin{enumerate}
	\item[•] The index set $I$ for the flattening stratification is the set of Hilbert polynomials $(f_+|f_-)$, and $I$ has a partial 
	ordering, defined by putting $(f_+|f_-)\le(g_+|g_-)$ whenever $f_+(n)\le g_+(n)$ and $f_-(n)\le g_-(n)$ for all $n\gg 0$. Then the 
	closure in $S$ of the subset $S_{(f_+|f_-)}$ is contained in the union of all $S_{(g_+|g_-)}$ where $(f_+|f_-)\le(g_+|g_-)$. 
	
\end{enumerate}

\begin{thr}[{Theorem 5.13 of \cite{FantechiIllusieGottsche2005}}]\label{FlatStrat_Projective}
	Let $\mathcal V$ be a locally-free sheaf on a superscheme $S$, and $\mathcal F$ be a finite presented sheaf on $\mathbb P(\mathcal V)$, then there exists flattening stratification for $\mathcal F$ over $S$.
\end{thr}
\begin{proof}
	We first assume that $S$ is affine. By Noetherian approximation, we can assume that $S$ is also Noetherian and $\mathcal F$ is coherent. If we can prove the proposition for affine Noetherian $S$, then the result automatically holds for general $S$ by functorial property of flattening stratification.
	
	\newl Let us explain how to modify the arguments in the proof of Theorem 5.13 of \textit{loc.cit.} so that it can be applied to its super analog: The special case corresponds to Corollary \ref{FlatStrat_Module}; Generic Flatness (\textit{loc.cit.} 5.12) is still true, since if the base is an integral scheme, then we can project $X$ down to its bosonic quotient $X_{\ev}$, and apply the generic flatness to $\mathcal F$ which is considered as a coherent sheaf on $X_{\ev}$; \textit{loc.cit.} 5.5 is replaced by its super analog (Lemma \ref{Flatness_Criterion}). Now the remaining step is to repeat the proof of \textit{loc.cit.} 5.13, which will be omitted here.
\end{proof}

\begin{corollary}[Incidence Superscheme]\label{Incidence_Superscheme}
	Let $S$ be a Noetherian superscheme and $\pi:X\to S$ is a projective morphism, $Y,Z\subset X$ are two closed sub superschemes, and $Y$ is flat over $S$. Then there is a closed sub superscheme $S'$ of $S$, such that for any morphism of superschemes $g:T\to S$, $Y_T$ is a closed sub superscheme of $Z_T$ if and only if $g$ factors through $S'$.
\end{corollary}

\begin{proof}
Let us first prove a weaker version: there is a set of disjoint locally closed sub superschemes $\{S_i\}$ of $S$, such that for any morphism of superschemes $g:T\to S$, $Y_T$ is a closed sub superscheme of $Z_T$ if and only if $g$ factors through some $S_i$.

\newl Let $g:T\to S$ be a morphism such that $Y_T$ is a closed sub superscheme of $Z_T$, then $(Y\times_X Z)_T=Y_T\times_{X_T} Z_T=Y_T$ is flat over $T$, so $g$ factors through one of the flattening strata $S'_i$ for the coherent sheaf $\mcal O_{Y\times_X Z}$. On the other hand, on the strata $S'_i$, $Y\times_X Z$ is flat over $S'_i$ and is a closed sub superscheme of $Y$, defined by an ideal $J_i$ which is flat over $S'_i$, then pull-back of $Y\times_X Z$ to $T$ is $Y_T$ exactly when the pull back of the ideal $J_i$ vanishes, i.e. $g$ factor through $S'_i\setminus \pi(\Supp J_i)$, the latter is open since $\Supp J_i$ and $\pi$ is proper. Now we take $S_i=S'_i\setminus \pi(\Supp J_i)$ and it satisfies the condition.

\newl Next we show that those $S_i$ are closed. By Valuative Criterion, it is enough to assume that $S=\Spec R$ where $R$ is a DVR\footnote{DVR stands for discrete valuation ring.}, and $(Y\times_X Z)_{\eta}=Y_{\eta}$ where $\eta$ is the generic point of $S$. Then the ideal of $\mcal O_Y$ defining the closed sub superscheme $Y\times_X Z$ is supported on the special fiber, hence it must be trivial because $Y$ is flat over $S$. Finally we take $S'=\cup_iS_i$ and it satisfies the condition.
\end{proof}

\subsubsection{The Quot Superscheme}
Let $S$ be a Noetherian superscheme, $\pi:X\to S$ is a quasi-projective morphism, $\mathcal L$ is relatively very ample line bundle of rank $(1|0)$ on $X$. Then for any $E\in \Coh(\mathcal O_X)$ and any pair of polynomials $\Phi(t),\Psi(t)\in \mathbb Q[t]$, the functor $\mathfrak{Quot}_{E/X/S}^{(\Phi|\Psi),\mathcal L}$ on the category of superschemes over $S$ is defined by 
\begin{align*}
(f:T\to S)\mapsto \{E_T\twoheadrightarrow \mcal F|\mcal F\text{ is flat over }T\text{ with Hilber polynomials }(\Phi|\Psi)\}.
\end{align*}

\begin{thr}
	The functor $\mathfrak{Quot}_{E/X/S}^{(\Phi|\Psi),\mathcal L}$ is represented by a quasi-projective $S$-superscheme $\Quot _{E/X/S}^{(\Phi|\Psi),\mathcal L}$. Moreover, if $\pi:X\to S$ is projective then $\Quot _{E/X/S}^{(\Phi|\Psi),\mathcal L}$ is projective.
\end{thr}

The proof will be carried out in several steps.

\begin{lem}\label{Open_Immersion}
	Let $X_1\hookrightarrow X_2$ be an open immersion between quasi-projective $S$-superschemes, $E\in \Coh(\mathcal O_{X_2})$, $\mathcal L$ is relatively very ample line bundle of rank $(1|0)$ on $X_2$. Then the natural transform $\mathfrak{Quot}_{E/X_1/S}^{(\Phi|\Psi),\mathcal L}\to\mathfrak{Quot}_{E/X_2/S}^{(\Phi|\Psi),\mathcal L}$ is an open immersion.
\end{lem}

\begin{proof}
	Given an $S$-superscheme $T$ and any point
	\begin{equation*}
	(\phi:E\to F)\in \mathfrak{Quot}_{E/X_2/S}^{(\Phi|\Psi),\mathcal L}(T).
	\end{equation*}
	The pullback of $F$ to $X_2\times _S U$ for any $T$-superscheme $U$ lies inside $X_1\times _S U$ if and only if the image of $U$ in $T$ does not intersect with $K=\pi\left(\Supp (F)\setminus X_1\right)$. Hence the $T$-value of $\mathfrak{Quot}_{E/X_1/S}^{(\Phi|\Psi),\mathcal L}\to\mathfrak{Quot}_{E/X_2/S}^{(\Phi|\Psi),\mathcal L}$ is represented by the complement of $K$ in $T$, which is open.
\end{proof}

\begin{lem}\label{Closed_Immersion}
	Let $\psi:E\to F$ be a surjective map between coherent sheaves on $X$. Then the natural transform $\mathfrak{Quot}_{F/X/S}^{(\Phi|\Psi),\mathcal L}\to\mathfrak{Quot}_{E/X/S}^{(\Phi|\Psi),\mathcal L}$ is a closed immersion.
\end{lem}

\begin{proof}
	Given an $S$-superscheme $T$ and any point
	\begin{equation*}
	(\phi:E\to G)\in \mathfrak{Quot}_{E/X/S}^{(\Phi|\Psi),\mathcal L}(T).
	\end{equation*}
	Let $K=\ker\psi$ and $\eta:K\to G$ be the composition of inclusion of $K$ in $E$ and projection $E\to G$. The pullback of $\phi$ to $X\times _S U$ for any $T$-superscheme $U$ factors through $F$ if and only if the pullback of $\eta$ to $X\times _S U$ is zero. Corollary \ref{Hom_Superscheme} says that $\eta$ gives rise to a section $\iota$ of a linear superscheme $\mathbb V(\mathcal Q)$ on $T$, and pullback to $X\times _S U$ is zero if and only if $U\to \mathbb V(\mathcal Q)$ lies in the zero section $\mathbf{0}$, i.e. $U\to T$ lies in the closed sub superscheme $\iota^{-1}(\mathbf{0})$. Hence the $T$-value of $\mathfrak{Quot}_{F/X/S}^{(\Phi|\Psi),\mathcal L}\to\mathfrak{Quot}_{E/X/S}^{(\Phi|\Psi),\mathcal L}$ is represented by $\iota^{-1}(\mathbf{0})$. 
\end{proof}

Representability of a Zariski functor can be proven Zariski locally on $S$, so we can assume that $S$ is affine and there is an open embedding $X\hookrightarrow Y$ followed by closed embedding $Y\hookrightarrow \mathbb P^{n|m}$, such that $\mathcal L^{\otimes a}$ is the pullback of $\mathcal O(1)$ on $\mathbb P^{n|m}$. Lemma \ref{Open_Immersion} implies the existence of prolongation of $E$ to a coherent sheaf $F$ on $Y$, and Lemma \ref{Closed_Immersion} implies that we can replace $F$ by $\pi^*(\pi_*F(N))(-N)$ for large enough $N$, in fact $\pi_*F(N)$ can be replaced by a free sheaf $\mathcal O_S^{p|q}$ which surjects onto it. This step reduces the original question to showing the representability of 
\begin{equation*}
\mathfrak{Quot}_{\mathcal O^{p|q}/\mathbb P_S^{n|m}/S}^{(\widetilde{\Phi}|\widetilde{\Psi}),\mathcal O(1)},
\end{equation*}
where $\widetilde{\Phi}(t)=\Phi(a(t+N))$ and $\widetilde{\Psi}(t)=\Psi(a(t+N))$. For simplicity, we keep using $\Phi$ and $\Psi$ instead of $\widetilde{\Phi}$ or $\widetilde{\Psi}$.

\newl By $m$-regularity (Theorem \ref{m-Regularity}), there exists $m$ such that for any geometric point $s$ of $S$, and any quotient $\phi:\mathcal O^{p|q}\to \mathcal F$ on $\mathbb P_s^{n|m}$, the sheaves $\mathcal O^{p|q}(r)$, $\mathcal F(r)$, $\mathcal G(r)$ (where $\mathcal G(r)=\ker\phi(r)$) are generated by their global sections and have trivial $\H^i$ for $i>0$, whenever $r>m$. Then for any $S$-superscheme $T$ and 
\begin{equation*}
\left(\phi:\mathcal O^{p|q}\to \mathcal F\right)\in \mathfrak{Quot}_{\mathcal O^{p|q}/\mathbb P_S^{n|m}/S}^{(\Phi|\Psi),\mathcal O(1)}(T),
\end{equation*}
there is a short exact sequence 
\begin{equation*}
0\to \pi_*\mathcal G(r)\to \pi_*\mathcal O^{p|q}(r)\to\pi_*\mathcal F(r)\to 0,
\end{equation*}
which gives rise to a point in $\Gras(\Phi(r)|\Psi(r),\pi_*\mathcal O^{p|q}(r))$.

\begin{lem}
	The natural transform $i:\mathfrak{Quot}_{\mathcal O^{p|q}/\mathbb P_S^{n|m}/S}^{(\Phi|\Psi),\mathcal O(1)}\to \Gras(\Phi(r)|\Psi(r),\pi_*\mathcal O^{p|q}(r))$ is represented by locally closed sub superscheme of $\Gras(\Phi(r)|\Psi(r),\pi_*\mathcal O^{p|q}(r))$.
\end{lem}

\begin{proof}
	Let the kernel of the canonical map $\pi^*\pi_*\mathcal O^{p|q}(r)\to \mathcal O^{p|q}(r)$ be $\mathcal K$, then the image of $\mathcal K$ in $\pi^*\pi_*\mathcal F(r)$ is the kernel of canonical map $\pi^*\pi_*\mathcal F(r)\to \mathcal F(r)$, hence $\mathcal F(r)$ is uniquely determined by the quotient $\pi_*\mathcal O^{p|q}(r)\to\pi_*\mathcal F(r)$. In other words, $i$ is a monomorphism.
	
	\newl Let $\mathcal Q$ be the universal quotient of $\pi_*\mathcal O^{p|q}(r)$ of rank $(\Phi(r)|\Psi(r))$, let $\mathcal N$ be the cokernel of the canonical map $\mathcal K\to\pi^*\mathcal Q$, then the formation of $\mathcal N$ commutes with base change and comes from a $\mathcal F(r)$ if and only if it is flat over the base and has characteristic polynomial $(\Phi(t+r)|\Psi(t+r))$. It follows that the lemma is a consequence of the flattening stratification (Theorem \ref{FlatStrat_Projective}).
\end{proof}

Finally it remains to show that if $\pi$ is projective then $\Quot _{E/X/S}^{(\Phi|\Psi),\mathcal L}$ is proper over $S$. This follows from the lemma below.

\begin{lem}[Valuative Criterion]
	Let $A$ be a discrete valuation ring, with fractional field $K$. Let $X$ be a proper superscheme on $\Spec A$ and $E\in \Coh(\mathcal O_X)$. If $E_K\to F$ is a quotient of $E_K$ on the generic fiber $X_K$, then there exists a quotient $E\to G$ on $X$ such that $G_K=F$ and $G$ is flat over $A$.
\end{lem}

\begin{proof}
	Take $G$ to be the image of $E$ in $j_*F$ under the canonical map $E\to j_*F$ where $j:X_K\to X$ is the open immersion. Clearly $G$ restricts to $F$. $G$ is torsion free as an $A$-module since it is a submodule of a torsion free $A$-module $j_*F$, so $G$ is flat over $A$.
\end{proof}

The followings are corollaries to the existence of Quot superschemes, which are analogous to the classical case, and the proofs are formally identical to those in the \textit{loc.cit.}, which will be omitted here.
\begin{thr}[Theorem 5.23 of \cite{FantechiIllusieGottsche2005}]\label{Superscheme_of_Morphisms}
	Let $S$ be a Noetherian superscheme, let $X$ be a projective superscheme over $S$, and let $Y$ be quasi-projective superscheme over $S$. Assume moreover that X is flat over $S$. Then the functor $$\mathfrak{Map}_S(X,Y):(f:T\to S)\mapsto \Map_T(X_T,Y_T),$$ is represented by an open sub superscheme $\Map_S(X,Y)$ of $\Hilb_{X\times _S Y/S}$.
\end{thr}

\begin{thr}[Theorem 5.25 of \cite{FantechiIllusieGottsche2005}]\label{Superscheme_of_Schematic_Quotient}
	Let $S$ be a Noetherian superscheme, and let $X \to S$ be a quasi-projective morphism. Let $f: R\to X\times_S X$ be schematic equivalence relation on $X$ over $S$, such that the projections $f_1,f_2 :R\rightrightarrows X$ are proper and flat. Then a  schematic quotient $X\to Q$ exists over $S$. Moreover, $Q$ is quasi-projective over $S$, the morphism $X\to Q$ is faithfully flat and projective, and the induced morphism 
	$(f_1,f_2):R\to X\times _Q X$ is an isomorphism.
\end{thr}

\subsubsection{Deformation Theory and Super Analog of Schlessinger's Theorem}
Let $k$ be a fixed algebraically-closed ground field, and let the category of local Artinian $k$-superalgebras be $\sArt_{k}$, and let the category of local complete Noetherian $k$-superalgebras be $\sLoc_{k}$. We consider a covariant functor $F:\sArt_{k}\to \Sets$. One such example is taking an $R\in \sLoc_{k}$ and letting $h_R:A\to \Hom_{\sLoc_{k}} (R,A)$. By Yoneda Lemma, a natural transform $\varphi:h_R\to F$ is given by a compatible family $\{\xi_n\in F(R/\mathfrak{m}_R^n)\}$.
\begin{definition}\label{Definition_Versality}
	Suppose that $F:\sArt_{k}\to \Sets$ is a covariant functor, $R\in \sLoc_{k}$ and $\varphi:h_R\to F$ is a natural transform, represented by $\xi:=\{\xi_n\in F(R/\mathfrak{m}_R^n)\}$. Then 
	\begin{enumerate}
		\item[•] $(R,\xi)$ is called a \textbf{versal} family if $\varphi:h_R\to F$ is strongly surjective, i.e. $\forall A\in \sLoc_{k}$, $h_R(A)\to F(A)$ is surjective, and for every surjection $B\to A$ in $\sLoc_{k}$, $h_R(B)\to h_R(A)\times_{F(A)}F(B)$ is surjective.
		\item[•] $(R,\xi)$ is called a \textbf{miniversal} family if it's versal, and maps $$h_R(k[\theta])\to F(k[\theta])\text{ , }h_R(k[t]/(t^2))\to F(k[\theta]/(t^2)),$$are bijective.
		\item[•] $(R,\xi)$ is called a \textbf{universal} family if $\varphi$ is an isomorphism.
	\end{enumerate}
\end{definition}
Note that analogous to the classical case, if the the functor $F:\sArt_{k}\to \Sets$ satisfies that $F(k)$ has just one element and 
\begin{equation}
F(A\times_k k[\theta])\to F(A)\times F(k[\theta])\text{ , }F(A\times_k k[t]/(t^2))\to F(A)\times F(k[t]/(t^2)),
\end{equation}
are bijective for all $A\in \sArt_k$, then $F(k[t,\theta]/(t^2,t\theta))$ has a natural structure of $k$-super vector space, which shall be denoted by $T_F=T_F^+\oplus T_F^-$, where $T_F^+=F(k[t]/(t^2)), T_F^-=F(k[\theta])$.

\begin{thr}[Schlessinger's Theorem]\label{Schlessinger's Theorem}
	The functor $F:\sArt_{k}\to \Sets$ has a miniversal family if and only if
	\begin{enumerate}
		\item[$(S_0)$] $F(k)$ has just one element.
		\item[$(S_1)$] $F(A'\times_A A'')\to F(A')\times_{F(A)}F(A'')$ is surjective for every small \footnote{An extension $0\to I\to A''\to A\to 0$ is called small if $I$ is one dimensional.} extension $A''\to A$.
		\item[$(S_2)$] The map in $(S_1)$ is bijective when $A=k$ and $A''=k[\theta]$ or $k[t]/(t^2)$.
		\item[$(S_3)$] $T_F$ is a finite-dimensional $k$-super vector space.
	\end{enumerate}
	Moreover, $F$ has a universal family if and only if in addition
	\begin{enumerate}
		\item[$(S_4)$] For every small extension $p:A''\to A$ with even kernel (resp. odd kernel) and for every $\eta\in F(A)$ which $p^{-1}(\eta)$ is nonempty, the group action of $T_F^+$ (resp. $T_F^-$) on $p^{-1}(\eta)$ is bijective.
	\end{enumerate}
\end{thr}
The proof is almost the same as in the classical case (e.g. Theorem 16.2 of \cite{Hartshorne2010}), except that all ideals in the original proof should be $\mathbb Z_2$-graded in the super case, and the proof in the \textit{loc.cit.} works formally.

\subsection{Stable $\mscr{N}=1$ SUSY Curves}
Assume that a triple of natural numbers $(\g,\ns,\ra)$ satisfies $\ra\in 2\mathbb Z_{\ge 0}$, and $$2\g-2+\ns+\ra>0.$$
\begin{definition}[{\bf Stable SUSY Curves}]\label{def:stable SUSY curves}
	Given a superscheme $S$, a family of stable SUSY curves of genus $\g$ with $\ns$ {\normalfont NS} punctures and $\ra$ {\normalfont R} punctures over $S$, is a finite presented proper flat relatively Cohen-Macaulay $S$-superscheme $X$, such that the smooth open locus $X^{\text{\normalfont sm}}\subset X$ is of relative dimension $(1|1)$ and is dense on each geometric fiber, together with
	\begin{enumerate}
		\item Sections $E=E_1\sqcup\cdots\sqcup E_{\sns}$ inside $X^{\text{\normalfont sm}}$. $E_i$ will be called the $i$'th {\normalfont NS} punctures;
		\item A closed sub superscheme $\mscr D=\mscr D_1\sqcup \cdots\sqcup \mscr D_{\sra}$ inside $X^{\text{\normalfont sm}}$, proper and flat over $S$, of codimension $(1|0)$, and $\mscr D_{i,s}$ is a point for every geometric point $s\in S$. $\mscr D_i$ will be called the $i$'th {\normalfont R} puncture;
		\item An $\mathcal O_S$-linear derivation 
		\begin{equation}
		\delta:\mathcal O_X\to \omega_{X/S}(\mscr D),
		\end{equation}
		or equivalently a homomorphism $\delta:\Omega^1_{X/S}\to \omega_{X/S}(\mscr D)$. $\delta$ is sometimes called the supersymmetric structure, or the SUSY structure;
	\end{enumerate}
	such that
	\begin{enumerate}
		\item[A.] For every geometric fiber $X_s$, $(X_s,\mcal O_{X_s}^+, E_s, \mscr D_s)$ is a punctured stable curve of genus $\g$, and $\delta^-$ induces isomorphism of $\mcal O_{X_s}^+$ modules$$\delta^-:\mathcal O_{X_s}^-\cong \omega_{X_s/s}^-(\mscr D_s).$$Consequently the homomorphism $\delta:\Omega^1_{X/S}\to \omega_{X/S}(\mscr D)$ is surjective on the smooth locus of $X$.
		\item[B.] Restricted on $X^{\text{\normalfont sm}}$, the dual of the surjective homomorphism $\delta:\Omega^1_{X/S}\to \omega_{X/S}(\mscr D)$, write 
		\begin{equation*}
		0\to D\to T_{X/S}\to T_{X/S}/D\to 0,
		\end{equation*}
		where $D=\omega^{\vee}_{X/S}(-\mscr D)$, induces an injective homomorphism $D^{\otimes 2}\to T_{X/S}/D$ via $D_1\otimes D_2\mapsto \frac{1}{2}[D_1,D_2]$ and the image is $ T_{X/S}/D(-\mscr D)$.
	\end{enumerate}
\end{definition}

\begin{definition}\label{def:the definition of isomorphism of super-Riemann surfaces}
	An isomorphism between stable SUSY curves $(X,E,\mscr D,\delta)$ and $(X',E',\mscr D',\delta')$ over $S$ is an $S$-isomorphism of superschemes $f:X\to X'$ such that $f^{-1}E'=E, f^{-1}\mscr D'=\mscr D$, and $f^*\delta'=\delta$.
\end{definition}

\begin{definition}\label{def:the definition of the stack of super-Riemann surfaces}
	The category of stable SUSY curves is the fibered groupoid over the category of superschemes, where the fiber over a superscheme $S$ is the groupoid of stable SUSY curve of genus $\g$ with $\ns$ NS punctures and $\ra$ R punctures. Denote this category by $\overline{\mscr{SM}}_{\sg,\sns,\sra}$.
\end{definition}

The following lemma is a direct consequence of the definition of stable SUSY curves:
\begin{lem}\label{Lemma_Limit_Stable_SUSY_Curves}
	Let $(I,\le)$ be a directed set, suppose that $\{A_{i}\}_{i\in I}$ is an directed system of superrings, with inductive limit $A$, then the natural functor 
	\begin{equation*}
	\lim_{\substack{\longrightarrow\\i}}\overline{\mscr{SM}}_{\sg,\sns,\sra}(\Spec A_i)\to\overline{\mscr{SM}}_{\sg,\sns,\sra}(\Spec A),
	\end{equation*}
	is an equivalence of categories.
\end{lem}

\begin{proof}
	{\bf\small Fully faithful:} first forgetting the derivation, then the natural map 
	\begin{equation*}
	\lim_{\substack{\longrightarrow\\i}}\Map_{\Spec A_i}(X_i,Y_i)\to\Map_{\Spec A}(X,Y),
	\end{equation*}
	is an equivalence of sets, by Lemma \ref{Lemma_Limit_of_Morphisms}. Now a morphism $f$ preserves the derivation if and only if $f^*\delta_Y-\delta_X:\Omega^1_{X/S}\to \omega_{X/S}(\mscr D)$ is zero, then the result follows from Lemma \ref{Lemma_Limit_of_Map_Between_Sheaves}.
	
	\newl {\bf\small Essentially surjective:} the functor is essentially surjective if we forget the conditions A and B in the Definition \ref{def:stable SUSY curves}, this follows from the Lemmas \ref{Lemma_Limit_of_Morphisms}, \ref{Lemma_Limit_of_Superschemes}, \ref{Lemma_Limit_of_Map_Between_Sheaves}, and Proposition \ref{Limit_Flatness}. Note that $X\to S$ is proper if and only if $X_{\text{bos}}\to S_{\text{bos}}$ is proper, so we can apply Theorem 8.10.5 of \cite{GrothendieckEGAIV1967} to ensure that there exists $i\in I$ such that $X_i\to S_i$ is proper.
	
	\newl Now if the derivation $\delta$ on $X$ satisfies condition A, then the derivation $\delta_i$ on $X_i$ automatically satisfies condition A, since a homomorphism of sheaves on a variety over a field is an isomorphism if and only if it is an isomorphism under some field extension. If the derivation $\delta$ on $X$ satisfies condition B, then it follows that the natural map $D_i^{\otimes 2}\to T_{X_i/S_i}/D_i$ is injective on $X^{\text{\normalfont sm}}_i$ since it is a map between locally free $\mcal O_{X_i}$ sheaves and is injective on geometric fibers $X^{\text{\normalfont sm}}_{i,s}$ for all geometric point $s$. Moreover after increasing $i$, the composition $D_i^{\otimes 2}\to T_{X_i/S_i}/D_i\to T_{X_i/S_i}/D_i\otimes \mcal O_{\mscr D_i}$ becomes zero, since it is zero after taking limit, so the image of  $D_i^{\otimes 2}$ in $T_{X_i/S_i}/D_i$ is a subsheaf of $T_{X_i/S_i}/D_i(-\mscr D_i)$, thus it must be $T_{X_i/S_i}/D_i(-\mscr D_i)$ since the induced map $D_i^{\otimes 2}\to T_{X_i/S_i}/D_i(-\mscr D_i)$ is an isomorphism on geometric fibers $X^{\text{\normalfont sm}}_{i,s}$. This concluds that $\delta_i$ satisfies condition B.
\end{proof}

By usual argument (see e.g. Lemma $3.1$ of \cite{DonagiWitten2013a}), the condition B in Definition \ref{def:stable SUSY curves} tells that away from an R puncture, locally there is a coordinate system $(z|\theta)$, such that $\delta$ can be written as 
\begin{equation*}
\delta:f\mapsto (\partial_{\theta}f+\theta\partial_zf)[\mathrm{d}z|\mathrm{d} \theta],
\end{equation*}
and nearby a R puncture there is a local coordinate system $(z|\theta)$, such that $z=0$ is the location of R puncture and $\delta$ can be written as 
\begin{equation*}
\delta:f\mapsto (\partial_{\theta}f+\theta z\partial_zf)\left[\frac{\mathrm{d}z}{z}\bigg{\vert}\mathrm{d} \theta\right].
\end{equation*}
On the other hand, suppose that $S$ is a bosonic scheme, then condition A in the Definition \ref{def:stable SUSY curves} tells that $(X,\mcal O_{X}^+, E, \mscr D)$ is a pointed stable curve of genus $\g$. The Cohen-Macaulay condition implies that $\mcal O_{X}^-$ is a relatively torsion free $\mcal O_{X}^+$ sheaf, and it is of rank one because it is a line bundle on $U$. This implies that 
\begin{equation}
\mcal O_{X}^-\cdot\mcal O_{X}^-=0,
\end{equation}
since it is generically zero and $\mcal O_{X}^+$ is relatively Cohen-Macaulay. Moreover the odd part of $\delta$ induces an isomorphism $$\delta^-:\mathcal O_{X}^-\cong \omega_{X/S}^-(\mscr D),$$ since $\delta^-_s$ is an isomorphism on each fiber $X_s$. Note that the relative dualizing complex (in fact a sheaf) is $$\omega_{X/S}=R\underline{\Hom}_{\mathcal O_X^+}(\mathcal O_X,\omega_{X_{\text{\normalfont bos}}/S})=\underline{\Hom}_{\mathcal O_X^+}(\mathcal O_X,\omega_{X_{\text{\normalfont bos}}/S}),$$ and it decomposes as $\mathcal O_X^+$-module into
\begin{equation*}
\omega_{X/S}^+=\omega_{X_{\text{\normalfont bos}}/S}\text{ , }\omega_{X/S}^-\cong\underline{\Hom}_{\mathcal O_X^+}(\mathcal O_X^-,\omega_{X_{\text{\normalfont bos}}/S}).
\end{equation*}
This realizes $(\mathcal O_X^-,\delta^-)$ as a spin structure on the punctured stable curve $(X,\mcal O_{X}^+, E, \mscr D)$. Finally, the condition B Definition \ref{def:stable SUSY curves} uniquely determines the even part of $\delta$, in fact $\delta ^+: \mathcal O_X^+\to \omega_{X/S}^+(\mscr D)$ is the composition of exterior differential $d:\mathcal O_X^+\to \omega_{X/S}^+$ with natural embedding $\omega_{X/S}^+\hookrightarrow \omega_{X/S}^+(\mscr D)$, because this is the case on the open $S$-dense locus $U$ (see the coordinate presentation of $\delta$ above), and $\mcal O_X\to j_*\mcal O_U$ is injective where $j:U\hookrightarrow X$ is the open embedding. As a summary,
\begin{proposition}
	Given a bosonic scheme $S$, a family of stable $\mscr N=1$ SUSY curves of genus $\g$ with $\ns$ {\normalfont NS} punctures and $\ra$ {\normalfont R} punctures over $S$, is a superscheme $(X,\mathcal O_X)$, flat over $S$, such that $\mathcal O_X=\mathcal O_X^+\oplus \mathcal O_X^-$, $\mathcal O_X^-$ has zero multiplication with itself, $X_{\text{\normalfont bos}}:=(X,\mathcal O_X^+)$ is a relative stable curve with puncturing divisors $E\sqcup \mscr D$, with $\deg E_s=\ns,\deg\mscr D_s=\ra$, and $\mathcal O_X^-$ is a rank one relatively torsion free sheaf over $X_{\text{\normalfont bos}}$ with an isomorphism
	\begin{equation}
	a:\mathcal O_X^-\to \underline{\Hom}_{\mathcal O_X^+}(\mathcal O_X^-,\omega _{X_{\text{\normalfont bos}}/S}(\mscr{D})).
	\end{equation}
\end{proposition}

\begin{remark}\label{Remark_Superspace_BosTruncationOfSRS}
	From the proposition, it is straightforward to see that the data of an stable $\mscr{N}=1$ punctured SUSY curve over a scheme $S$ is equivalent to the data of a relative punctured spin curve over $S$ {\normalfont(Definition \ref{Defn_Marked_RelSpinCurve})}.
\end{remark}
Recall that on a stable spin curve $(\mcal C,E,\mcal D,\mscr E)$ over a bosonic base scheme $S$, the line bundle $\mcal L=\omega_{\mcal C/S}(E+\mcal D)$ is relatively ample and $\mcal L^{\otimes 3}$ is relatively very ample. Consider a stable SUSY curve $X$ over $S$ with underlying spin curve $(\mcal C,E,\mcal D,\mscr E)$, since $p:X\to \mcal C$ is finite, we see that $p^*\mcal L$ is a relatively ample line bundle of rank $(1|0)$ on $X$. Consequently, we have

\begin{lem}\label{Stable_SUSY_Curve_Is_Locally_Projective}
	Let $S$ be a Noetherian affine superscheme and $X\to S$ is a stable SUSY curve, then $X$ is projective over $S$.
\end{lem}
\begin{proof}
	Suppose that $\mathfrak{n}_S^{N+1}=0$ and the restriction of $X$ on $(S,\mcal O_S/\mathfrak{n}^N_S)$ has an ample line bundle $\mcal L$, we first lift $\mcal L$ to $X$: This can be done locally on an affine covering $U_i$, since $\mcal L$ is locally isomorphic to the structure sheaf; from local to global we need to lift the transition function $\varphi_{ij}$ from $\Gamma(U_{ij},\mcal O_X/\mathfrak{n}^N_S\mcal O_X)_+$ to $\Gamma(U_{ij},\mcal O_X)_+$, we make an arbitrary choice $\widetilde{\varphi}_{ij}$; then 
	\begin{equation*}
	\widetilde{\varphi}_{ij}\circ \widetilde{\varphi}_{jk}\circ \widetilde{\varphi}_{ki}-1\in \Gamma(U_{ijk},\mathfrak{n}^N_S\mcal O_X)_+,
	\end{equation*}
	defines a \v Cech cocyle valued in the coherent sheaf $\mathfrak{n}^N_S\mcal O_X$, but $S$ is affine and the dimensions of fibers $X_s$ are $1$ hence $\H^2$ of a coherent sheaf vanishes; so we can modify $\widetilde{\varphi}_{ij}$ by an element in $1+\Gamma(U_{ij},\mathfrak{n}^N_S\mcal O_X)_+$ to force it satisfying cocycle conditions, hence local liftings glue to a global one. Now the lift $\widetilde {\mcal L}$ is ample since $X\to S$ is proper and $\mcal L_s$ is ample for all geometric point $s\in S$.
\end{proof}

\subsubsection{\' Etale Descent for $\overline{\mscr{SM}}_{\sg,\sns,\sra}$}
Suppose that $B$ is a Noetherian superring, $S=\Spec B$, $B_1$ is an \' etale $B$ superalgebra, let $B_2=B_1\otimes_B B_1$, $B_3=B_1\otimes_B B_1\otimes_B B_1$, $S_i=\Spec B_i$. Let $p_1,p_2$ be two canonical projections from $S_2$ to $S_1$. 

\begin{lem}\label{Lemma_Descent_for_Homomorphisms}
	Suppose that $M,M'$ are $B$-modules, let $M_i=M\otimes _B B_i$, $M'_i=M'\otimes _B B_i$, then the sequence
	\begin{equation}
	0\to \Hom_B(M,M')\to \Hom_{B_1}(M_1,M'_1)\rightrightarrows\Hom_{B_1}(M_2,M'_2),
	\end{equation}
	is exact.
\end{lem}

\begin{proof}
	Let $B=B^+\oplus B^-$m where $B^+$ and $B^-$ are its even and odd parts. Considering similar decomposition for $B_i$, then $B^+\to B^+_1$ is \' etale, $B_2^+=B_1^+\otimes _{B^+}B_1^+$, and $B_i^-=B_i^+\otimes _{B^+}B^-$, by Lemma \ref{Lemma_Equiv_Defn_of_Etale}. If a $B$-homomorphism $\varphi:M\to M'$ is zero after base change to $B_1$, then it is zero as a $B^+$-homomorphism, by classical descent theory, thus $\varphi=0$. If there is a $B_1$-homomorphism $\psi:M_1\to M'_1$ whose base changes to $B_2$ in two ways agree, then there exists a $B^+$-homomorphism $\rho:M\to M'$ whose base change to $B_1^+$ is $\psi$ (as $B_1^+$-homomorphism). Note that a $B^+$-homomorphism is a $B$-homomorphism if it commutes with the $B^-$ action on $M$ and $M'$. Now the base change of $\rho$ to $B_1^+$ commutes with the action of $B_1^-$, so $\rho$ commutes with the $B^-$ action, by classical descent thery. Hence we get a $B$-homomophism $\rho:M\to M'$ whose base change to $B_1$ is $\psi$.
\end{proof}

\begin{lem}\label{Lemma_Descent_for_Modules}
	Suppose that $N$ is a $B_1$ module with an ismorphism 
	\begin{equation*}
	\phi:p_1^*N\overset {\sim}{\longrightarrow} p_2^*N,
	\end{equation*}
	satisfying cocycle condition. Then there is a $B$ module $M$ with isomorphism $\psi:M\otimes_B B_1\overset {\sim}{\rightarrow} N$ such that $\phi$ is the canonical isomorphism $p_2^*\psi\circ p_1^*\psi^{-1}:p_1^*N\overset {\sim}{\rightarrow}p_2^* N$. Moreover, $N$ is a flat (resp. finite type) $B_1$-module if and only if $M$ is a flat (resp. finite type) $B$-module.
\end{lem}

\begin{proof}
	Let $B=B^+\oplus B^-$ where $B^+$ and $B^-$ are its even and odd parts, similar for $B_i$, then $B^+\to B^+_1$ is \' etale, $B_2^+=B_1^+\otimes _{B^+}B_1^+$, and $B_i^-=B_i^+\otimes _{B^+}B^-$, by Lemma \ref{Lemma_Equiv_Defn_of_Etale}. Then it follows that $p_i^*N\cong p_{i,\text{ev}}^*N$ when $N$ is considered as a $\mathbb Z_2$-graded $B_1^+$-module, and $p_{i,\text{ev}}:S_{2,\text{ev}}\to S_{1,\text{ev}}$ are two natural projections. Moreover, the isomorphism $\phi$ preserves $\mathbb Z_2$-grading. By classical descent theory, we have a $B^+$-module $M=M^+\oplus M^-$ with isomorphism $\psi:M^{\pm}\otimes_{B^+} B_1^+\overset {\sim}{\rightarrow} N^{\pm}$ such that $\phi$ is the canonical isomorphism $p_{2,\text{ev}}^*\psi\circ p_{1,\text{ev}}^*\psi^{-1}:p_{1,\text{ev}}^*N^{\pm}\overset {\sim}{\rightarrow}p_{2,\text{ev}}^* N^{\pm}$. Note that $\mathbb Z_2$-graded $B^+$ module $M=M^+\oplus M^-$ has a $B$-module structure if and only if there is a pair of $B^+$-homomorphism $a:B^-\otimes_{B^+}M^{\pm}\to M^{\mp}$ such that the diagram
	\begin{center}
		\begin{tikzcd}
		B^-\otimes_{B^+}B^-\otimes_{B^+}M^{\pm}\arrow[r,"\text{Id}\otimes a"] \arrow[dr,"\mu\otimes \text{Id}",swap] &B^-\otimes_{B^+}M^{\mp}\arrow[d,"a"]\\
		&M^{\pm}
		\end{tikzcd}
		,
	\end{center}
	commutes. Now the existence of homomorphism $a$ comes from the classical descent theory, and the commutativity of the diagram follows as well. The statement of flatness comes from the fact that $B\to B_1$ is faithfully flat, and statement of finiteness comes from reducing $B$ to $B/\mathfrak{n}_B$ and use the classcal descent theory.
\end{proof}

\begin{remark}
	In fact we can relax the condition $B\to B_1$ being \' etale to that $B\to B_1$ is flat and $\mathfrak{n}_B$ generates $\mathfrak{n}_{B_1}$, then the trick of decomposing $B=B^+\oplus B^-$ still applies and {\normalfont Lemmas \ref{Lemma_Descent_for_Homomorphisms} and \ref{Lemma_Descent_for_Modules}} is still true under this condition.
\end{remark}

Consider a stable SUSY curve $X_1$ over $S_1$, and an isomorphism 
\begin{equation*}
\varphi: X_1\times_{S_1,p_1}S_2\overset {\sim}{\longrightarrow} X_1\times_{S_1,p_2}S_2,
\end{equation*}
satisfying the cocycle condition. 

\begin{lem}\label{Lemma_Open_Affine_Cover_for_Descent}
	There exists an open affine cover $\{U_{\alpha}\to X_1\}$, such that $\varphi$ maps $U_{\alpha}\times_{S_1,p_1}S_2$ to $U_{\alpha}\times_{S_1,p_2}S_2$.
\end{lem}

\begin{proof}
	According to the affineness criterion Proposition \ref{Chevalley's Criterion}, it suffices to prove this lemma under the additional assumption that $B$ is purely bosonic\footnote{Then it follows from the definiton of \' etale morphism that $B_i$ are also bosonic.}, then $X_1$ is a spin curve $(\mcal C_1,E,\mscr D,\mscr E)$ together with the derivation induced from the structure map of the spinor sheaf $\mscr E$. By descent theory for classical schemes, we get a stable SUSY curve $X$ on $S$ with isomrphism $\psi:X\times _S S_1\to X_1$ such that $\varphi$ is the canonical isomorphism $p_2^*\psi\circ p_1^*\psi^{-1}:X_1\times_{S_1,p_1}S_2\to X_1\times_{S_1,p_2}S_2$. Take an open affine cover $\{V_{\alpha}\}$ of $X$, then the base change of $V_{\alpha}$ is what we want.
\end{proof}

\begin{proposition}\label{Etale_Descent_Stable_SUSY_Curve}
	There exists a stable SUSY curve $X$ over $S$ with an isomorphism $\psi:X\times _S S_1\to X_1$ such that $\varphi$ is the canonical isomorphism $p_2^*\psi\circ p_1^*\psi^{-1}:X_1\times_{S_1,p_1}S_2\overset {\sim}{\rightarrow}  X_1\times_{S_1,p_2}S_2$.
\end{proposition}

\begin{proof}
	We take the cover $\{U_{\alpha}\to X_1\}$ in the Lemma \ref{Lemma_Open_Affine_Cover_for_Descent}, then $\mcal O_X(U_{\alpha})$ is a flat $B_1$-superalgebra with a gluing isomorphism when base changed to $B_2$. By Lemmas \ref{Lemma_Descent_for_Homomorphisms} and \ref{Lemma_Descent_for_Modules}, there exists a flat and finite type $B$-superalgebra $C_{\alpha}$ with isomorphism $\psi_{\alpha}:C_{\alpha}\otimes _B B_1\overset {\sim}{\rightarrow} \mcal O_X(U_{\alpha})$ such that $\varphi|_{U_{\alpha}\times_{S_1,p_1}S_2}$ is $p_{2,\alpha}^*\psi\circ p_{1,\alpha}^*\psi^{-1}$. Moreover $C_{\alpha}$ is relatively Cohen-Macaulay since its geometric fibers are Cohen-Macaulay. By Lemma \ref{Lemma_Descent_for_Homomorphisms} again, we obtain a Zariski gluing data for $\{\Spec C_{\alpha}\}$, thus there exists an $S$-superscheme $X$ with an isomorphism $\psi:X\times _S S_1\to X_1$ and $\varphi=p_2^*\psi\circ p_1^*\psi^{-1}$. $X$ is proper and its smooth locus $X^{\text{sm}}$ is dense in every geometric fiber, since $X\times _S S_{\text{bos}}$ is the descent for $X_1\times _{S_{1}} S_{1,\text{bos}}$ and classical descent theory is applied. The puncture divisors $E,\mscr D$ are detrmined by their ideal of definition, so they also descent to $X$, by Lemma \ref{Lemma_Descent_for_Modules}. We also have a homomorphism $\delta:\Omega^1_{X/S}\to \omega_{X/S}(\mscr D)$ by Lemma \ref{Lemma_Descent_for_Homomorphisms}. $\delta$ satisfies Condition A in Definition \ref{def:stable SUSY curves} since it is a condition on geometric fibers; $\delta$ satisfies Condition B as well, since $D^{\otimes 2}\to T_{X/S}/D$ is injective after base change to $S_1$, and the image is $ T_{X/S}/D(-\mscr D)$ because it is after base change to $S_1$.
\end{proof}

\subsubsection{Local Structures of Stable SUSY Curves}
Before giving the charaterization of the local structures of stable SUSY curves, let us define some canonical local models. Take $S=\Spec A$, where $A=A^+\oplus A^-$ is a superring
\begin{enumerate}
	\item Smooth Canonical Model: $X=\Spec A[z|\theta]$, with derivation
	\begin{equation*}
	\delta:f\mapsto (\partial_{\theta}f+\theta\partial_zf)[\mathrm{d}z|\mathrm{d} \theta].
	\end{equation*}
	\item R Puncture Canonical Model: $X=\Spec A[z|\theta]$, with derivation 
	\begin{equation*}
	\delta:f\mapsto (\partial_{\theta}f+\theta z\partial_zf)\left[\frac{\mathrm{d}z}{z}\bigg{\vert}\mathrm{d} \theta\right].
	\end{equation*}
	\item R Node Canonical Model: $X=\Spec A[z_1,z_2|\theta]/(z_1z_2+t)$ where $t\in A^+,\alpha\in A^-$, with derivation 
	\begin{equation*}
	\delta:f\mapsto 
	\begin{cases}
	(\partial_{\theta}f+\theta z_1\partial_{z_1}f)\left[\frac{\mathrm{d}z_1}{z_1}\bigg{\vert}\mathrm{d} \theta\right]  & \text{on } D(z_1) \\
	(\partial_{\theta}f-\theta z_2\partial_{z_2}f)\left[-\frac{\mathrm{d}z_2}{z_2}\bigg{\vert}\mathrm{d} \theta\right]   & \text{on } D(z_2)
	\end{cases}.
	\end{equation*}
	\item NS Node Canonical Model: $X=\Spec A[z_1,z_2|\theta_1,\theta_2]/(z_1z_2+t^2,z_1\theta_2-t\theta_1,z_2\theta_1+t\theta_2,\theta_1\theta_2)$ where $t\in A^+$, with derivation 
	\begin{equation*}
	\delta:f\mapsto 
	\begin{cases}
	(\partial_{\theta_1}f+\theta_1 \partial_{z_1}f)[\mathrm{d}z_1|\mathrm{d} \theta_1] & \text{on } D(z_1) \\
	(\partial_{\theta_2}f+\theta_2 \partial_{z_2}f)[\mathrm{d}z_2|\mathrm{d} \theta_2] & \text{on } D(z_2)
	\end{cases}.
	\end{equation*}
\end{enumerate}
Let $S$ be a bosonic scheme and $X/S$ is a stable SUSY curve. Suppose that $S=\Spec A$ where $A$ is a complete local Noetherian ring with algebraically-closed residue field $\kappa$. Take a point $p$ in the special fiber of $X$, and take completion of $\mcal O_X$ at $p$, then according to the local structure of spin curves (Proposition \ref{Deform_Spin_Moduli}), there are three possibilities:
\begin{enumerate}
	\item[a.] $p$ is a smooth point: $\widehat{\mcal O}_{X,p}\cong\Spec A[\![ z|\theta]\!]$.
	\item[b.] $p$ is an R node: $\widehat{\mcal O}_{X,p}\cong\Spec A[\![z_1,z_2|\theta]\!]/(z_1z_2+t)$ where $t\in A$.
	\item[c.] $p$ is an NS node: $\widehat{\mcal O}_{X,p}\cong\Spec A[\![z_1,z_2|\theta_1,\theta_2]\!]/(z_1z_2+t^2,z_1\theta_2+t\theta_1,z_2\theta_1-t\theta_2,\theta_1\theta_2)$ where $t\in A$.
\end{enumerate}
To figure out the derivation, let us first write down a set of generators of $\omega_{X/S}$ as $\widehat{\mcal O}^+_{X,p}$ module in each case:
\begin{enumerate}
	\item[a.] $p$ is a smooth point: $\omega_{X/S}$ is generated by $\mathfrak b^-=[\mathrm{d}z|\mathrm{d} \theta]$ and $\mathfrak b^+=\theta[\mathrm{d}z|\mathrm{d} \theta]$.
	
	\item[b.] $p$ is a R node: $\omega_{X/S}$ is generated by 
	\begin{equation*}
	\mathfrak b^-=
	\begin{cases}
	\left[\frac{\mathrm{d}z_1}{z_1}\bigg{\vert}\mathrm{d} \theta\right]  & \text{on } D(z_1) \\
	\left[-\frac{\mathrm{d}z_2}{z_2}\bigg{\vert}\mathrm{d} \theta\right]   & \text{on } D(z_2)
	\end{cases}
	, \qquad 
	\mathfrak b^+=
	\begin{cases}
	\theta\left[\frac{\mathrm{d}z_1}{z_1}\bigg{\vert}\mathrm{d} \theta\right]  & \text{on } D(z_1) \\
	\theta\left[-\frac{\mathrm{d}z_2}{z_2}\bigg{\vert}\mathrm{d} \theta\right]   & \text{on } D(z_2)
	\end{cases}.
	\end{equation*}
	\item[c.] $p$ is a NS node: $\omega_{X/S}$ is generated by
	\begin{equation*}
	\mathfrak b_1^-=
	\begin{cases}
	\left[\mathrm{d}z_1{\vert}\mathrm{d} \theta_1\right]  & \text{on } D(z_1) \\
	\frac{t}{z_2}\left[\mathrm{d}z_2{\vert}\mathrm{d} \theta_2\right]   & \text{on } D(z_2)
	\end{cases}
	,\qquad 
	\mathfrak b_2^-=
	\begin{cases}
	-\frac{t}{z_1}\left[\mathrm{d}z_1{\vert}\mathrm{d} \theta_1\right]  & \text{on } D(z_1) \\
	\left[\mathrm{d}z_2{\vert}\mathrm{d} \theta_2\right]   & \text{on } D(z_2)
	\end{cases}.
	\end{equation*}
	and 
	\begin{equation*}
	\mathfrak b^+=
	\begin{cases}
	\frac{\theta_1}{z_1}\left[\mathrm{d}z_1{\vert}\mathrm{d} \theta_1\right]  & \text{on } D(z_1) \\
	-\frac{\theta_2}{z_2}\left[\mathrm{d}z_2{\vert}\mathrm{d} \theta_2\right]   & \text{on } D(z_2)
	\end{cases}.
	\end{equation*}
\end{enumerate}
By construction, if $p$ is a smooth point but not an R node, then $\delta$ maps $\theta$ to $\mathfrak{b}^-$ and maps $z$ to $\mathfrak{b}^+$; otherwise $\delta$ maps $\theta_i$ to $\mathfrak{b}^-_i$ and maps $z_i$ to $\pm z_i\mathfrak{b}^+$. It's easy to check that this is exactly the derivation in the canonical models. Hence we have 
\begin{lem}\label{Lemma_Local_Str_on_Complete_Ring}
	Let $S=\Spec A$ where $A$ is a complete local Noetherian ring and $X/S$ is a stable SUSY curve. Take a point $p$ in the special fiber of $X$, then the completion of $\mcal O_X$ at $p$ together with derivation $\delta$ is isomorphic to the completion of one of the canonical models at the origin of the special fiber.
\end{lem}

The next goal is to extend this result to a more general situation where $A$ can be a complete local Noetherian superring. The idea is to first prove it for local Artinian $A$, then take the inverse limit. For local Artinian $A$, use the induction on the length of $A$, more precisely, suppose that there is an extension 
\begin{equation*}
0\to \mcal J\to A\to A/\mcal J\to 0,
\end{equation*}
where $\mcal J$ is an ideal of $A$, purely bosonic or purely fermionic, such that the maximal ideal $\mathfrak{m}_A$ annihilates it. If Lemma \ref{Lemma_Local_Str_on_Complete_Ring} is true for $A/\mcal{J}$, then we want to show that Lemma \ref{Lemma_Local_Str_on_Complete_Ring} is true for $A$. Denote the residue field of $A$ by $\kappa$.

\newl To begin with, let us temporarily forget the derivation and try to figure out the underlying scheme structure of $\widehat{\mcal O}_{X,p}$, assuming that $\widehat{\mcal O}_{X,p}/\mcal{J}\widehat{\mcal O}_{X,p}$ is isomorphic to the completion of one of the canonical models at the origin of the special fiber. The following list of possibilities is a direct consequnce of flatness:
\begin{enumerate}
	\item[a.] $p$ is a smooth point: $\widehat{\mcal O}_{X,p}\cong\Spec A[\![ z|\theta]\!]$.
	\item[b.] $p$ is an R node: $\widehat{\mcal O}_{X,p}\cong\Spec A[\![z_1,z_2|\theta]\!]/(z_1z_2+t'+\eta \theta)$ where $t'\in A^+,\eta\in \mcal J$ such that $t'\equiv t \pmod {\mcal{J}}$.
	\item[c.] $p$ is an NS node: $\widehat{\mcal O}_{X,p}\cong\Spec A[\![z_1,z_2|\theta_1,\theta_2]\!]/(z_1z_2+t_1t_2+a,z_1\theta_2+t_1\theta_1+b,z_2\theta_1-t_2\theta_2+c,\theta_1\theta_2+d)$ where $t_i\in A^+$ such that $t_i\equiv t \pmod {\mcal{J}}$ and
	\begin{equation*}
	a,b,c,d\in B:=\mcal{J}\otimes_{\kappa}\widehat{\mcal O}_{X,p}/\mathfrak{m}_A\widehat{\mcal O}_{X,p}=\mcal{J}\widehat{\mcal O}_{X,p}.
	\end{equation*}
\end{enumerate}
In fact, suppose that $Y$ is stable SUSY curve over $\Spec A/\mcal J$, then for any such $p\in Y$, take an open affine neighborhood $W_p$ of $p$ such that $W_p$ contains no other node (if $p$ is a node) or puncture divisor, then by Theorem 2.1.2 of \cite{Vaintrob199001}, the obstruction of lifting $W_p$ to a flat superscheme over $\Spec A$ is an element in $\H^2(T^{\bullet}_{W_{p,s}})$ where $W_{p,s}$ is the special fiber of $W_p$. Note that $\H^i(T^{\bullet}_{W_{p,s}})$ is supported at $p$ for $i>0$, since $W_p\setminus \{p\}$ is smooth. If the obstruction vanishes, then the set of  deformation modulo infinitesimal automorphism $\text{Def}_A(W_p,\mcal J)$ is a torsor under $\H^1(T^{\bullet}_{W_{p,s}})$, and we also have 
\begin{align}\label{Deformation_Local_to_Global}
\text{Def}_A(W_p,\mcal J)=\text{Def}_A(\widehat{\mcal O}_{Y,p},\mcal J).
\end{align}
In fact, the obstruction always vanishes. For the situation a, i.e. $p$ is a smooth point, $\H^i(T^{\bullet}_{W_{p,s}})$ vanishes for $i>0$; for situation b, i.e. $p$ is an R node, $\H^2(T^{\bullet}_{W_{p,s}})$ vanishes since the relation defining $\widehat{\mcal O}_{Y,p}$ in $A[\![z_1,z_2|\theta]\!]$ is $z_1z_2+t$ which is not a zero-divisor, i.e. there is no syzygy; for the situation c, i.e. $p$ is an NS node, the obstruction of lifting $W_p$ is identified with the obstruction of lifting $\Spec \widehat{\mcal O}_{Y,p}$, which vanishes because there exists a lifting: $A[\![z_1,z_2|\theta_1,\theta_2]\!]/(z_1z_2+t_1t_2,z_1\theta_2+t_1\theta_1,z_2\theta_1-t_2\theta_2,\theta_1\theta_2)$\footnote{Notice that this is actually a germ of a nodal curve over $A$ defined by $A[\![z_1,z_2]\!]/(z_1z_2+t_1t_2)$, together with a rank one torsion-free sheaf $\mscr E$ on it generated by $\theta_1,\theta_2$ with relation $z_1\theta_2+t_1\theta_1=0,z_2\theta_1-t_2\theta_2=0$, and obvious relation $\theta_1\theta_2=0$ to ensure that $\mscr E\cdot\mscr E=0$.}.

\newl Following \cite{Deligne198709}, we can put more constraints on the elements $a,b,c,d$ in the case c. Namely, $a,b,c,d$ are components of an element in the $\H^1$ of the chain complex
\begin{center}
	\begin{tikzcd}
	B\partial_{z_i}\oplus B\partial_{\theta _i} \arrow[r,"\frac{\partial f_j}{\partial z_i}\oplus \frac{\partial f_j}{\partial \theta _i}"] & B\partial_{f_j} \arrow[r,"h_{kj}"] & B\partial_{S_k},
	\end{tikzcd}
\end{center}
where $f_j$ are relations of $\widehat{\mcal O}_{X,p}/\mathfrak{m}_A\widehat{\mcal O}_{X,p}$ and $S_k=\sum _j h_{kj}f_j$ are the syzygies, namely
\begin{equation*}
\begin{cases}
f_1= z_1z_2\\
f_2= z_1\theta_2\\
f_3= z_2\theta_1\\
f_4= \theta_1\theta_2
\end{cases}
,\qquad 
\begin{cases}
S_1=\theta_1f_1-z_1f_3\\
S_2=\theta_2f_1-z_2f_2\\
S_3=z_1f_4-\theta_1f_2\\
S_4=z_2f_4+\theta_2f_3
\end{cases}.
\end{equation*}
So we get four equations:
\begin{equation*}
\theta_1a=z_1c\text{ , }\theta_2a=z_2b\text{ , }z_1d=\theta_1b\text{ , }z_2d=-\theta_2c,
\end{equation*}
modulo the image of the matrix
\begin{align*}
M=\begin{bmatrix}
z_2  & z_1 & 0 & 0 \\
\theta_2 & 0 & 0 & z_1 \\
0  & \theta_1 & z_2 & 0 \\
0 & 0 & \theta_2 & -\theta_1
\end{bmatrix},
\end{align*}
which corresponds to coordinate transformation. If $N$ is odd, then $d\in (\theta_1,\theta_2)$ and is even, so it can be written as $d=\lambda_1\theta_1+\lambda_2\theta_2$, and a coordinate transform $\theta_1\mapsto \theta_1-\lambda_2,\theta_2\mapsto \theta_2-\lambda_1$ will set $d$ to zero; if $N$ is even, then $b,c\in (\theta_1,\theta_2)$, so the last two equations become $z_1d=z_2d=0$, hence $d=0$. So in any case we can set $d=0$. In the following discussion, we only consider coordinate transform whose image is in $B\partial_{f_1}\oplus B\partial_{f_2}\oplus B\partial_{f_3}$, for example the image of the matrix 
\begin{align*}
\begin{bmatrix}
z_2  & z_1 & 0 & 0 \\
\theta_2 & 0 & 0 & -z_1\theta_1 \\
0  & \theta_1 & z_2\theta_2 & 0
\end{bmatrix}.
\end{align*}
Using this coordinate transformation, together with the observation that any $\text{constant}\times \theta_1$ term of $b$ can be absorbed by $t_1\theta_1$ and any $\text{constant}\times \theta_2$ term of $c$ can be absorbed by $t_2\theta_2$, we can transform $b$ and $c$ to Laurent polynomials in $z_1,z_2$ simultaneously, i.e.
\begin{equation*}
b=f_1(z_1)+z_2f_2(z_2),\qquad c=z_1g_1(z_1)+g_2(z_2).
\end{equation*}
Moreover, the equations $\theta_1b=\theta_2c=0$ spells out that $$f_1(z_1)\theta_1+tf_2(z_2)\theta_2=g_2(z_2)\theta_2-tg_1(z_1)\theta_1=0,$$hence $$f_1(z_1)=tf_2(z_2)=g_2(z_2)=tg_1(z_1)=0.$$
Let us turn to consider equations
\begin{equation*}
\theta_1a=z_1c, \qquad \theta_2a=z_2b.
\end{equation*}
Since $b,c$ has no $\theta_i$ dependence, the equations become
\begin{equation*}
\theta_1a=z_1c=\theta_2a=z_2b=0.
\end{equation*}
We thus see that $a\in (\theta_1,\theta_2)$ and $f_2(z_2)=g_1(z_1)=0$, i.e. $b=c=0$. Now we can use coordinate transforms with images in $B\partial_{f_1}$ to modify $a$, for example we can choose $\theta_i$ times the first and the second rows of the matrix $M$ hence all terms in $a$ involving $z_i\theta_j$ can be removed, leaving
\begin{equation*}
a=\eta_1\theta_1+ \eta_2\theta_2,
\end{equation*}
where $\eta_i\in \mcal{J}$. To sum up, we have that
\begin{enumerate}
	\item[•] When $p$ is an NS node, $\widehat{\mcal O}_{X,p}\cong\Spec A[\![z_1,z_2|\theta_1,\theta_2]\!]/(z_1z_2+t_1t_2+\eta_1\theta_1+ \eta_2\theta_2,z_1\theta_2+t_1\theta_1,z_2\theta_1-t_2\theta_2,\theta_1\theta_2)$ where $t_i\in A^+$ such that $t_i\equiv t \pmod {\mcal{J}}$ and $\eta_i\in \mcal{J}$.
\end{enumerate}
Next we take the derivation $\delta$ into account, in fact we have
\begin{lem}\label{Lemma_Extension_of_Derivation}
	$\widehat{\mcal O}_{X,p}$ discribed above has a SUSY structure derivation $\delta$ extending the derivation on $\widehat{\mcal O}_{X,p}/\mcal{J}\widehat{\mcal O}_{X,p}$ if and only if
	\begin{enumerate}
		\item[a.] If $p$ is a smooth point, there is no condition.
		\item[b.] If $p$ is a R node, then $\eta=0$.
		\item[c.] If $p$ is a NS node, then $\eta_i=0$ and $t_1=t_2$.
	\end{enumerate}
	In other words, derivations extend exactly when $\widehat{\mcal O}_{X,p}$ is isomorphic to the completion of one of the canonical models at the origin of the special fiber. Moreover, the choice of extension is unique in all three cases, in the sense that there exists coordinate transform under which derivations are transformed to that of the corresponding canonical models.
\end{lem}

Assume this lemma for now, then it concludes the induction step and we have the following generalization of Lemma \ref{Lemma_Local_Str_on_Complete_Ring}.
\begin{proposition}\label{Local_Str_on_Complete_SuperRing}
	Let $S=\Spec A$ where $A$ is a complete local Noetherian superring and $X/S$ is a stable SUSY curve. Take a point $p$ in the special fiber of $X$, then the completion of $\mcal O_X$ at $p$ together with derivation $\delta$ is isomorphic to the completion of one of the canonical models at the origin of the special fiber.
\end{proposition}

\begin{proof}[Proof of {\normalfont Lemma \ref{Lemma_Extension_of_Derivation}}] We consider each cases separately 
	
	\newl \textbf{If $p$ is smooth}, $\widehat{\mcal O}_{X,p}$ is already isomorphic to that of canonical models so we can always choose the derivation in the smooth or R puncture canonical models, and the uniqueness follows from the condition B of Definition \ref{def:stable SUSY curves}.
	
	\newl\textbf{If $p$ is a R node}, write $$\delta=\delta_0+\delta',$$where $\delta_0$ is
	\begin{equation*}
	f\mapsto
	\begin{cases}
	(\partial_{\theta}f+\theta z_1\partial_{z_1}f)\left[\frac{\mathrm{d}z_1}{z_1}\bigg{\vert}\mathrm{d} \theta\right],  & \text{on } D(z_1) \\
	(\partial_{\theta}f-\theta z_2\partial_{z_2}f) \left[-\frac{\mathrm{d}z_2}{z_2}\bigg{\vert}\mathrm{d} \theta\right],  & \text{on } D(z_2)
	\end{cases},
	\end{equation*}
	and 
	\begin{equation*}
	\delta'\in \text{Der}_{\kappa}\left(\mathcal O_U,\mcal{J}\otimes_{\kappa}\omega_{U/\Spec \kappa}\right),
	\end{equation*}
	where $U=\Spec( \widehat{\mcal O}_{X,p}/\mathfrak{m}_A\widehat{\mcal O}_{X,p})\setminus \{p\}$ is a disjoint union of $\Spec \kappa (\!(z_1)\!)$ and $\Spec \kappa (\!(z_2)\!)$. So we can write down $\delta'$ as 
	\begin{equation*}
	\delta' f=(xDf+yD^2f)\mathfrak{b}^-,\qquad 
	D=
	\begin{cases}
	(\partial_{\theta}f+\theta z_1\partial_{z_1}f)& \text{on } D(z_1) \\
	(\partial_{\theta}f-\theta z_2\partial_{z_2}f) & \text{on } D(z_2)
	\end{cases},
	\end{equation*}
	where $x,y\in\mcal{J}\otimes_{\kappa}\mcal O_U$ such that $x$ is even and $y$ is odd. Using Lemma 1.1 of \cite{Deligne198709}, we have that $x+Dy=0$, so $\delta'$ can be written as 
	\begin{equation*}
	\delta': f\mapsto -D(y\cdot Df)\mathfrak{b}^-.
	\end{equation*}
	Consider $\delta'\theta=(1-Dy)\mathfrak{b}^-\in \omega_{X/S}$, this implies that $y\in \mcal J [\![z_1]\!]\times\mcal J [\![z_2]\!]$ or $\theta\mcal J [\![z_1]\!]\times \theta\mcal J [\![z_2]\!]$ depending on the parity of $\mcal J$. Now consider
	\begin{equation*}
	\delta z_1=
	\begin{cases}
	(\theta +y)z_1\mathfrak{b}^- & \text{on } D(z_1) \\
	\frac{\eta-\theta t'}{z_2}\mathfrak{b}^- & \text{on }D(z_2)
	\end{cases},\qquad 
	\delta z_2=
	\begin{cases}
	\frac{\eta+\theta t'}{z_1}\mathfrak{b}^- & \text{on } D(z_1) \\
	-\left(\theta+y\right)z_2\mathfrak{b}^- & \text{on }D(z_2)
	\end{cases}
	\end{equation*}
	So $\delta z_1-(\theta+y)z_1\mathfrak{b}^-$ is an element in $\omega_{X/S}$ which is zero on $D(z_1)$ and is $\eta/z_2\mathfrak{b}^-$ on $D(z_2)$, and $\delta z_2+(\theta+y)z_2$ is zero on $D(z_2)$, and is $\eta/z_1\mathfrak{b}^-$ on $D(z_1)$. However, $\eta/z_1$ is not an element of $\widehat{\mcal O}_{X,p}$, unless $\eta=0$. Hence it must be that $\eta=0$.\\
	
	Assume that $\mcal J$ is fermionic then $t'=t$ and $y=(f_1(z_1),f_2(z_2))$, where $f_1,f_2$ are Laurent polynomials in $z_1,z_2$ valued in $\mcal J$ respectively, then we can write action of $\delta$ on generators by
	\begin{alignat*}{2}
	\theta&\mapsto (1-\theta z_1f'_1(z_1)-\theta z_2f'_2(z_2))\mathfrak{b}^-,
	\\
	z_1&\mapsto (\theta+f_1(z_1))z_1\mathfrak{b}^-,
	\\
	z_2&\mapsto -(\theta+f_2(z_2))z_2\mathfrak{b}^-.
	\end{alignat*}
	Take coordinate transformation $\theta\mapsto \theta +f_1(z_1)+f_2(z_2)-f_2(0)\equiv\theta'$ and $z_2\mapsto z_2(1+(f_2(0)-f_1(0))\theta)\equiv z_2'$, and note that this does not change $\mathfrak{b}^-$. Under the new coordinate $(z_1,z'_2,\theta')$, $\delta$ acts on generators by 
	\begin{equation*}
	\theta'\mapsto \mathfrak{b}^-,\qquad z_1\mapsto \theta'z_1\mathfrak{b}^-,\qquad z'_2\mapsto -\theta'z'_2\mathfrak{b}^-.
	\end{equation*}
	It is thus the derivation of the R node canonical model.
	
	\newl Now assume that $\mcal J$ is bosonic then $y=\theta(f_1(z_1),f_2(z_2))$, where $f_1,f_2$ are Laurent polynomials in $z_1,z_2$ valued in $\mcal J$ respectively, then
	\begin{equation*}
	\delta \theta=
	\begin{cases}
	(1-f_1(z_1))\mathfrak{b}^- & \text{on } D(z_1) \\
	(1-f_2(z_2))\mathfrak{b}^- & \text{on }D(z_2)
	\end{cases}.
	\end{equation*}
	So we must have $f_1(0)=f_2(0)$ in order that $\delta \theta\in \omega_{X/S}$, thus $y\in \mcal J \widehat{\mcal O}_{X,p}$. we can write action of $\delta$ on generators by
	\begin{alignat*}{2}
	\theta&\mapsto (1-\partial_{\theta}y)\mathfrak{b}^-, 
	\\
	z_1&\mapsto (\theta+y)z_1\mathfrak{b}^-,
	\\
	z_2&\mapsto -(\theta+y)z_2\mathfrak{b}^-.
	\end{alignat*}
	Take the coordinate transformation $\theta\mapsto \theta +y=:\theta'$, note that this does not change $\mathfrak{b}^-$. Under the new coordinate $(z_1,z_2,\theta')$, $\delta$ acts on generators by
	\begin{equation*}
	\theta'\mapsto \mathfrak{b}^-,\qquad 
	z_1\mapsto \theta'z_1\mathfrak{b}^-,\qquad 
	z_2\mapsto -\theta'z_2\mathfrak{b}^-.
	\end{equation*}
	It is thus the derivation of the R node canonical model. This concludes the situation b.
		
	\newl \textbf{If $p$ is a NS node}, write $$\delta=\delta_0+\delta',$$where $\delta_0$ is
	\begin{equation*}
	\delta_0:f\mapsto 
	\begin{cases}
	(\partial_{\theta_1}f+\theta_1 \partial_{z_1}f)[\mathrm{d}z_1|\mathrm{d} \theta_1], & \text{on } D(z_1), \\
	(\partial_{\theta_2}f+\theta_2 \partial_{z_2}f)[\mathrm{d}z_2|\mathrm{d} \theta_2], & \text{on } D(z_2)
	\end{cases},
	\end{equation*} 
	and 
	\begin{equation*}
	\delta'\in \text{Der}_{\kappa}\left(\mathcal O_U,\mcal{J}\otimes_{\kappa}\omega_{U/\Spec \kappa}\right),
	\end{equation*}
	where $U=\Spec( \widehat{\mcal O}_{X,p}/\mathfrak{m}_A\widehat{\mcal O}_{X,p})\setminus \{p\}$ is a disjoint union of $\Spec \kappa (\!(z_1)\!)$ and $\Spec \kappa (\!(z_2)\!)$. The same argument in the situation b shows that $x+Dy=0$, $y\in \mcal J [\![z_1]\!]\times\mcal J [\![z_2]\!]$ or $\theta\mcal J [\![z_1]\!]\times \theta\mcal J [\![z_2]\!]$ depending on the parity of $\mcal J$, and $\delta'$ can be written as
	\begin{equation}
	\delta': f\mapsto 
	\begin{cases}
	-D(y\cdot Df)[\mathrm{d}z_1|\mathrm{d} \theta_2], & \text{on } D(z_1)\\
	-D(y\cdot Df)[\mathrm{d}z_2|\mathrm{d} \theta_2], & \text{on } D(z_2)
	\end{cases}.
	\end{equation}
	Consequently $\delta'f$ has nonnegative power of $z_1$ (respectively nonnegative power of $z_2$) on $D(z_1)$ (respectively on $D(z_1)$), for all $f\in \widehat{\mcal O}_{X,p}/\mathfrak{m}_A\widehat{\mcal O}_{X,p}$. Now assume that $\mcal J$ is fermionic, then $t_i=t$. Consider
	\begin{equation*}
	\delta_0z_1=
	\begin{cases}
	\theta_1[\mathrm{d}z_1|\mathrm{d} \theta_1], & \text{on } D(z_1)\\
	\frac{\eta_2-z_1\theta_2}{z_2}[\mathrm{d}z_2|\mathrm{d} \theta_2], & \text{on } D(z_2)
	\end{cases}
	\end{equation*}
	So $\delta z_1-z_1\mathfrak{b}^+-\delta'z_1$ is zero on $D(z_1)$ and is $\eta_2/z_2$ on $D(z_2)$. Since $\delta'z_1$ has nonnegative power of $z_2$ on $D(z_2)$, thus $\eta_2$ must be zero. Similarly $\eta_1=0$. Assume that $\mcal J$ is bosonic, then $\eta_i=0$ automatically. Consider
	\begin{equation*}
	\delta_0\theta_1=
	\begin{cases}
	[\mathrm{d}z_1|\mathrm{d} \theta_1], & \text{on } D(z_1)\\
	\frac{t_2}{z_2}[\mathrm{d}z_2|\mathrm{d} \theta_2], & \text{on } D(z_2)
	\end{cases}.
	\end{equation*}
	Note that a set of generators of $\omega_{X/S}$ is
	\begin{alignat*}{2}
	\mathfrak b_1^-&=
	\begin{cases}
	\left[\mathrm{d}z_1{\vert}\mathrm{d} \theta_1\right],  & \text{on } D(z_1) \\
	\frac{t_1}{z_2}\left[\mathrm{d}z_2{\vert}\mathrm{d} \theta_2\right],   & \text{on } D(z_2)
	\end{cases},
	\\
	\mathfrak b_2^-&=
	\begin{cases}
	-\frac{t_2}{z_1}\left[\mathrm{d}z_1{\vert}\mathrm{d} \theta_1\right],  & \text{on } D(z_1) \\
	\left[\mathrm{d}z_2{\vert}\mathrm{d} \theta_2\right],   & \text{on } D(z_2)
	\end{cases},
	\\
	\mathfrak b^+&=
	\begin{cases}
	\frac{\theta_1}{z_1}\left[\mathrm{d}z_1{\vert}\mathrm{d} \theta_1\right],  & \text{on } D(z_1) \\
	-\frac{\theta_2}{z_2}\left[\mathrm{d}z_2{\vert}\mathrm{d} \theta_2\right],   & \text{on } D(z_2)
	\end{cases}.
	\end{alignat*}
	So $\delta_0\theta_1-\mathfrak{b}_1^-$ is zero on $D(z_1)$ and is $\frac{t_2-t_1}{z_2}$ on $D(z_2)$. Since $\delta'\theta_1$ has nonnegative power of $z_2$ on $D(z_2)$, this concludes that $t_1-t_2$ must be zero. Then we can write down the action of $\delta$ on generators:
	\begin{equation*}
	\delta \theta_i=(1-Dy)\mathfrak{b}_i^-,\qquad \delta z_i=(z_i\theta_i+y)\mathfrak{b}_i^-.
	\end{equation*}
	Suppose that $y=(f_1(z_1),f_2(z_2))$, where $f_1,f_2$ are Laurent polynomials in $z_1,z_2$ respectively, then the constant term of $\delta z_1$ is $f_1(0)$, so $f_1(0)=0$ in order that $\delta z_1\in \omega_{X/S}$. Similarly $f_2(0)=0$. Thus $y=(z_1\widetilde{f}_1(z_1),z_2\widetilde{f}_2(z_2))$, and $\delta z_i=(\theta_i+\widetilde f_i(z_i))z_i\mathfrak{b}_i^-$. Now take the coordinate transform $\theta_i\mapsto\theta_i+\widetilde f_i(z_i)\equiv \theta_i'$, note that this does not change $\mathfrak{b}_i^-$. Under the new coordinate $(z_1,z_2,\theta_1',\theta_2')$, $\delta$ acts on generators by 
	\begin{equation*}
	\theta_i'\mapsto \mathfrak{b}_i^-,\qquad 
	z_i\mapsto \pm z_i\mathfrak{b}^+.
	\end{equation*}
	It is thus the derivation of the NS node canonical model. This concludes the situation c.
\end{proof}

\subsubsection{Infinitesimal Automorphisms of Stable SUSY Curves}
Let the sheaf of automorphism group of $X$ which induces identity modulo $\mcal J$ be $\mcal A$, then if $X$ is smooth then it is a classical fact that 
\begin{equation*}
\mcal A=\left(T^+_{X_{s}}(-E-\mscr D)\oplus\Pi T^+_{X_{s}}(-E-\mscr D)\otimes_{\mcal O^+_{X_s}} \mscr E\right)\otimes_{\kappa} \mcal J,
\end{equation*}
where $T^+_{X_{s}}$ is the even tangent sheaf of the special fiber $X_s$, and $\mscr E$ is the spinor sheaf. If $X$ is a general stable SUSY curve, then we need to take R and NS nodes into account. In fact, suppose that an automorphism at a formal neighborhood of a node comes from a derivation $xD+yD^2$ where $x$ is odd and $y$ is even, then it must preserves the superconformal structure, i.e. $[D,xD+yD^2]\equiv 0\pmod {D}$, hence $x=\frac{1}{2}Dy$. In local coordinate we can write it as
\begin{enumerate}
	\item[•] If $p$ is a R node: when $\mcal J$ is bosonic
	\begin{equation*}
	xD+yD^2=z_1f_1(z_1)\partial_{z_1}+z_2f_2(z_2)\partial_{z_2}+\frac{z_1f'_1(z_1)+z_2f'_2(z_2)}{2}\theta\partial_{\theta},
	\end{equation*}
	where $f_i(z_i)\in \mathcal J[\![z_i]\!]$. When $\mcal J$ is fermionic 
	\begin{equation*}
	xD+yD^2=\theta g(z_1,z_2)(z_1\partial_{z_1}+z_2\partial_{z_2})+g(z_1,z_2)\partial_{\theta},
	\end{equation*}
	where $g(z_1,z_2)\in \mathcal J[\![z_1,z_2]\!]/(z_1z_2)$. Note that $z_1\partial_{z_1}+z_2\partial_{z_2}$ is a generator of $(\omega^+_{X_{s},p})^{-1}$, and $\partial_{\theta}$ is a generator of the dual of spin structure $\mscr E_p$.
	\item[•] If $p$ is a NS node: when $\mcal J$ is bosonic
	\begin{equation*}
	xD+yD^2=f_1(z_1)\partial_{z_1}+f_2(z_2)\partial_{z_2}+\frac{1}{2}f'_1(z_1)\theta_1\partial_{\theta_1}+\frac{1}{2}f'_2(z_2)\theta_2\partial_{\theta_2},
	\end{equation*}
where $f_i(z_i)\in z_i\mathcal J[\![z_i]\!]$. When $\mcal J$ is fermionic $$xD+yD^2=g_1(z_1)(\partial_{\theta_1}-\theta_1\partial_{z_1})+g_2(z_2)(\partial_{\theta_2}-\theta_2\partial_{z_2}),$$where $g_i(z_i)\in \mathcal J[\![z_i]\!]$. Note that $\partial_{\theta_i}-\theta_i\partial_{z_i}$ is a generator of the dual of spin structure on the $i$'th component of the normalization at $p$, and the direct sum of their pushforward along the normalization map is the dual of spin structure $\mscr E_p$.
\end{enumerate}
so we see that
\begin{align*}
\widehat{\mcal A}_p=\left(T^+_{X_{s},p}\oplus\Pi (\omega^+_{X_{s},p})^{-1}\otimes \mscr E_p\right)\otimes_{\kappa} \mcal J\cong\left(T^+_{X_{s},p}\oplus\Pi \mscr E^{\vee}_p\right)\otimes_{\kappa} \mcal J,
\end{align*}
if $p$ is a R node or NS node. Combine the smooth case and the nodes by formal gluing, we get
\begin{proposition}\label{Infinitesimal_Automorphism}
	The sheaf of infinitesimal automorphism group $\mcal A$ is isomorphic to
	\begin{align}
	\mcal A=\left(T^+_{X_{s}}(-E-\mscr D)\oplus\Pi \mscr E^{\vee}\right)\otimes_{\kappa} \mcal J,
	\end{align}
\end{proposition}

\begin{corollary}\label{Absence_of_Global_Infinitesimal_Automorphism}
	There is no global infinitesimal automorphism, neither even nor odd.
\end{corollary}

\begin{proof}
	Let $\pi:C\to X_{s,\text{bos}}$ be the normalization map of the bosonic truncation of $X_s$, then 
	\begin{equation*}
	T^+_{X_{s}}(-E-\mscr D)=\pi_*\left(\omega^{-1}_\mscr C(-E-\mscr D-\mcal R-\mcal N)\right),
	\end{equation*}
	where $\mcal R$ (resp. $\mcal N$) is the preimage of R nodes (resp. NS nodes) in $\mscr C$. Since $X_{s,\text{bos}}$ is a stable curve, $\omega^{-1}_C(-E-\mscr D-\mcal R-\mcal N)$ has negative degree on every connected component, hence $$\Gamma(X,T^+_{X_{s}})(-E-\mscr D)=\Gamma(\mscr C,\omega^{-1}_{\mscr C}(-E-\mscr D-\mcal R-\mcal N))=0.$$
	On the other hand, there is a short exact sequence
	\begin{equation*}
	 0\to (\omega^+_{X_{s},p})^{-1}(-E-\mscr D)\otimes_{\mcal O^+_{X_s}} \mscr E\to \pi_*\left(\omega^{-1}_{\mscr C}(-E-\mscr D-\mcal R-\mcal N)\otimes \mscr E_{\mscr C}\right)\to \bigoplus_{p\in \text{ R Node}}\kappa\to 0,
	\end{equation*}
	where $\mscr E_{\mscr C}$ is a square root of $\omega_{\mscr C}(\mscr D+\mcal R)$. On each connected component $\mscr C_i$, 
	\begin{equation*}
	\deg \mscr E_{\mscr C}=\frac{1}{2}\deg \omega_{\mscr C_i}(\mscr D+\mcal R)\le \frac{1}{2}\deg \omega_{\mscr C_i}(E+\mscr D+\mcal R+\mcal N),
	\end{equation*}
	consequently
	\begin{equation*}
	\deg \omega^{-1}_{\mscr C_i}(-E-\mscr D-\mcal R-\mcal N)\otimes \mscr E_{\mscr C_i}\le -\frac{1}{2}\deg \omega_{\mscr C_i}(E+\mscr D+\mcal R+\mcal N)<0.
	\end{equation*}
	Then it follows that
	\begin{equation*}
	\Gamma(X,(\omega^+_{X_{s},p})^{-1}(-E-\mscr D)\otimes_{\mcal O^+_{X_s}} \mscr E)\subset\Gamma(\mscr C,\omega^{-1}_{\mscr C}(-E-\mscr D-\mcal R-\mcal N)\otimes \mscr E_{\mscr C})=0.
	\end{equation*}
\end{proof}

\subsubsection{Isomorphisms of Stable SUSY Curves}

\begin{proposition}\label{Representablity_of_Diagonal}
	Suppose that $S$ is a Noetherian superscheme and $X,Y$ are two stable SUSY curves over $S$, then the functor
	$$\mathfrak{Isom}_S(X,Y):(f:T\to S)\mapsto \Isom_T(X_T,Y_T),$$ is represented by a finite unramified superscheme over $S$.
\end{proposition}
\begin{proof}
	Since representability can be proven Zariski locally, so we may additionally assume that $S$ is affine, so $X$ is projective over $S$ according to Lemma \ref{Stable_SUSY_Curve_Is_Locally_Projective}. First consider the natural transform $\mathfrak{Isom}_S(X,Y)\to \mathfrak{Map}_S(X,Y)$ (see Theorem \ref{Superscheme_of_Morphisms}). According to \textit{loc.cit.} $\mathfrak{Map}_S(X,Y)$ is represented by a locally finite type $S$-superscheme $\Map_S(X,Y)$, we claim that $\mathfrak{Isom}_S(X,Y)$ is represented by a locally closed sub superscheme of $M\equiv \Map_S(X,Y)$. Take the universal map $\Phi:X_M\to Y_M$. The locus $U$ on $M$ where $\Phi$ is an isomorphism of superschemes is open. Restricting on $U$, and let the puncture of $X$ (resp. $Y$) be $E,\mcal D$ (resp. $E',\mcal D'$), then the incidence relations $\Phi^{-1}E'_i\subset E_i$, $E_i\subset \Phi^{-1}E'_i$, $\Phi^{-1}\mcal D'_i\subset \mcal D_i$ and $\mcal D_i\subset\Phi^{-1}\mcal D'_i$ are represented by closed sub superschemes $U_{\alpha}$ of $U$ (Corollary \ref{Incidence_Superscheme}). Take intersections of these $U_{\alpha}$ and we obtain a closed sub superscheme $V$ of $U$ which represents isomorphisms between $X$ and $Y$ as superschemes as preserves punctures. Finally the preservation of derivation can be interpreted as $\Phi^*\delta'-\delta=0$, so it is represented by the intersection of the section induced by $\Phi^*\delta'-\delta$ and the zero section in the Hom superscheme $\mathfrak{Hom}_V(\Omega^1_{X_V/V},\omega_{X_V/V})$. Note that moduli spaces showing up in these steps are locally finite presented, whence $\mathfrak{Isom}_S(X,Y)$ is represented by a locally finite type $S$-superscheme $\Isom_S(X,Y)$.
	
	\newl Next, if $T$ is a bosonic scheme, then isomorphisms between $X_T$ and $Y_T$ are one to one corresponds to isomorphisms between their underlying spin curves. By Yoneda Lemma, this implies that 
	\begin{equation*}
	\Isom_S(X,Y)_{\bos}\cong\Isom_{S_{\bos}}(\mscr C,\mscr C'),
	\end{equation*}
	where $\mscr C,\mscr C'$ are underlying spin curves of $X\times_S S_{\bos},Y\times_S S_{\bos}$. Consequently $\Isom_S(X,Y)_{\bos}$ is finite over $S_{\bos}$ hence $\Isom_S(X,Y)$ is finite over $S$ by Nakayama Lemma.
	
	\newl Finally, Corollary \ref{Absence_of_Global_Infinitesimal_Automorphism} implies that for any geometric point $s$ of $S$ and any $\phi\in \Isom_s(X_s,Y_s)$, the tangent space of $\Isom_S(X,Y)_s$ at $\phi$ is trivial, hence $\Isom_S(X,Y)$ is unramified over $S$. This concludes the proof.
\end{proof}

\subsubsection{\' Etale Local Structures of Stable SUSY Curves}

Let us discuss the \'etale local structures of stable SUSY curves. 
\begin{lem}\label{Lemma_Artin_Approximation}
	Suppose that $A$ is a strictly Henselian local Noetherian superring, such that $A/\mathfrak{n}_A$ is a G-ring. Then there exists an inductive system $\{B_{\lambda}\}_{{\lambda}\in L}$ of smooth $A$ superalgebras, such that $$\widehat{A}=\lim_{\substack{\longrightarrow\\{\lambda}}}B_{\lambda}.$$
\end{lem}

\begin{proof}
	Let $A=A^+\oplus A^-$ be its bosonic and fermionic parts. $A/\mathfrak{n}_A$ being a G-ring implies that $A^+$ is also a G-ring, so $\widehat{A}^+=\widehat{A^+}$ is a regular-$A^+$ algebra. By Popescu's Theorem, there exists an inductive system $\{B'_{\lambda}\}_{{\lambda}\in L}$ of smooth $A^+$ algebras, such that $$\widehat{A^+}=\lim_{\substack{\longrightarrow\\{\lambda}}}B'_{\lambda}.$$Then $B_{\lambda}=B'_{\lambda}\otimes _{A^+}A$ is what we want.
\end{proof}

\begin{proposition}\label{Etale_Local_Str}
	Suppose that $S$ is a Noetherian superscheme such that $S_{\text{\normalfont bos}}$ is a G-scheme. Let $X\to S$ be a stable SUSY curve, then there exists an \' etale cover $\{V_i\to S\}_{i\in I}$ and for each $i\in I$ there exists a finite set $J_i$ and an \' etale cover $\{\rho_{ij}:U_{ij}\to X_{i}\}_{j\in J_i}$ where $X_i=X\times _S V_i$, such that
	\begin{enumerate}
		\item Every $U_{ij}$ has an \' etale morphism $\varphi_{ij}:U_{ij}\to Y$ where $Y$ is one of the canonical models such that $\rho_{ij}^*\delta_X=\varphi_{ij}^*\delta_{\text{\normalfont can}}$.
		\item For every pair of indices $j,j'\in J_i$, $j\neq j'$, $U_{ij}\times _{X_i}U_{ij'}$ is smooth over $V_i$ and contains no marking divisors.
	\end{enumerate}
\end{proposition}

\begin{proof}
	By Lemma \ref{Lemma_Limit_Stable_SUSY_Curves}, it suffices to prove the proposition with additional assumtion that $S=\Spec A$ where $A$ is a strict Henselian local Noetherian superring. Since $S_{\text{bos}}$ is a G-scheme, we can apply Lemma \ref{Lemma_Artin_Approximation} to furthermore assume that $S=\Spec A$ where $A$ is a complete local Noetherian superring with algebraically closed residue field\footnote{$B_i$ in the Lemma \ref{Lemma_Artin_Approximation} are only smooth, not \' etale, but we can always choose a closed point $p$ on the special fiber of $\Spec B_i\to \Spec A$ and lift a system of parameters of $p$ on the special fiber to the germ of $B_i$ at $p$, after taking quotient of these elements, we get a local superring which is \' etale over $A$ hence isomorphic to $A$ since $A$ is strictly henselian.}. Let the set of nodes and marking divisors of the special fiber $X_s$ be $J_0$, then for every $p\in J_0$, there exists an isomorphism
	\begin{equation*}
	\varphi_p:\Spec \widehat{\mcal O}_{X,p}\cong \Spec \widehat{\mcal O}_{Y,0},
	\end{equation*}
	where $Y$ is one of the canonical models and $0$ is the origin of its special fiber over $\Spec A$, by Proposition \ref{Local_Str_on_Complete_SuperRing}. Choose an affine neighborhood of $p$, denoted by $W_p$, such that it does not contain any point in the set $J_0\setminus \{p\}$. Now $X_{\text{bos}}$ is a G-scheme, so Lemma \ref{Lemma_Artin_Approximation} implies that there exists an directed system $\{B_{\lambda}\}_{\lambda\in L}$ of smooth $\mcal O_{X}(W_p)$ superalgebras such that $$\widehat{\mcal O}_{X,p}=\lim_{\substack{\longrightarrow\\ \lambda}}B_{\lambda}.$$Denote the morphism $\Spec B_{\lambda}\to W_p$ by $\rho_{p,{\lambda}}$. Next we prove the NS node case and the rest are similar. There exists ${\lambda}\in L$ such that $\exists z_1,z_2,\theta_1,\theta_2\in B_{\lambda}$ whose images in $\widehat{\mcal O}_{X,p}$ are the corresponding generators of the maximal ideal and satisfies the relations in the NS node canonical model, this defines a morphism of $A$ superschemes $\varphi_{p,\lambda}:\Spec B_{\lambda}\to Y$. Moreover, after increasing ${\lambda}$, there is a set of generators $\mathfrak{b}_1^-,\mathfrak{b}_2^-,\mathfrak{b}^+\in \rho_{p,{\lambda}}^*\omega_{X/S}$, whose images in $\omega_{\Spec \widehat{\mcal O}_{X,p}/S}$ are the corresponding generators in the NS node canonical model, such that $\rho_{p,{\lambda}}^*\delta_X:\rho_{p,{\lambda}}^*\Omega^1_{X/S}\to \rho_{p,{\lambda}}^*\omega_{X/S}$ acts on generators by
	\begin{alignat*}{2}
	\mathrm{d}z_{1}&\to z_1\mathfrak{b}^+, &\qquad \mathrm{d}z_{2}&\to -z_2\mathfrak{b}^+,
	\\
	\mathrm{d}\theta_{1}&\to \mathfrak{b}_1^-,&\qquad \mathrm{d}\theta_{2}&\to \mathfrak{b}_2^-.
	\end{alignat*}
	Use the same notation $p$ to denote the image of the closed point of $\Spec \widehat{\mcal O}_{X,p}$ in $\Spec B_{\lambda}$, then lift a system of homogeneous parameters of the fiber of $\rho_{p,{\lambda}}$ at point $p$ to an open affine neighborhood $U'_{p,\lambda}$ and take the closed sub superscheme of $U'_{p,\lambda}$ defined by these elements, denoted by $U''_{p,\lambda}$, then $U''_{p,\lambda}$ is \' etale over $W_p$ at $p$. We keep using the notation $\varphi_{p,\lambda}$ to denote its restriction to $U''_{p,\lambda}$. Then taking completion at $p$ induces isomorphism $\widehat{\mcal O}_{U''_{p,\lambda},p}\cong \widehat{\mcal O}_{Y,0}$, hence $\varphi_{p,\lambda}$ is \' etale at $p$. Shrinking $U''_{p,\lambda}$ if needed, we have two morphisms $\rho_{p,\lambda}:U''_{p,\lambda}\to U_p$ and $\varphi_{p,\lambda}:U''_{p,\lambda}\to Y$ such that $\rho_{p,\lambda}^*\delta_X=\varphi_{p,\lambda}^*\delta_{\text{\normalfont can}}$ since they agree on generators of $\Omega^1_{U''_{p,\lambda}/S}$.
	
	\newl Applying the same procedure to other elements in the set $J_0$, we found a finite set of superschemes $\{U_p\}_{p\in J_0}$ with \' etale morphisms $\rho_p:U_p\to X,\varphi_p:U_p\to Y$ such that $\rho_p^*\delta_X=\varphi_p^*\delta_{\text{\normalfont can}}$. Moreover $U_p\times_X U_{p'}$ is smooth over $S$ and contains no puncturing divisors if $p\neq p'$, this follows from the construction that $U_p$ are superschemes over $W_p$ for some open affine neighborhood of $p$ which contains no element $p'\in J_0\setminus \{p\}$, so the special fiber of $U_p\times_X U_{p'}$ is smooth and contains no marking divisor. The remaining issue is that $\bigcup _{p\in J_0}U_p\to X$ might not be a covering, this can be cured by enlarging $J_0$ to include all points in $X_s\setminus \bigcup _{p\in J_0}W_p$, note that this is a finite set since $X_s$ is one dimensional.
\end{proof}

\subsubsection{Absence of Obstruction for Lifting Stable SUSY Curves}

Recall the setup for discussing local models of stable SUSY curves: suppose that there is an extension 
\begin{equation*}
0\to \mcal J\to A\to A/\mcal J\to 0,
\end{equation*}
of local Artinian superrings, where $\mcal J$ is an ideal of $A$, purely bosonic or purely fermionic, such that the maximal ideal $\mathfrak{m}_A$ annihilates it. We also assume that the residue field of $A$, denote by $\kappa$, is algebraically closed. 

\begin{proposition}\label{Absence_of_Obstruction}
	Suppose that there is a stable SUSY curve $Y$ over $\Spec A/\mcal J$, then there exists a stable SUSY curve $X$ over $\Spec A$ such that its base change to $\Spec A/\mcal J$ is isomorphic to $Y$.
\end{proposition}

\begin{proof}
	\textbf{Step 1.} We prove that Zariski-locally on $Y$ there always exists such lifting. Let the deformation space of $\mcal O_{Y,p}$ (see the notation in equation \eqref{Deformation_Local_to_Global}) as a germ of SUSY curve be $\sDef_A(\mcal O_{Y,p},\mcal J)$, we want to show that this is not empty. According to Equation \ref{Deformation_Local_to_Global}, we have a commutative diagram of maps between sets
	\begin{center}
		\begin{tikzcd}
		\sDef _A(\mcal O_{Y,p},\mcal J)\arrow[r] \arrow[d] & \sDef  _A(\widehat{\mcal O}_{Y,p},\mcal J)\arrow[d]\\
		\text{Def}_A(\mcal O_{Y,p},\mcal J) \arrow[r, equal] &\text{Def}_A(\widehat{\mcal O}_{Y,p},\mcal J).
		\end{tikzcd}
	\end{center}
	We claim that the map $\sDef _A(\mcal O_{Y,p},\mcal J)\rightarrow \sDef _A(\widehat{\mcal O}_{Y,p},\mcal J)$ is bijective.
	\begin{enumerate}
		\item[•] Injectivity: Suppose that there are two deformations whose completions are isomorphic, then the underlying deformation of superscheme structures are isomorphic, denoted by $\mcal O'_{Y,p}$, by the identification $\text{Def}_A(\mcal O_{Y,p},\mcal J) =\text{Def}_A(\widehat{\mcal O}_{Y,p},\mcal J)$. Let the derivations be $\delta_1,\delta_2$, using Lemma \ref{Lemma_Artin_Approximation}, there exists an \' etale morphism $\mcal O'_{Y,p}\to C$ such that $C$ is local, and $\pi^*\delta_1,\pi^*\delta_2$ are related by infinitesimal automorphism $\alpha$ on $\Spec C$, then $\pi_1^*\alpha=\pi_2^*\alpha$ because they both transform $q^*\delta_1$ to $q^*\delta_2$, see notations below
		\begin{center}
			\begin{tikzcd}
			\Spec C\otimes _{\mcal O'_{Y,p}}C\arrow[r,"\pi_1"] \arrow[d,"\pi_2"]\arrow[rd,"q"] & \Spec C\arrow[d,"\pi"]\\
			\Spec C\arrow[r,"\pi"] &\Spec \mcal O'_{Y,p}
			\end{tikzcd},
		\end{center}
		Hence $\alpha$ descents to an infinitesimal automorphism on $\Spec \mcal O'_{Y,p}$ and transform $\delta_1$ to $\delta_2$.
		\item[•] Surjectivity: Suppose that $(\widehat{\mcal O}'_{Y,p},\delta')$ is a deformation of $(\widehat{\mcal O}_{Y,p},\delta)$ whose underlying superscheme is the completion of a deformation of $\mcal O_{Y,p}$, demoted by $\mcal O'_{Y,p}$. Using Lemma \ref{Lemma_Artin_Approximation}, there exists an \' etale morphism $\mcal O'_{Y,p}\to C$ such that $C$ is local and $\delta'$ descends to a SUSY derivation $\delta_C$ on $\Spec C$. Then $\pi_1^*\delta_C,\pi_2^*\delta_C$ are isomorphic derivations because their completions are isomorphic, hence they are related by an infinitesimal automorphism $\alpha$, see notations below
		\begin{center}
			\begin{tikzcd}
			\Spec C\otimes _{\mcal O'_{Y,p}}C\arrow[r,"\pi_1"] \arrow[d,"\pi_2"] & \Spec C\arrow[d,"\pi"]\\
			\Spec C\arrow[r,"\pi"] &\Spec \mcal O'_{Y,p}
			\end{tikzcd},
		\end{center}
		Apparently $\alpha$ satifies cocycle condition when pulled back to $\Spec C\otimes _{\mcal O'_{Y,p}}C\otimes _{\mcal O'_{Y,p}}C$, hence defines an element in the cohomology group
		\begin{equation*}
		\H^1\left((\Spec \mcal O'_{Y,p})_{\text{\' et}},\mcal A\right)_+,
		\end{equation*}
		which is trivial since $\mcal A$ is coherent and $\Spec \mcal O'_{Y,p}$ is affine. So there exists an \' etale morphism $C\to D$ and a SUSY derivation $\delta_D$ which is isomorphic to the pull bcak of $\delta_C$ to $\Spec D$, such that $f^*_1\delta_D=f^*_2\delta_D$ where $f_1,f_2$ are two projections from $\Spec D\otimes _{\mcal O'_{Y,p}}D$ to $\Spec C$. Hence $\delta_D$ descends to $\Spec \mcal O'_{Y,p}$.
	\end{enumerate}
	Now we take an element of $\sDef _A(\widehat{\mcal O}_{Y,p},\mcal J)$ (see Proposition \ref{Local_Str_on_Complete_SuperRing}), which induces a deformation of $W_p$ as a superscheme (see equation \eqref{Deformation_Local_to_Global} and notations there), denote the deformation of $W_p$ by $W'_p$, and by the identification $\sDef _A(\mcal O_{Y,p},\mcal J)\overset{\sim}{\rightarrow} \sDef _A(\widehat{\mcal O}_{Y,p},\mcal J)$ there exists a SUSY derivation $\delta_p$ on $\Spec \mcal O_{W'_p,p}$. Extend $\delta_p$ to a Zariski open neighborhood of $p$ and we are done.
	
	\newl \textbf{Step 2.} For smooth open affine $W\subset Y$, $W$ deforms uniquely, since it deforms uniquely as superscheme with puncturing divisor\footnote{The deformation space of smooth superscheme $W$ is $\H^1 (W,T_{W_s})$, by Theorem 2.1.2 of \cite{Vaintrob199001}, which is trival when $W$ is affine.}, and the derivation is also unique up to infinetesimal automorphism, since for any point $p\in W$, $\Spec \widehat{\mcal O}_{W,p}$ deforms uniquely as germ of SUSY curves.
	
	\newl \textbf{Step 3.} Take an open affine cover $\{W_j\}_{j\in J}$ of $Y$ such that $W_j$ contains at most one node or puncture divisor, and $W_j\cap W_{k}$ is smooth and contains no node or puncture divisor, whenever $j\neq k$. Shrinking $W_j$ if necessary, we assume that there exists deformations $W'_j$. $W_j\cap W_{k}$ is smooth hence has a unique deformation. Thus $W'_j\cap W_{k}$ and $W'_k\cap W_j$ are related by an isomorphism $\phi_{jk}$. Consider
	\begin{equation*}
	\phi_{jk}\circ\phi_{kl}\circ \phi_{lj}\in \Gamma(W_j\cap W_{k}\cap W_{l},\mcal A)_+,
	\end{equation*}
	which apparently satisfies cocycle condition, thus it defines an element in $\H^2 (Y,\mcal A)_+$ which is trivial because $\dim Y=1$. This means that after a refinement of $\{W_j\}_{j\in J}$, we can modify $\phi_{jk}$ by elements in $\Gamma(W_j\cap W_{k},\mcal A)_+$ to force it satisfying cocycle condition, thus $\{W'_j\}$ glue to a deformation of $Y$ as a stable SUSY curve.
\end{proof}

In fact, the choice of gluing data in the last step of proof is in one-to-one correspondence with $\H^1(Y,\mcal A)_+=\H^1(Y_{s,\text{bos}},\mcal A)_+$, so we have 
\begin{corollary}
	There is a fiber sequence of sets
	\begin{align}\label{Fiber_Sequence_of_Deformation}
	\H^1(Y_{s,\text{\normalfont bos}},\mcal A)_+\to \sDef _A(Y,\mcal J)\to \prod_{p\in \text{\normalfont NS Nodes}}\sDef _A(\widehat{\mcal O}_{Y,p},\mcal J)\times\prod_{p\in \text{\normalfont R Nodes}}\sDef _A(\widehat{\mcal O}_{Y,p},\mcal J).
	\end{align}
\end{corollary}
When $A=B\times _{\kappa}\kappa [\theta]$ or $A=B\times _{\kappa}\kappa [t]/(t^2)$ and $\mcal J$ is the ideal $(\theta)$ or $(t)$, there is a canonical $\kappa$-vector space structure on $\sDef _A(Y,\mcal J)$ with zero being trivial extension and the fiber sequence \eqref{Fiber_Sequence_of_Deformation} becomes a short exact sequence of vector spaces 
\begin{align*}
0\to \H^1(Y_{s,\text{\normalfont bos}},\mcal A)_+\to \sDef _A(Y,\mcal J)\to \bigoplus_{p\in \text{\normalfont NS Nodes}}(\mcal J_p)_+\oplus\bigoplus_{p\in \text{\normalfont R Nodes}}(\mcal J_p)_+\to 0.
\end{align*}
Therefore, we can use it to compute the dimension of $\sDef _A(Y,\mcal J)$:
\begin{enumerate}
	\item[•] $A=B\times _{\kappa}\kappa [t]/(t^2)$, $\H^1(Y_{s,\text{\normalfont bos}},\mcal A)_+=\H^1(Y_{s,\text{\normalfont bos}},T^+_{Y_{s}}(-E-\mscr D))$, use the normalization in the proof of Corollary \ref{Absence_of_Global_Infinitesimal_Automorphism}, there is isomorphism 
	\begin{equation*}
	\H^1(Y_{s,\text{\normalfont bos}},T^+_{Y_{s}}(-E-\mscr D))\cong\bigoplus_i\H^1(C_i,\omega^{-1}_{C_i}(-E-\mscr D-\mcal R-\mcal N)).
	\end{equation*}
	So the dimension of $\H^1(Y_{s,\text{\normalfont bos}},\mcal A)_+$ can be computed 
	\begin{align*}
	\dim_{\kappa} \H^1(Y_{s,\text{\normalfont bos}},\mcal A)_+&=\sum_i (3g_i-3)+\ns+\ra+2\#(\text{NS Node})+2\#(\text{R Node})\\
	&=3\g-3+\ns+\ra-\#(\text{NS Node})-\#(\text{R Node}).
	\end{align*}
	The spaces $(\mcal J_p)_+$ are one dimensional so the dimension of even deformations can be computed
	\begin{align}\label{Even_Tangent_Space}
	\dim_{\kappa} \sDef _A(Y,\mcal J)=3\g-3+\ns+\ra. 
	\end{align}
	\item[•] $A=B\times _{\kappa}\kappa [\theta]$, $\H^1(Y_{s,\text{\normalfont bos}},\mcal A)_+=\H^1(Y_{s,\text{\normalfont bos}},(\omega^+_{X_{s},p})^{-1}(-E-\mscr D)\otimes \mscr E)$, using the short exact sequence
	\begin{align*}
	0\to \bigoplus_{p\in \text{R Node}}\kappa&\to\H^1(Y_{s,\text{\normalfont bos}},(\omega^+_{X_{s},p})^{-1}(-E-\mscr D)\otimes \mscr E)\\
	&\to\bigoplus_i\H^1(C_i,\omega^{-1}_{C_i}(-E-\mscr D-\mcal R-\mcal N)\otimes \mscr E_{C_i})\to 0,
	\end{align*}
	we see that the dimension of $\H^1(Y_{s,\text{\normalfont bos}},\mcal A)_+$ is
	\begin{alignat}{2}\label{Odd_Tangent_Space}
	\dim_{\kappa} \H^1(Y_{s,\text{\normalfont bos}},\mcal A)_+&=\sum_i (2g_i-2)+\ns+\frac{\ra}{2}+2\#(\text{NS Node})+2\#(\text{R Node})\nonumber
	\\
	&=2\g-2+\ns+\frac{\ra}{2}
	\end{alignat}
	The spaces $(\mcal J_p)_+$ are trivial so the dimension of odd deformations is
	\begin{align*}
	\dim_{\kappa} \sDef _A(Y,\mcal J)=2\g-2+\ns+\frac{\ra}{2}.
	\end{align*}
\end{enumerate}
In these two situations, the restriction maps are actually isomorphisms
\begin{alignat}{2}\label{Schlessinger_Second_Condition}
\sDef_{B\times _{\kappa}\kappa [t]/(t^2)}(Y,(t))&\cong \sDef_{\kappa [t]/(t^2)}(Y_s,(t)),\nonumber
\\ 
\sDef_{B\times _{\kappa}\kappa [\theta]}(Y,(\theta))&\cong \sDef_{\kappa [\theta]}(Y_s,(\theta)).
\end{alignat}
Since the restriction of local deformation parameters $$\sDef_{B\times _{\kappa}\kappa [t]/(t^2)}(\widehat{\mcal O}_{Y,p},(t))\to \sDef_{\kappa [t]/(t^2)}(\widehat{\mcal O}_{Y_s,p},(t))$$ is an isomorphism as can be seen from the Proposition \ref{Local_Str_on_Complete_SuperRing} (same for the other case), and the \v Cech cocycles restricts isomorphically, so the isomorphisms \eqref{Schlessinger_Second_Condition} follows from the Five Lemma.

\subsection{Moduli of Stable $\mscr{N}=1$ SUSY Curves}\label{subsec:moduli of stable SUSY curves}
For simplicity, we will work over base field $\mathbb C$. Analogous to the non-supersymmetric case, $\overline{\mscr{SM}}_{\sg,\sns,\sra}$ is an algebraic superstack, more precisely

\begin{thr}\label{Moduli_of_Stable_SUSY_Curve_is_DM}
	Assume that a triple of natural numbers $(\g,\ns,\ra)$ satisfies $\ra\in 2\mathbb Z_{\ge 0}$, and $2\g-2+\ns+\ra>0$.
	\begin{enumerate}
		\item[$(1)$] $\overline{\mscr{SM}}_{\sg,\sns,\sra}$ is a smooth proper Deligne-Mumford superstack of dimension $$\left(3\g-3+\ns+\ra\big|2\g-2+\ns+\frac{\ra}{2}\right).$$ Moreover its bosonic truncation $\overline{\mscr{SM}}_{\sg,\sns,\sra;\text{\normalfont bos}}$ is the moduli stack $\overline{\mscr{S}}_{\sg,\sns,\sra}$ of punctured stable spin curves.
		
		\item[\hypertarget{second part of the moduli problem}{$(2)$}] $\overline{\mscr{SM}}_{\sg,\sns,\sra}$ has a coarse moduli superspace $\overline{\mathcal{SM}}_{\sg,\sns,\sra}$ which is actually bosonic and is also the coarse moduli space for the bosonic quotient $\overline{\mscr{SM}}_{\sg,\sns,\sra;\text{\normalfont ev}}$. Moreover its reduced subspace $\overline{\mathcal{SM}}_{\sg,\sns,\sra;\text{\normalfont red}}$ is the coarse moduli space $\overline{\mathcal{S}}_{\sg,\sns,\sra}$ of punctured stable spin curves.
	\end{enumerate}
\end{thr}

\begin{proof}
	\textbf{(1):} Proposition \ref{Etale_Descent_Stable_SUSY_Curve} shows that $\overline{\mscr{SM}}_{\sg,\sns,\sra}$ is an \' etale superstack, and Lemma \ref{Lemma_Limit_Stable_SUSY_Curves} implies that $\overline{\mscr{SM}}_{\sg,\sns,\sra}$ is locally finite presented. Moreover, Proposition \ref{Representablity_of_Diagonal} says that the diagonal of $\overline{\mscr{SM}}_{\sg,\sns,\sra}$ is representable and is finite and unramified. Assuming that $\overline{\mscr{SM}}_{\sg,\sns,\sra}$ has an \' etale atlas, then $\overline{\mscr{SM}}_{\sg,\sns,\sra}$ is a Deligne-Mumford superstack. Smoothness follows from Proposition \ref{Absence_of_Obstruction}, and the computation of tangent space is done in equations \eqref{Even_Tangent_Space} and \eqref{Odd_Tangent_Space}. Remark \ref{Remark_Superspace_BosTruncationOfSRS} implies that for any Noetherian scheme $T$ over $\mathbb C$, there is an isomorphism $\Hom_{\CSp}(T,\overline{\mscr{S}}_{\sg,\sns,\sra})\cong\Hom_{\SSp}(S(T),\overline{\mscr{SM}}_{\sg,\sns,\sra})$ which is functorial in $T$, then it follows from Lemma \ref{Superspace_Adjointness} and Yoneda Lemma that
	\begin{equation*}
	\overline{\mscr{SM}}_{\sg,\sns,\sra;\text{\normalfont bos}}=\overline{\mscr{S}}_{\sg,\sns,\sra}.
	\end{equation*}
	Hence the properness of $\overline{\mscr{SM}}_{\sg,\sns,\sra;\text{\normalfont bos}}$ follows from the properness of $\overline{\mscr{S}}_{\sg,\sns,\sra}$ (Theorem \ref{Main}).\\
	
	\textbf{(2):} For a stable SUSY curve $\pi:X\to S$, consider automorphism $\Gamma_X:X\to X$ acting by $1$ on $\mcal O_X^+$ and $-1$ on $\mcal O_X^-$, similarly $\Gamma_S:S\to S$ is defined. They are compatible in the sense that the diagram
	\begin{center}
		\begin{tikzcd}
		X\arrow[r,"\Gamma_X"] \arrow[d,"\pi"] & X\arrow[d,"\pi"]\\
		C \arrow[r,"\Gamma_C"] &C
		\end{tikzcd},
	\end{center}
	commutes. $\Gamma_X$ preserves punctures $E_i$, since $E_i$ is locally in coordinate system $(z|\theta)$ defined by the ideal $(z,\theta)$ which is preserved by $\Gamma_X$. Similarly $\Gamma_X$ preserves punctures $\mscr D_i$. $\Gamma_X$ commutes with derivation $\delta$, since $\delta$ in the smooth canonical model can be written as $f\mapsto (\partial_{\theta}f+\theta\partial_zf)[\mathrm{d}z|\mathrm{d}\theta]$ which is invariant under $\theta\mapsto-\theta$, similar for other three canonical models. Hence we have an embedding $\text{Id}_S\times \Gamma_S:S\to \Isom_{S\times S}(p_1^*X,p_2^*X)$ where $p_1,p_2:S\times S\to S$ are two projections. The the proof of \hyperlink{second part of the moduli problem}{$(2)$} follows from the same argument for the moduli of smooth SUSY curves (Proposition 5.4 and Corollary 5.5 of \cite{CodogniViviani201706}).
\end{proof}

It remains to construct an \' etale atlas of $\overline{\mscr{SM}}_{\sg,\sns,\sra}$, which will be done in the next two sections.

\subsubsection*{Universal Deformation}
For a stable SUSY curve $X$ over $\Spec \mbb C$, define a functor $F_X:\sArt _{\mbb C}\to \Sets$ by sending $A\in \sArt _{\mbb C}$ to the set of equivalence class of $(Y,\phi)$ where $Y$ is a stable SUSY curve over $A$ and $\phi:Y_{\mbb C}\to X$ is an isomorphism, with the obvious equivalence relation. From the definition it is easy to see that $F_X$ satisfies the Schlessinger's condition $(S_0)$; Proposition \ref{Absence_of_Obstruction} implies that $F_X$ satisfies $(S_1)$ and isomorphism \eqref{Schlessinger_Second_Condition} is the condition $(S_1)$. The tangent space of $F_X$ is of dimension
\begin{equation*}
\left(3\g-3+\ns+\ra|2\g-2+\ns+\frac{\ra}{2}\right).
\end{equation*}
So $(S_3)$ does also hold. We claim that condition $(S_4)$ is true for $F_X$ as well. Suppose that $p:A\to A/\mcal J$ is an \textit{even} small extension, and $(Y,\phi)\in F_X(A/\mcal J)$, then $p^{-1}(Y,\phi)$ is the affine space $\sDef_A(Y,\mcal J)$ with a transitive action of $T^+_{F_X}$\footnote{Transitivity comes from the condition $(S_1)$.}. The action is actually faithful, since the affine space 
\begin{equation*}
\Lambda:=\prod_{p\in \text{\normalfont NS Nodes}}\sDef _A(\widehat{\mcal O}_{Y,p},\mcal J)\times\prod_{p\in \text{\normalfont R Nodes}}\sDef _A(\widehat{\mcal O}_{Y,p},\mcal J),
\end{equation*}
has dimension $\#(\text{R Nodes})+\#(\text{NS Nodes})$, so the stabilizer of a given point $\lambda\in\Lambda$ under the induced action of $T^+_{F_X}$ has dimension $\dim_{\mbb C}T^+_{F_X}-\#(\text{R Nodes})-\#(\text{NS Nodes})$, which is exactly $\dim_{\mbb C}\H^1(X,\mcal A)_+$. The action of $\text{Stab}_{\lambda}T^+_{F_X}$ on $\H^1(X,\mcal A)_+$ is thus faithful, and as such the action of $T^+_{F_X}$ on $\sDef_A(Y,\mcal J)$ is also faithful. The odd small extension case is straightforward, since the fiber sequence \eqref{Fiber_Sequence_of_Deformation} degenerates to first two terms thus realizing $\sDef_A(Y,\mcal J)$ as an affine space modeled on $\H^1(X,\mcal A)_+=T^-_{F_X}$. To summarize, we have the following result
\begin{proposition}\label{Universal_Deformation}
	Functor $F_X:\sArt _{\mbb C}\to \Sets$ is isomorphic to the functor $h_R$ where 
	\begin{alignat}{2}
	R&=\mbb C[\![t_1,\cdots,t_{a}|\theta_{1},\cdots,\theta_{b}]\!], \nonumber
	\\ 
	a&=3\g-3+\ns+\ra,\nonumber
	\\
	b&=2\g-2+\ns+\frac{\ra}{2}.
	\end{alignat}
	The isomorphism is realized by a compatible family $\{(Y_n,\phi_n)\in F_X(R/\mathfrak{m}_R^n)\}$.
\end{proposition}
\begin{proof}
	$F_X$ has a universal family by applying Schlessinger's Theorem \ref{Schlessinger's Theorem} to what we have explained. $R$ is formally smooth because the functor $F_X$ is formally smooth (Proposition \ref{Absence_of_Obstruction}), hence $R$ is the formal power series $\mbb C$-superalgebra generated by the dual of $T_{F_X}$.
\end{proof}

\begin{corollary}\label{Formal_Neighbohood_of_Moduli}
	There exists a morphism $\Spec \mbb C[\![t_1,\cdots,t_{a}|\theta_{1},\cdots,\theta_{b}]\!]\to \overline{\mscr{SM}}_{\sg,\sns,\sra}$ whose reduction modulo $\mathfrak{m}_R^n$ is equivalent to $(Y_n,\phi_n)$.
\end{corollary}
\begin{proof}
	Since $X$ is projective over $\Spec \mbb C$ (Lemma \ref{Stable_SUSY_Curve_Is_Locally_Projective}) and the ample line bundle lifts to $Y_n$ in a compatible way for all $n\ge 1$, so we get a compatible system of closed embeddings $Y_n\hookrightarrow \mbb P^{M|N}_{R/\mathfrak{m}_R^n}$ for some integers $M,N$. By Grothendieck's Existence Theorem \ref{Grothendieck's_Existence_Theorem}, there exists a proper $R$-superscheme $Y$ with punctures $E,\mcal D$ and map $\delta:\Omega^1_{Y/R}\to \omega_{Y/R}$ and an isomorphism $\phi:Y_{\mbb C}\cong X$, and the reduction of $(Y,\phi)$ to $R/\mathfrak{m}_R^n$ is equivalent to $(Y_n,\phi_n)$. The flatness of $Y$ is a consequence of local flatness criterion.
\end{proof}

Note that the restriction of the universal family $\{(Y_n,\phi_n)\in F_X(R/\mathfrak{m}_R^n)\}$ to $\Art_{\mbb C}$ (category of Artinian bosonic $\mbb C$-algebras) is a universal family for the deformations of $X$ as spin curves, so the restriction of the morphism $\Spec \mbb C[\![t_1,\cdots,t_{a}|\theta_{1},\cdots,\theta_{b}]\!]\to \overline{\mscr{SM}}_{\sg,\sns,\sra}$ to $\Spec \mbb C[\![t_1,\cdots,t_{a}]\!]$ factors through some affine \' etale chart $U_0$ of $\overline{\mscr{S}}_{\sg,\sns,\sra}$, which induces an isomorphism between $\mbb C[\![t_1,\cdots,t_{a}]\!]$ and the completion of $\mcal O_{U_0}$ at a closed point $p$. Let $U$ be the superscheme $(U_0,\wedge^{\bullet}\mcal O^b_{U_0})$, i.e. formally adding independent odd variables $\{\theta_1,\cdots,\theta_n\}$ to the coordinate ring of $U_0$, so there is an isomorphism $\widehat{\mcal O}_{U,p}\cong \mbb C[\![t_1,\cdots,t_{a}|\theta_{1},\cdots,\theta_{b}]\!]$. Combining Lemma \ref{Lemma_Artin_Approximation} and Lemma \ref{Lemma_Limit_Stable_SUSY_Curves}, we obtain a superscheme $V$ smooth over $U$ with a morphism $\Spec \mbb C[\![t_1,\cdots,t_{a}|\theta_{1},\cdots,\theta_{b}]\!]\to V$ and a stable SUSY curve $Y_V$ over $V$ such that the universal family $Y$ is isomorphic to the pull back of $Y_V$. Denote the image of the close point of $\Spec \mbb C[\![t_1,\cdots,t_{a}|\theta_{1},\cdots,\theta_{b}]\!]$ in $V$ by $q$, consider the the embeddings
\begin{center}
	\begin{tikzcd}
	\Spec \mbb C[\![t_1,\cdots,t_{a}|\theta_{1},\cdots,\theta_{b}]\!]/\mathfrak{m}^2 \arrow[r]\arrow[d,hook]& \Spec \mcal O_V(V)/\mathfrak{m}_q^2 \arrow[r] \arrow[d,hook] & \Spec \mcal O_U(U)/\mathfrak{m}_p^2 \arrow[d,hook]\arrow[ld,dotted]\\
	\Spec \mbb C[\![t_1,\cdots,t_{a}|\theta_{1},\cdots,\theta_{b}]\!] \arrow[r] & V \arrow[r] & U
	\end{tikzcd},
\end{center}
defined by modulo $\mathfrak{m}^2$. Since the composition in the top row is an isomorphism, this gives an embedding of $\Spec \mcal O_U(U)/\mathfrak{m}_p^2$ into $\Spec \mcal O_V(V)/\mathfrak{m}_q^2$. Extending the image of this embedding to a neighborhood of $q$, we obtain an locally closed sub superscheme $W\subset V$ and $W$ is \' etale over $U$. The restriction of $Y_V$ to $W$ induces a homomophism $\mbb C[\![t_1,\cdots,t_{a}|\theta_{1},\cdots,\theta_{b}]\!]\to \widehat{\mcal O}_{W,q}$ which is an isomorphism modulo $\mathfrak{m}^2$, thus it is an isomorphism because any automoprhism of $\mbb C[\![t_1,\cdots,t_{a}|\theta_{1},\cdots,\theta_{b}]\!]$ which induces isomorphism on tangent space is an isomorphism. To summarize, we have the following
\begin{proposition}\label{Neighborhood_of_Moduli}
	For every $p\in \overline{\mscr{SM}}_{\sg,\sns,\sra}(\mbb C)$, there exists an affine smooth $\mbb C$-variety $W$ with an element $P\in \overline{\mscr{SM}}_{\sg,\sns,\sra}(W)$, and a closed point $q\in W$ such that $P|_q$ is isomorphic to $p$. Furthermore,
	\begin{enumerate}
		\item[$(1)$] The induced map $P_{\bos}:W_{\bos}\to \overline{\mscr{S}}_{\sg,\sns,\sra}$ is \' etale.
		\item[$(2)$] The induced map $\widehat{P}:\widehat{\mcal O}_{W,q}\to \overline{\mscr{SM}}_{\sg,\sns,\sra}$ identifies $\widehat{\mcal O}_{W,q}$ with the base of universal deformation of $p$.
	\end{enumerate}
\end{proposition}
The next goal is to prove that $P:W\to \overline{\mscr{SM}}_{\sg,\sns,\sra}$ is \' etale (shrink $W$ if necessary), for which we will need more ingredients, discussed in the next section.

\subsubsection*{The Odd Tangent Space}
Suppose that $S$ is a affine smooth $\mbb C$-variety of positive dimension and $\pi:X\to S$ is a stable SUSY curve and is smooth on a open subscheme $S_0\subset S$. Denote the underlying spin curve of $X$ by $(\mcal C,E,\mcal D,\mscr E)$. Let $\mcal J$ be a finite free sheaf on $S$ and let the trivial extension of $S$ by $\mcal J$ be $S'$. Assume that stable SUSY curve $\pi':X'\to S'$ is a lift of $X\to S$, then we have
\begin{lem}
	The sheaf $\mcal A$ of automorphisms of $X'\to S'$ which induce identity on $X\to S$ is 
	\begin{equation*}
	\omega^{-1}_{\mcal C/S}(-E-\mscr D)\otimes \mscr E\otimes \pi^*\mcal J\cong \mscr E^{\vee}\otimes \pi^*\mcal J.
	\end{equation*}
\end{lem}

\begin{proof}
	On the smooth locus $X^{\text{sm}}$, it is well-known that (see Sections 4.2.1 and 4.2.3 of \cite{Witten2012b}) $\mcal A=T_{\mcal C/S}(-E-\mscr D)\otimes \mscr E\otimes \pi^*\mcal J\cong \omega^{-1}_{\mcal C/S}(-E-\mscr D)\otimes \mscr E\otimes \pi^*\mcal J$. It remains to determine the local structure of $\mcal A$ at nodes. Let the closed subscheme of R node (resp. NS node) be $Z_{\text R}$ (resp. $Z_{\text NS}$). Suppose that $p\in X$ is closed and is a R node, then by the same argument in the proof of Lemma \ref{Lemma_Extension_of_Derivation}, we see that in an open affine neighborhood $W_p$ of $p$, an element in $\Gamma(W_p,\mcal A)$ is determined by derivation of the form $yD^2$, which is a rational section of $T_{\mcal C/S}\otimes\mscr E\otimes \pi^*\mcal J$ on the smooth locus $W_p\setminus Z_{\text R}$. Notice that $W_p$ is normal, $Z_{\text R}$ has codimension at least 2, and $\omega^{-1}_{\mcal C/S}\otimes\mscr E\otimes \pi^*\mcal J$ is reflexive and its restriction to $W_p\setminus Z_{\text R}$ is ismorphic to $T_{\mcal C/S}\otimes \mscr E\otimes \pi^*\mcal J$, so 
	\begin{equation*}
	\Gamma(W_p,\mcal A)=\Gamma(W_p,\omega^{-1}_{\mcal C/S}\otimes\mscr E\otimes \pi^*\mcal J),
	\end{equation*}
	by algebraic Hartogs's Theorem. The similar statement is also true for an open neighborhood of a NS node, whence the result follows by gluing.
\end{proof}

The object we will concern about is the space of deformations $\sDef _{S'}(X,\mcal J)$. Note that if $\mcal J\to \mcal K$ is a homomorphism of free $\mcal O_S$ sheaves, then there is a canonical map $\sDef _{S'}(X,\mcal J)\to \sDef _{S'}(X,\mcal K)$ defined by pull back of liftings of $X$. This endows $\sDef _{S'}(X,\mcal J)$ with a canonical $\mcal O_S$-module structure via
\begin{enumerate}
	\item[•] Addition: Induced by the addition map $\mcal J\times \mcal J\to \mcal J,(x,y)\mapsto x+y$.
	\item[•] Zero element: Induced by the trivial map $0\to \mcal J$.
	\item[•] Inverse: Induced by inverse map $\mcal J\to \mcal J,x\mapsto -x$.
	\item[•] $\mcal O_S$-action: For any element $f\in \mcal O_S(S)$, the action of $f$ on $\sDef _{S'}(X,\mcal J)$ is induced by $\mcal J\to \mcal J,x\mapsto f\cdot x$.
\end{enumerate}
The goal of this section is to determine the $\mcal O_S$-module structure of $\sDef _{S'}(X,\mcal J)$. Indeed, we have 
\begin{proposition}\label{The_Odd_Tangent_Sheaf}
	$\sDef _{S'}(X,\mcal J)$ is isomorphic to 
	\begin{equation*}
	\H^1(\mscr C,\mscr E^{\vee})\otimes_{\mcal O_S}\mcal J.
	\end{equation*}
	Moreover, the isomorphism is functorial in $\mcal J$ in the sense that for a homomorphism $\mcal J\to \mcal K$ of free $\mcal O_S$ sheaves, the diagram
	\begin{center}
		\begin{tikzcd}
		\sDef _{S'}(X,\mcal J)\arrow[r] \arrow[d] & \sDef _{S'}(X,\mcal K)\arrow[d]\\
		\H^1(\mcal C,\mscr E^{\vee})\otimes_{\mcal O_S}\mcal J\arrow[r] &\H^1(\mcal C,\mscr E^{\vee})\otimes_{\mcal O_S}\mcal K
		\end{tikzcd},
	\end{center}
	commutes.
\end{proposition}
\begin{proof}
	The proof basically follows the same idea of Proposition \ref{Absence_of_Obstruction}. Again we have $\sDef _A(\mcal O_{Y,p},\mcal J)\overset {\sim}{\rightarrow} \sDef _A(\widehat{\mcal O}_{Y,p},\mcal J)$ and the latter is a one element set, since the odd deformation is trivial for all canonical models. So the same gluing argument shows that $\sDef _{S'}(X,\mcal J)$ is a torsor under the action of $\H^1(\mscr C,\mscr E^{\vee})\otimes_{\mcal O_S}\mcal J$. It has a distinguished element which is the trivial deformation, so we have a canonical isomorphism of sets $\H^1(\mscr C,\mscr E^{\vee})\otimes_{\mcal O_S}\mcal J\overset{\sim}{\rightarrow}\sDef _{S'}(X,\mcal J)$. It is easy to see that this isomorphism commutes with the Abelian group structure and $\mcal O_S$-action, since the deformations are charaterized by a \v Cech cochain valued in $\mscr E^{\vee}\otimes \pi^*\mcal J$, whose groups structures and $\mcal O_S$-action are exactly given by the description above. Now the functoriality in $\mcal J$ comes from the functoriality of \v Cech cochains.
\end{proof}
Observe that $\mscr E^{\vee}$ is flat over $S$ and its restrictions to fibers have cohomology concentrated in degree $1$, so $\H^1(\mscr C,\mscr E^{\vee})$ is a locally-free $\mcal O_S$-module. From the \v Cech cochain description of $\sDef _{S'}(X,\mcal J)$, we see that 
\begin{corollary}\label{Restriction_Functoriality}
	The isomorphism $\sDef _{S'}(X,\mcal J)\overset {\sim}{\rightarrow}\H^1(\mscr C,\mscr E^{\vee})\otimes_{\mcal O_S}\mcal J$ commutes with restriction to fiber map, i.e. the diagram
	\begin{center}
		\begin{tikzcd}
		\sDef _{S'}(X,\mcal J)\arrow[r,"\text{\normalfont res}"] \arrow[d] & \sDef _{s'}(X_s,\mcal J(s))\arrow[d]\\
		\H^1(\mcal C,\mscr E^{\vee})\otimes_{\mcal O_S}\mcal J\arrow[r,"\text{\normalfont res}"] &\H^1(\mcal C_s,\mscr E^{\vee}_s)\otimes_{\mbb C}\mcal J(s)
		\end{tikzcd},
	\end{center}
	commutes, where $s\in S$ is a closed point, $s'$ is the superscheme constructed from $s$ trvially extended by the fiber $\mcal J(s)$.
\end{corollary}

\subsubsection*{An \' Etale Atlas of $\overline{\mscr{SM}}_{\sg,\sns,\sra}$}
\begin{lem}\label{Openness_of_Versality}
	Keep using the same notation as {\normalfont Proposition \ref{Neighborhood_of_Moduli}}, then there exists an open affine neighborhood $S$ of $q$ in $W$ such that $\forall s\in S$, the completion of $P$ at $s$ is a universal deformation family of $P|_s$.
\end{lem}
\begin{proof}
	From the construction of $W$, we already know that the pullback of $P$ to $\widehat{\mcal O}_{W,s,\bos}$ is a universal deformation family of $P|_s$ as spin curves (since $W\to \overline{\mscr{S}}_{\sg,\sns,\sra}$ is \' etale). To show that the pullback of $P$ to  $\widehat{\mcal O}_{W,s}$ is a universal deformation family of $P|_s$ as stable SUSY curves, it suffices to show that the homomorphism $\mbb C[\![t_1,\cdots,t_{a}|\theta_{1},\cdots,\theta_{b}]\!]\to \widehat{\mcal O}_{W,s}$ induced by $\widehat{P}_s$ is an isomorphism on tangent spaces. The even tangent spaces are obviously isomorphic, so it remains to show that this induces isomorphism on odd tangent spaces.\\
	
	Consider the restriction of $P$ to $W_1\equiv\Spec \mcal O_W(W)/\mathfrak{n}_W^2$, it gives rise to an element 
	\begin{equation*}
	[P]\in\sDef_{W_1}(P_{W_{0}},\mcal O_{W_{0}}^{\oplus b}),
	\end{equation*}
	where $W_0=W_{\bos}$ is the bosonic trunction. Hence by Proposition \ref{The_Odd_Tangent_Sheaf}, it corresponds to an element in $\H^1(\mscr C,\mscr E^{\vee})\otimes _{\mcal O_{W_0}}\mcal O_{W_0}^{\oplus b}$, or equivalently 
	\begin{equation*}
	[P]:\H^1(\mscr C,\mscr E^{\vee})^{\vee}\to \mcal O_{W_0}^{\oplus b},
	\end{equation*}
	where $\mscr C$ is the spin curve undelying $P|_{W_0}$ with spinor sheaf $\mscr E$. Because the restriction commutes with the identification between $\sDef$ and $\H^1$ (Corollary \ref{Restriction_Functoriality}), the restriction of the homomorphism $[P]$ to the fiber $q$ is the homomorphism 
	\begin{equation*}
	[P_q]:\H^1(\mscr C_q,\mscr E_q^{\vee})^{\vee}\to \mbb C^{\oplus b},
	\end{equation*}
	where the vector space $\mbb C^{\oplus b}$ is spanned by $\{\theta_i\}$, i.e. the cotangent space of the universal deformation base of $P_q$, thus $[P_q]$ is isomorphism. It follows that there exists an open affine neighborhood $S$ of $q$ in $W$ such that $[P]$ is an isomorphism on $S_{\bos}$. Then $\forall s\in S$, the homomoprhism $[P_s]:\H^1(\mscr C_s,\mscr E_s^{\vee})^{\vee}\to \mbb C^{\oplus b}$ is isomorphism, whence $\widehat{P}_s$ induces an isomorphism between the universal deformation base of $P|_s$ and $\widehat{\mcal O}_{S,s}$.
\end{proof}

\begin{proposition}\label{Etale_Covering}
	The morphism $S\to \overline{\mscr{SM}}_{\sg,\sns,\sra}$ in {\normalfont Lemma \ref{Openness_of_Versality}} is representable and \' etale.
\end{proposition}
\begin{proof}
	Representability of $S\to \overline{\mscr{SM}}_{\sg,\sns,\sra}$ is a consequence of representability of the diagonal of $\overline{\mscr{SM}}_{\sg,\sns,\sra}$ (Proposition \ref{Representablity_of_Diagonal}). To show that it is \' etale is equivalent to show that for every small extension $p:A\to A/\mcal J$ of local Artinian $\mbb C$-superalgebras, and for every commutative diagram 
	\begin{center}
		\begin{tikzcd}
		\Spec A/\mcal J\arrow[r,"f"] \arrow[d,"i",hook] & S\arrow[d]\\
		\Spec A \arrow[r,"g"] & \overline{\mscr{SM}}_{\sg,\sns,\sra}
		\end{tikzcd},
	\end{center}
	there exists a unique $\widetilde f:\Spec A \to S$ making the diagram still commutative. Let the image of the closed point of $\Spec A/\mcal J$ under $f$ be $s\in S$, then the pull back of stable SUSY curve $P$ over $S$ to $\Spec A/\mcal J$, denoted by $P_{A/\mcal J}$, extends to a stable SUSY curve $P_{A}$ over $\Spec A$, where $P_{A}$ is the defined by the morphism $g$. Since $\widehat{P}_s$ over $\widehat{\mcal O}_{S,s}$ is a universal deformation, there exists a unique homomorphism $\widetilde{f}:\Spec A \to \Spec \widehat{\mcal O}_{S,s}$ which extends $f:\Spec A/\mcal J \to \Spec \widehat{\mcal O}_{S,s}$ such that $\widetilde{f}^*P\cong P_{A}$.
\end{proof}

Now for every object $p\in \overline{\mscr{SM}}_{\sg,\sns,\sra}(\mbb C)$, we can construct such $S\to \overline{\mscr{SM}}_{\sg,\sns,\sra}$ in the Proposition \ref{Etale_Covering} that $p$ is equivalent to $q\in S(\mbb C)$, thus we obtained an \' etale covering of $\overline{\mscr{SM}}_{\sg,\sns,\sra}$ and concludes the proof of Theorem \ref{Moduli_of_Stable_SUSY_Curve_is_DM}.

\subsubsection{Gluing of Punctures}
Similar to the spin curves, we can define a set of gluing operations functorially for stable SUSY curves. For NS punctures, this is easy, since NS punctures are sections of the base superscheme, so we can use this to define a gluing diractly. More precisely, suppose that $X_1$ and $X_2$ are stable SUSY curves over a base superscheme $S$, and that $\ns^1$ and $\ns^2$ are nonzero, then we can pick one NS puncture from each curve, say that the $i_1$'th NS puncutre $E_{i_1}$ on $X_1$ and $i_2$'th NS puncture $E_{i_2}$ on $X_2$. The projection $E_{i_1}\to S$ is an isomorphism, same for $E_{i_2}$, so we can use these isomorphims to define a canonical identification between $\varphi:E_{i_1}\cong E_{i_2}$, and glue them to get a new superscheme $X=X_1\underset{E_{i_1}}{\cup_{\varphi}} X_2$. Note that there are two short exact sequence of sheaves 
\begin{center}
	\begin{tikzcd}
	0\arrow[r] &\mcal O_X \arrow[r]& p_{1*}\mcal O_{X_1}\oplus p_{2*}\mcal O_{X_2} \arrow[d,"\delta_1\oplus\delta_2"] \arrow[r] &\mcal O_S \arrow[r] &0 &\\
	&0\arrow[r] &p_{1*}\omega_{X_1/S}\oplus p_{2*}\omega_{X_2/S} \arrow[r] & \omega_{X/S} \arrow[r] & \mcal O_{S}\arrow[r] &0
	\end{tikzcd},
\end{center}
where $p_i:X_i\to X$ are natural projections. So we can define the derivation $\delta$ on $X$ as the composition 
\begin{equation*}
\mcal O_X \to p_{1*}\mcal O_{X_1}\oplus p_{2*}\mcal O_{X_2}\overset{\delta_1\oplus\delta_2}{\longrightarrow}p_{1*}\omega_{X_1/S}\oplus p_{2*}\omega_{X_1} \to \omega_{X/S}.
\end{equation*}
This construction is functorial hence giving rise to an operation on corresponding moduli spaces:
\begin{equation*}
\mathfrak m^{\text{NS}} _{i_1,i_2}:\overline{\mscr{SM}}_{\sg_1,\sns^1,\sra^1}\times \overline{\mscr{SM}}_{\sg_2,\sns^2,\sra^2}\to \partial\overline{\mscr{SM}}_{\sg_1+\sg_2,\sns^1+\sns^2-2,\sra^1+\sra^2}.
\end{equation*}
This can be defined for a single stable SUSY curve with at least two NS puncutres as well:
\begin{equation*}
\mathfrak g^{\text{NS}}_{i,j}:\overline{\mscr{SM}}_{\sg,\sns,\sra}\to \partial\overline{\mscr{SM}}_{\sg+1,\sns-2,\sra}.
\end{equation*}
Note that these maps depends on the choice of $i_1$, $i_2$, $i$, and $j$, in an equivariant way. For example, if $\sigma_{i_1,i_1'}\in S_{\sns^1}$ switches $i_1$ and $i_1'$ (it acts $\overline{\mscr{SM}}_{\sg_1,\sns^1,\sra^1}$ naturally), then we have $\mathfrak m _{i_1',i_2}=\mathfrak m^{\text{NS}} _{i_1,i_2}\circ \sigma_{i_1,i_1'}$.

\newl For gluing R divisors, the $\mbb Z/2$-ambiguity is lifted to the ambiguity of choosing an isomorphism between the structure sheaf of two R divisors. The solution is also similar to the spin curve case: we add a canonical isomorphism {\it by hand}. Suppose that $X_1$ and $X_2$ are stable SUSY curves over a base superscheme $S$, and that $\ra^1$ and $\ra^2$ are nonzero, we can pick one R divisor from each curve, $\iota_1:\mscr D_{i_1}\subset X_1$ and $\iota_2:\mscr D_{i_2}\subset X_2$. Note that $\mscr D_i$ are smooth over $S$ of relative dimension $(0|1)$. The derivations $\delta_1,\delta_2$ induce $\mcal O_S$-linear derivations 
\begin{alignat*}{2}
d_1&:\mcal O_{\mscr D_{i_1}}{\rightarrow}\omega_{\mscr D_{i_1}/S}, 
\\
d_2&:\mcal O_{\mscr D_{i_2}}{\rightarrow}\omega_{\mscr D_{i_2}/S},
\end{alignat*}
as can be easily seen from a local computation. Consider the superscheme $$G_{i_1,i_2}\subset \Isom_S(\mscr D_{i_1},\mscr D_{i_2}),$$ parametrizing isomorphisms from $\mscr D_{i_1}$ to $\mscr D_{i_2}$ which commutes with $d_i$. Locally on $S$, after choosing coordinate system of $\mscr D_{i_1},\mscr D_{i_2}$, denoted by $\theta_1,\theta_2$, we can assume that $d_1,d_2$ can be written as 
$$d_i:\theta_i\mapsto [1/\mathrm{d}\theta_i],\quad 1\mapsto 0,$$
so for any $S$-superscheme $T$, $G_{i_1,i_2}(T)$ is locally generated by homomorphims of form 
\begin{equation*}
\theta_2\mapsto\{\pm\theta_1+a|a\in \mcal O_T^-\}.
\end{equation*}
Therefore, $G_{i_1,i_2}$ is locally on $S$ isomorphic to two copies of $\mbb A^{0|1}$. Note that $G_{i_1,i_2,\bos}$ on $S_{\bos}$ is exactly the $G_{i_1,i_2}$ we constructed in \ref{app: gluing marked points}.

\newl Base change to the superscheme $G_{i_1,i_2}$, there is a universal isomorphism $\alpha:\mscr D_{i_1}\cong\mscr D_{i_2}$ which commutes with $d_i$, so we can take the gluing of $X_1\times _S G_{i_1,i_2}$ and $X_2\times _S G_{i_1,i_2}$ to be $X=X_1\underset{\mscr D_{i_1}}{\cup_{\alpha}} X_2$. Note that there are two short exact sequence of sheaves connected by derivations:
\begin{center}
	\begin{tikzcd}
	0\arrow[r] &\mcal O_X \arrow[r]& p_{1*}\mcal O_{X_1}\oplus p_{2*}\mcal O_{X_2} \arrow[d,"\delta_1\oplus\delta_2"] \arrow[r,"\alpha+\text{Id}"] &\mcal O_{\mscr D_{i_2}} \arrow[r]\arrow[d,"d_2"] &0 \\
	0\arrow[r] &\omega _{X/S} \arrow[r] & p_{1*}\omega _{X_1/S}\oplus p_{2*}\omega _{X_2/S} \arrow[r,"\alpha+\text{Id}"] & \omega _{\mscr D_{i_2}/S}\arrow[r] &0
	\end{tikzcd},
\end{center}
such that the diagram commutes. So there is a derivation $\delta:\mcal O_X\to \omega_{X/S}$ and this gives $X$ a structure of stable SUSY curve. The construction is also functorial, hence giving rise to an operation on corresponding moduli spaces:
\begin{center}
	\begin{tikzcd}
	\mscr{G}_{i_1,i_2}\arrow[r,"\mathfrak m^{\text{R}} _{i_1,i_2}"] \arrow[d,"\pi_{i_1,i_2}" swap] &  \partial\overline{\mscr{SM}}_{\sg_1+\sg_2,\sns^1+\sns^2,\sra^1+\sra^2-2}\\
	\overline{\mscr{SM}}_{\sg_1,\sns^1,\sra^1}\times \overline{\mscr{SM}}_{\sg_2,\sns^2,\sra^2}
	\end{tikzcd}.
\end{center}
This can be defined for a single punctured SUSY curve with at least two NS puncturing divisors as well:
\begin{center}
	\begin{tikzcd}
	\mscr{G}_{i,j}\arrow[r,"\mathfrak g^{\text{R}} _{i,j}"] \arrow[d,"\pi_{i,j}" swap] &  \partial\overline{\mscr{SM}}_{\sg+1,\sns,\sra-2}\\
	\overline{\mscr{SM}}_{\sg,\sns,\sra}.
	\end{tikzcd}.
\end{center}
Note that $\pi_{i_1,i_2}$ and $\pi_{i,j}$ are locally isomorphic to a $\mbb Z/2\ltimes \mbb A^{0|1}$-bundle. Similar to the gluing of NS punctures, these maps depends on the choice of $i_1,i_2$ and $i,j$ in an equivariant way.

\section{Discussion, and Future Directions}\label{subsec:conclusion}
In this las section, we would like to point out few interesting research directions
\begin{itemize}
	\item As we have shown that genus-$\g$ heterotic-string and type-II-superstring field theory vertices exists for any type of boundary components.
	
	\newl It would be interesting to explicitly construct these vertices. An attempt to construct type-II-superstring field theory vertices using a generalization of the minimal-area problem of ordinary Riemann surfaces \cite{Zwiebach199005,Zwiebach199012,Zwiebach199102a,Zwiebach199102b,Zwiebach199111,Zwiebach199206,HeadrickZwiebach1806a,HeadrickZwiebach1806b,NaseerZwiebach201903} to the case of $\mscr{N}=1$ super-Riemann surface is considered in\footnote{We would like to thank Ted Erler who pointed us to this paper.} \cite{JurcoMunster1303}. 
	
	\item We also proved that the solution to the BV QME in $\mcal{F}(\mscr{S})$ are unique up to homotopies in the homotopy category of BV algebras. 
	
	\newl It is shown by Costello that the space of homotopy-equivalent solutions to the BV QME in the moduli of bordered ordinary Riemann surfaces is a point \cite{Costello200509}. This together with \cite{HataZwiebach199301} shows that the gauge-fixed action of bosonic-string field theory is unique. For another investigation in proving the uniqueness of bosonic-string field theory see \cite{MunsterSachs1208}.   
	
	\newl As we have explained in \hyperlink{choice:string vertex}{\small\bf a choice of string vertex}, the physics of string field theory should not be dependent on the choice of string vertices. On the other hand, we have proven the existence of string vertices by defining a quasi-isomorphism to the BV algebra of the complexes associated to the bordered spin-Riemann surfaces, and the solution to the BV QME in the latter BV algebra is unique up to homotopy. It would be very interesting to find the space of homotopy-equivalent solution to the BV QME that give rises to string vertices. This shed some light on the question of dependence of resulting string field theory on the choice of string vertex. One possibility is the following. Assuming that the space of homotopy-equivalent solutions of the BV QME that give rise to string vertices has a trivial mapping-class group (the group of large diffeomorphisms), then any two solutions of the BV quantum master equation can be related by a sequence of small deformations of one of them. Then, a result similar to \cite{HataZwiebach199301} for the heterotic-string and type-II superstring field theories would show that the two gauge-fixed actions constructed using the two string vertices are the same upon using two different gauge-fixing conditions. 
	
	\item  We have proven the existence of type-II-superstring vertices. Two related questions are as follows
	\begin{enumerate}
		\item In the presence of D-branes, we need to deal with surfaces with boundaries with both open-string punctures on the boundaries and closed-string punctures on the interior of the surface \cite{Zwiebach199802,MoosavianSenVerma201907}. One way to deal with such surfaces is to consider the double of the surface. We refer to section 7 of \cite{Witten2012b} for details. So, one might be able to prove the existence of type-II-superstring vertices in the presence of D-branes by considering the doubled surface. As we have shown, there is always a solution to the BV QME in the BV algebra of complexes associated to moduli space of such or corresponding bordered surfaces. 
		
		\item Another interesting and relevant direction is the case of unoriented string field theory \cite{DeWolfe199708,MoosavianSenVerma201907}. Since this theory is the orientifold of type-IIB-superstring theory, and the proof of the existence of the string vertex for the latter has been given here, it would be interesting to prove the existence of string vertices for unoriented string field theory as well.
	\end{enumerate}
	\end{itemize}
We leave these and related questions for future works.

{\vskip .8 cm}
\noindent {\bf Acknowledgments} It is our pleasure to thank Kevin Costello whose works and ideas shaped the main motivation behind this work. We also thank him for discussion and guidance on this work. A special gratitude goes to Ted Erler who provided detailed and critical comments on various parts of an earlier version of this manuscript. We also thank Ivo Sachs for useful discussions. SFM is grateful to Davide Gaiotto for his support during his stay at Perimeter Institute. Research at Perimeter Institute is supported by the Government of Canada through Industry Canada and by the Province of Ontario through the Ministry of Research and Innovation.

\appendix

\section{The Picture-Changing Operation}\label{app:the picture-changing operation}

In this appendix, we define the picture-changing operation and explain its various incarnations. We also explain its relation to the natural operation defined on space of form on a supermanifold. Following \cite{VerlindeVerlinde1987a, Belopolsky1997b, Belopolsky1997c}, we then argue that all these incarnations are different manifestations of the same concept.

\subsection{The Canonical Framework: Friedan-Shenker-Martinec Apprach}

In this section, we explain the picture-changing operation in the Friedan-Shenker-Martinec (FSM) formulation of covariant superstring perturbation theory \cite{FriedanMartinecShenker1986}.

\newl The fixing of superconformal invariance of the worldsheet theory of RNS-superstring theory will give rise to two sets of antighost-ghost systems 
\begin{enumerate}
	\item The fermionic reparametrization ghosts corresponding to a $bc$ system with fields having conformal dimensions $\Delta_b=2$ and $\Delta_c=-1$.
	\item The bosonic local supersymmetry ghosts corresponding to a $\beta\gamma$ system with fields having conformal dimensions $\Delta_\beta=3/2$ and $\Delta_\gamma=-1/2$.
\end{enumerate}
The $\beta\gamma$ system has a first-order action and as such there is no ground state for the system \footnote{This can be seen as follows. The action of the mode $\gamma_{\frac{1}{2}}$ of $\gamma(z)$ field on the $\mtt{SL}(2,\mbb{C})$-invariant vacuum $|0\rangle$ is not zero. Since $\gamma_{\frac{1}{2}}$ is a bosonic mode, it can be applied to $|0\rangle$ arbitrarily. However, this action reduced the $L_0$ eigenvalue, i.e. the energy, of the resulting state by $\frac{1}{2}$. In this way one can construct states with arbitrary negative energy and as a result there is no ground state for the system. To cure this problem, one has to construct a state that is annihilated by all of the positive modes. For more details, see section 5.2 of \cite{FriedanMartinecShenker1986}.}. The resolution of the problem was found in \cite{FriedanMartinecShenker1986}. The idea is that instead of $\beta\gamma$ system, one considers its bosonization in terms of the bosonic field $\phi$ and the fermionic fields $\eta$ and $\xi$ using the following identifications
\begin{equation}\label{eq:bosonization formulas}
\beta(z)\mapsto \partial_z\xi(z)e^{-\phi(z)},\qquad \qquad \gamma(z)\mapsto\eta(z)e^{\phi(z)}. 
\end{equation}
One can then define a family of states parametrized by $\phi$-charge $n$ of the corresponding state 
\begin{equation}\label{eq:state with picture number n}
|n\rangle \equiv \lim\limits_{z\longrightarrow 0}e^{n\phi(z)}|0\rangle,
\end{equation}
where $|0\rangle$ is the $\mtt{SL}(2,\mbb{C})$-invariant vacuum. $n$ is called the {\it picture number} \cite{NeveuSchwarzThorn1971a,BrinkOliveRebbiScherk1973a}. The action of $\beta$ and $\gamma$ does not change $n$. This means that $\{|n\rangle\}_{n\in\mbb{Z}}$ can be used to construct representations of the algebra of modes of $\beta$ and $\gamma$ fields, since they have picture number zero. Representations constructed using different $n$ are inequivalent. This is a crucial property of the bosonic nature of the $\beta\gamma$ system. The choice of a picture is thus a choice of vacuum for the $\beta\gamma$ system \cite{FriedanMartinecShenker1986}. 

\newl  One can expand these fields as 
\begin{equation}
\beta(z)=\sum\limits_{k}\frac{\beta_k}{z^{k+\frac{3}{2}}}, \qquad \qquad 
\gamma(z)=\sum\limits_{k}\frac{\gamma_k}{z^{k-\frac{1}{2}}},
\end{equation}
where $k\in\frac{1}{2}+\mbb{Z}$ (for NS sector) or $k\in\mbb{Z}$ (for R sector). From these, we conclude that
\begin{equation}
\beta_k=\bigointsss\frac{dz}{2\pi\mfk{i}}\,z^{k+\frac{1}{2}}\beta(z), \qquad \qquad  \gamma_k=\bigointsss\frac{dz}{2\pi\mfk{i}}\,z^{k-\frac{3}{2}}\gamma(z). 
\end{equation}
These modes satisfy the commutation relation 
\begin{equation}
[\gamma_m,\beta_n]=\delta_{m,-n}.
\end{equation}
The other commutators are zero. There is a ghost current $\mscr{Q}_{\text{gh}}\equiv \sum_n\boldsymbol{:}\beta_n\gamma_{-n}\boldsymbol{:}$, where $\boldsymbol{:}\cdot \boldsymbol{:}$ denotes the normal ordering. The ghost-charges of modes are
\begin{equation}
[\mscr{Q}_{\text{gh}},\beta_n]=+n\beta_n, \qquad [\mscr{Q}_{\text{gh}},\gamma_n]=-n\gamma_n.
\end{equation} 
We can also introduce the operators that interpolate between the ghost-see levels. These operators are formally written as $\delta(\beta_n)$ and $\delta(\gamma_n)$. The action of these operators on a vacuum with picture number $n$ is defined via \cite{VerlindeVerlinde198804}
\begin{equation}
\delta(\beta_{-n-\frac{3}{2}})|n\rangle=|n+1\rangle, \qquad  \delta(\gamma_{n+\frac{1}{2}})|n\rangle=|n-1\rangle.
\end{equation}
We can thus assign picture-numbers $+1$ to $\delta(\beta_n)$ and $-1$ to $\delta(\gamma_n)$.  

\newl To determine the picture number of the NS and R vacua, we consider the state \eqref{eq:state with picture number n} and act with the modes of $\beta$ and $\gamma$ fields. The action of the mode $\beta_k$ on the state \eqref{eq:state with picture number n} gives
\begin{alignat}{2}
\beta_k|n\rangle&\sim\bigointsss\frac{dz}{2\pi\mfk{i}}\, z^{k+\frac{1}{2}}\beta(z)e^{n\phi(0)}=\bigointsss\frac{dz}{2\pi\mfk{i}}\, z^{k+\frac{1}{2}}\partial_z\xi(z)e^{-\phi(z)}e^{n\phi(0)} \nonumber
\\
&=\bigointsss\frac{dz}{2\pi\mfk{i}}\, z^{k+\frac{1}{2}+n}\boldsymbol{:}\partial_z\xi(z)e^{-\phi(z)}e^{n\phi(0)}\boldsymbol{:}.
\end{alignat} 
This vanishes only if $k+\frac{1}{2}+n\ge 0$. Similarly, the action of the mode $\gamma_k$ on the state \eqref{eq:state with picture number n}  gives
\begin{alignat}{2}
\gamma_k|n\rangle&\sim\bigointsss\frac{dz}{2\pi\mfk{i}}\, z^{k-\frac{3}{2}}\gamma(z)e^{n\phi(0)}=\bigointsss\frac{dz}{2\pi\mfk{i}}\, z^{k-\frac{3}{2}}\eta(z)e^{\phi(z)}e^{n\phi(0)} \nonumber
\\
&=\bigointsss\frac{dz}{2\pi\mfk{i}}\, z^{k-\frac{3}{2}-n}\boldsymbol{:}\eta(z)e^{\phi(z)}e^{n\phi(0)}\boldsymbol{:}. 
\end{alignat}
This vanishes only if $k-\frac{3}{2}-n\ge 0$. We want to choose $n$ such that $\beta_k|n\rangle=\gamma_k|n\rangle=0$ for $k>0$. From the possible values of $k$ for the NS and R sectors, it is clear that $n$ is integer for the NS sector and half-integer for the R sector. The only picture numbers that satisfy these inequalities are $n=-1$ for states in the NS sector, and $n=-\frac{1}{2},-\frac{3}{2}$ for states in the R sector. The R vacua are obtained by the action of spin fields
\begin{equation}
|\Sigma_{\pm}\rangle\equiv \lim\limits_{z\longrightarrow 0} e^{\pm\frac{1}{2}\phi(z)}|0\rangle
\end{equation}
on the NS ground state with picture number $-1$, i.e. spin fields interpolate between NS and R ground states. The picture numbers $-1$ and $-\frac{1}{2}$ are called canonical picture numbers\footnote{They give rise to the canonical fermionic dimension for the superstack. Doing superstring theory on other negative pictures corresponds to the modification of the fermionic dimension of superstack. For more details see section $4.3$ of \cite{Witten2012b}.}. It turns out that R states with picture-number $-\frac{3}{2}$ are useful for formulation of off-shell amplitudes in superstring theory \cite{Sen2014b} and also for the formulation of heterotic and type-II closed-superstring field theories \cite{Sen2016a}. The main reason is that a formulation using the states from R sector with picture number $-\frac{3}{2}$ avoids the introduction of the so-called {\it inverse picture-changing operators}.

Once a particular vacuum $|n\rangle$ is chosen, one can build the spectrum of the ghost system upon it. Therefore, the question is that whether these inequivalent description (i.e. inequivalent representation of the algebra of modes) of the spectrum build upon different vacua with different picture numbers are related. It is clear from the formulas \eqref{eq:bosonization formulas} that the superstring BRST complex does not contain states involving the zero-mode $\xi_0$ of the $\xi$ field, states like $\xi_0|\psi\rangle$. Consider the state $|\psi'\rangle\equiv \mscr{Q}\xi_0|\psi\rangle$, for a BRST-closed state $|\psi\rangle$, where $\mscr{Q}$ is the superstring BRST operator. $|\psi'\rangle$ has three properties, $1)$ it does not contain the zero-mode of $\xi(z)$ field, which can be seen by using the explicit form of $\mscr{Q}$, $2)$ it is obviously BRST-closed, i.e. it is an element of the BRST cohomology and as such it is a physical state, $3)$ it is not a representative of the trivial class of the superstring BRST cohomology since $\xi_0|\psi\rangle$ is not an element of the superstring BRST complex (it does contain zero-mode of $\xi(z)$ field and, as such, it is not part of the spectrum of the $\beta\gamma$ system). One can thus define the so-called {\it picture-changing operator} (PCO) \footnote{Here, we are not dealing with the more complicated issue of vertical integration which involves the introduction of a modified picture-changing operator \cite{Sen2015a}. See also \cite{ErlerKonopka2017a} for an alternative approach to the vertical integration procedure from the point of view of large Hilbert space.}. If $\mscr{V}_n(z)$ denotes a vertex operator built on a vacuum with picture number $n$, then 
\begin{equation}
\mscr{V}_{n+1}(z)=\boldsymbol{\mcal{X}}[\mscr{V}_n(z)]\equiv\{\mscr{Q},\xi(z)\mscr{V}_n(z)\},
\end{equation}
where $\boldsymbol{\mcal{X}}$ is the PCO (for explicit examples of the construction of vertex operators in two different pictures using the PCO, see section 5.3 of \cite{FriedanMartinecShenker1986}). The explicit form of PCO is thus
\begin{equation}\label{eq:PCO FSM}
\boldsymbol{\mcal{X}}(z)=\{\mscr{Q},\xi(z)\}.
\end{equation}
Thus, the answer to the above question is that the vertex operators associated to the same physical state built upon vacua with different picture number are related by the picture-changing operation \cite{FriedanMartinecShenker1986}. In other words, if we think of physical states with a particular picture number as elements of the BRST cohomology of the BRST complex defined by that picture number (picture numbers are part of the grading of the complex \cite{Belopolsky1997b}), the PCO maps the BRST cohomologies of the BRST complexes with different picture numbers into each other. It can be shown that the computation of {\it on-shell} superstring scattering amplitudes using different picture numbers gives the same result \cite{FriedanMartinecShenker1986}. 

\newl It turns out that picture-changing operators naturally enters the superstring path integral by integration over the odd moduli of superstack \cite{VerlindeVerlinde1987a}. We turn to this point in the next section.

\subsection{The Path Integral Framework: Verlinde-Verlinde Approach}\label{subsubsec:VV interpretation of PCO}

In this section, we review the results of \cite{VerlindeVerlinde1987a} due to Verlinde and Verlinde (VV). It is argued there that the picture-changing operation arises in the string measure as the integration over odd moduli of superstack. We will be brief and refer to \cite{VerlindeVerlinde1987a} or section 3.6.1 and 3.6.2 of \cite{Witten2012c} for more details. 

\newl The simplest amplitude is the genus-$\g$ superstring vacuum amplitude given by \cite{Martinec1987a}
\begin{equation}
\mcal{Z}=\bigintsss_{\mscr{SM}_{\sg}}\mcal{Z}(\mbs{m},\widehat{\mbs{m}})=\bigintsss_{\mscr{SM}_{\sg}} \mcal{Z}(X,\mcal{G};\mbs{m},\widehat{\mbs{m}})|\mcal{Z}(B,C;\mbs{m},\widehat{\mbs{m}})|^2,
\end{equation}
where $\mscr{SM}_{\sg}$ is the superstack of genus-$\g$ super-Riemann surfaces $\mcal{R}(\mbs{m},\widehat{\mbs{m}})$, and $\mbs{m}$ and $\widehat{\mbs{m}}$ collectively denote the even and odd moduli of $\mcal{R}(\mbs{m},\widehat{\mbs{m}})$, and 
\begin{alignat}{2}
\mcal{Z}(X,\mcal{G};\mbs{m},\widehat{\mbs{m}})&\equiv \bigintsss [\mscr{D}X]\,e^{-S[X,\mcal{G}]},\nonumber
\\
\mcal{Z}(B,C;\mbs{m},\widehat{\mbs{m}})&\equiv \bigintsss [\mscr{D}B,\mscr{D}C]\,e^{-S[B,C]}\left(\prod_{i=1}^{3\sg-3}\langle\mu_i,B\rangle\prod_{a=1}^{2\sg-2}\delta(\langle\widehat{\mu}_a,B\rangle)\right).
\end{alignat}
Here $\mcal{G}$ denotes the gravitino field. The fields $X$, $B$, and $C$ have the following expansions
\begin{equation}
X^\mu(z|\theta)=x^\mu(z)+\theta \psi^\mu(z), \qquad B(z|\theta)=\beta(z)+\theta b(z), \qquad C(z|\theta)=c(z)+\theta \gamma(z).
\end{equation}
$x^\mu(z)$ are space-time coordinates, $\psi^\mu(z)$ are space-time fermions, $b(z)$ and $c(z)$ are fields of the diffeomorphism ghost system, and $\beta(z)$ and $\gamma(z)$ are fields of the superconformal ghosts system. $\{\mu_i,\widehat{\mu}_a\}$ is a basis of super-Beltrami differentials and the respective pairings $\langle\mu_i,B\rangle$ and $\langle\widehat{\mu}_a,B\rangle$ are defined by integration over $\mcal{R}$. 

\newl The integration over the odd coordinates of the superstack can be done as follows. It turns out that the dependence of the action on the odd moduli of super-Riemann surfaces $\mcal{R}(\mbs{m},\widehat{\mbs{m}})$ is contained in the gravitino field $\mcal{G}$. Let us explain the reason. Consider a surface $\mcal{R}(\mbs{m},\mbs{0})$, i.e. a surface that does not have any odd moduli. These moduli can be turned on by doing a quasi-superconformal transformation generated by an odd vector field $\widehat{v}(z,\bm\bar{z})$ \cite{VerlindeVerlinde198804}
\begin{alignat}{2}\label{eq:the quasi-superconformal transformation of a split surface}
z\quad\longrightarrow \quad z'&=z+\theta\widehat{v}(z,\bm\bar{z}), \nonumber
\\
\theta\quad \longrightarrow\quad \theta'&=\theta +\widehat{v}(z,\bm\bar{z})+\frac{1}{2}\theta\widehat{v}(z,\bm\bar{z})\partial_z\widehat{v}(z,\bm\bar{z}).
\end{alignat} 
The fields on the surface $\mcal{R}(\mbs{m},\widehat{\mbs{m}})$, obtained from $\mcal{R}(\mbs{m},\mbs{0})$ by a quasi-superconformal transformation like \eqref{eq:the quasi-superconformal transformation of a split surface}, can be written in terms of the field on $\mcal{R}(\mbs{m},\mbs{0})$ by a field-redefinition depending on $\widehat{v}(z,\bm\bar{z})$. Once the matter action $S[X,\mcal{G}]$ on $\mcal{R}(\mbs{m},\widehat{\mbs{m}})$ is written, it turns out that it can be identified with the Brink-Di Vecchia-Howe action for massless fields $x^\mu(z)$ and $\psi^\mu(z)$ coupled to $\mscr{N}=1$ supergravity in the Wess-Zumino gauge \cite{BrinkDiVecchiaHowe197612}, in which gravitino field can be identified with the odd vector field $\widehat{v}(z,\bm\bar{z})$ as follows. Since $\widehat{v}(z,\bm\bar{z})$ is a $(-\frac{1}{2},0)$-differential, on a genus-$\g$ $\mscr{N}=1$ super-Riemann surface, it can be expanded in a basis of odd moduli $\{\widehat{m}^a\}$ as follows
\begin{equation}
\widehat{v}(z,\bm\bar{z})=\sum_{a=1}^{2\sg-2}\widehat{m}^av_a(z,\bm\bar{z}).
\end{equation}
On the other hand, the gravitino field\footnote{Consider a super-Riemann surface $\mcal{R}$ with $\sns$ NS punctures located at $p_1,\cdots,p_{\sns}$ and $\sra$ R punctures located at $q_1,\cdots,q_{\sra}$. As has been explained in section 5.4 of \cite{Witten2012c}, the gravitino field $\mcal{G}$ is an element of $H^1(\mcal{R}_{\text{red}},\widehat{L})$, where $$\widehat{L}\simeq L\otimes\mcal{O}\left(-\sum\limits_{a=1}^{\sns}p_a\right),\qquad L^2\simeq T\mcal{R}_{\text{red}}\otimes \mcal{O}\left(-\sum\limits_{a=1}^{\sra}q_a\right),$$ up to $\widehat{L}$-valued gauge transformations. $\mcal{R}_{\text{red}}$ is the reduced space of the super-Riemann surface $\mcal{R}$. The dimension of $H^1(\mcal{R}_{\text{red}},\widehat{L})$ is $2\sg-2+\sns+\frac{1}{2}\sra$, i.e. the fermionic dimension of the superstack. In the above computation, $\sns=0=\sra$, and $\{\widehat{m}^a\}$ is a basis for $H^1(\mcal{R}_{\text{red}},T\mcal{R}_{\text{red}}^{\frac{1}{2}})$. See also section 2.2 of \cite{Witten2012b}.} can be similarly expanded
\begin{equation}
\mcal{G}=\sum\limits_{a=1}^{2\sg-2}\widehat{m}^a\mcal{G}_a(z,\bm\bar{z}),
\end{equation}
where the differentials $\mcal{G}_a(z,\bm\bar{z})$ can be chosen to be independent of both even and odd moduli of $\mcal{R}(\mbs{m},\widehat{\mbs{m}})$. One can then identify 
\begin{equation}
\mcal{G}_a(z,\bm\bar{z})=2\overline{\partial}v_a(z,\bm\bar{z}).
\end{equation}
Therefore, $\mcal{G}$ is a $(-\frac{1}{2},1)$- differential. In the basis chosen for the gravitino modes $\mcal{G}_a(z,\bm\bar{z})$, $\langle\mu_i,B\rangle_{\mcal{R}(\mbs{m},\widehat{\mbs{m}})}=\langle\mu_i,b\rangle_{\mcal{R}(\mbs{m},\mathbf{0})}$, and $\langle\widehat{\mu}_a,B\rangle_{\mcal{R}(\mbs{m},\widehat{\mbs{m}})}=\langle\mcal{G}_a,\beta\rangle_{\mcal{R}(\mbs{m},\mathbf{0})}$\footnote{In general, we have
	$$\langle\mu_i,B\rangle_{\mcal{R}(\mbs{m},\widehat{\mbs{m}})}=\langle\mu_i,b\rangle_{\mcal{R}(\mbs{m},\mathbf{0})}+\sum\limits_{a=1}^{2\vsg-2}\widehat{m}^a\langle\boldsymbol{\nabla}_i\mcal{G}_a,\beta\rangle_{\mcal{R}(\mbs{m},\mathbf{0})},$$ where $\boldsymbol{\nabla}_i\mcal{G}_a$ denotes the covariant derivative of the gravitino mode $\mcal{G}_a$ with respect to the even moduli of the surface. The explicit form of  $\boldsymbol{\nabla}_i\mcal{G}_a$ is not important for the following conclusion. Since we have chosen $\mcal{G}_a$s to be independent of even moduli $\{m^a\}$, and hence $\boldsymbol{\nabla}_i\mcal{G}_a$ vanishes identically.}. Using this choice, the worldsheet action $S[X,\mcal{G};B,C]_{\mcal{R}(\mbs{m},\widehat{\mbs{m}})}$ on $\mcal{R}(\mbs{m},\widehat{\mbs{m}})$ can be expanded as
\begin{equation}
S[X,\mcal{G};B,C]_{\mcal{R}(\mbs{m},\widehat{\mbs{m}})}=S[X;B,C]_{\mcal{R}(\mbs{m},\mathbf{0})}+\sum\limits_{a=1}^{2\sg-2}\widehat{m}^a\langle\mcal{G}_a,T_F\rangle,
\end{equation}
where $T_F$ is the worldsheet supercurrent, the fermionic counterpart of the usual energy-momentum tensor $T_{B}$\footnote{This means that the OPEs of $T_B$ and $T_F$, i.e. $T_B(z)T_B(w)$, $T_B(z)T_F(w)$, and $T_F(z)T_F(w)$, generates a copy of the $\mscr{N}=1$ superconformal algebra. They have the following mode-expansion
	\begin{equation*}
	T_{B}(z)=\sum\limits_{n}\frac{L_n}{z^{n+2}}, \qquad \qquad T_F(z)=\sum\limits_{r}\frac{G_r}{z^{n+\frac{3}{2}}}.
	\end{equation*}
	$r\in\mbb{Z}+\frac{1}{2}$ (for NS sector) or $r\in\mbb{Z}$ (for R sector). The modes satisfy the NS or R extensions of the Virasoro algebra.}, of the total superconformal field theories of matter and ghosts whose explicit form is not important for us. Expanding the exponential, we get
\begin{equation}
\mcal{Z}(\mbs{m},\widehat{\mbs{m}})=\bigintsss \mscr{D}[X;B,C] \left[\prod_{i=1}^{3\sg-3}\langle\mu_i,b\rangle\prod_{a=1}^{2\sg-2}\left(1+\widehat{m}^a\langle \mcal{G}_a,T_F\rangle\right)\delta(\langle\mcal{G}_a,\beta\rangle)\right]e^{-S[X;B,C]_{\mcal{R}(\mbs{m},\mathbf{0})}}.
\end{equation}
In this final form, the integration over $\{\widehat{m}^a\}$ can be easily done. The result is
\begin{equation}
\mcal{Z}=\bigintsss_{\mscr{SM}_{\sg}}\prod_{a=1}^{2\sg-2}d\widehat{m}^a\,\mcal{Z}(\mbs{m},\widehat{\mbs{m}})=\bigintsss_{\mscr{S}_{\sg}}\bigintsss \mscr{D}[X;B,C]\left[\prod_{i=1}^{3\sg-3}\langle\mu_i,b\rangle_{\mcal{R}(\mbs{m},\mathbf{0})}\prod_{a=1}^{2\sg-2}\boldsymbol{\mcal{X}}(\mcal{G}_a)\right]e^{-S[X;B,C]_{\mcal{R}(\mbs{m},\mathbf{0})}},
\end{equation} 
where $\mscr{S}_{\sg}$ is the moduli space of split surfaces $\mcal{R}(\mbs{m},\mathbf{0})$, and $\boldsymbol{\mcal{X}}$ is the {\it picture-changing operator} given by
\begin{equation}\label{eq:PCO VV I}
\boldsymbol{\mcal{X}}(\mcal{G}_a)\equiv\langle \mcal{G}_a,T_F\rangle\,\delta(\langle\mcal{G}_a,\beta\rangle)\Big|_{\mcal{R}(\mbs{m},\mathbf{0})}.
\end{equation}
Therefore, the picture-changing operation is the operation of absorbing a zero-mode of $\beta$ ($\delta(\langle\mcal{G}_a,\beta\rangle)$) together with the integration over the corresponding odd moduli of superstack. This analysis can be repeated for the scattering amplitudes \cite{VerlindeVerlinde1987a}. A special choice of the gravitino modes $\mcal{G}_a(z,\bm\bar{z})$ is the delta-function support at points $z_a$
\begin{equation}
\mcal{G}_a(z,\bm\bar{z})=\delta^{(2)}(z-z_a;\bm\bar{z}-\bm\bar{z}_a)\qquad \qquad a=1,\cdots,2\g-2.
\end{equation}
Using this choice, \eqref{eq:PCO VV I} can be written as
\begin{equation}\label{eq:PCO VV II}
\boldsymbol{\mcal{X}}(z_a)=\delta(\beta(z_a))T_F(z_a)=\delta(\beta(z_a))\{\mscr{Q},\beta(z_a)\}, 
\end{equation}
where we have used the fact that $\{\mscr{Q},\beta(z)\}=T_{F}(z)$. From the relation $\frac{d\Theta(\beta(z))}{d\beta(z)}=\delta(\beta(z))$, this expression can be formally written as
\begin{equation}\label{eq:PCO VV III}
\boldsymbol{\mcal{X}}(z_a)=\{\mscr{Q},\Theta(\beta(z_a))\}. 
\end{equation}
The locations of the support of the gravitino modes $\mcal{G}_a(z,\bm\bar{z})$ are the locations of insertion of PCOs. Considering the fact that $\Theta(\beta(z))$ can be identified with $\xi(z)$\footnote{One way to determine this is as follows. By considering the integral representation of the delta function, we have
	\begin{equation*}
	\beta(z)\delta(\beta(w))\sim (z-w)\partial_w\beta(w)\delta(\beta(w))=(z-w)\partial_w\Theta(w),
	\end{equation*}
	where we have used the fact that $d\Theta(\beta(z))=d\beta(z)\delta(\beta(z))$, which implies that $\partial_z\Theta(\beta(z))=\partial_z\beta(z)\delta(\beta(z))$. On the other hand, using the bosonization formula $\delta(\beta(z))\mapsto \exp(\phi(z))$ and \eqref{eq:bosonization formulas}, we find that
	\begin{equation*}
	\beta(z)\exp(\phi(w))\sim (z-w)\partial_w\xi(w). 
	\end{equation*}
	Comparing the above equations, we can identify $\xi(z)=\Theta(\beta(z))$.} \cite{VerlindeVerlinde1987a}, this form matches with \eqref{eq:PCO FSM}, the FSM form of the picture-changing operator.

\newl One is also interested in understanding the meaning of the picture-changing operation in the superstack of $\mscr{N}=1$ Super-Riemann surfaces. It turns out that the correct interpretation of the picture-changing operation can be understood by a proper generalization of the concept of differential forms on ordinary manifolds to the so-called {\it pseudoforms} on supermanifolds. The picture-changing operation can then be understood as a natural operation on pseudoforms \cite{Belopolsky1997b, Belopolsky1997c}. We will explain this viewpoint in the next section.

\subsection{The Supergeometry Framework: Belopolsky Approach}

In this section, we explain the geometrical meaning of the picture-changing operation. For the full exposition, we will refer to the original papers \cite{Belopolsky1997b,Belopolsky1997c} in which the relation between FSM's interpretation of the picture-changing operation and the operation of forms on supermanifolds has been elucidated.

\newl In the theory of supermanifolds, the concept of differential forms is not enough to define the integration theory. The reason is that for odd coordinates $\theta_i$ and $\theta_j$, the differential $d\theta_i\wedge d\theta_j$ is even under $i\longleftrightarrow j$. This means that there is no differential form of top degree on a supermanifold. To properly define the integration on supermanifold, one has to generalize the concept of differential and top-degree forms. 

\newl The simplest integrand for the integration over a supermanifold is the volume form. To define a volume form on a supermanifold, one tries to generalize the idea of integration over the bosonic manifolds. A volume form on an ordinary manifold is a non-degenerate nowhere-vanishing section of the canonical sheaf of the manifold. The idea can be generalized to supermanifold: a natural object that can be integrated over a supermanifold is a section of the canonical (or Berezinian) sheaf of the supermanifold. It turns out that defining such a section requires a generalization of the concept of differential form on an ordinary manifold to the so-called {\it pseudoforms} on supermanifolds \cite{BernshteinLeites1978a}. If $(x|\theta)\equiv (x^1\cdots x^m|\theta^1\cdots \theta^n)$ denotes the local coordinate on a supermanifold, then an $(p|q)$-pseudoform $\boldsymbol{\Omega}$ can be written locally as
\begin{equation}\label{eq:generic (p|q)-pseudoform}
\boldsymbol{\Omega}=\Omega(x|\theta) dx^1\cdots dx^r d\theta^1\cdots d\theta^s \delta^{(n_1)}(d\theta^1)\cdots\delta^{(n_q)}(d\theta^q),
\end{equation}
where $\delta^{(n_i)}(d\theta^i)$ denotes the $n_i$\textsuperscript{th} derivative of the delta function with respect to $d\theta^i$, and $p\equiv r+s-\sum_{a=1}^q n_a$. There are two degrees associated to a pseudo-form: $p$ is called {\it the form degree}, corresponding to the usual degree of a differential form, and  $-q$, minus of the number of fermionic dimension of the supermanifold, is called {\it the picture number}, which does not have an analog in the theory of differential forms on ordinary manifolds. 

\newl We would like to have an object that can be integrated over a supermanifold. On an ordinary manifold, a natural object that can be integrated over is a section of the canonical sheaf of the manifold\footnote{Here, we consider real manifold. For a complex manifold, a section of the canonical sheaf provides the holomorphic volume form.}. These objects transform as a density on the manifold, i.e. under a coordinate transformations a density is multiplied by the Jacobian determinant of the coordinate transformation. Therefore, to be able to define the analogue of a density on a supermanifold, we need a supergeometric analog of the Jacobian determinant. This notion is the called {\it Berezinian} of a coordinate transformation\footnote{For more on the notion of density on a supermanifold see section 7 of chapter 4 of \cite{Manin1997a}.}. It thus turns out that proper objects that can be integrated over an $(m|n)$-dimensional supermanifold are sections of the the canonical (or Berezinian) sheaf of the supermanifold, the so-called {\it integral forms} of top degree \cite{BernshteinLeites1978a, BernshteinLeites1977a}. These are special cases of pseudoforms \eqref{eq:generic (p|q)-pseudoform} with the restriction $s=0$, $n_i=0$, $r=m$ and $q=n$. A top-degree integral form $\boldsymbol{\omega}$ can thus locally be written as
\begin{equation}
\boldsymbol{\omega}=\omega(x|\theta)dx^1\cdots dx^m \delta(d\theta^1)\cdots\delta(d\theta^n). 
\end{equation}
It has form number $m-n$ and picture number $-n$. We note that $\delta(dx)=dx$ has picture number zero since $dx$ is a {\it fermionic} variable. The fundamental properties of such forms are that 1) $\boldsymbol{\omega}dx^i=\boldsymbol{\omega}d\theta^j=0$ for $i=1,\cdots,m$ and $j=1,\cdots,n$, and 2) $\delta(d\theta^a)\delta(d\theta^b)=-\delta(d\theta^b)\delta(d\theta^a)$ for $a,b=1,\cdots,q$\footnote{We refer to section $3.3.2$ of \cite{Witten2012a} for the details.}, therefore, top-degree integral forms are indeed the objects that can be integrated over a supermanifold of dimension $(m|n)$. 

\newl {\it What is the significance of integral forms in the superstring theory?} In superstring theory, the supermanifolds we are interested in is some integration cycle $\Gamma$ inside the product $\mscr{M}_L\times\mscr{M}_R$ of spaces $\mscr{M}_L$ and $\mscr{M}_R$ that parameterize the left- and right-moving sectors of relevant theories, as we have explained in Sections \ref{subsec:the case of heterotic-string vertices} and \ref{subsec:the case of type-II-superstring vertices}. The fermionic dimension of the superstack of genus-$\g$ super-Riemann surfaces is equals to the space of meromorphic and multivalued quadratic $\frac{3}{2}$-superdifferentials associated to the variation of odd moduli \cite{Knizhnik1986a}. The space of such differentials has dimension $2\g-2$. On the other hand, each vertex operator from the NS sector has picture number $-1$, i.e. it comes with a factor of $\delta(\gamma)$, and each vertex operator from the R sector has picture number $-\frac{1}{2}$ since the spin field, i.e. the vacuum of the $\beta\gamma$ system, has picture number $-\frac{1}{2}$. Therefore, the fermionic dimension of the superstack is $2\g-2+\ns+\frac{1}{2}\ra$, i.e. the superstring measure must be a top-degree integral form with picture number $-(2\g-2+\ns+\frac{1}{2}\ra)$. The superstring measure associated to an arbitrary scattering process involving $\n=(\ns,\ra)$ external states from either sectors, given by the path integral over the worldsheet superfields, naturally produces a top-degree integral form on $\Gamma$ \cite{ Alwarez-GaumeNelsonGomezSieraVafa1988,Witten2012b,Witten2012c}
\begin{equation}
\langle\!\langle\mscr{V}_1\cdots\mscr{V}_{\sn}\rangle\!\rangle_{\mcal{R}} \equiv \bigintsss \mscr{D}(X,\mcal{G};B,C)\,\,\mscr{V}_1\cdots\mscr{V}_{\sn}\,\,e^{-S},
\end{equation}
where $\mcal{R}$ is a fixed genus-$\g$ surface with $\n=(\ns,\ra)$ punctures, $X$, $B$, $C$ are superfields of the matter and ghost superconformal field theories respectively, $\mscr{V}_i$s are appropriate vertex operators associated to the external states, and $S$ is the action describing the worldsheet theory (the gauge-fixed version of the RNS-superstring action or any other mater-ghost superconformal field theory with appropriate supersymmetry and central charges). Therefore, $\langle\!\langle\mscr{V}_1\cdots\mscr{V}_{\sn}\rangle\!\rangle_{\mcal{R}}$ can be integrated over $\Gamma$ to give $\mcal{A}_{\bf String}(1,\cdots,\n)$, the genus-$\g$ stringy contribution to the scattering amplitude of $\n$ arbitrary external states
\begin{equation}
\mcal{A}_{\bf String}(1,\cdots,\n)=\bigintsss_{\Gamma}\langle\!\langle\mscr{V}_1\cdots\mscr{V}_{\sn}\rangle\!\rangle_{\mcal{R}}.
\end{equation}

\newl After defining the pseudoforms, the next step is to define operations on them. Some of these operations may change the form degree or the picture number of pseudoform. In particular, one is interested to define an operator that changes the picture number of a generic pseudoform. To define such an operator, consider an odd vector field $\widehat{V}$, even tangent vectors $\{v_1,\cdots,v_p\}$ and odd tangent vectors $\{\widehat{v}_1,\cdots,\widehat{v}_q\}$, which collectively denoted as $\boldsymbol{v}$, on a $(m|n)$-dimensional supermanifold. We can define the analog of the Lie derivative along $\widehat{V}$ on a $(p|q)$ form $\boldsymbol{\omega}$
\begin{equation}
[\mscr{L}_{\widehat{V}}\boldsymbol{\omega}](v_1,\cdots,v_p|\widehat{v}_1,\cdots,\widehat{v}_{q})\equiv \widehat{V}^I\frac{\partial \boldsymbol{\omega}(\boldsymbol{v})}{\partial x^I}+(-1)^{[A]}v^I_A \frac{\partial \widehat{V}^J}{\partial x^I}\frac{\partial\boldsymbol{\omega}(\boldsymbol{v})}{\partial v^J_A},
\end{equation}
where $A=1,\cdots,(p|q)$ and $I,J=1,\cdots,(m|n)$, the dimension of the supermanifold. $[\,\,\cdot\,\,]$ denotes the statistics of the quantity inside it. This operation does not change the form degree or the picture number of pseudoforms. Next, we define an operation that changes the form degree of a pseudoform but preserves its picture number. It can thus be interpreted as the supergeometry analog of the de Rham differential. It is given by
\begin{alignat}{2}
[\mathbf{d}\boldsymbol{\omega}](v_1,\cdots,v_{p+1}|\widehat{v}_1,\cdots,\widehat{v}_{q})&\equiv (-1)^pv_{p+1}^I\frac{\partial\boldsymbol{\omega}(\boldsymbol{v})}{\partial x^I}+(-1)^{[I][A]}v_A^J\frac{\partial^2\boldsymbol{\omega}(\boldsymbol{v})}{\partial x^J\partial v^I_A},
\end{alignat}
$\mathbf{d}$ maps $(p|q)$-forms to $(p+1|q)$-forms. Finally, we define an operation that preserves the form number of pseudoforms but changes their picture number. It is given by 
\begin{equation}
[\delta(\boldsymbol{i}_{\widehat{V}})\boldsymbol{\omega}](v_1,\cdots,v_p|\widehat{v}_1,\cdots,\widehat{v}_{q-1})\equiv (-1)^p\boldsymbol{\omega}(v_1,\cdots,v_p|\widehat{V},\widehat{v}_1,\cdots,\widehat{v}_{q-1}).
\end{equation}
$\delta(\boldsymbol{i}_{\widehat{V}})$ maps a $(p|q)$ form to a $(p|q-1)$ form. Note that $\boldsymbol{i}_{\widehat{V}}$ is the usual interior product with the vector $\widehat{V}$. It might be tempting to identify $\delta(\boldsymbol{i}_{\widehat{V}})$ with the picture-changing operation in the superstring theory. However, it turns out that the relevant operation is a little more subtle than this operation. One defines the so-called {\it picture-changing operator} on the space of pseudoforms as follows
\begin{equation}\label{eq:PCO Belopolsky}
\boldsymbol{\mcal{X}}_{\widehat{V}}\equiv \frac{1}{2}\left(\delta(\boldsymbol{i}_{\widehat{V}})\mscr{L}_{\widehat{V}}+\mscr{L}_{\widehat{V}}\delta(\boldsymbol{i}_{\widehat{V}})\right)=\{\mathbf{d},\Theta(\boldsymbol{i}_{\widehat{V}})\},
\end{equation}
where $\Theta(x)$ is the step function. In the last identity, the fact that $\mscr{L}_{\widehat{V}}=\{\mathbf{d},\boldsymbol{i}_{\widehat{V}}\}$ has been used. It is clear from the first equality that $\boldsymbol{\mcal{X}}_{\widehat{V}}$ maps $(p|q)$ forms to $(p|q-1)$ forms, and indeed changes the picture number of a form on a supermanifold. Considering these operations, we can give the complex of $(p|q)$-forms on a supermanifold in figure \ref{fig:the complex of form on a supermanifold}.
\begin{figure}[]\centering 
	\begin{tikzcd}
	& 0\arrow[r, "\mbf{d}"]\arrow[d, "\mbs{\Xi}" left,shift left=-1] & \Omega^{(0|0)} \arrow[d, "\mbs{\Xi}" left,shift left=-1]\arrow[r,"\mbf{d}"] & \cdots\arrow[d, "\mbs{\Xi}" left,shift left=-1]\arrow[r,"d"] & \Omega^{(m|0)}\arrow[d, "\mbs{\Xi}" left,shift left=-1]\arrow[r,"\mbf{d}"] & \Omega^{(m+1|0)} \arrow[d, "\mbs{\Xi}" left,shift left=-1]\arrow[r,"\mbf{d}"] & \cdots 
	\\
	\cdots	\arrow[r, "\mbf{d}"] & \vdots \arrow[d, "\mbs{\Xi}" left,shift left=-1]\arrow[u, "\mbs{\chi}" right,shift right=1] \arrow[r,"\mbf{d}"] & \vdots \arrow[d, "\mbs{\Xi}" left,shift left=-1]\arrow[u, shift right=1, "\mbs{\chi}" right] \arrow[r,"\mbf{d}"] & \cdots\arrow[d, "\mbs{\Xi}" left,shift left=-1]\arrow[u, "\mbs{\chi}" right,shift right=1]\arrow[r,"\mbf{d}"] & \vdots \arrow[d, "\mbs{\Xi}" left,shift left=-1]\arrow[u, "\mbs{\chi}" right,shift right=1]\arrow[r,"\mbf{d}"] \arrow[r,"\mbf{d}"] & \vdots \arrow[d, "\mbs{\Xi}" left,shift left=-1]\arrow[u, "\mbs{\chi}" right,shift right=1]\arrow[r,"\mbf{d}"] & \cdots 
	\\
	\cdots\arrow[r, "\mbf{d}"] & \Omega^{(-1|s)} \arrow[d, "\mbs{\Xi}" left,shift left=-1]\arrow[u, "\mbs{\chi}" right,shift right=1]\arrow[r,"\mbf{d}"] & \Omega^{(0|s)}\arrow[d, "\mbs{\Xi}" left,shift left=-1]\arrow[u, "\mbs{\chi}" right,shift right=1]\arrow[r,"\mbf{d}"] & \cdots \arrow[d, "\mbs{\Xi}" left,shift left=-1]\arrow[u, "\mbs{\chi}" right,shift right=1]\arrow[r,"\mbf{d}"] & \Omega^{(m|s)}\arrow[d, "\mbs{\Xi}" left,shift left=-1]\arrow[u, "\mbs{\chi}" right,shift right=1]\arrow[r,"\mbf{d}"] & \Omega^{(m+1|s)} \arrow[u, "\mbs{\chi}" right,shift right=1]\arrow[d, "\mbs{\Xi}" left,shift left=-1]\arrow[r,"\mbf{d}"] & \cdots 
	\\
	\cdots	\arrow[r, "\mbf{d}"] & \vdots \arrow[d, "\mbs{\Xi}" left,shift left=-1]\arrow[u, "\mbs{\chi}" right,shift right=1] \arrow[r,"\mbf{d}"] & \vdots \arrow[d, "\mbs{\Xi}" left,shift left=-1]\arrow[u, shift right=1, "\mbs{\chi}" right] \arrow[r,"\mbf{d}"] & \cdots\arrow[d, "\mbs{\Xi}" left,shift left=-1]\arrow[u, "\mbs{\chi}" right,shift right=1]\arrow[r,"\mbf{d}"] & \vdots \arrow[d, "\mbs{\Xi}" left,shift left=-1]\arrow[u, "\mbs{\chi}" right,shift right=1]\arrow[r,"\mbf{d}"] & \vdots \arrow[d, "\mbs{\Xi}" left,shift left=-1]\arrow[u, "\mbs{\chi}" right,shift right=1]\arrow[r,"\mbf{d}"] & \cdots
	\\
	\cdots\arrow[r, "\mbf{d}"] & \Omega^{(-1|n)} \arrow[u, "\mbs{\chi}" right,shift right=1]\arrow[r,"\mbf{d}"] & \Omega^{(0|n)}\arrow[u, "\mbs{\chi}" right,shift right=1]\arrow[r,"\mbf{d}"] & \cdots \arrow[u, "\mbs{\chi}" right,shift right=1]\arrow[r,"\mbf{d}"] & \Omega^{(m|n)}\arrow[u, "\mbs{\chi}" right,shift right=1]\arrow[r,"\mbf{d}"] & 0 \arrow[u, "\mbs{\chi}" right,shift right=1].
	\end{tikzcd}
	\caption{The complex of forms on a $(m|n)$-dimensional supermanifold. The first row is the complex of superforms. the intermediate complexes are complexes of pseudoforms with various picture numbers. The last row is the complex of integral forms, which as it is clear, there is a notion of top-degree integral forms. We have denoted two natural operations on forms on a supermanifold: 1) the picture-changing operation $\mbs{\chi}$, which acts from bottom to top and decreases the picture-number of forms by one, and 2) the picture-changing operation $\mbs{\Xi}$, which acts from top to bottom and increases the picture-number of forms by one. The space $\Omega^{(m|n)}$ is the space of top-degree integral forms. The elements of this space can naturally be integrated over an $(m|n)$-dimensional supermanifold.} \label{fig:the complex of form on a supermanifold}
\end{figure}

\newl From \eqref{eq:PCO Belopolsky}, and the fact that de Rham differential $\mathbf{d}$ on the superstack can be interpreted as the superstring BRST operator\footnote{The precise statement is as follows. Let $\mathbf{d}$ denotes the de Rham differential acting on the superstack, and $\mathbf{\Omega}(\Psi_1,\cdots,\Psi_{\vsn})$ is the off-shell superstring measure on the superstack for scattering of $\n$ external off-shell states $\{\Psi_1,\cdots,\Psi_{\vsn}\}$. Then, $$\mathbf{d}\mathbf{\Omega}_{p-1}(\Psi_1,\cdots,\Psi_{\vsn})=\sum\limits_{a=1}^{\vsn}\mathbf{\Omega}_{p}(\Psi_1,\cdots,\mscr{Q}\Psi_a,\cdots,\Psi_{\vsn}).$$
	$p$ and $p-1$ denotes the form degree of the integral form $\mathbf{\Omega}$. For more details see section 3 of \cite{Belopolsky1997a}.}, and the fact that $\Theta(\beta(z))$ can be identified with $\xi(z)$, it seems reasonable to guess that this operator is the same as the superstring picture-changing operation. This turns out to be the case \cite{Belopolsky1997c}. We will explain this fact in the next section.

\subsection{The Unified Perspective: From Supermanifolds to Superstrings}

In this section, we explain the relation between the Belopolsky's picture-changing operator \eqref{eq:PCO Belopolsky} and the superstring picture-changing operator. We will not explain the details of this relationship since it requires the explanation of the relation between semi-infinite forms defined over a suitable algebra and the ghost sector of the superstring theory. This section is just for the completion of discussion and we just quote the results from \cite{Belopolsky1997b,Belopolsky1997c} where the relation between FSM-VV approaches and the semi-infinite forms has been elucidated.

\newl It turns out that the space of semi-infinite forms on the algebra of superconformal vector fields on the supercircle $S^{1|1}$ can be identified with the ghost sector of superstring theory. Let us denote the even and odd vector fields concentrated at $z$ by $l(z)$ and $g(z)$. Then, superconformal ghosts can be identified as follows
\begin{alignat}{2}\label{eq:ghost-semiforms identification}
b(z)&\mapsto \boldsymbol{i}_{l(z)},   \qquad&\qquad c(z)&\mapsto\boldsymbol{e}_{l^\vee(z)}, \nonumber
\\
\beta(z)&\mapsto \boldsymbol{i}_{g(z)},   \qquad&\qquad \gamma(z)&\mapsto\boldsymbol{e}_{g^\vee(z)}. 
\end{alignat} 
$v^\vee$ is the covector dual to the tangent vector $v$. $\boldsymbol{e}_{v^\vee}$ is the operation of exterior product defined using $v^\vee$. For a $(p|q)$-form $\boldsymbol{\omega}$, it is defined as follows
\begin{equation}
[\boldsymbol{e}_{v^\vee}\boldsymbol{\omega}](v_1,\cdots,v_{p+1}|\widehat{v}_1,\cdots,\widehat{v}_q)\equiv \left(-1\right)^p\left[v^\vee(v_{p+1})-(-1)^{[v^\vee][A]}v^\vee(v_A)v_{p+1}^I\frac{\partial}{\partial v_A^I}\right]\boldsymbol{\omega}.
\end{equation}
Comparing \eqref{eq:PCO Belopolsky} and \eqref{eq:ghost-semiforms identification} shows that we can identify $\delta(\boldsymbol{i}_{\widehat{V}})$ with $\delta(\beta(z))$. The odd vector fields that appears in the FSM approach is thus $g(z)$, the odd generator of the super-Virasoro algebra. To complete the identification, we need to identify $\mscr{L}_{\widehat{V}}$, the operation of Lie derivative along the odd vector field $\widehat{V}=g(z)$, with some conformal field in the ghost sector of superstring theory. It turns out that the Lie derivative $\mscr{L}_{g(z)}$ can be identified with $T_F(z)$. Therefore, the picture-changing operator \eqref{eq:PCO Belopolsky} can be identified with the following conformal field
\begin{equation}\label{eq:PCO identification}
\boldsymbol{\mcal{X}}_{\widehat{V}}\mapsto\boldsymbol{\mcal{X}}_{g(z)}\equiv \frac{1}{2}\left(\delta(\beta(z))T_F(z)+T_F(z)\delta(\beta(z))\right)=\{\mscr{Q},\Theta(\beta(z))\}. 
\end{equation}  
Since $\delta(\beta(z))$ is a conformal field with conformal dimension $-\frac{3}{2}$\footnote{Note that $\beta(z)$ is a conformal field with conformal dimension $+\frac{3}{2}$.}, $\boldsymbol{\mcal{X}}_{g(z)}$ has conformal dimension zero. The final identity is the FSM \eqref{eq:PCO FSM} and VV \eqref{eq:PCO VV III} expressions for the picture-changing operator. Note that this is a formal identification. The picture-changing operator \eqref{eq:PCO FSM} acting on a vertex operator {\it increases} the picture-number of the associated vertex operator. However, \eqref{eq:PCO Belopolsky} acts on forms on a supermanifold and {\it decreases} their picture number. 

\newl In summary, we found that various interpretations of the picture-changing operation, the FSM's interpretation (i.e. an operation that changes the picture number of the vertex operator and maps the superstring BRST cohomology with a particular picture number to another superstring BRST cohomology with a different picture number), VV's interpretation (i.e. the combined operation of the absorption of a zero-mode of the $\beta(z)$ field and the integration over the corresponding odd moduli of the superstack of $\mscr{N}=1$ super-Riemann surfaces), and the Belopolsky's interpretation (i.e. an operation that acts on the space of pseudoforms and changes their picture number) match. This concludes our discussion of the picture-changing operation. 

\section{Moduli Stacks of Stable Spin Curves}\label{app:compactification of the moduli stack of spin curves}

In this appendix, we review some details on the construction of the moduli stack of punctured spin curves over $\Spec \mathbb Z[1/2]$, following \cite{Cornalba1989a,Jarvis1998a,Jarvis2000a}. In \cite{Jarvis1998a,Jarvis2000a}, Jarvis discusses higher-spin structures and they are more complicated than $r=2$ case, i.e. ordinary spin structures. For convenience, we focus on $r=2$ since only ordinary spin structures are relevant for string theory. In fact, all four versions of spin structures in \cite{Jarvis1998a,Jarvis2000a}, i.e. quasi-spin $\text{QSpin}_{r,{\sg}}$, spin $\overline{\text{Spin}}_{r,{\sg}}$, pure spin $\text{Pure}_{r,{\sg}}$, and finally the $\overline{\mscr{S}}^{1/r}_{\sg}$, coincide when $r=2$:$$\text{QSpin}_{2,{\sg}}=\overline{\text{Spin}}_{2,{\sg}}=\text{Pure}_{2,{\sg}}=\overline{\mscr{S}}^{1/2}_{\sg}.$$

\subsection{Definitions and Examples}
We begin by recalling some definition. We begin with the definition of spin curve with and without punctures.

\begin{definition}[{\bf Spin Curve}]\label{Defn_SpinCurve}
	A spin curve of genus $\g$ is a triple $(\mscr{C},\mscr E,b)$, where $\mscr{C}$ is a stable curve of genus $\g$, $\mscr E$ is a rank one torsion-free sheaf A coherent $\mcal{O}_X$-module $\mscr{F}$ is called a torsion-free or relatively-torsion-free sheaf on a family of stable or semi-stable curves $f:X\longrightarrow S$ over a scheme $S$ if it is a flat module over $S$ in such a way that on $\mscr{C}$ of degree $\g-1$ with a homomorphism $b:\mscr{E}^{\otimes 2}\to \omega _{\mscr{C}}$ which is an isomorphism on the smooth locus of $\mscr{C}$. $\omega _\mscr{C}$ is the dualizing sheaf on $\mscr{C}$.
\end{definition}

\begin{definition}[\textbf{Family of Spin Curves}]\label{Defn_RelSpinCurve}
	A family of spin curves of genus $\g$ over the base scheme $T$ is a triple $(\mscr{C},\mscr E,b)$, where $\mscr{C}\to T$ is a relative stable curve, $\mscr E$ is a finite-presented sheaf on $\mscr{C}$ which is \textbf{flat over $T$}, with a morphism $b:\mscr E^{\otimes 2}\to \omega _{\mscr{C}/T}$, such that for every geometric point $t$ of $T$, $(\mscr{C}_t,\mscr{E}_t,b_t)$ is a stable curve. $\omega _{\mscr{C}/T}$ is the relative dualizing sheaf on $\mscr{C}$.
\end{definition}

\begin{remark}
	Since relative stable curves are relative Gorenstein, dualizing complex $\omega_{\mscr{C}/T}^{\bullet}$ is quasi-isomorphic to a line bundle $\omega _{\mscr{C}/T}$, we call it the relative dualizing sheaf.
\end{remark}

\begin{remark}Let $X$ be a topological space and $S$ be a scheme. Suppose that $X\longrightarrow S$ is a locally finite-presented morphism between schemes. We call a finite-presented sheaf $\mscr F$ on $X$ \textbf{relatively torsion-free} if $\mscr F$ is flat over $S$ and $\mscr F_s$ is torsion-free on $X_s$ for every geometric point $s$ of $S$, i.e. the associated primes of the induced $\mscr F_s$ on each fiber $X_s\equiv X\underset{S}{\times}\Spec \mbb{K}(s)$, for some algebraically-closed field $\mbb{K}$ and all $s\in S$, do have height zero\footnote{Let $M$ be an $R$-module over a commutative ring $R$. An associated prime ideal of $M$ is a prime ideal $I_p$ that is the annihilator of a nonzero element $\mfk{m}\in M$. The height of a prime ideal is the maximal length of a chain of prime ideals contained in $I_p$, which is also called the Krull dimension of the localization $R_{I_p}$.}. In {\normalfont Definition \ref{Defn_RelSpinCurve}}, $\mscr E$ is relatively torsion-free. 
\end{remark}
A useful fact is the following:
\begin{lem}\label{Dual_of_TorsionFree}
	Suppose that $X\longrightarrow T$ is a relative stable curve. If $\mscr{E}\in \QCoh (\mathcal O_X)$, where $\QCoh(\mcal{O}_X)$ denotes the category of quasi-coherent sheaves on $X$, whose structure sheaf is $\mcal{O}_X$\footnote{A local-ringed space $(X,\mcal{O}_X)$ is a topological space $X$ together with a sheaf of rings $\mcal{O}_X$, called the structure sheaf, on $X$.}, is relatively torsion free, and $\mscr F\in \QCoh (\mathcal O_X)$ is locally-free of finite rank, then $\underline{\RHom}^{\bullet}(\mscr{E},\mathcal F)$\footnote{Here,  $\underline{\RHom}^{\bullet}$ denotes the derived Hom functor and the $\underline{\RHom}$ means it is an object in the derived category of sheaves, not just literally taking the $\Ext ^{\bullet}$}, is quasi-isomorphic to $\underline{\Hom}(\mscr{E},\mscr F)$\footnote{Here, $\Hom$ denotes the Hom functor and $\underline{\Hom}$ denotes the sheaf of Homs, not just the global Hom between two sheaves.}, which is relatively torsion free and its construction commutes with arbitrary base change.
\end{lem}

\begin{proof}There is a canonical morphism $\underline{\Hom}(\mscr{E},\mscr F)\to \underline{\RHom}^{\bullet}(\mscr{E},\mscr F)$ defined by truncation, and it's a quasi-isomorphism if it is locally on $X$, so we can assume that both $X$ and $T$ are affine. Our strategy is to prove it first for Noetherian $T$ then show that it remains true after base change.\\
	
	\textbf{When $T$ is a spectrum of a field}, i.e. $T=\Spec \mbb{K}$ for some field $\mbb{K}$, then $X$ is a stable curve over $\mbb{K}$. Since $X$ is Gorenstein of dimension 1 and $\mscr F$ is locally free, injective dimension of $\mscr F$ is 1, hence $\underline{\RHom}^{\bullet}(\mscr{E},\mscr F)$ is concentrated in degree 0 and 1. Moreover, $\mscr{E}$ is torsion free thus torsionless, so locally $\mscr{E}$ can be embedded into a free sheaf 
	\begin{equation*}
	\mscr{E}|_U\hookrightarrow \mathcal O_U^{\oplus n},
	\end{equation*}
	which implies that there is a surjective morphism
	\begin{equation*}
	\Ext^1(\mathcal O_U^{\oplus n},\mscr F|_U)\twoheadrightarrow \Ext^1(\mscr{E}|_U,\mscr F|_U).
	\end{equation*}
	Hence $\Ext^1(\mscr{E}|_U,\mscr F|_U)=0$ and
	\begin{equation*}
	\underline{\Hom}(\mscr{E},\mscr F)\cong \underline{\RHom}^{\bullet}(\mscr{E},\mscr F).
	\end{equation*}
	Moreover $\underline{\Hom}(\mscr{E},\mscr F)$ is torsion free (in fact reflexive).\\
	
	\textbf{For general Noetherian $T$}, take any geometric point $t$ of $T$, denote the natural morphism $X_t\to X$ by $i_t$. For any integer $N>1$, consider a complex of finite free $\mathcal O_{X}$ modules $$\cdots \to 0\to \mscr E^{-N}\to \mscr E^{-N+1}\to \cdots \to \mscr E^{0}$$with trivial cohomologies at degree $\{-N+1,\cdots, -1\}$ and $H^0$ being $\mscr E$. Recall that we assume that $X$ is affine, hence there exists such complex. Note that there is a distinguished triangle in the derived category:
	\begin{center}
		\begin{tikzcd}
		\mscr K^{-N}[N] \arrow[r] & \mscr E^{\bullet} \arrow[r]& \mscr E \arrow[r,"+1"] &,
		\end{tikzcd}
	\end{center}
	where $\mscr K^{-N}$ is the kernel of $\mscr E^{-N}\to \mscr E^{-N+1}$. Hence we have a quasi-isomorphism 
	\begin{equation}
	\tau _{\le N-1}\underline{\RHom}^{\bullet}(\mscr{E},\mscr F)\cong \tau _{\le N-1}\underline{\RHom}^{\bullet}(\mscr{E}^{\bullet},\mscr F),
	\end{equation}
	where $\tau_{\le N-1}$ is the truncation functor. Moreover this quasi-isomorphism is stable under arbitrary base change $T'\to T$, since derived base change of $\mscr K^{-N}[N]$ has nonzero cohomologies only in degree below $-N+1$. Now take the base change $i_t:X_t\to X$, we have 
	\begin{equation}
	Li_t^* \underline{\RHom}^{\bullet}(\mscr{E}^{\bullet},\mscr F)\cong \underline{\RHom}^{\bullet}(i_t^*\mscr{E}^{\bullet},i_t^*\mscr F),
	\end{equation}
	since both $\mscr E^{\bullet}$, $\mscr F$, and $\underline{\RHom}^{\bullet}(\mscr{E}^{\bullet},\mscr F)$ are represented by complex of free modules. The case of $T$ being $\Spec \mbb{K}(t)$ implies that $\tau_{\le N-1}\underline{\RHom}^{\bullet}(i_t^*\mscr{E}^{\bullet},i_t^*\mscr F)$ is concentrated in degree zero and torsion free. Since $t$ can be arbitrary, we see that the truncation of perfect complex, $\tau_{\le N-1}\underline{\RHom}^{\bullet}(\mscr{E}^{\bullet},\mscr F)$, is again perfect, concentrated in degree zero, and whose degree zero cohomology is relatively torsion free. Hence the same is true for $\tau_{\le N-1}\underline{\RHom}^{\bullet}(\mscr{E},\mscr F)$. Take $N\to \infty$ and we conclude the proof for arbitrary Noetherian $T$.\\
	
	\textbf{For general scheme $T$}, Note that we can take Noetherian approximation and assume that there is a Noetherian affine $X_0\to T_0$ with $\mscr E_0$ and $\mscr F_0$ satisfying conditions of this lemma, such that $X$ and $\mscr E$ and $\mscr F$ are pull-back of a morphism $\phi:T\to T_0$. We have shown a quasi-isomorphism $$\tau _{\le N-1}\underline{\RHom}^{\bullet}(\mscr{E}_0,\mscr F_0)\cong \tau _{\le N-1}\underline{\RHom}^{\bullet}(\mscr{E}_0^{\bullet},\mscr F_0)$$ in step 2 which is stable under arbitrary base change. We also know that the perfect complex $i_t^*\underline{\RHom}^{\bullet}(\mscr{E}_0,\mscr F_0)$ has non-trivial cohomologies only in degree $0$ and $N$, for arbitrary geometric point $t$, hence $\underline{\Ext}^{N}(\mscr{E}_0,\mscr F_0)$ is also flat over $T_0$. Hence we see that $$\tau _{\le N-1}\underline{\RHom}^{\bullet}(L\phi^*\mscr{E}_0^{\bullet},L\phi^*\mscr F_0)\cong L\phi^*\tau _{\le N-1}\underline{\RHom}^{\bullet}(\mscr{E}_0^{\bullet},\mscr F_0)$$is concentrated in degree zero and whose degree zero cohomology is relatively torsion free. Take $N\to \infty$ and we conclude the proof for arbitrary scheme $T$.
\end{proof}

\begin{definition}[\textbf{Isomorphism of Spin Curves}]\label{Defn_Isom_RelSpinCurve}
	An isomorphism between two families of spin curves over $T$ denoted by $(\mscr{C},\mscr{E},b)$ and $(\mscr{C}',\mscr{E}',b')$ is a pair $(\phi, \psi)$, where $\phi:\mscr{C}'\longrightarrow \mscr{C}$ is an isomorphism between spin curves, $\psi$ is an isomorphism between $\mscr{E}'$ and $\phi^* \mscr{E}$ which is compatible with $b$.
\end{definition}

\begin{example}
	For any family of spin curve such that $2$ is invertible in the base scheme, there always exists an automorphism $(\mbb{I},-1)$, for the identity map $\mbb{I}$.
\end{example}

\begin{definition}[\textbf{Stack of Spin Curves}]\label{Defn_Stack_SpinCurve}
	Fix a scheme $S\equiv \Spec \mathbb Z[1/2]$ and an integer $\g\ge 2$, define the following prestack in $(\Sch_S)_{\text{\' et}}$\footnote{$(\Sch_S)_{\text{\' et}}$ denotes the category of schemes over $S$ endowed with the \'etal topology.}
	\begin{align*}
	(T\in \Sch _S)\mapsto 
	\begin{cases}
	\text{\normalfont\bf Obj}: \text{Families of Spin curves of genus-}\g \text{ over T},\\
	\text{\normalfont\bf Mor}: \text{Isomorphisms between families of spin curves}.
	\end{cases}
	\end{align*}
	It's in fact a stack: every stable curve $\mscr C\to S$ is canonically polarized by dualizing sheaf $\omega_{\mscr C/S}$. We denote it by $\overline{\mscr{S}}_{\sg}$. 
\end{definition}

\noindent We also need the following generalization:

\begin{definition}[\textbf{Marked Spin Curve}]\label{Defn_Marked_RelSpinCurve}
	A family of punctured spin curves of genus $\g$ over the base scheme $T$, of puncturing type $\mathbf m\equiv(m_1,\cdots, m_{\sn})$ is a quadruple $(\mscr{C},\mbs{\mscr{D}},\mscr{E},b)$, where $\mscr{C}\longrightarrow T$ is a relative semi-stable curve, $\mathbf \mbs{\mscr{D}}\equiv (\mscr{D}_1,\cdots,\mscr{D}_{\sn})$ are $\n$ sections of $\mscr{C}\longrightarrow T$ with images lying in the smooth locus and making $(\mscr{C},\mbs{\mscr{D}})$ into an $\n$-punctured stable curve, $\mscr{E}$ is a rank-$1$ relatively torsion-free sheaf on $\mscr{C}$, with a morphism $b:\mscr{E}^{\otimes 2}\to \omega _{\mscr{C}/T}\left(\sum_i m_i\mscr{D}_i\right)$\footnote{$\omega _{\mscr{C}/T}\left(\sum_i m_i\mscr{D}_i\right)$ means that sections of the canonical sheaf have poles of order $m_i$ along the divisors $\mscr{D}_i$ defined by the puncture $\mfk{p}_i$. For more on this see \hyperlink{notation for stack of punctured spin curve}{the definition of the notation $\mscr{S}^{\mbf{m}}_{\sg,\sn}$} in the Introduction.}, such that for every geometric point $t$ of $T$
	\begin{enumerate}
		\item $\deg \mscr{E}_t=\g-1+\frac{1}{2}\sum\limits_{i=1}^{\sn} m_i$,
		\item $b_t$ is isomorphism on the smooth locus of $\mscr{C}_t$.
	\end{enumerate}
\end{definition}

It turns out that there is an equivalent definition which is sometimes more convenient:
\begin{lem}\label{Equiv_a_b}
	Suppose that $(\mscr{C},\mbs{\mscr{D}})/T$ is a relative $\n$-punctured stable curve. $\mscr{E}$ is a rank-1 relatively torsion free sheaf on $\mscr{C}$, of degree $\g-1+\frac{1}{2}\sum_i m_i$, then there is a one-to-one correspondence between 
	\begin{enumerate}
		\item Homomorphisms $b:\mscr{E}^{\otimes 2}\to \omega _{\mscr{C}/T}(\sum_i m_i\mscr{D}_i)$ such that $b_t$ is isomorphism on the smooth locus of $\mscr{C}_t$ for any geometric point $t$ of $T$;
		\item Isomorphisms $a:\mscr{E}\to \underline{\Hom}(\mscr{E},\omega _{\mscr{C}/T}(\sum_i m_i\mscr{D}_i))$.
	\end{enumerate}
\end{lem}

\begin{proof}
	First of all, according to Lemma \ref{Dual_of_TorsionFree}, $\underline{\Hom}(\mscr{E},\omega _{\mscr{C}/T}(\sum_i m_i\mscr{D}_i))$ is a rank-1 relatively torsion-free sheaf on $\mscr{C}$, and we have $$\deg \underline{\Hom}\left(\mscr{E}_t,\omega _{\mscr{C}_t/t}\left(\sum_i m_i\mscr{D}_i\right)\right)=\deg \omega _{\mscr{C}_t/t}\left(\sum_i m_i\mscr{D}_i\right)-\deg \mscr{E}_t=\g-1+\frac{1}{2}\sum\limits_{i=1}^{\sn} m_i=\deg \mscr{E}_t.$$
	By adjunction, there is a one to one correspondence between
	\begin{enumerate}
		\item Homomorphisms $b:\mscr{E}^{\otimes 2}\to \omega _{\mscr{C}/T}(\sum_i m_i\mscr{D}_i)$ such that $b_t$ is isomorphism on the smooth locus of $\mscr{C}_t$ for any geometric point $t$ of $T$;
		\item Homomorphisms $c:\mscr{E}\to \underline{\Hom}\left(\mscr{E},\omega _{C/T}\left(\sum_i m_i\mscr{D}_i\right)\right)$ such that $c_t$ is isomorphism on the smooth locus of $\mscr{C}_t$ for any geometric point $t$ of $T$.
	\end{enumerate}
	Since $c_t$ is an isomorphism on the smooth locus of $\mscr{C}_t$, we see that $\ker c_t=0$, this follows from the fact that a subsheaf of torsion free sheaf is torsion free, hence if $\ker c_t$ is not empty, $\Supp (\ker c_t)$ has dimension 1 so it must have nontrivial intersection with the smooth locus, which is absurd. The equation $\deg \underline{\Hom}(\mscr{E}_t,\omega _{\mscr{C}_t/t}(\sum_i m_i\mscr{D}_i))=\deg \mscr{E}_t$ implies that any homomorphism $c_t$ which is isomorphism on smooth locus is actually an isomorphism, otherwise $\coker (c_t)\neq 0$ and
	\begin{align*}
	\deg \underline{\Hom}\left(\mscr{E}_t,\omega _{\mscr{C}_t/t}\left(\sum_i m_i\mscr{D}_i\right)\right)&=\chi \left(\underline{\Hom}\left(\mscr{E}_t,\omega _{\mscr{C}_t/t}\left(\sum_i m_i\mscr{D}_i\right)\right)\right)-\chi \left(\mathcal O_{C_t}\right)\\
	&=\chi\left(\mscr{E}_t\right)-\chi\left(\mathcal O_{C_t}\right)+\dim _{k(t)}\coker (c_t)>\deg \mscr{E}_t,
	\end{align*}
	where $\chi$ denotes the Euler number. This is a contradiction. Hence $c$ is forced to be an isomorphism.
\end{proof}

\begin{definition}[\textbf{Stack of Punctured Spin Curves}]\label{Defn_Stack_Marked_SpinCurve}
	Fix a scheme $S\equiv \Spec \mathbb Z[1/2]$, a pair of natural numbers $(\g,\n)$ such that $2\g-2+\n> 0$, and an $\n$-tuple $\mathbf m=(m_1,m_2,\cdots ,m_{\sn})$ such that $\sum_i m_i$ is even, define the following prestack in $(\Sch_S)_{\text{\' et}}$:
	\begin{align*}
	(T\in \Sch _S)\mapsto 
	\begin{cases}
	\textbf{\normalfont Obj}: \text{Families of spin curves of genus g over T of puncturing type }{\mathbf m},\\
	\textbf{\normalfont Mor}: \text{Isomorphisms between families of maked spin curves}.
	\end{cases}
	\end{align*}
	It's in fact a stack: every marked stable curve $\mscr C\to S$ is canonically polarized by $\omega_{\mscr C/S}(\sum_iD_i)$. We denote it by $\overline{\mscr{S}}_{\sg,\sn}^{\mbf{m}}$.
\end{definition}

\begin{remark}\label{Relation_Different_m}
	There is an obvious relation between these $\bSgnm${\normalfont :} shifting $\mscr{E}$ to $\mscr{E}\otimes \mathcal O_{\mscr{C}}(\mscr{D}_i)$, we get an isomorphism between stacks{\normalfont :}
	\begin{equation*}
	\overline{\mscr{S}}^{\mathbf {m}}_{\sg,\sn},\cong \overline{\mscr{S}}^{\mathbf {m'}}_{\sg,\sn},
	\end{equation*}
	where $\mathbf{m'}\equiv \mathbf{m}+(0,\cdots,0,2,0,\cdots,0)$.
\end{remark}

\begin{remark}\label{Forgetful_Map_to_Moduli_of_Curves}
	There is an obvious morphism from $\overline{\mscr{S}}^{\mathbf {m}}_{\sg,\sn}$ to $\overline{\mscr{M}}_{\sg,\sn}$ by forgetting the spin structure. 
\end{remark}

\noindent The main result of this appendix is the following

\begin{thr}\label{Main}
	Suppose that $2\g-2+n>0$. Then, $\overline{\mscr{S}}^{\mathbf {m}}_{\sg,\sn}$ is a smooth proper Deligne-Mumford stack over $\Spec \mathbb Z[1/2]$ and contains an open dense substack $\mscr{S}^{\mathbf {m}}_{\sg,\sn}$. $\overline{\mscr{S}}^{\mathbf {m}}_{\sg,\sn}$ is flat and quasi-finite over $\overline{\mscr{M}}_{\sg,\sn}$. Moreover the number of connected components of geometric fibers of $\overline{\mscr{S}}^{\mathbf {m}}_{\sg,\sn}\to \Spec \mathbb Z[1/2]$ is $\gcd (2,m_1,m_2,\cdots,m_n)$.
\end{thr}
In the following sections, we discuss the proof of this Theorem. The different parts of the Theorem is proven in Corollary \ref{cor:compactified spin moduli is a separated DM stack}, Corollary \ref{cor:smoothness of compactified spin moduli}, and Proposition \ref{prop:the number of connected components of compactified spin moduli}. 
\subsection{Rank-One Torsion-Free Sheaves}
Recall the classical fact of local classification of torsion-free sheaves on semistable curves, due to Faltings:

\begin{thr}[\textbf{Theorem 3.5 of} \cite{Faltings1996a}, \textbf{The Rank-One Case}]\label{Faltings}
	Let $R$ be a complete local Noetherian ring with the maximal ideal $\mathfrak m$, then any rank-one relative torsion-free sheaf which is not free over the ring $A=R\dol x,y\dor /(xy-\pi)$, $\pi\in \mathfrak m$ is isomorphic to $E(p,q)$, for $p,q\in \mathfrak m$ and $pq=\pi$, defined as the image of
	\begin{align*}
	\alpha=
	\begin{bmatrix}
	x       & q  \\
	p       & y 
	\end{bmatrix}
	: A^{\oplus 2}\to A^{\oplus 2}.
	\end{align*}
	Moreover $E(p,q)\cong E(p',q')$ if and only if $\exists u\in R^{\times}$ such that $p'=up$, $q'=q/u$.
\end{thr}

\begin{remark}
	This theorem is also true for $R$ being a Henselian local G-ring, by Artin's Approximation Theorem.
\end{remark}

In fact the module $E(p,q)$ in the Theorem {\normalfont \ref{Faltings}} has a nice geometric description. Consider the projective scheme 
\begin{align}\label{Blow_Up_Equation}
X=\Proj A[z,w]/(xz-pw,qz-yw).
\end{align}
We can check that $X$ is flat over $R$: on the chart $D_+(w)$, $X$ is given by $$\Spec A[T]/(qT-y,p-xT)=\Spec R[x,y,T]^{\wedge}/(qT-y,p-xT).$$Since $\{xT,y\}$ is a regular sequence on the special fiber, $R[x,y,T]/(qT-y,p-xT)$ is flat over $R$ by local criterion for flatness, so its completion  is also flat over $R$. Note that the special fiber of $X\to \Spec R$ is a $\mathbb P^1$-arc connecting closed points of $\Spec \mbb{K}[\![x]\!]$ and $\Spec \mbb{K}[\![y]\!]$ with nodal singularity, where $\mbb{K}$ is the residue field of $R$, while $f:X\to \Spec A$ on the open locus {\normalfont (}possibly empty{\normalfont )} $\Spec R-V(p,q)=D(p)\cup D(q)$ is an isomorphism.\\

Now we can push forward the line bundle $\mathcal O(1)$ on $X$ to $\Spec A$, which can be computed explicitly: $$f_*\mathcal O(1)=\ker \left(A[T]/(qT-y,p-xT)\oplus A[S]/(q-yS,pS-x)\to A[T,T^{-1}]/(qT-y,p-xT)\right),$$where the map is given by embedding the first direct summand naturally and send $S$ to $T^{-1}$ meanwhile multiply the second direct summand by $T$. Notice that $A[T]/(qT-y,p-xT)$ has a free resolution as $A[T]$-module: 
\begin{center}
	\begin{tikzcd}
	\cdots\arrow[r] &A[T]^{\oplus 2}\arrow[r,"\alpha"]& A[T]^{\oplus 2}\arrow[r,"\beta"] &A[T]^{\oplus 2}\arrow[r,"\alpha"] & A[T]^{\oplus 2}\arrow[r,"\beta_1"]& A[T],
	\end{tikzcd}
\end{center}
where 
\begin{align*}
\alpha=
\begin{bmatrix}
x       & q  \\
p       & y 
\end{bmatrix},\quad
\beta=
\begin{bmatrix}
y       & -q  \\
-p       & x 
\end{bmatrix},\quad
\beta_1=
\begin{bmatrix}
qT-y       & p-xT  
\end{bmatrix}.
\end{align*}
This is a resolution since it is a resolution modulo $\mathfrak{m}$ and $A[T]$ is flat over $R$. For the same reason, $A[S]/(q-yS,pS-x)$ has a free resolution as $A[S]$-module{\normalfont :}
\begin{center}
	\begin{tikzcd}
	\cdots\arrow[r] &A[S]^{\oplus 2}\arrow[r,"\alpha"]& A[S]^{\oplus 2}\arrow[r,"\beta"] &A[S]^{\oplus 2}\arrow[r,"\alpha"] & A[S]^{\oplus 2}\arrow[r,"\beta_2"]& A[S],
	\end{tikzcd}
\end{center}
where 
\begin{align*}
\beta_2=
\begin{bmatrix}
q-yS       & pS-x
\end{bmatrix},
\end{align*}
and $A[T,T^{-1}]/(qT-y,p-xT)$ has a free resolution as $A[T,T^{-1}]$-module{\normalfont :}
\begin{center}
	\begin{tikzcd}
	\cdots\arrow[r]& A[T,T^{-1}]^{\oplus 2}\arrow[r,"\beta"] &A[T,T^{-1}]^{\oplus 2}\arrow[r,"\alpha"] & A[T,T^{-1}]^{\oplus 2}\arrow[r,"\beta_3"]& A[T,T^{-1}],
	\end{tikzcd}
\end{center}
where 
\begin{align*}
\beta_3=
\begin{bmatrix}
qT-y       & p-xT  
\end{bmatrix},
\end{align*}
It's easy to see that the original map lifts to a map between complexes{\normalfont :}
\begin{center}
	\begin{tikzcd}
	\cdots\arrow[r] &A[T]^{\oplus 2}\oplus A[S]^{\oplus 2} \arrow[r,"\alpha\oplus \alpha"] \arrow[d,"\text{(Id,Id)}"]& A[T]^{\oplus 2}\oplus A[S]^{\oplus 2}\arrow[r,"\beta_1\oplus\beta_2"]\arrow[d,"\text{(Id,Id)}"]& A[T]\oplus A[S]\arrow[d,"\text{(Id,T)}"]\\
	\cdots\arrow[r] &A[T,T^{-1}]^{\oplus 2}\arrow[r,"\alpha"] & A[T,T^{-1}]^{\oplus 2}\arrow[r,"\beta_3"]& A[T,T^{-1}].
	\end{tikzcd}
\end{center}
Obviously, vertical arrows are surjective, so we end up with a chain of kernels{\normalfont :}
\begin{center}
	\begin{tikzcd}
	\cdots\arrow[r] & A^{\oplus 2}\arrow[r,"\beta"]& A^{\oplus 2}\arrow[r,"\alpha"] &A^{\oplus 2}\arrow[r,"\beta"] & A^{\oplus 2}.
	\end{tikzcd}
\end{center}
It is easy to see that this chain is exact except for the terminal place, and it gives rise to an exact sequence
\begin{center}
	\begin{tikzcd}
	A^{\oplus 2}\arrow[r,"\beta"]& A^{\oplus 2}\arrow[r] & f_*\mathcal O(1)\arrow[r] &0.
	\end{tikzcd}
\end{center} 
As a result, $f_*\mathcal O(1)\cong \im (\alpha)=E(p,q)$.

\begin{remark}
	The identification $f_*\mathcal O(1)\cong \coker (\beta)$ implies the following isomorphism of $A$-algebras:$$\bigoplus_{n\ge 0}\Sym ^n (E(p,q))\cong \bigoplus_{n\ge 0}\Sym ^n (A^{\oplus 2})/\text{\normalfont relations}=A[z,w]/(xz-pw,qz-yw),$$via identifying two generators of $A^{\oplus 2}$ with $z,w$, hence the relations are exactly $xz-pw=0$ and $qz-yw=0$. Consequently $$X\cong \mathbb P(E(p,q)).$$
\end{remark}

\subsection{Cornalba's Definition of Spin Curves}
Above local characterization of rank one torsion free sheaves indicates that the rank one torsion free sheaf defnining the spin structure should be a pushforward of a locally free sheaf on a certian "blow-up" of the family of curves. Indeed, suppose that $\pi:\mscr C\to T$ is a family of spin curves with puncturing divisors $\{\mscr D_1,\mscr D_2,\cdots,\mscr D_n\}$ and spin structure $\mscr E$, then $\widetilde{\pi}:\widetilde{\mscr C}:=\mathbb P(\mscr E)\to T$ is a flat family of nodal curves (not necessarily stable) over $T$, together with a distinct line bundle $\mathcal L=\mathcal O(1)$ such that $$f_*\mathcal L\cong \mscr E,$$where $f:\widetilde{\mscr C}\to \mscr C$ is the natural projection. Moreover the structure map $b:\mscr E^{\otimes 2}\to \omega _{\mscr C/T}(\sum_{i}m_i\mscr D_i)$ induces a homomorphism by adjunction
\begin{equation*}
\alpha:\mathcal L^{\otimes 2}\to \omega _{\widetilde{\mscr C}/T}\left(\sum_{i}m_iD_i\right),
\end{equation*}
which is an isomorphism outside of the locus of exceptional curves.

\newl In \cite{Cornalba1989a}, Cornalba defines a spin curve over $T$ as a  \textit{semi-stable} curve $\mscr C$ over $T$ with a line bundle $\mcal L$ and a homomorphism $\alpha:\mcal L^{\otimes 2}\to \omega _{{\mcal C}/T}$ such that for all geometric point $t$ of $T$, $b$ is an isomorphism outside of the rational components of $\mscr C_t$, and $\mcal L_t$ has degree one on each rational component. An isomorphism between spin curves $\mscr C,\mscr C'$ is defined to be a pair of isomorphisms $\sigma:\mscr C\to \mscr C'$ and $\varphi:\sigma^*\mcal L'\to \mcal L$ which is compatible with the canonical isomorphism of $\sigma^*\omega_{\mscr C'/T}\to \omega_{\mscr C/T}$. In this way we can define the stack of spin curves over a base scheme $S$ by 
\begin{align*}
(T\in \Sch _S)\mapsto 
\begin{cases}
\text{\normalfont\bf Obj}: \text{Spin curves of genus-}\g \text{ in Cornalba's sense over $T$},\\
\text{\normalfont\bf Mor}: \text{Isomorphisms between families of spin curves}.
\end{cases}
\end{align*}
It is in fact a stack: every spin curve $\mscr C$ over $T$ with spinor bundle $\mcal L$ is canonically polarized by $\omega_{\mscr C/S}\otimes\mcal L$. We denote it by $\overline{\mscr{S}}'_{\sg}$. By what we discussed above, there is a canonical morphism $p:\overline{\mscr{S}}_{\sg}\to \overline{\mscr{S}}'_{\sg}$ defined by$$(\mscr C,\mscr E)\mapsto(\mbb P(\mscr E),\mcal O(1)).$$
It is easy to see from the above {\it blowing-up} construction that over an algebraic closed field $\mbb K$, $p$ is an equivalence between categories. In fact, $p$ is an isomorphism between stacks, proven below.
\begin{lem}\label{Fully_Faithfulness}
	$p:\overline{\mscr{S}}_{\sg}\to \overline{\mscr{S}}'_{\sg}$ is fully faithful.
\end{lem}
\begin{proof}
	Fix a scheme $T$ and let $(\mscr C_1,\mscr E_1)$ and $(\mscr C_2,\mscr E_2)$ be two objects in $\overline{\mscr{S}}_{\sg}(T)$, and let $(\widetilde{\mscr C}_i,\mcal L_i)$ be $(\mbb P(\mscr E_i),\mcal O(1))$. We ought to show that $\Isom_T((\mscr C_1,\mscr E_1),(\mscr C_2,\mscr E_2))\to \Isom_T((\widetilde{\mscr C}_1,\mcal L_1),(\widetilde{\mscr C}_2,\mcal L_2))$ is bijective.
	
	\newl \textbf{Injectivity:} Suppose that there are two isomorphisms $\phi,\psi:(\mscr C_1,\mscr E_1)\to (\mscr C_2,\mscr E_2)$ which give rise to the same isomorphism $\varphi:(\widetilde{\mscr C}_1,\mcal L_1)\to(\widetilde{\mscr C}_2,\mcal L_2)$, then the $\phi$ and $\psi$ induces the same map for underlying topolgical space, because $\widetilde{\mscr C}_1\to {\mscr C}_1$ is surjective. Furthermore, $\phi^{\#}:\phi^{-1}\mcal O_{{\mscr C}_2}\to \mcal O_{{\mscr C}_1}$ is the same as $\pi_{1,*}(\varphi^{\#}):\pi_{1,*}\varphi^{-1}\mcal O_{\widetilde{\mscr C}_2}\to \pi_{1,*}\mcal O_{\widetilde{\mscr C}_1}$, and so is $\psi^{\#}$, thus $\phi^{\#}=\psi^{\#}$, i.e. $\phi$ and $\psi$ are the same for the underlying stable curve. Similar argument for spin structure shows that $\phi$ agrees with $\psi$ on spinor sheaves, hence $\phi=\psi$.
	
	\newl \textbf{Surjectivity:} Suppose that there is an isomorphism $\varphi:(\widetilde{\mscr C}_1,\mcal L_1)\to(\widetilde{\mscr C}_2,\mcal L_2)$. Then exceptional curves of $\widetilde{\mscr C}_1$ are identified with exceptional curves of $\widetilde{\mscr C}_2$ under $\varphi$, so they are contracted to single points by the map $\pi_2\circ \varphi$, thus there is a continuous map $\psi$ between topological spaces $\mscr C_1$ and $\mscr C_2$ (because $\pi_1$ is submersive) such that $\psi\circ \pi_1=\pi_2\circ \varphi$, hence $\psi$ is a homeomorphism. It follows that $\psi^{-1}\mcal O_{{\mscr C}_2}=\pi_{1,*}\varphi^{-1}\mcal O_{\widetilde{\mscr C}_2}$. Define $\psi^{\#}:\psi^{-1}\mcal O_{{\mscr C}_2}\to \mcal O_{{\mscr C}_1}$ by $\pi_{1,*}(\varphi^{\#})$, which is an isomorphism between sheaves of rings and also $\mcal O_T$-linear, thus $\psi$ is a isomorphism between $T$-schemes. Finally, define the isomorphism between spinor sheaves by $\pi_{1,*}$ of the isomorphism between $\varphi^*\mcal L_2$ and $\mcal L_1$.
\end{proof}

\begin{lem}\label{Relative_Representability}
	The diagonal of $\overline{\mscr{S}}'_{\sg}$ is represented by separated and locally quasi-finite scheme.
\end{lem}

\begin{proof}
	Since spin curves are projective, $\Isom_T(\mscr C_1,\mscr C_2)$ is a locally closed subscheme of $\Hilb_{\mscr C_1\times_T\mscr C_2/T}$, which is separated and locally finite type over $T$. For each geometric point $t$ of $T$, $\Isom _t(\mscr C_{1,t},\mscr C_{2,t})$ has finite many points (Lemma 2.2 of \cite{Cornalba1989a}), so $\Isom_T(\mscr C_1,\mscr C_2)$ is locally quasi-finite.
\end{proof}

\begin{proposition}
	$p:\overline{\mscr{S}}_{\sg}\to \overline{\mscr{S}}'_{\sg}$ is isomorphism.
\end{proposition}

\begin{proof}
	By Lemma \ref{Fully_Faithfulness} and \ref{Relative_Representability}, for every scheme $T$ and a $T$-point $g:T\to\overline{\mscr{S}}'_{\sg}$, $p^{-1}T$ is represented by an algebraic space, which is separated and locally finite type over $T$. For a geometric point $s$ of $\overline{\mscr{S}}'_{\sg}$, $p^{-1}(s)$ has only one point, thus $p$ is universally injective. Furthermore, comparing Proposition \ref{Deform_Spin_Moduli} and the construction of universal deformation in \cite{Cornalba1989a}, we see that $p$ induces isomorphism between universal deformation basis, thus $p$ is \' etale. Applying Zariski's Main Theorem \cite{zariskimaintheorem} to the separated, universally injective, and \' etale morphism $p$, we see that $p$ is an isomorphism.
\end{proof}

\subsection{Stack of Prespin Structures and Representability of $\overline{\mscr{S}}^{\mbf{m}}_{\sg,\sn}$}

In this section, we describe the stack of prespin curves. We begin by giving the definition of the relevant concepts.

\begin{definition}[\textbf{Prespin Curve}]\label{Defn_RelPrespinCurve}
	A family of prespin curves of genus $\g$ with $\n$ punctures over the base scheme $T$, of puncturing number $N\in 2\mathbb Z$, is a triple $(\mscr{C},\mbs{\mscr D},\mscr{E})$, where $(\mscr{C},\mbs{\mscr D})\longrightarrow T$ is an $\n$-punctured stable curve, $\mscr{E}$ is a rank-1 relatively torsion-free sheaf on $C$, such that for every geometric point $t$ of $T$
	$$\deg \mscr{E}_t=\g-1+\frac{1}{2}N.$$
\end{definition}

\begin{definition}[\textbf{Stack of Prespin Curves}]\label{Defn_Stack_PrespinCurve}
	Fix a scheme $S\equiv \Spec \mathbb Z[1/2]$, a pair of natural numbers $(\g,\n)$ such that $2\g-2+\n> 0$, define the following prestack in $(\Sch_S)_{\text{\' et}}$:
	\begin{align*}
	(T\in \Sch _S)\mapsto 
	\begin{cases}
	\textbf{\normalfont Obj}: \text{Families of (\g,\n) prespin curves over T of marking number N},\\
	\textbf{\normalfont Mor}: \text{Isomorphisms between families of prespin curves}/(\mbb{I},-1).
	\end{cases}
	\end{align*}
	Note that we do not distinguish two isomorphisms if they are related by $(\mbb{I},-1)$. We denote the stackification of the above prestack by $\spspin$. There is a morphism $\pi: \spspin\to \cscurve$ forgetting the sheaf $\mscr{E}$, and there is a morphism $p: \csspinm\to \csspinN$ forgetting $b:\mscr{E}^{\otimes 2}\to \omega_{\mscr{C}/T}\left(\sum_im_i\mscr{D}_i\right)$, where $N=\sum_i m_i$.
\end{definition}

\begin{proposition}\label{Prespin_to_Mg}
	$\pi: \spspin\to \cscurve$ is relatively represented by finite-type quasi-separated Artin stacks with separated diagonal.
\end{proposition}

\begin{proof}
	For any $T$-point of $\cscurve$, i.e. an $\n$-punctured stable curve $(\mscr{C},\mbs{\mscr{D}})/T$, $\spspin\times_{\cscurve}T$ is the stack of rank-1 relatively torsion-free sheaves of degree $\g-1+N/2$ on $(\Sch _T)_{\text{\' et}}$. Since rank-1 relatively torsion free-sheaves are obviously semistable, so $\spspin\times_{\cscurve}T$ is a substack of $\mathcal M^{1,\sg-1+N/2}_{\mscr{C}/T}$, i.e. stack of rank-1 semistable sheaves of degree $\g-1+N/2$. It is well-known that $\mathcal M^{1,\sg-1+N/2}_{\mscr{C}/T}$ is a finite-type quasi-separated Artin stack with separated diagonal. Since being relatively torsion-free is an open condition on the base, $\spspin\times_{\cscurve}T$ is a finite-type quasi-separated Artin stack with separated diagonal. 
\end{proof}

\begin{proposition}\label{Spin_to_Prespin}
	$p: \csspinm\to \spspin$ is relatively represented by finite type quasi-affine schemes.
\end{proposition}

\begin{proof}
	For any $T$-point of $\spspin$, i.e. a $n$-pointed prespin curve $(C,\mbs{\mscr D},\mathcal E)/T$, $\csspinm\times _{\spspin}T$ is the functor of isomorphisms $a:\mathcal E\to \underline{\Hom}(\mathcal E,\omega _{C/T}(\sum_i m_i\mscr{D}_i))$. The next lemma says that $\csspinm\times _{\spspin}T$ is represented by a finite type quasi-affine scheme, hence $p: \csspinm\to \spspin$ is relatively represented by finite type quasi-affine schemes.
\end{proof}

\begin{lem}
	Suppose that $X$ is proper over a Noetherian scheme $S$, $\mscr F$ and $\mscr G$ are coherent sheaves on $X$ and flat over $S$, then $\underline{\Isom}(\mscr F,\mscr G)$ is represented by a scheme finite type and quasi-affine over $S$.
\end{lem}

\begin{proof}
	The functor $T\mapsto \Hom _T(\mscr F_T,\mscr G_T)$ is represented by a scheme finite type and affine over $S$ (see section 7.7.8, 7.7.9 of \cite{EGA3} ), denoted by $M$. The locus in $M$ where $\mscr F\to \mscr G$ is surjective is open since it's the complement of the image of $\Supp (\coker (\mscr F\to \mscr G))$ in $M$, the latter is closed since $X$ is proper over $S$. Now in the locus $U\subset M$ where $\mscr F\to \mscr G$ is surjective, the kernel of $\mscr F\to \mscr G$, denoted by $\mscr K$, is flat over $S$, since both $\mscr F$ and $\mscr G$ are flat over $S$. Hence $\underline{\Isom}(\mscr F,\mscr G)$ is represented by the vanishing locus of $\mscr K$, i.e. the complement of the image of $\Supp (\mscr K)$ in $U$, which is open by the properness of $X$.
\end{proof}

\noindent Combining propositions \ref{Prespin_to_Mg} and \ref{Spin_to_Prespin}, we conclude that

\begin{corollary}
	$\csspinm$ is represented by a finite-type quasi-separated Artin stack over $\Spec \mathbb Z[1/2]$ with separated diagonal.
\end{corollary}

\subsection{Diagonal of $\csspinm$}

\noindent We are going to to describe the diagonal $\Delta_f$ of $\csspinm$ by investigating the diagonal of morphism $f:\csspinm\to \cscurve$.

\begin{proposition}\label{Aut_Field}
	Suppose that $\mbb{K}$ is an algebraically-closed field, $p$ is a $\mbb{K}$-point of $\csspinm$, represented by a punctured spin curve $(\mscr{C},\mbs{\mscr{D}},\mscr{E},b)$ over $\mbb{K}$. Let $\{N_1,\cdots,N_r\}$ be the set of nodes where $\mscr{E}$ is not locally-free, $\widetilde{\mscr{C}}$ be the curve $\mscr{C}$ normalized at $\{N_1,\cdots,N_r\}$, then the group of automorphisms of $(\mscr{C},\mbs{\mscr{D}},\mscr{E},b)$ which induce identity on $(\mscr{C},\mbs{\mscr{D}})$ is
	\begin{equation*}
	(\mathbb Z/2\mathbb Z)^{\oplus s},
	\end{equation*}
	where $s$ is the number of connected components of $\widetilde{\mscr{C}}$.
\end{proposition}

\begin{proof}
	It is easy to see that there is a one-to-one correspondence between  automorphisms of $(\mscr{C},\mbs{\mscr{D}},\mscr{E},b)$ and automorphisms of $(\widetilde{\mscr{C}},\mbs{\mscr{D}},\pi ^*\mscr{E},\pi^*b)$,where $\pi:\widetilde{\mscr{C}}\longrightarrow\mscr{C}$ is the normalization morphism at NS punctures. Since $\pi ^*\mscr{E}$ is locally free, automophism of $\pi ^*\mscr{E}$ is given by $\mbb{K}^{\times s}$, hence those automorphisms compatible with $\pi^*b$ are sqaure roots of $(1,\cdots,1)\in \mbb{K}^{\times s}$, i.e. $(\mathbb Z/2\mathbb Z)^{\oplus s}$.
\end{proof}

\begin{proposition}
	$\Delta_f$ is unramified.
\end{proposition}

\begin{proof}
	Suppose that $\mbb{K}$ is an algebraically-closed field, $p$ is a $\mbb{K}$-point of $\csspinm$, represented by a punctured spin curve $(\mscr{C},\mbs{\mscr{D}},\mscr{E},b)$ over $\mbb{K}$. We want to show that the automorphism group scheme of $(\mscr{C},\mbs{\mscr{D}},\mscr{E},b)$ which fixes $(\mscr{C},\mbs{\mscr{D}})$ is \' etale, or equivalently, the tangent space at identity is trivial. 
	
	Let $D_{\epsilon}\equiv\Spec \mbb{K}[\epsilon]/\epsilon^2$. Consider an automorphism $\alpha$ of $\mscr{E}\times D_{\epsilon}$ on $\mscr{C}\times D_{\epsilon}$ which reduces to identity on the special fiber. Then $\alpha$ induces an automorphism $\widetilde {\alpha}$ of $\pi^* \mscr{E}\times D_{\epsilon}$ on $\widetilde {\mscr{C}}\times D_{\epsilon}$, and $\widetilde {\alpha}$ reduces to identity on the special fiber, so $$\widetilde {\alpha}\in (1+\mbb{K}\epsilon)^{\oplus s}.$$Moreover $\widetilde {\alpha}^2=\mbb{I}$, which implies that $\widetilde {\alpha}=\mbb{I}$, since $2$ is invertible in $\mbb{K}$. As a result we see that $\alpha=\mbb{I}$, i.e. the tangent space at identity is trivial. 
\end{proof}

\begin{proposition}
	$\Delta_f$ is proper.
\end{proposition}

\begin{proof}
	Let $R$ be a complete discrete valuation ring with algebraically-closed residue field $\mbb{K}$ and fraction field $\mbb{K}_{\eta}$, $(\mscr{C},\mbs{\mscr{D}},\mscr{E},b)$ is a relative punctured spin curve on $\Spec R$, and $\alpha_{\eta}$ is an automorphism of $\mscr{E}$ on generic fiber which is compatible with $b$. Using the valuative criterion, we need to show that $\alpha_{\eta}$ extends to an automorphism $\alpha$ on the whole $\mscr{C}$ and is still compatible with $b$.
	
	Suppose that $\{N_1,\cdots,N_r\}$ is the set of nodes of generic fiber $\mscr{C}_{\eta}$ where $\mscr{E}_{\eta}$ is not locally free. We can assume that $N_i$ is a geometric point of $\mbb{K}_{\eta}$ otherwise we can always take a finite extension of $R$ to make it a geometric point. Note that the valuative criterion still holds for a finite extension of $R$. So $N_i$ extends to a section of $R$, which is still denoted by $N_i$. Completion of the local ring of $\mscr{C}$ at the closed point of $N_i$ is $[\![R\dol x,y\dor /(xy)]\!]$. According to the Faltings' classification theorem (Theorem \ref{Faltings}), the completion of $\mscr{E}$ at the closed point of $N_i$ is of type $(0,0)$. Let $\pi:\widetilde{\mscr{C}}\to \mscr{C}$ be the normalization of $\mscr{C}$ at sections $\{N_1,\cdots,N_r\}$, then the adjunction $\mscr{E}\to \pi_*\pi^*\mscr{E}$ is an isomorphism. Thus, it is enough to extend $\pi^*\alpha_{\eta}$ to an automorphism of $\pi^*\mscr{E}$. According to Proposition \ref{Aut_Field}, $\pi^*\alpha_{\eta}$ is a direct product of $\pm 1$ on connected components of $\widetilde{\mscr{C}}_{\eta}$, so we can extend it tautologically.
\end{proof}

As a corollary, we have the following

\begin{corollary}\label{cor:compactified spin moduli is a separated DM stack}
	$\csspinm$ is a separated Deligne-Mumford stack.
\end{corollary}

\subsection{Deformation Theory of $\csspinm$}

In this section, we discuss the deformation theory of $\csspinm$. We begin with the following proposition 

\begin{proposition}\label{Deform_Spin_Moduli}
	Suppose that $\mbb{K}$ is an algebraically-closed field, $p$ is a $\mbb{K}$-point of $\csspinm$, represented by a punctured spin curve $(\mscr{C},\mbs{\mscr{D}},\mscr{E},b)$ over $\mbb{K}$. Let $\{N_1,\cdots,N_r\}$ be the set of nodes where $\mscr{E}$ is not locally free, then the universal deformation of $(\mscr{C},\mbs{\mscr{D}},\mscr{E},b)$ is the pull-back of the universal deformation of punctured stable curve $(\mscr{C},\mbs{\mscr{D}})$ under homomorphism
	\begin{equation}
	\mathfrak{o}_{\mbb{K}}[\![ t_1,\cdots,t_r,t_{r+1},\cdots,t_{3\sg-3+\sn}]\!]\to \mathfrak{o}_{\mbb{K}}[\![ \tau_1,\cdots,\tau_r,t_{r+1},\cdots,t_{3\sg-3+\sn}]\!],
	\end{equation}
	$\mfk{o}_{\mbb{K}}$ is the Cohen ring with residue field $\mbb K$ {\normalfont \cite{DeligneMumford1969a}}, by sending $t_i$ to $\tau_i^2$ for $i\in \{1,2,\cdots, r\}$. 
\end{proposition}

\begin{proof}
	Recall that it is shown in \cite{DeligneMumford1969a} that the universal deformation of $(\mscr{C},\mbs{\mscr{D}})$ is
	\begin{equation*}
	R\equiv\mathfrak{o}_k[\![t_1,\cdots,t_{3\sg-3+\sn}]\!].
	\end{equation*}
	At a node of $\mscr{C}$, the complete local ring of the universal deformation of $\mscr{C}$ is $$R\dol x,y\dor /(xy-t_j),$$for some index $j$. According to the Faltings' classification theorem (Theorem \ref{Faltings}), any deformation of the completion of $\mscr{E}$ at $N_i$ is uniquely determined by a parameter $\tau_i$ such that $\tau_i^2=t_i$. On the other hand, for any Artinian local ring $A$ with residue field $\mbb{K}$, and choose a deformation of $(\mscr{C},\mbs{\mscr{D}})$ with base scheme $\Spec A$, denote it by $(\mscr{C}_A,\mbs{\mscr{D}}_A)$, we claim that $\mscr{E}|_{\mscr{C}-\{N_1,\cdots,N_r\}}$ has a unique deformation to $\mscr{C}_A-\{N_1,\cdots,N_r\}$ as a square root of $\omega_{\mscr{C}_A}(\sum_im_i\mbs{\mscr{D}}_{A,i})$. This can be proven as follows. First we notice that there is an exact sequence\footnote{Note that $U_{\text{\' et}}$ means the we equip $U$ with \'etale topology, and cohomologies are taken with respect to the \' etale topology of $U$. $\Pic (U)$ is the Picard group of $U$.}
	\begin{center}
		\begin{tikzcd}
		\H ^1(U_{\text{\' et}}, \mu_2)\arrow[r] &\Pic (U)\arrow[r,"2"] &\Pic (U)\arrow[r]&\H ^2(U_{\text{\' et}}, \mu_2),
		\end{tikzcd}
	\end{center}
	where $U=\mscr{C}-\{N_1,\cdots,N_r\}$, $\mu_2$ is the sheaf of the $2$\textsuperscript{nd} roots of unity, and the first and the last terms are $2$-torsion groups. There is an exact sequence on $U_A$, the open subscheme of $\mscr{C}_A$ with underlying topological space $U$, as well. They are related by restriction map:
	\begin{center}
		\begin{tikzcd}
		\H ^1(U_{\text{\' et}}, \mu_2)\arrow[r] \arrow[d,equal] &\Pic (U_A)\arrow[d,"\text{res}",two heads]\arrow[r,"2"] &\Pic (U_A)\arrow[d,"\text{res}",two heads]\arrow[r]&\H ^2(U_{\text{\' et}}, \mu_2)\arrow[d,equal] \\
		\H ^1(U_{\text{\' et}}, \mu_2)\arrow[r] &\Pic (U)\arrow[r,"2"] &\Pic (U)\arrow[r]&\H ^2(U_{\text{\' et}}, \mu_2).
		\end{tikzcd}
	\end{center}
	Notice that the kernel of restriction map, $\ker (\text{res})$ has a composition series of quotient of $\mbb K$-linear spaces where $2$ is invertible. Since $\omega_{\mscr{C}}(\sum_im_i\mbs{\mscr{D}}_{i})$ has a square root, from the commutative diagram above, we see that $\omega_{\mscr{C}_A}(\sum_im_i\mbs{\mscr{D}}_{A,i})$ also has a square root, we want to show that there is a one-to-one correspondence between square roots of $\omega_{\mscr{C}_A}(\sum_im_i\mbs{\mscr{D}}_{A,i})$ and square roots of $\omega_{\mscr{C}}(\sum_im_i\mbs{\mscr{D}}_{i})$. Suppose that there are two line bundles on $\mscr{C}_A$ denoted by $\mathcal L$ and $\mathcal L'$, both square to $\omega_{\mscr{C}_A}(\sum_im_i\mbs{\mscr{D}}_{A,i})$, and their restrictions to $\mscr{C}$ are isomorphic, then $\mathcal L^{-1}\otimes \mathcal L'\in\ker (\text{res})$. However, $\mathcal L^{-1}\otimes \mathcal L'$ is in the image of $\H ^1(U_{\text{\' et}},\mu_2)$ which is $2$-torsion, so $\mathcal L^{-2}\otimes \mathcal L'^{2}\cong \mathcal O_{U_A}$. This implies that $\mathcal L^{-1}\otimes \mathcal L'\cong \mathcal O_{U_A}$ since $2$ is invertible in $\ker (\text{res})$. Thus we see that there is a unique deformation of $\mscr{E}|_{U}$.
	
	Finally, we are going to show that given any set of deformations of the completion of $\mscr{E}$ at $\{N_1,\cdots,N_r\}$, there exists a unique deformation of $\mscr{E}$ which gives rise to deformations of the completion of $\mscr{E}$ at $\{N_1,\cdots,N_r\}$. This is achieved by observing that all we need is the formal gluing at generic points of $\widehat{\mathcal O}_{N_i,\mscr C_A}$, the completion of ${\mathcal O}_{N_i,\mscr C_A}$ at the maximal ideal. Note that $\omega_{\mscr{C}_A}(\sum_im_i\mbs{\mscr{D}}_{A,i})$ naturally carries a formal gluing data, which is a set $\{a_1,b_1,\cdots,a_r,b_r\}\in A(\!(x)\!)^{\times 2r}$. Since we already have formal gluing data coming from $\mscr{E}$ when restrict to $\mscr{C}$, i.e. a set of square roots of $\{\bar a_1,\bar b_1,\cdots,\bar a_r,\bar b_r\}\in \mbb{K}(\!(x)\!)^{\times 2r}$ where $\bar a_i$ (resp. $\bar b_i$) is $a_i$ (resp. $b_i$) modulo the maximal ideal of $A$. Since $2$ is invertible in $A$, square roots lift uniquely in $A(\!(x)\!)$, hence there exists a unique deformation of formal gluing data. 
	
	To sum up, we have shown that formal deformation of $\mscr{E}$ is uniquely determined by deformations of completions $\widehat{\mscr{E}}_{N_i}$, which is determined by a square root of $t_i$. Hence the universal deformation of $(\mscr{C},\mbs{\mscr{D}},\mscr{E},b)$ is the pull-back of the universal deformation of punctured stable curve $(\mscr{C},\mbs{\mscr{D}})$ under homomorphism
	\begin{equation}
	\mathfrak{o}_{\mbb{K}}[\![ t_1,\cdots,t_r,t_{r+1},\cdots,t_{3\sg-3+\sn}]\!] \to \mathfrak{o}_{\mbb{K}}[\![ \tau_1,\cdots,\tau_r,t_{r+1},\cdots,t_{3\sg-3+\sn}]\!],
	\end{equation}
	by sending $t_i$ to $\tau_i^2$ for $i\in \{1,2,\cdots, r\}$.
\end{proof}

\noindent We denote the preimage of $\scurve$ under the natural projection $f:\csspinm\to\cscurve$ by $\sspin$, then we can easily see from the proposition \ref{Deform_Spin_Moduli} that

\begin{corollary}\label{cor:smoothness of compactified spin moduli}
	$\csspinm$ is a smooth stack, and contains $\sspin$ as an open dense substack. Moreover $f:\csspinm\to \cscurve$ is flat and quasi-finite.
\end{corollary}

\subsection{Properness of $\csspinm$}

In this section, we prove the following proposition.

\begin{proposition}
	The projection $f:\csspinm\to \cscurve$ is proper.
\end{proposition}

\begin{proof}
	Consider the following data: a complete discrete valuation ring (DVR) $(R,\pi )$ with algebraically-closed residue field $\mbb{K}$ and generic point $\Spec \mbb{K}(\eta)$, a relative punctured curve $(\mscr{C},\mbs{\mscr{D}})$ on $\Spec R$ such that $\mscr{C}_{\eta}$ is smooth, and a relative punctured spin curve $(\mscr{C}_{\eta},\mbs{\mscr{D}}_{\eta},\mscr{E}_{\eta},b_{\eta})$ on $\Spec  \mbb{K}(\eta)$. We are about to show that there exists a pair $(\mscr{E},b)$ extending $(\mscr{E}_{\eta},b_{\eta})$ on the generic fiber which makes $(\mscr{C},\mbs{\mscr{D}},\mscr{E},b)$ a relative punctured spin curve.
	
	Denote the set of generic points of the special fiber $\mscr{C}_s$ by $\{\xi_1,\cdots ,\xi_e\}$, and the local ring of $\mscr{C}$ at $\{\xi_1,\cdots ,\xi_e\}$ is $\{R_1,\cdots ,R_e\}$. Note that $R_i$ are DVRs. Since $\mscr{E}_{\eta}$ is torsion-free, its localiztion at $R_{i\eta}$ extends naturally to a free module of rank one on $R_i$, and can be furthermore extended to an open neighborhood of $\Spec R_i$. Assume that those open neighborhoods do not intersect with each other. Let $j:U\hookrightarrow \mscr{C}$ be the union of $\mscr{C}_{\eta}$ with those neighborhoods, and the extension of $\mscr{E}_{\eta}$ be $\mscr{E}_U$. Since $\mscr{C}$ is Cohen-Macaulay\footnote{This basically means that the depth (maximal length of regular sequence) of every local ring is equal to its Krull dimension.}, $\mscr{E}_U$ is torsion-free, and $C-U$ is of codimension $2$, so $j_*\mscr{E}_U$ is coherent, and satisfies the Serre's $S_2$ condition. Define $\mscr{E}:=j_*\mscr{E}_U$. It is torsion-free over $R$ hence flat over $R$. As a result, $\mscr{E}_s$ satisfies the Serre's $S_1$ condition, i.e. it is torsion-free. So $\mscr{E}$ is relatively torsion-free of rank one, and has the same degree as $\mscr{E}_{\eta}$.
	
	For the morphism $b$, it is equivalent to find an extension of the isomorphism $a_{\eta}:\mscr{E}_{\eta}\cong \underline{\Hom}(\mscr{E}_{\eta},\omega _{\mscr C_{\eta}}(\sum_i m_i\mscr{D}_i))$ (Lemma \ref{Equiv_a_b}). Notice that $a_{\eta}$ extends to $U$ (up to a shrinking of $U$), and up to multiplication by a power of uniformizer $\pi$ on each neighborhood $U_i$, we can assume that restriction of $a$ to $R_i$ is isomorphism. By adjunction
	$$\Hom\left(\mscr{E},j_*\underline{\Hom}\left(\mscr{E}_{U},\omega _{\mscr C_{U}}\left(\sum_i m_i\mscr{D}_i\right)\right)\right)\cong\Hom\left(\mscr{E}_{U},\underline{\Hom}\left(\mscr{E}_{U},\omega _{\mscr C_{U}}\left(\sum_i m_i\mscr{D}_i\right)\right)\right).$$ 
	There is a globally-defined homomorphism $a:\mscr{E}\to \underline{\Hom}(\mscr{E},\omega _{\mscr C}(\sum_i m_i\mscr{D}_i))$ which is isomorphism on $R_i$. The same argument in Lemma \ref{Equiv_a_b} shows that $a$ is an isomorphism. This completes the proof.
\end{proof}

\subsection{Connected Components of $\csspinm$}
We are going to prove the last statement of Theorem \ref{Main}: the number of connected components of geometric fibers of $\csspinm\to \Spec \mathbb Z[1/2]$ is
\begin{align*}
\gcd (2,m_1,m_2,\cdots,m_n).
\end{align*}
Since $\csspinm$ is smooth and proper over $\Spec \mathbb Z[1/2]$, the number of connected components of its geometric fibers is locally constant, so it is enough to prove the statement for the base scheme being $\Spec \mathbb C$. Using Remark \ref{Relation_Different_m}, we can assume that $m_i$ are either zero or one, so the statement is equivalent to:
\begin{proposition}\label{prop:the number of connected components of compactified spin moduli}
	The number of connected components of $\csspinm$ over $\Spec \mathbb C$ is $2$ if all $m_i$ are zero, and $1$ if some of $m_i$ is $1$.
\end{proposition}

\begin{remark}
	When all $m_i$ are zero, the number of connected components of $\csspinm$ is at least 2, since $\mscr{E}\mapsto \dim(H^0(\mscr{E})) \mod 2$ is locally-constant on $\sspin$ and both 1 and 0 show up as value of this function {\normalfont\cite{Mumford1971a}}. In fact it is not hard to show that $\mscr{E}\mapsto \dim(H^0(\mscr{E}))\mod 2$ is locally-constant on $\csspinm$, we give a sketch here: 
	
	\newl Assume that $(\mscr{C},\mbs{\mscr{D}},\mscr{E},b)$ is a punctured spin curve over algebraically-closed field $\mbb{K}$. Consider a divisor $\mathfrak a=\sum_{j=1}^N P_j$ consisting of $N$ distinct smooth points, such that each irreducible component has at least $\g$ points,  where $\g$ is the genus of the curve. Then an argument similar to the one in {\normalfont\cite{Mumford1971a}} shows that, $W_1\equiv \Gamma(\mscr{E}(\mathfrak{a}))$\footnote{$\Gamma$ is used to denote the space of global sections, which is the same as $\H^0$.} and $W_2\equiv \Gamma(\mscr{E}/\mscr{E}(-\mathfrak{a}))$ are subspaces of $V\equiv \Gamma(\mscr{E}(\mathfrak{a})/\mscr{E}(-\mathfrak{a}))$, and $W_1\cap W_2=\Gamma(\mscr{E})$. Moreover
	\begin{equation*}
	\dim W_1=\dim W_2=\frac{1}{2}\dim V=N.
	\end{equation*}
	Define the quadratic form $q$ on $V$ by
	\begin{equation*}
	q(\mathbf{s})=\sum_{j=1}^N\Res _{P_j}Q(s_i),
	\end{equation*}
	where $\mathbf{s}=(s_1,s_2,\cdots, s_N)$ is a collection of sections of $\mscr{E}(\mathfrak{a})_{P_j}/\mscr{E}(-\mathfrak{a})_{P_j}$, and $Q:\mscr{E}(\mathfrak{a})\to \omega_\mscr{C}(2\mathfrak{a})$ is the quadratic morphism induced by $b:\mscr{E}^{\otimes 2}\longrightarrow \omega_\mscr{C}$. The same argument in {\normalfont\cite{Mumford1971a}} shows that $W_1$ and $W_2$ are isotropic and of half the dimension of $V$, and hence the dimension modulo 2 of their intersection is locally constant, for the same reason in {\normalfont\cite{Mumford1971a}}.
	
	The spin structure $\mscr{E}$ with $h^0(\mscr{E})\mod 2=0$ {\normalfont (}resp. $h^0(\mscr{E})\mod 2=1${\normalfont)} is called even spin structure {\normalfont(}resp. odd spin structure{\normalfont)}. Recall that for smooth curve of genus $\g$, there are $2^{\sg-1}(2^{\sg}+1)$ even spin structures and $2^{\sg-1}(2^{\sg}-1)$ odd spin structures {\normalfont\cite{Mumford1971a}}.
\end{remark}

The strategy of the proof of Proposition \ref{prop:the number of connected components of compactified spin moduli} is the following: Utilizing the fact that $\cscurve$ (equivalently $\scurve$) is connected over $\Spec \mathbb C$ \cite{DeligneMumford1969a}, we investigate the monodromy action of $\pi_1(\scurve)$ (the mapping-class group of a genus-$\g$ Riemann surface with $\n$ punctures) on the fiber of $\sspin\to \scurve$, and prove by induction on $\g$ to reduce to $\g=1$ and $\g=0$, of which the proof is straightforward. We deal with the induction step first.

\begin{lem}
	Suppose that $\g>1$. If Proposition {\normalfont\ref{prop:the number of connected components of compactified spin moduli}} is true for genus from $1$ to $\g-1$, then Proposition {\normalfont\ref{prop:the number of connected components of compactified spin moduli}} is true for genus $\g$.
\end{lem}

\begin{proof}
	The idea is reducing to the boundary of $\cscurve$, and degenerate the curve into components with lower genus.
	
	\newl \textbf{Case $\mathbf m=\mathbf 0$:} We want to show that for a generic smooth curve of genus $\g$ with $\n$ punctures, its even spin structures (resp. odd spin structures) lie in the same connected component of $\overline{\mscr{S}}_{\sg,\sn}$. Consider two kinds of degenerations of this curve: 
	\begin{enumerate}
		\item $\mscr{C}=E\cup \mscr{C}'$ where $E$ is an elliptic curve, and $\mscr{C}'$ is smooth of genus $\g-1$;
		\item $\mscr{C}$ is irreducible with one node.
	\end{enumerate}
	
	For a degeneration of the first kind, it is easy to see that every spin structure $\mscr{E}$ comes from pushing forward spin structures on $E\amalg \mscr{C}'$, the disjoint union of the elliptic curve $E$ and the curve $\mscr{C}'$. The set of even spin structures is thus the disjoint union of two subsets of size $3\cdot 2^{\sg-2}(2^{\sg-1}+1)$ and $2^{\sg-2}(2^{\sg-1}-1)$, coming from pushing forward of even-even and odd-odd spin structures. We denote them by $A_1$ and $A_2$. Assuming that Proposition 7.8 holds for genus from $1$ to $\g-1$, we see that each one of these subsets lies in a connected component of $\overline{\mscr{S}}_{\sg,\sn}$. For a degeneration of the second kind, it is easy to see that every spin structure $\mscr{E}$ comes from either gluing square root of $\omega_{\widetilde {\mscr{C}}}(r+s)$ or pushing forward spin structure for $\omega_{\widetilde {\mscr{C}}}$, where $\widetilde {\mscr{C}}$ is the normalization of $\mscr{C}$, $r$ and $s$ are preimages of the node of $\mscr{C}$. As a result, the set of even spin structures is disjoint union of two subsets of size $2^{2\sg-2}$ and $2^{\sg-2}(2^{\sg-1}+1)$. We denote them by $B_1$ and $B_2$. Note that $B_2$ is the degeneration of twice as many spin structures of a generic smooth curve, since Proposition \ref{Deform_Spin_Moduli} shows that universal deformation of spin structure in $B_2$ is a double covering of universal deformation of $\mscr{C}$, and Proposition \ref{Aut_Field} show that the automorphism is trivial, so every spin structure from $B_2$ is the degeneration of a pair of spin structures. Using the assumption, we see that there are two subsets of even spin structures of a generic smooth curve of size $2^{2\sg-2}$ and $2^{\sg-1}(2^{\sg-1}+1)$, each lies in a connected component of $\overline{\mscr{S}}_{\sg,\sn}$. We denote them by $B_1$ and $B_2'$. Combining results from two kinds of degenerations, we see that $A_1$ must intersects both $B_1$ and $B_2'$, hence even spin structures of a generic smooth curve lie in a connected component of $\overline{\mscr{S}}_{\sg,\sn}$. For odd spin structures, it is essentially the same. The difference is that now $A_1$ has size $3\cdot 2^{\sg-2}(2^{\sg-1}-1)$, $A_2$ has size $2^{\sg-2}(2^{\sg-1}+1)$, corresponding to pushing forward of even-odd, and odd-even spin structures; $B_1$ has size $2^{2\sg-2}$, and $B_2$ has size $2^{\sg-2}(2^{\sg-1}-1)$, so $B_2'$ has size $2^{\sg-1}(2^{\sg-1}-1)$; the rest of the arguments is the same.
	
	\newl \textbf{Case $\mathbf m\neq \mathbf 0$:} Consider the following degeneration of a generic smooth punctured curve: 
	\begin{enumerate}
		\item[•] $\mscr{C}=E\cup \mscr{C}'$ where $E$ is an elliptic curve with 1 punctured point $\mscr{D}_i$ such that $m_i=1$, and $\mscr{C}'$ is smooth of genus $\g-1$ with $\n-1$ punctures.
	\end{enumerate}
	It is easy to see that every spin structure of $\mscr{C}$ is a line bundle and is the gluing of square roots of $\omega_{E}(\mscr{D}_i+r)$ and $\omega_{\mscr{C}'}(s+\sum_{j\neq i}m_j\mscr D_j)$, where $r$ and $s$ are preimage of the node under the projection $E\amalg \mscr{C}'\longrightarrow \mscr{C}$. Using the assumption, we see that all spin structures of $\mscr{C}$ lie in the same connected component of $\csspinm$.
\end{proof}

\begin{lem}
	Proposition {\normalfont\ref{prop:the number of connected components of compactified spin moduli}} is true for $\g=0$.
\end{lem}

\begin{proof}
	Since for smooth curves of genus zero, the spin structure is unique, so $\mscr{S}_{0,\sn}^{\mathbf{m}}\to \mscr{M}_{0,\sn}$ is an isomorphism and the connectedness is obvious.
\end{proof}

\begin{lem}
	Proposition {\normalfont\ref{prop:the number of connected components of compactified spin moduli}} is true for $\g=1$.
\end{lem}

\begin{proof}
	We fix an elliptic curve $E$ with $\n$ punctures, the set of spin structures is an $\mathbb F_2$-affine space modeled on $\H^1(E_{\text{\' et}},\mu_2)$, where $\H^1(E_{\text{\' et}},\mu_2)$ is the first \'etale cohomology group of $E_{\text{\' et}}$ with coefficients in the sheaf of the $2$\textsuperscript{nd} roots of unity. Then the monodromy action of $\pi_1(\mscr M _{1,\sn})=\text{MCG}_{1,\sn}$ on the set of spin structures quotient by transportation is a subgroup of $\SL _2(\mathbb F_2)=\Sp _2(\mathbb F_2)$. Note that the homomorphism $\text{MCG}_{1,\sn}\to \Sp _2(\mathbb F_2)$ factors through $\Sp _2(\mathbb Z)$, and this is well-known to be surjective \cite{MeeksPatrusky1978a}. Since $\Sp _2(\mathbb Z)\to \Sp _2(\mathbb F_2)$ is surjective, we see that $\text{MCG}_{1,\sn}\to \Sp _2(\mathbb F_2)$ is surjective.
	
	\newl \textbf{Case $\mathbf m= \mathbf 0$:} The space of spin structures is exactly $\H^1(E_{\text{\' et}},\mu_2)$, hence $\text{MCG}_{1,\sn}$ acts on by $\SL _2(\mathbb F_2)$ and it is easy to see that the only odd spin structure is the trivial one and even spin structures are permuted under $\SL _2(\mathbb F_2)$ action. As a result, even spin structures lie in the same connected component of $\mscr S _{1,\sn}$.
	
	\newl \textbf{Case $\mathbf m\neq \mathbf 0$:} We want to show that $\text{MCG}_{1,\sn}$ acts on the set of spin structures transitively. We first show that the action has no fixed point, in other word, for any spin structure $\mscr{E}$, there exists another spin structure $\mscr{E}'$ such that they are in the same connected component of $\mscr{S}_{1,\sn}^{\mathbf m}$.
	
	Consider the the degeneration $\mscr C=F_1\cup F_2$ where $F_i\cong\mathbb P^1$, they intersect at two nodes $r_1$ and $r_2$, there is one puncturing point $\mscr D_1\in F_1(\mathbb C)$ with $m_1=1$ and all other punctures are on the $F_2$. It is easy to see that there are two spin structures of $\mscr{C}$, $\mscr{E}_1$ and $\mscr{E}_2$, such that $\mscr{E}_i$ is the push-forward of the spin structure for $\omega_{\mscr{C}_i}(\sum_j m_j\mscr D_j)$ on curve $\mscr{C}_i$, where $\mscr{C}_i$ is the normalization of $\mscr{C}$ at the node $r_i$. Using Proposition \ref{Aut_Field}, we see that $\mscr{E}_i$ has trivial automorphism, hence according to Proposition \ref{Deform_Spin_Moduli}, $\mscr{E}_i$ is the degeneration of two spin structures of a generic smooth curve. As a result, every spin structure of a generic smooth curve share the connected component of $\mscr{S}_{1,\sn}^{\mathbf m}$ with another spin structure. Now, if the action of $\text{MCG}_{1,\sn}$ on the set of spin structures is not transitive, then the only possibility is that there are two pairs of spin structures, each pair is stabilized and permuted by $\text{MCG}_{1,\sn}$. However, this implies that the image of $\text{MCG}_{1,\sn}$ is of order at most $4$, which is absurd because a further quotient by transportation subgroup has order $6$. This concludes the proof of the case $\mathbf m\neq\mathbf 0$.
\end{proof}

\subsection{Gluing Punctures}\label{app: gluing marked points}
To simplify notation, we will use $\overline{\mscr{S}}_{\sg,\ns,\sra}$ to denote the moduli space of spin curves which is used to be called $\overline{\mscr{S}}_{\sg,\sn+|\mathbf m|}^{\mathbf m}$ , i.e. curves of genus $\g$ with $\n$ NS marked points and $\ra$ R marked points. Similarly for uncompactified version, $\mscr{S}_{\sg,\sn,\sra}:=\mscr{S}_{\sg,\sns+|\mathbf m|}^{\mathbf m}$. And we define the boundary of $\overline{\mscr{S}}_{\sg,\sns,\sra}$ to be $\overline{\mscr{S}}_{\sg,\sns,\sra}\setminus \mscr{S}_{\sg,\sns,\sra}$, with notation $\partial \overline{\mscr{S}}_{\sg,\sns,\sra}$.

Suppose that $\mscr C_1$ and $\mscr C_2$ are punctured spin curves over a base $S$, and that $\ns^1$ and $\ns^2$ are nonzero, then we can pick one NS divisor from each curve, say that the $i_1$'th NS divisor $\mscr D_{i_1}$ on $\mscr C_1$ and $i_2$'th NS divisor $\mscr D_{i_2}$ on $\mscr C_2$. The projection $\mscr D_{i_1}\to S$ is an isomorphism, same for $\mscr D_{i_2}$, so we can use these isomorphims to define a canonical identification between $\varphi:\mscr D_{i_1}\cong \mscr D_{i_2}$, and glue them to get a new curve $\mscr C=\mscr C_1\underset{\mscr D_{i_1}}{\cup_{\varphi}} \mscr C_2$, and the spin structure $\mathcal{E}$ on $\mscr C$ is the direct sum of pushforward of spin structures $\mscr{E}_1$ and $\mscr{E}_2$ via the gluing map $\mscr C_1\sqcup \mscr C_2\to C$. This construction is functorial hence giving rise to an operation on corresponding moduli spaces:
$$\mathfrak m^{\text{NS}} _{i_1,i_2}:\overline{\mscr{S}}_{\sg_1,\sns^1,\sra^1}\times \overline{\mscr{S}}_{\sg_2,\sns^2,\sra^2}\to \partial\overline{\mscr{S}}_{\sg_1+\sg_2,\sns^1+\sns^2-2,\sra^1+\sra^2}.$$
This can be defined for a single punctured spin curve with at least two NS puncturing divisors as well:
$$\mathfrak g^{\text{NS}}_{i,j}:\overline{\mscr{S}}_{\sg,\sns,\sra}\to \partial\overline{\mscr{S}}_{\sg+1,\sns-2,\sra}.$$
Note that these maps depends on the choice of $i_1$, $i_2$, $i$, and $j$, in an equivariant way. For example, if $\sigma_{i_1,i_1'}\in S_{\sns^1}$ switches $i_1$ and $i_1'$ (it acts $\overline{\mscr{S}}_{\sg_1,\sns^1,\sns^1}$ naturally), then we have $\mathfrak m _{i_1',i_2}=\mathfrak m^{\text{NS}} _{i_1,i_2}\circ \sigma_{i_1,i_1'}$.

\newl We move on to gluing R divisors: suppose that $\mscr C_1$ and $\mscr C_2$ are punctured spin curves over a base $S$, and that $m_1$ and $m_2$ are nonzero, we can pick one R puncture from each curve, $\iota_1:S\overset{\sim}{\rightarrow} \mscr D_{i_1}\subset C_1$ and $\iota_2:S\overset{\sim}{\rightarrow} \mscr D_{i_2}\subset C_2$, then $\iota_1^*\mscr E_1$ and $\iota_2^*\mscr E_2$ are two line bundles on $S$ with isomorphisms $b_1:(\iota_1^*\mscr E_1)^{\otimes 2}\cong \mcal O_S,b_2:(\iota_2^*\mscr E_2)^{\otimes 2}\cong \mcal O_S$. If we try to glue $\mscr C_1$ and $\mscr C_2$, surely we can still get a curve as in the NS case: $\mscr C=\mscr C_1\underset{\mscr D_{i_1}}{\cup_{\varphi}} \mscr C_2$. However, there is no canonical isomorphism between $\iota_1^*\mscr E_1$ and $\iota_2^*\mscr E_2$, every isomorphism can be multiplied by $-1$ to get another isomorphism. A way out is to adding a choice {\it by hand}, namely, consider the line bundle $\mcal L_{12}=\iota_1^*\mscr E_1^{\vee}\otimes\iota_2^*\mscr E_2$ with the isomorphism $b_{12}=b_1^{-1}\otimes b_2:\mcal L_{12}^{\otimes 2}\cong \mcal O_S$, then we have a morphism between schemes $f:\mbb V(\mcal L^{\vee}_{12})\to \mbb A^1_S$ by sending a section of $\mcal L_{12}$ to the image of its square under $b_{12}$. The scheme 
\begin{equation*}
G_{i_1,i_2}\equiv f^{-1}(\mathbf{1}),
\end{equation*}
parametrizes isomorphisms from $\iota_1^*\mscr E_1$ to $\iota_2^*\mscr E_2$ which commutes with $b_i$. Locally on $S$, after choosing trivializations of $\iota_1^*\mscr E_1,\iota_2^*\mscr E_2$, $f$ becomes the morphisms $\mbb A^1\to \mbb A^1$ by sending a section to its square, hence locally $G_{i_1,i_2}$ is disjoint union of two copies of the base scheme, and globally there is a $\mbb Z/2$-action on $G_{i_1,i_2}$ via multiplication by $-1$, thus $G_{i_1,i_2}$ is a $\mbb Z/2$-torsor.

\newl By base change to the scheme $G_{i_1,i_2}$, there is a universal isomorphism $\alpha:\iota_1^*\mscr E_1\cong\iota_2^*\mscr E_2$ which commutes with $b_i$, so we can define a gluing of $\mscr C_1\times _S G_{i_1,i_2}$ and $\mscr C_2\times _S G_{i_1,i_2}$ by taking the spinor sheaf $\mscr E$ on $\mscr C=\mscr C_1\underset{\mscr D_{i_1}}{\cup_{\varphi}} \mscr C_2$ to be the kernel of 
\begin{equation*}
p_{1*}\mscr E_1\oplus p_{2*}\mscr E_2\to \iota_1^*\mscr E_1\oplus\iota_2^*\mscr E_2\overset{(\alpha, \text{Id})}{\longrightarrow}\iota_2^*\mscr E_2,
\end{equation*}
where $p_i:\mscr C_i\to \mscr C$ is the canonical projection. It is easy to see that $\mscr E$ is locally-free in a neighborhood of the R node formed by gluing, and 
\begin{equation*}
\deg \mscr E=\deg \mscr E_1+\deg\mscr E_2-1.
\end{equation*}
Furthermore, it follows from the definition of $\mscr E$ and $\alpha$ commuting with $b_i$ that the composition 
\begin{alignat*}{2}
\mscr E^{\otimes 2}&\to (p_{1*}\mscr E_1\oplus p_{2*}\mscr E_2)^{\otimes 2}\to p_{1*}\mscr E_1^{\otimes 2}\oplus p_{2*}\mscr E_2^{\otimes 2}
\\
&\overset{b_1\oplus b_2}{\longrightarrow}p_{1*}\omega_{\mscr C_1/S}(\mscr D_{i_1}+E)\oplus p_{2*}\omega_{\mscr C_2/S}(\mscr D_{i_2}+H)\to \iota_2^*\omega_{\mscr C_2/S}(\mscr D_{i_2}+H),
\end{alignat*}
is zero, where $E$ and $H$ are other puncturing divisors on $\mscr C_1$ and $\mscr C_2$. The kernel of $p_{1*}\omega_{\mscr C_1/S}(\mscr D_{i_1}+E)\oplus p_{2*}\omega_{\mscr C_2/S}(\mscr D_{i_2}+H)\to \iota_2^*\omega_{\mscr C_2/S}(\mscr D_{i_2}+H)$ is $\omega_{\mscr C/S}(E+H)$. We thus have a homomorphism $b:\mscr E^{\otimes 2}\to \omega_{\mscr C/S}(E+H)$ which is an isomorphism on the smooth locus. This confirms that we obtain a spin structure on $\mscr C$, and the construction is also functorial. Consequently, $G_{i_1,i_2}$ forms a relative scheme $\mscr{G}_{i_1,i_2}$ over the product moduli space $\overline{\mscr{S}}_{\sg_1,\sns^1,\sra^1}\times \overline{\mscr{S}}_{\sg_2,\sns^2,\sra^2}$. Hence giving rise to an operation on corresponding moduli spaces:
\begin{center}
	\begin{tikzcd}
	\mscr{G}_{i_1,i_2}\arrow[r,"\mathfrak m^{\text{R}} _{i_1,i_2}"] \arrow[d,"\pi_{i_1,i_2}" swap] &  \partial\overline{\mscr{S}}_{\sg_1+\sg_2,\sns^1+\sns^2,\sra^1+\sra^2-2}\\
	\overline{\mscr{S}}_{\sg_1,\sns^1,\sra^1}\times \overline{\mscr{S}}_{\sg_2,\sns^2,\sra^2}.
	\end{tikzcd}
\end{center}
This can be defined for a single punctured spin curve with at least two NS marking divisors as well:
\begin{center}
	\begin{tikzcd}
	\mscr{G}_{i,j}\arrow[r,"\mathfrak g^{\text{R}} _{i,j}"] \arrow[d,"\pi_{i,j}" swap] &  \partial\overline{\mscr{S}}_{\sg+1,\sns,\sra-2}\\
	\overline{\mscr{S}}_{\sg,\sns,\sra}.
	\end{tikzcd}
\end{center}
Note that $\pi_{i_1,i_2}$ and $\pi_{i,j}$ are $\mbb Z/2$-torsors. Similar to the gluing of NS divisors case, these maps depends on the choice of $i_1,i_2$ and $i,j$ in an equivariant way.

\bibliography{References}
\bibliographystyle{JHEP}
	
\end{document}